\documentclass[letterpaper, 11pt]{article}
\usepackage[margin=1in]{geometry}
\usepackage{color, xcolor, graphicx, appendix}
\usepackage[bookmarks, colorlinks=true, plainpages = false, citecolor= darkblue, linkcolor = chicago-maroon, anchorcolor = red, urlcolor = chicago-maroon]{hyperref}
\definecolor{chicago-maroon}{RGB}{128,0,0}
\usepackage{url}
\usepackage{amsmath, amsfonts, amsthm, amssymb, bm, bbm, verbatim, dsfont, mathtools, mathrsfs}
\definecolor{darkblue}{rgb}{0.0, 0.0, 0.55}
\usepackage{etoolbox}
\usepackage{almendra}
\usepackage{authblk}
\usepackage{array}
\usepackage{multirow}
\usepackage{multicol}
\providecommand{\keywords}[1]{\textbf{Keywords: } #1}
\usepackage{float}
\usepackage{pgf,tikz}
\usetikzlibrary{arrows,shapes.arrows,shapes.geometric,shapes.multipart,fit,automata,decorations.pathmorphing,positioning,shapes.swigs}
\tikzset{
	>=stealth',
	true/.style={
		rectangle,
		draw=black, very thick,
		text width=6.5em,
		minimum height=2em,
		text centered,
		fill=gray, opacity = 0.5},
	punkt/.style={
		rectangle,
		rounded corners,
		draw=black, very thick,
		text width=6.5em,
		minimum height=2em,
		text centered},
	est/.style={
		circle,
		draw=black, very thick,
		text centered},
	shade/.style={
		circle,
		draw=black, very thick, fill=gray!50,
		text centered},
	weight/.style={
		circle,
		draw=black, very thick,
		text width=6.5em,
		minimum height=2em,
		text centered},
	pil/.style={
		->,
		thick,
		shorten <=2pt,
		shorten >=2pt,},
	double/.style={
		<->,
		thick,
		shorten <=2pt,
		shorten >=2pt,},
	dash/.style={
		dashed,
		thick,
		shorten <=2pt,
		shorten >=2pt,},
	dashdouble/.style={
		<->,
		dashed,
		thick,
		shorten <=2pt,
		shorten >=2pt,}
}

\usepackage{tikz-cd}
\makeatletter 
\newcolumntype{C}[1]{>{\centering\arraybackslash}p{#1}}

\usepackage{epstopdf}
\usepackage{fullpage}
\usepackage{algorithm}
\usepackage{algorithmicx}
\usepackage{subcaption}

\usepackage[nameinlink]{cleveref}
\usepackage[round, compress]{natbib}

\usepackage[utf8]{inputenc} 
\usepackage[T1]{fontenc}    
\usepackage{url}            
\usepackage{booktabs}     
\usepackage{nicefrac}
\usepackage{xcolor}
\usepackage{pdfpages}
\usepackage{scalerel}
\usepackage{dashrule}
\usepackage{lmodern}
\usepackage{fancyhdr}
\usepackage{tabu}
\usepackage{enumitem}

\def\pa{\mathrm{pa}}
\def\de{\mathrm{de}}
\def\ch{\mathrm{ch}}

\def\Pr{\mathsf{Pr}}

\def\Var{\mathsf{Var}}

\renewcommand{\(}{\left(}

\renewcommand{\[}{\left[}

\renewcommand{\tilde}{\widetilde}

\usepackage{xr,xspace}
\usepackage{todonotes}

\usepackage{caption,subcaption,soul}
\usepackage{algpseudocode}
\makeatletter
\usepackage{romanbar}

\theoremstyle{plain}
\newtheorem{theorem}{Theorem}
\newtheorem{lemma}{Lemma}
\newtheorem{proposition}{Proposition}
\newtheorem{corollary}{Corollary}

\theoremstyle{definition}

\newtheorem{example}{Example}

\newtheorem{problem}{Problem}

\newtheorem{remark}{Remark}

\newtheorem*{remark*}{Remark}

\newcommand{\argmin}{\arg\min}
\usepackage{algorithm}
\usepackage{algpseudocode}

\usepackage{xspace, prettyref}

\newcommand{\diff}{{\mathrm d}}

\newcommand\indep{\protect\mathpalette{\protect\independenT}{\perp}}
\def\independenT#1#2{\mathrel{\rlap{$#1#2$}\mkern2mu{#1#2}}}
 
\def\Var{\mathrm{Var}}

\newrefformat{eq}{(\ref{#1})}
\newrefformat{chap}{Chapter~\ref{#1}}
\newrefformat{sec}{Section~\ref{#1}}
\newrefformat{alg}{Algorithm~\ref{#1}}
\newrefformat{fig}{Fig.~\ref{#1}}
\newrefformat{tab}{Table~\ref{#1}}
\newrefformat{rmk}{Remark~\ref{#1}}
\newrefformat{clm}{Claim~\ref{#1}}
\newrefformat{def}{Definition~\ref{#1}}
\newrefformat{cor}{Corollary~\ref{#1}}
\newrefformat{lmm}{Lemma~\ref{#1}}
\newrefformat{prop}{Proposition~\ref{#1}}
\newrefformat{prob}{Problem~\ref{#1}}
\newrefformat{app}{Appendix~\ref{#1}}
\newrefformat{hyp}{Hypothesis~\ref{#1}}
\newrefformat{thm}{Theorem~\ref{#1}}

\newcommand{\sfM}{{\mathsf{M}}}

\newcommand{\calG}{{\mathcal{G}}}
\newcommand{\calH}{{\mathcal{H}}}
\newcommand{\calI}{{\mathcal{I}}}

\newcommand{\calP}{{\mathcal{P}}}

\newcommand{\calX}{{\mathcal{X}}}
\newcommand{\calY}{{\mathcal{Y}}}
\newcommand{\calZ}{{\mathcal{Z}}}

\newcommand{\Tr}{\mathsf{Tr}}

\renewcommand{\tilde}{\widetilde}

\newcommand{\bbE}{{\mathbb{E}}}

\newcommand{\bbP}{{\mathbb{P}}}

\newcommand{\bbR}{{\mathbb{R}}}

\newcommand{\dis}{\mathrm{dis}}
\newcommand{\pre}{\mathrm{pre}}

\newcommand{\mb}{\mathrm{mb}}

\usepackage{tcolorbox}
\newcommand{\opt}{\mathrm{opt}}

\makeatletter
\def\ubar#1{\underline{\sbox\tw@{$#1$}\dp\tw@\z@\box\tw@}}
\makeatother

\def\leftarrowCirc{\hbox{$\leftarrow$}\kern-1.5pt\hbox{$\circ$}}
\def\Circrightarrow{\hbox{$\circ$}\kern-1.5pt\hbox{$\rightarrow$}}
\def\Circleftarrow{\hbox{$\circ$}\kern-1.5pt\hbox{$\leftarrow$}}
\def\rightarrowCirc{\hbox{$\rightarrow$}\kern-1.5pt\hbox{$\circ$}}

\def\I{\mathrm{I}}

\def\M{\mathrm{M}}

\def\p{\mathrm{p}}

\def\eps{\epsilon}

\usepackage{namedtensor}

\usepackage{titletoc}

\usepackage{bibunits}
\makeatletter
\newcounter{mybibunit}
\AtBeginEnvironment{bibunit}{\stepcounter{mybibunit}}
\renewcommand{\hyper@natlinkstart}[1]{%
  \Hy@backout{#1}%
  \hyper@linkstart{cite}{cite.\themybibunit @#1}%
  \def\hyper@nat@current{#1}%
}
\renewcommand{\hyper@natlinkbreak}[2]{%
  \hyper@linkend#1\hyper@linkstart{cite}{cite.\themybibunit @#2}%
}
\renewcommand{\hyper@natanchorstart}[1]{%
  \Hy@raisedlink{\hyper@anchorstart{cite.\themybibunit @#1}}%
}
\makeatother

\usepackage{bbold}
\usepackage[mathcal]{euscript}
\usepackage{setspace}
\usepackage{cancel}

\def\nuis{\mathrm{nuis}}
\def\sp{\operatorname{span}}
\def\Tr{\operatorname{Tr}}
\def\np{\mathrm{np}}
\def\eff{\mathrm{eff}}
\def\opt{\mathrm{opt}}
\def\I{\mathbb{I}}
\def\M{\mathtt{M}}
\def\interior{\operatorname{int}}
\def\sub{\operatorname{sub}}

\def\pre{\operatorname{pre}}

\def\calI{\mathcal{I}}
\def\eps{\epsilon}

\usepackage{makecell}

\def\mytitle{Toward a Semiparametric Efficiency Theory under Equality Constraints in Nested Markov Models}

\begin{document}

\allowdisplaybreaks

\title{\mytitle}

\author[1]{Razieh Nabi\thanks{Email: \href{razieh.nabi@emory.edu}{razieh.nabi@emory.edu}}}
\author[1]{Anna Guo\thanks{Email: \href{anna.guo@emory.edu}{anna.guo@emory.edu}}}
\author[2]{Lin Liu\thanks{Email: \href{linliu@sjtu.edu.cn}{linliu@sjtu.edu.cn}}}

\affil[1]{Department of Biostatistics and Bioinformatics, Emory University, Atlanta, GA, USA}
\affil[2]{Institute of Natural Sciences, MOE--LSC, School of Mathematical Sciences, CMA--Shanghai, SJTU--Yale Joint Center for Biostatistics and Data Science, Shanghai Jiao Tong University, Shanghai, China}

\date{\today}
    
\maketitle


\begin{abstract}
Probabilistic models of Directed Acyclic Graphs (DAGs) with latent variables, or latent-variable DAGs, impose equality constraints on the observed data distribution beyond ordinary conditional independencies. These so-called Verma constraints arise naturally in nested Markov models associated with the latent projection of latent-variable DAGs, known as Acyclic Directed Mixed Graphs. While nested Markov models have been extensively studied from the perspectives of graphical representation and causal identification, their implications for semiparametric efficiency theory remain comparatively less well understood. In this paper, we develop results toward establishing a semiparametric framework for statistical models defined by Verma constraints. Our starting point is the observation that nested Markov constraints admit representations as weighted conditional moment restrictions under post-fixing distributions induced by graphical fixing operations. We show that these fixing operations induce weighted orthogonality relations in $L^2 (\bbP)$, thereby converting Verma constraints into explicit tangent-space restrictions.
Building on this representation, we characterize tangent-space orthocomplement for statistical models defined only by a single nested Markov constraint through residualized weighted moment functions. This geometric formulation yields Hilbert-space characterizations of semiparametric efficient influence functions and efficiency bounds via orthogonal projection and equivalent minimum-variance formulations. We further discuss how the resulting geometric perspective naturally extends to models involving multiple nested Markov constraints, for which we characterize a subspace of the orthocomplement as sums of the corresponding weighted orthogonality relations, but still leaving the actual tangent space an open problem.  More broadly, our results suggest that conditional moment constraints provide a natural bridge between nested graphical structure and semiparametric Hilbert-space geometry and pave the way toward establishing a complete semiparametric efficiency theory for general nested Markov models. We illustrate the framework through several canonical examples of latent-variable DAGs. 
\end{abstract}


\newpage

\onehalfspacing
\section{Introduction}
\label{sec:intro}

Probabilistic directed acyclic graph (DAG) models with latent variables induce equality constraints on the observed data distribution. Some of these constraints correspond to ordinary conditional independencies implied by graphical factorization, while others arise from more subtle restrictions that are not expressible as ordinary conditional independence relations. These latter constraints, known as Verma constraints, arise naturally in acyclic directed mixed graphs (ADMGs), which summarize latent-variable DAGs through latent projection, and constitute a defining feature of nested Markov models \citep{verma1990equivalence, robins1999testing, richardson2003markov, shpitser2022multivariate, richardson2023nested}. Over the last two decades, nested Markov models have emerged as a unifying framework for studying latent-variable DAG models. They underpin modern causal identification theory in the presence of hidden confounding \citep{tian2002general, shpitser2008complete} and provide a systematic characterization of equality constraints, including Verma constraints, implied by latent-variable DAGs \citep{evans2018margins, richardson2023nested, zhao2025statistical}. 

While nested Markov models have been extensively studied from the perspectives of graphical representation, equality constraints, and nonparametric causal identification, their implications for semiparametric efficiency theory remain comparatively less developed. In semiparametric theory, structural assumptions restrict the space of allowable data-generating distributions and thereby determine the tangent space, efficient influence function, and semiparametric efficiency bound for a target parameter \citep{van2003unified, tsiatis2006semiparametric}. For DAGs without latent variables, these restrictions typically arise through conditional independence relations and admit relatively transparent tangent-space characterizations \citep{rotnitzky2020efficient, guo2023variable}. More recently, \citet{phung2026characterization} studied tangent-space orthocomplements for semiparametric models defined by collections of ordinary marginal and conditional independence restrictions. Their framework, however, explicitly excludes the generalized independence constraints that distinguish nested Markov models from ordinary Markov models. 

Several recent works have demonstrated that graphical structures can yield meaningful efficiency gains. In causal DAG models, graphical criteria have been developed to identify optimal adjustment strategies \citep{witte2020efficient, rotnitzky2020efficient, henckel2022graphical}. In latent-variable settings, semiparametric efficiency theory has been developed for particular classes of ADMGs and identification functionals whose relevant observed-data restrictions reduce to ordinary conditional independence constraints \citep{bhattacharya2022semiparametric, guo2023flexible, guo2024average, guo2025causal}. Additional efficiency gains have been established under specialized structural restrictions with no-direct-effect assumptions that induce nested Markov restrictions \citep{robins1999testing, caniglia2019emulating, liu2021efficient}. More recently, \citet{guo2023flexible} and \citet{guo2025causal} studied two prominent ADMG examples involving mean-scale Verma constraints, deriving tangent-space characterizations and efficiency bounds for the corresponding average causal effects. Despite these advances, existing results are largely tailored to specific graphical structures and estimands. A general semiparametric framework for nested Markov models that generalize Verma constraints is lacking.

There are several obstacles toward establishing a complete semiparametric efficiency theory for nested Markov models.
 The first is that Verma constraints encoded in ADMG models are not naturally expressed as ordinary conditional mean restrictions under the observed data law, hindering straightforward parameterizations as in DAG models. Instead, the relevant orthogonality relations emerge only after applying graphical fixing operations that induce reweighted post-fixing distributions. From a geometric perspective, these fixing operations generate weighted orthogonality relations in $L^2 (\bbP)$ through inverse-probability weighting operators (or fixing operators), so that the tangent-space restrictions associated with Verma constraints take the form of weighted residual orthogonality conditions rather than ordinary conditional independence restrictions.

The importance of such fixing-based representations has already been recognized in other inferential settings. For example, \citet{bhattacharya2022testability} exploited a Verma constraint arising in a latent-variable graphical model to develop parametric specification tests. 
Building on the same graphical structure, \citet{guo2023flexible} developed flexible semiparametric procedures, including doubly robust tests, for assessing assumptions encoded by generalized independence constraints. 
\citet{thams2023statistical} showed how dormant independence constraints can be tested by reweighting the observed distribution so that the generalized independence becomes an ordinary independence restriction in a shifted distribution. \citet{guo2025rank} and \citet{dhawan2026debiased} subsequently developed subsampling-based tests to testing Verma constraints under such reweighting schemes. 
These developments highlight the growing role of Verma constraints in statistical inference beyond causal identification. 

At the same time, the fixing-based representation of a Verma constraint suggests a natural connection with the literature on semiparametric models defined by conditional moment restrictions \citep{ai2003efficient, ai2012semiparametric, chen2026local}. Bridging this literature with the fixing operator in ADMGs provides a route from nested Markov restrictions to the explicit characterization of tangent-space geometry.

A second obstacle is that nested Markov models, even with only a single Verma constraint, do not have natural parameterizations except in a few special subclasses \citep{shpitser2012parameter, shpitser2018acyclic, evans2018margins}. The problem becomes even more challenging when multiple constraints are present. This difficulty often precludes standard textbook strategies for constructing parametric submodels \citep{van2003unified} to prove that a candidate for a tangent space of a nested Markov model is indeed \emph{the tangent space}, rather than its superset. The standard strategies generally conjecture linear perturbations of the observed data law or each conditional component of its Bayes factorization as a candidate submodel. In more complicated scenarios, it is sometimes indispensable to postulate nonlinear perturbations, such as polynomials of higher degrees \citep{dong2026marginal}. 

It is worth clarifying that characterizing a nontrivial superset of the tangent space is nonetheless useful for statistical practice: Although projecting an influence function onto this superset may not lead to the unique efficient influence function, the resulting projection still yields an influence function with no greater variance, by the Pythagorean theorem. To our knowledge, even a characterization of such a superset is lacking in the literature, except for a few special cases \citep{liu2021efficient, phung2026characterization}.

\subsection*{Contributions and organization}

In this paper, we develop a semiparametric framework for statistical models defined by Verma constraints. Our starting point is the observation that nested Markov constraints admit representations as weighted conditional moment restrictions under post-fixing distributions induced by graphical fixing operations. This representation allows tools from semiparametric efficiency theory for moment-restricted models to be brought into the study of nested Markov models. We show that fixing operations induce weighted orthogonality relations in $L^2(\bbP)$ and thereby convert Verma constraints into explicit restrictions on the tangent space. 

Building on this formulation, we characterize tangent-space orthocomplements for statistical models defined by a single Verma constraint. Our analysis proceeds primarily through orthocomplement characterizations, since the weighted score restrictions induced by graphical fixing operations admit particularly tractable representations in the dual Hilbert-space geometry. In particular, pathwise differentiation of the weighted moment restrictions generates explicit orthocomplement elements through residualized weighted moment functions associated with the fixing structure of the underlying ADMG. 

These orthocomplement characterizations lead directly to derivations of efficient influence functions, and in turn, semiparametric efficiency bounds. Given an arbitrary influence function under the nonparametric model, the efficient influence function under the Verma-constrained model is obtained through orthogonal projection onto the restricted tangent space, or equivalently through a minimum-variance characterization over admissible orthocomplement directions. This yields a unified Hilbert-space perspective that connects graphical fixing operations, Verma constraints, conditional moment restrictions, tangent-space geometry, and semiparametric efficiency. 

Tangent spaces and the resulting semiparametric efficiency bound are conditional on the correctness of the underlying semiparametric model and, in our case, the correctness of the ADMG. The recent literature on statistics and economics discusses whether less efficient estimators that are robust against potential model misspecification should be used instead \citep{andrews2025purpose}. However, \citet{adusumilli2026you} demonstrated that semiparametrically efficient estimators may remain preferable for follow-up decision making under \emph{local} misspecification of the underlying model, providing further motivation for studying efficiency within structured graphical models.

We illustrate the framework through several canonical examples, including the Napkin graph \citep{pearl2009causality, pearl2018book} and an extended front-door model \citep{pearl2009causality, bhattacharya2022testability}, the latter of which is closely related to no-direct-effect restrictions and the g-null hypothesis studied by \citet{robins1997estimation, robins1999testing}. In each case, we derive explicit orthocomplement representations and corresponding efficient influence functions, highlighting how nested Markov structure induces additional efficiency gains relative to models that ignore the associated Verma constraints. We further illustrate the proposed estimators in real-data applications based on the Framingham Heart Study and the Finnish Life Course study; these analyses are reported in the supplementary materials. 
Although our primary focus is on models defined by a single nested Markov constraint, we also discuss extensions to models involving multiple nested Markov constraints, where tangent spaces arise through intersections of constraint-specific score spaces and orthocomplements combine through corresponding weighted orthogonality relations.

The main contributions of this paper are as follows. 
(i) We show that Verma constraints admit representations as weighted conditional moment restrictions under post-fixing distributions induced by graphical fixing operations. We use this to characterize the tangent-space orthocomplements induced by a single Verma constraint through its residualized weighted moment constraints. Importantly, we demonstrate the existence of parametric submodels in this semiparametric model, whose scores are only restricted by being orthogonal to the conjectured tangent-space orthocomplements via differentiating the weighted moment constraint as in the calculus of influence functions \citep{hines2022demystifying, kennedy2024semiparametric}. Constructing such parametric submodels is not necessarily a trivial task, and in the case of Verma constraint, often involves constructing nonlinear paths locally around the true data generating law \citep{dong2026marginal}.  
(ii) We develop projection and minimum-variance characterizations of semiparametric efficient influence functions under Verma-constrained models. 
(iii) We introduce indexed functionals whose invariance in the relevant variable index captures mean-scale implications of Verma constraints, and discuss how these lower-dimensional restrictions can be used to construct valid variance-improved influence functions without requiring projection onto the full Verma orthocomplement. 
(iv) We discuss how the single-Verma-constraint framework developed here serves as a building block toward semiparametric inference in more general nested Markov models.
(v) We illustrate the framework through canonical examples, including the Napkin graph and extended front-door/no-direct-effect models, and evaluate the resulting estimators through simulation studies and real-data applications, highlighting the efficiency gains that can be obtained by exploiting Verma constraints.

The remainder of the paper is organized as follows. Section~\ref{sec:review} reviews background material on nested Markov models, fixing operations, and semiparametric theory. Section~\ref{subsec:constraint} formulates Verma constraints as weighted conditional moment restrictions under post-fixing laws. Section~\ref{subsec:tangent_space} develops the corresponding tangent-space and orthocomplement geometry under a single nested Markov constraint. 
Section~\ref{subsec:semipar_efficiency} develops projection and minimum-variance characterizations of semiparametric efficient influence functions under Verma-constrained models.
Section~\ref{subsec:mean_scale_verma} studies indexed functional invariance induced by Verma restrictions and its implications for semiparametric efficiency. 
Section~\ref{sec:beyond_single_constraint} discusses extensions beyond the single-constraint setting. 
Section~\ref{sec:applications} illustrates the framework through several canonical examples.  
Section~\ref{sec:sims} presents simulation studies, and 
Section~\ref{sec:conclusion} concludes. 
The real-data applications and all proofs are deferred to the supplementary materials.  

\section{Review of Essential Concepts}
\label{sec:review}

\subsection{A review of nested Markov models and acyclic directed mixed graphs}
\label{sec:nnm}

We let $\calG \equiv \calG(V;W)$ denote a conditional ADMG, where $V$ and $W$ denote the sets of random and fixed vertices, respectively. We write $O$ for the collection of observed variables $V$ and $W$. Following standard graphical-model notation, we do not distinguish between a vertex and its associated random variable. Thus, capital letters may refer either to vertices (or collections of vertices) or to the corresponding random variables, depending on context. For example, $X \subseteq V$ denotes a subset of vertices, whereas $X \indep Y$ denotes independence between the corresponding random variables. Lowercase letters denote realizations of random variables. 

The notation $\phi_R(\cdot)$, for $R \subseteq V$, denotes the fixing operator, whose argument may be either a kernel $q_{V \mid W}$ or an ADMG $\calG(V;W)$. We also adopt the standard familial terminology of graphical models, including parents $\pa_{\calG}(\cdot)$, children $\ch_{\calG}(\cdot)$, and descendants $\de_{\calG}(\cdot)$. For any vertex $X \in V$, we let $\dis_{\calG}(X)$ denote the district of $X$, namely the set of vertices connected to $X$ by bidirected paths.

We now elaborate on the definition of fixing operation and what type of vertices can be classified as being \emph{fixable}, the two important concepts related to the theory of ADMG and nested Markov models  \citep{richardson2023nested}. A vertex $R \in V$ is said to be fixable if $\dis_{\calG} (R) \cap \left\{\de_{\calG}(R) \setminus \{R\}\right\} = \emptyset$. Given a kernel $q_{V \mid W}$ associated with ADMG $\calG (V; W)$, the corresponding
probabilistic operation of fixing $R \in V$, denoted by $\phi_{R} (q_{V \mid W}; \calG)$, is defined as follows: 
\begin{align}
\phi_{R} (q_{V \mid W}; \calG) \coloneqq \frac{q_{V \mid W} (V \mid W)}{q_{V \mid W} (R \mid \mb_{\calG} (R), W)} \equiv q_{V \setminus R \mid W \cup R}(V \setminus R \mid W \cup R), \label{eq:fixing_operation}
\end{align}
where $\mb_{\calG} (R)$ denotes the Markov blanket of $R$ and is defined as the union of $\dis_{\calG} (R)$ and the parents of $\dis_{\calG} (R)$, excluding $R$ itself. Fixing a vertex $R$ in an ADMG $\calG$, denoted by $\phi_{R} (\calG)$, means removing all edges pointing to $R$, including both directed and bidirected ones, leading to a new ADMG $\calG (V \setminus \{R\}; W \cup \{R\})$. According to \citet{richardson2023nested}, the kernel $q_{V \setminus \{R\} \mid W \cup \{R\}}$ is also nested-Markov factorized with respect to the post-fixing ADMG $\calG (V \setminus \{R\}; W \cup \{R\})$. When $W = \emptyset$, we use $p_{V} \equiv p$ to denote the joint probability density function of all random variables $V$. Intuitively speaking, fixing operation $\phi_{R} (\cdot)$ corresponds to a change of distribution from the original to the one that intervenes vertex $R$, via inverse probability weighting; and in the post-fixing ADMG, the vertex $R$ becomes fully exogenous. The subscript in the fixing operation can be a set of vertices, say $R \subseteq V$. For more details on nested Markov models constraints, see Appendix~\ref{app:NMM}. 

\subsection{A review of semiparametric theory}
\label{sec:semi}

Throughout, $\bbP$ denotes the observed-data distribution, and $\bbE[\cdot]$ denotes expectation with respect to $\bbP$ unless otherwise indicated. Expectations taken under other probability laws, such as post-fixing laws or parametric submodels, are indicated by appropriate subscripts.

Given a smooth target functional $\psi: \calP \rightarrow \bbR$ that is differentiable in the sense of \citet{van1991differentiable}, it is well known that \emph{regular and asymptotically linear} (RAL) estimators are a class of estimators for which standard nominal $(1 - \alpha$)-Wald-type confidence intervals can be easily constructed and have uniformly correct asymptotic coverage probability. Any RAL estimator $\tilde{\psi}_{n}$ of $\psi \equiv \psi (\bbP)$, based on $n$ i.i.d. observed data $\{O_{i}\}_{i = 1}^{n} \overset{\rm i.i.d.}{\sim} \bbP$ drawn from the data generating distribution $\bbP \in \calP$, has the following asymptotic linear representation: 
\begin{align*}
\sqrt{n} (\tilde{\psi}_{n} - \psi) = \frac{1}{\sqrt{n}} \sum_{i = 1}^{n} \tilde{\varphi} (O_{i}) + o_{\bbP} (1),
\end{align*}
where $\tilde{\varphi} (\cdot) \equiv \tilde{\varphi}_{\bbP} (\cdot)$ is the influence function of the RAL estimator $\tilde{\psi}_{n}$. Therefore, we have a bijection between the space of all RAL estimators and the space of all influence functions of $\psi$. In the sequel, to reduce clutter, we will drop the qualification ``smooth'' when we refer to the target functional.

Modern semiparametric theory addresses two main challenges in the estimation and inference of $\psi (\cdot)$ \citep{bickel1998efficient, van2003unified, tsiatis2006semiparametric}. Given a target functional $\psi (\cdot)$, the theory tells us that the space of all influence functions $\varphi (\cdot) \equiv \varphi_{\bbP} (\cdot)$ of $\psi (\bbP)$, locally at any $\bbP$ in the interior of $\calP$, denoted by $\interior (\calP)$, corresponds to the orthocomplement to the nuisance tangent space at $\bbP$, denoted by $\Lambda_{\bbP, \nuis}^{\perp}$. By definition, the nuisance tangent space $\Lambda_{\bbP, \nuis}$ consists of the span of all scores $s (\cdot)$ associated with (one-dimensional) parametric submodels $\bbP_{t} \in \calP$ such that $t \mapsto \psi (\bbP_{t})$ is locally stationary at $t = 0$: that is,
\begin{align*}
\left. \frac{\diff}{\diff t} \right\vert_{t = 0} \psi (\bbP_{t}) = 0.
\end{align*}
Hence, if a statistician has a target functional in mind, she can find all RAL estimators of $\psi$ at $\bbP$ by characterizing $\Lambda_{\bbP, \nuis}$; see \citet{van2003unified} for various examples. Different RAL estimators have different asymptotic variances. The smallest possible asymptotic variance of RAL estimators is called the \emph{semiparametric variance bound} (SVB). The RAL estimator that achieves the SVB is the unique semiparametric efficient estimator whose influence function is the efficient influence function, denoted by $\varphi^{\eff} (\cdot) \equiv \varphi^{\eff}_{\bbP} (\cdot)$ and can be found by minimizing the variance of all influence functions $\varphi (\cdot) \in \Lambda_{\bbP, \nuis}^{\perp}$.

However, in certain applications, one may not have a particular functional of interest to begin with. It will be difficult to apply the approach of characterizing the nuisance tangent space and its orthocomplement. An alternative and arguably more general approach is to directly characterize the tangent space $\Lambda_{\bbP}$ of the statistical model $\calP$ locally at $\bbP$. Once the statistician has settled on a target functional of interest $\psi$, she can first find an influence function $\varphi (\cdot)$ of $\psi$, and the efficient influence function $\varphi^{\eff} (\cdot) \equiv \Pi (\varphi \mid \Lambda_{\bbP}) (\cdot)$ is simply the projection of $\varphi$ onto the tangent space $\Lambda_{\bbP}$. Equivalently, the set of all influence functions of $\psi(\bbP)$, denoted by $\calI \{\psi(\bbP)\}$, equals the following linear variety
\begin{equation}
\label{space of IFs}
\calI \{\psi(\bbP)\} = \varphi + \Lambda_{\bbP}^{\perp}.
\end{equation}
Since our goal is to develop a semiparametric theory of nested Markov models associated with ADMGs without restricting to particular target functionals, we adopt this latter approach by characterizing the tangent space $\Lambda_{\bbP}$ at $\bbP$ in the interior of $\calP_{\calG}$ and its orthocomplement. When a model $\calP$ is characterized by certain (conditional) moment constraints, a common calculus to deriving the tangent space orthocomplement is by differentiating the constraints along one-dimensional parametric submodels indexed by the scalar perturbation parameter $t$ locally at $t = 0$. The output of going through this tangent space calculus is a subspace $\calH$ of the actual tangent space orthocomplement. To show the reverse direction, i.e. this subspace $\calH$ is the full space, one needs to establish the existence of a parametric submodel in $\calP$ with the corresponding score with the sole restriction of being orthogonal to $\calH$. This latter step can be challenging when the constraints are complicated.

\section{Information Geometry of Nested Markov Models}
\label{sec:geometry}

\subsection{A general formulation of a nested Markov (equality) constraint} 
\label{subsec:constraint}

As discussed in Section~\ref{sec:semi}, establishing a general semiparametric efficiency theory for causal graphical models with latent variables is to characterize the tangent space $\Lambda_{\bbP}$ of a given law $\bbP$ lying in the interior of the nested Markov model $\calP_{\calG}$ associated with an ADMG $\calG$. A crucial observation is that the equality constraints encoded by $\calP_{\calG}$ admit representations as weighted conditional moment restrictions in the sense of \citet{ai2003efficient, ai2012semiparametric}. We therefore begin by expressing a single nested Markov (a.k.a. Verma) constraint as a weighted conditional moment restriction under the observed law. This representation will serve as the basis for the tangent-space characterizations and efficiency calculations developed in later sections.  

Let $\calP_{\calG}$ denote a statistical model nested Markov factorized according to an ADMG $\calG(V,E)$. Let $X,Y,Z \subseteq V$ denote pairwise disjoint subsets of observed variables, and let $R \subseteq V$ denote a (possibly empty) valid fixing set of variables. Consider the nested Markov constraint
\begin{equation}\label{eq:general_verma_constraint}
X \indep Y \mid Z \;\;\; [\phi_R (p)].
\end{equation}%

A distinguishing feature of a Verma constraint is that the conditional independence relation does not generally hold under the observed law itself. Rather, it holds only after applying one or more graphical fixing operations. Thus, a Verma constraint may be viewed as an ordinary conditional independence restriction in a post-fixing kernel, transported back to the observed-data distribution through the fixing operator. Ordinary conditional independence restrictions correspond to the special case in which no fixing operation is required.

Although $\phi_R(p)$ is generally a kernel rather than a probability law, throughout we write $\bbP_{\phi_R}$ for the probability law obtained by combining the post-fixing kernel $\phi_R(p)$ with a specified reference distribution $\tilde p(R)$ over the fixed variables. Expectations under this law are denoted by $\bbE_{\phi_R}[\cdot]$. The following proposition shows that \eqref{eq:general_verma_constraint} admits a representation as a weighted conditional moment restriction under the observed law. 

\begin{proposition}[Weighted moment representation of a Verma constraint]
\label{prop:weighted_moment_representation}
Suppose that $X \indep Y \mid Z \quad [\phi_R (p)]$. Then, for every pair of square-integrable functions
$f \in L^{2} (\bbP_{X, Z})$ and $g \in L^{2} (\bbP_{Y, Z})$,
\begin{equation}\label{eq:general_equality_constraint}
\bbE_{\phi_R} \left[ \left\{ f (X, Z) - \bbE_{\phi_R} [f (X, Z) \mid Z] \right\} g(Y,Z) \right] = 0,
\end{equation}
or equivalently
\begin{equation}\label{eq:general_equality_constraint2}
\bbE_{\phi_R} \left[ f (X, Z) \left\{ g (Y,Z) - \bbE_{\phi_R}[g (Y, Z) \mid Z] \right\} \right] = 0. 
\end{equation}
Furthermore, expectations taken under the post-fixing law $\bbP_{\phi_R}$ admit the following re-weighted representation under the observed law $\bbP$. Writing $\M := (R, \mb_\calG(R))$, we have 
\begin{equation}\label{eq:kernel_expectations}
\bbE_{\phi_R} [A] = \frac{\bbE [\omega_R(\M) \, A]}{\bbE [\omega_R(\M)]}, \qquad
\bbE_{\phi_R} [A \mid B]
= \frac{\bbE [\omega_R(\M) \, A \mid B]}{\bbE [\omega_R(\M) \mid B]}, \qquad 
\omega_R(\M) \coloneqq \frac{\tilde{p} (R)}{p (R \mid \mb_{\calG} (R))}.
\end{equation}
Consequently, \eqref{eq:general_verma_constraint} induces a collection of weighted conditional moment restrictions under $\bbP \in \calP_{\calG}$. 
\end{proposition} 
See a proof in Appendix~\ref{app:prop:weighted_moment_representation}. 

\begin{remark}[\textit{Role of the reference density $\tilde p$}]
\label{rem:reference_density}
The probability law $\bbP_{\phi_R}$ depends on the choice of the reference distribution $\tilde p(R)$ over the fixed variables. Throughout, we assume that $\tilde p(R)$ is a known density (or mass function) satisfying $\tilde p(R)\ge 0$ and $\int \tilde p(r)\,\diff r = 1$, with support contained in the support of $p(R\mid \mb_{\calG}(R))$. We further assume positivity, $p(R\mid \mb_{\calG}(R)) > 0$ almost surely on the support of $\tilde p$, and sufficient integrability conditions to ensure that the weighted expectations and conditional expectations appearing in Proposition~\ref{prop:weighted_moment_representation} are well defined.
\end{remark}

\begin{remark}[\textit{Ordinary conditional independence as a special case}] \label{rem:ordinary independence as a special case} When $R = \emptyset$, no fixing operation is required, $\omega_R \equiv 1$, and the post-fixing conditional expectation operator $\bbE_{\phi_R} [\cdot\mid\cdot]$ reduces to the ordinary conditional mean operator. In this case, Proposition~\ref{prop:weighted_moment_representation} reduces to the familiar characterization of conditional independence through residual orthogonality. Thus, ordinary conditional independence restrictions arise as a special case of the proposed framework, whereas genuine Verma constraints induce weighted residual orthogonality through graphical fixing operations.
\end{remark}

Representation \eqref{eq:general_equality_constraint} places Verma-constrained models within the broader class of conditional moment restriction models. The distinguishing feature is that the moment restrictions are induced by graphical fixing operations rather than imposed directly under the observed law, yielding weights and residualization operators determined by the underlying ADMG. This representation forms the basis for the tangent-space characterization developed in the next subsection.

\subsection{Tangent space characterization under a single Verma constraint} \label{subsec:tangent_space}

Proposition~\ref{prop:weighted_moment_representation} represents a Verma constraint as a collection of weighted conditional moment restrictions indexed by square-integrable functions. We now characterize the local geometry of the resulting statistical model. Throughout this section, tangent spaces are viewed as closed subspaces of the Hilbert space
\begin{align*}
L_{0}^{2} (\bbP) = \Big\{
h \in L^{2} (\bbP): \bbE [h (O)] = 0 \Big\},
\end{align*}
equipped with inner product $\langle h_1, h_2 \rangle \coloneqq \bbE [h_1 (O) h_2 (O)]$. 

Let $\calP_{\calG}$ denote the statistical model defined by the single nested Markov constraint as in \eqref{eq:general_verma_constraint}. The following proposition characterizes the pathwise score restrictions induced by this constraint.

\begin{proposition}[Pathwise score restrictions]
\label{prop:pathwise_score_restrictions}
Let $\bbP$ lie in the interior of $\calP_{\calG}$, and let $\{\bbP_{t}: t \in (-\eps, \eps)\} \subseteq \calP_{\calG}$ be a regular parametric submodel through $\bbP$ at $t = 0$, with score $s (\cdot) = \left. \frac{\diff}{\diff t} \log p_{t} (\cdot) \right|_{t=0}$. For every pair of square-integrable functions $(f, g)$, the restriction $X \indep Y \mid Z \ [\phi_R(p_t)]$ implies
\begin{equation}\label{eq:pathwise_restriction_postfixing}
\frac{\diff}{\diff t}
\bbE_{\phi_R(p_t)}
\left[
\left\{
f(X,Z)
-
\bbE_{\phi_R(p_t)}[f(X,Z)\mid Z]
\right\}
g(Y,Z)
\right]
\Bigg|_{t=0}
=0.
\end{equation}
Equivalently, using the weighted representation in Proposition~\ref{prop:weighted_moment_representation},
\begin{equation}\label{eq:pathwise_restriction_observed}
\left.
\frac{\diff}{\diff t}
\bbE_{t}
[\psi_{f,g,t}(O)]
\right|_{t=0}
=0,
\end{equation}
where
\begin{equation}\label{eq:psi_fgt}
\psi_{f,g,t}(O)
\coloneqq
\omega_{R,t}(\M)
\bigg\{
f(X,Z)
-
\frac{
\bbE_{t}[\, \omega_{R,t}(\M)f(X,Z)\mid Z]
}{
\bbE_{t}[\, \omega_{R,t}(\M)\mid Z]
}
\bigg\} \,
g(Y,Z),
\end{equation}

\vspace{0.2cm} \noindent 
and $\omega_{R,t}(\M)={\tilde p(R)}/{p_t(R\mid \mb_{\calG}(R))}$, with $\M = (R, \mb_{\calG}(R))$. Consequently, pathwise differentiation of the constraint yields an element $\chi_{f,g}\in L_0^2(\bbP)$ satisfying
\begin{equation}\label{eq:generic_orthogonality}
\bbE[\chi_{f,g}(O)s(O)]= 0,
\end{equation}
for every score $s\in\Lambda_{\bbP}$. Hence $\chi_{f,g}\in\Lambda_{\bbP}^{\perp}$, where $\Lambda_{\bbP}$ denotes the tangent space of $\calP_{\calG}$ at $\bbP$.
\end{proposition}
A proof is given in Appendix~\ref{app:prop:pathwise_score_restrictions}.

Proposition~\ref{prop:pathwise_score_restrictions} shows that each Verma constraint induces an infinite collection of directions indexed by $(f,g)$ that must lie in the orthocomplement to the tangent space. However, the proposition does not yet provide an explicit representation of these directions or establish that they generate the entire orthocomplement. The key observation is that pathwise differentiation of the weighting operator and the post-fixing conditional expectations induces a residualization structure analogous to that arising in ordinary conditional independence models, albeit modified by the fixing weights. The following theorem makes this structure explicit and shows that these directions span the orthocomplement to the tangent space of the single-constraint model. 

\begin{theorem}[Orthocomplement tangent space under a nested Markov constraint]
\label{thm:orthocomplement_tangent_space}
Consider the statistical model $\calP_{\calG}$ defined by the single nested Markov constraint $X \indep Y \mid Z\ [\phi_R(p)]$, where $X$, $Y$, and $Z$ are pairwise disjoint sets of random variables, and $R$ is a valid fixing set. No disjointness between $R$ and $X\cup Y\cup Z$ is assumed. Under the regularity conditions stated in Proposition~\ref{prop:weighted_moment_representation}, the orthocomplement to the tangent space at $\bbP \in \interior (\calP_{\calG})$ is
\begin{equation}\label{eq:orthocomp}
\Lambda_{\bbP}^{\perp}
=
\overline{\sp}
\left\{
\chi_{f,g}(O):
f\in L^2(\bbP_{X,Z}),
\;
g\in L^2(\bbP_{Y,Z})
\right\},
\end{equation}
where the closure is taken in $L_0^2(\bbP)$ and
\begin{align}
\chi_{f,g}(O)
&\coloneqq
\widetilde\chi_{f,g}(O)
-
\bbE[\, \widetilde\chi_{f,g}(O)\mid \mb_{\calG} (R), R]
+
\bbE[\, \widetilde\chi_{f,g}(O)\mid \mb_{\calG}(R)],
\label{eq:chi_residualized}
\\[0.2cm]
\widetilde\chi_{f,g}(O)
&\coloneqq
\omega_R(\M)
\bigg\{
f(X,Z)
-
\frac{
\bbE[\, \omega_R(\M)f(X,Z)\mid Z]
}{
\bbE[\, \omega_R(\M)\mid Z]
}
\bigg\}
\bigg\{
g(Y,Z)
-
\frac{
\bbE[\, \omega_R(\M)g(Y,Z)\mid Z]
}{
\bbE[\, \omega_R(\M)\mid Z]
}
\bigg\},
\label{eq:chi_raw}
\end{align}
with $\omega_R(\M) := {\tilde p(R)}/{p(R\mid \mb_{\calG}(R))}$ and $\M := (R, \mb_{\calG}(R))$. 
\end{theorem}
A proof is given in Appendix~\ref{app:thm:orthocomplement_tangent_space}.

\begin{remark}[\textit{Coupled score directions under Verma constraints}]
\label{rem:coupled_score_directions}
The residualization operator in Theorem~\ref{thm:orthocomplement_tangent_space} expresses each orthocomplement direction as a sum of components belonging to the ordinary likelihood score spaces associated with the observed-data factorization,  according to a valid topological ordering of the variables in the ADMG. Consequently, the resulting directions may lie in a direct sum of several orthogonal score spaces. This does not imply, however, that the orthocomplement equals that direct sum. The raw weighted moment functions to which the residualization operator is applied have the structured multiplicative form induced by the Verma constraint, and therefore couple the resulting score components. Thus, the individual components cannot generally be varied independently. This distinguishes the tangent-space geometry of Verma constraints from that of ordinary conditional independence restrictions. For the latter, likelihood factorization often yields componentwise score restrictions and corresponding direct-sum characterizations. Verma constraints instead impose restrictions jointly across likelihood components, so their orthocomplements are generally structured subspaces of the corresponding direct sums.
\end{remark}

When $R=\emptyset$, no fixing operation is required, $\omega_R\equiv1$, the post-fixing expectation coincides with ordinary expectation under $\bbP$, and the residualization terms involving $R$ and $\mb_{\cal G}(R)$ disappear. Theorem~\ref{thm:orthocomplement_tangent_space} therefore reduces to the classical residual-product characterization below

\begin{corollary}[Orthocomplement tangent space under an ordinary conditional independence] \label{cor:ordinary_ci_orthocomplement}
For a statistical model defined by a conditional independence restriction $X \indep Y \mid Z$, the orthocomplement to the tangent space at $\bbP$ is $\Lambda_{\bbP}^{\perp} = \overline{\sp} \left\{ \chi_{f,g}(O): f\in L^2(\bbP_{X,Z}), \; g\in L^2(\bbP_{Y,Z})\right\}$, where 
\begin{align*}
\chi_{f,g}(O)
=
\left\{
f(X,Z)-\bbE[f(X,Z)\mid Z]
\right\}
\left\{
g(Y,Z)-\bbE[g(Y,Z)\mid Z]
\right\}.
\end{align*}
\end{corollary}

The preceding results serve as the building blocks for nested Markov models involving multiple equality constraints. In that setting, the tangent space is obtained as the intersection of the constraint-specific tangent spaces, while the orthocomplement is generated by the closed span of the corresponding orthocomplement directions. We return to this extension in Section~\ref{sec:beyond_single_constraint}.

\subsection{Semiparametric efficiency under a Verma constraint}
\label{subsec:semipar_efficiency}

The orthocomplement characterization in Theorem~\ref{thm:orthocomplement_tangent_space} immediately yields the semiparametric efficient influence function. Let $\psi(\bbP)$ be a pathwise differentiable parameter with nonparametric influence function $\varphi^{\mathrm{np}}$. Then the efficient influence function under the Verma-constrained model $\calP_{\calG}$, denoted by $\varphi^{\eff}$, is obtained by orthogonally projecting $\varphi^{\mathrm{np}}$ onto the tangent space. That is, 
\begin{equation}
\varphi^{\eff}(O)
=
\Pi\big(
\varphi^{\np}(O)
\mid
\Lambda_{\bbP}
\big)
=
\varphi^{\np}(O)
-
\Pi\big(
\varphi^{\np}(O)
\mid
\Lambda_{\bbP}^{\perp}
\big),
\label{eq:eif_projection}
\end{equation}
where $\Pi(\cdot\mid\Lambda_{\bbP})$ denotes orthogonal projection in $L_0^2(\bbP)$.

Equivalently, since every influence function under $\calP_{\calG}$ differs from $\varphi^{\np}(O)$ by an element of $\Lambda_{\bbP}^{\perp}$, the efficient influence function may be characterized variationally as the minimum-variance element of the affine space $\varphi^{\np}(O)+\Lambda_{\bbP}^{\perp}$. That is,
\begin{equation}
\varphi^{\eff}(O)
=
\arg\min_{\varphi \, \in \, \varphi^{\np}+\Lambda_{\bbP}^{\perp}}
\bbE[\varphi(O)^2].
\label{eq:eif_variational}
\end{equation}

The projection characterization in \eqref{eq:eif_projection} and the variational characterization in \eqref{eq:eif_variational} are equivalent Hilbert-space formulations of the semiparametric efficiency problem. Once the orthocomplement $\Lambda_{\bbP}^{\perp}$ has been characterized, efficient influence functions may be obtained either by explicit orthogonal projection or by solving the corresponding minimum-variance problem over $\varphi^{\mathrm{np}}+\Lambda_{\bbP}^{\perp}$.

Exact projection onto the orthocomplement $\Lambda_{\bbP}^{\perp}$ may be analytically intractable because this space is generally infinite dimensional. A natural approach is therefore to approximate the projection by restricting it to a finite-dimensional subspace generated by prespecified basis functions. One convenient construction is based on tensor-product basis expansions of the weighting functions $f$ and $g$; see Appendix~\ref{app:basis_approximation}. Given any positive integer $J$, we let $[J] := \{1, 2, \cdots, J\}$. The following theorem gives the resulting basis approximation to the efficient influence function. 

\begin{theorem}[Finite-dimensional approximation to the efficient influence function]
\label{theorem:locally_eff}

Let $f_1,\ldots,f_J\in L^2(\bbP_{X,Z})$ and $g_1,\ldots,g_K\in L^2(\bbP_{Y,Z})$ be prespecified basis functions, and define $\chi_{jk}(O) \equiv \chi_{f_j,g_k}(O)$, for $j \in [J]$ and $k \in [K]$, where $\chi_{f_j, g_k} (O)$ is the orthocomplement element defined in Theorem~\ref{thm:orthocomplement_tangent_space}. Collect these directions into the vector $\chi (O) = (\chi_{11} (O), \cdots, \chi_{1K} (O), \chi_{21} (O), \cdots, \chi_{JK} (O))^\top \in \bbR^{JK}$, and define the finite-dimensional subspace 
\begin{align*}
\Lambda_{\bbP}^{\perp, \sub} = \Big\{ \alpha^\top \chi (O): \alpha \in \bbR^{JK} \Big\}. 
\end{align*}
Then the orthogonal projection of the nonparametric IF $\varphi^{\np} (O)$ onto $\Lambda^{\perp, \sub}_{\bbP}$ is 
\begin{align*}
\Pi \big( \varphi^{\np} (O) \mid \Lambda^{\perp, \sub}_{\bbP} \big)
=
{\alpha^{\mathrm{opt}}}^\top \chi(O), 
\quad\text{where} \quad 
\alpha^\opt &=\arg\min_{\alpha \, \in \, \bbR^{JK}}\bbE\left[\left\{\varphi^\np(O)-\alpha^\top\chi(O)\right\}^2\right].
\end{align*}
If $\Sigma=\bbE\big[\chi(O)\chi(O)^\top\big]$ and $b=\bbE\big[\chi(O)\varphi^{\mathrm{np}}(O)\big]$, and $\Sigma$ is nonsingular, then $\alpha^{\mathrm{opt}}=\Sigma^{-1}b$.  Consequently, the finite-dimensional approximation to the efficient influence function is
\begin{align*}
\varphi^{\mathrm{eff,basis}}(O)=\varphi^{\mathrm{np}}(O)-\left({\alpha^{\mathrm{opt}}}\right)^\top \chi(O), 
\end{align*}
with corresponding variance reduction 
\begin{align*}
\Var\big[\varphi^{\mathrm{np}}(O)\big]-\Var\big[\varphi^{\mathrm{eff,basis}}(O)\big]
=
b^\top \Sigma^{-1} b.
\end{align*}
\end{theorem}
See a proof in Appendix~\ref{app:theorem:locally_eff}. 

\begin{remark}[\textit{Matrix representation of the basis approximation}]\label{rem:matrix_representation}
The $JK$ orthocomplement directions in Theorem~\ref{theorem:locally_eff} may equivalently be collected into the matrix $\mathsf X(O) = (\chi_{jk}(O))_{j,k} \in\bbR^{J\times K}$, so that $\chi(O) = \mathrm{vec} \{\mathsf{X}(O)\}$ under the corresponding vectorization convention. Accordingly, for any $\alpha\in\bbR^{JK}$, there exists a matrix $\Theta\in\bbR^{J\times K}$ with $\alpha = \mathrm{vec} (\Theta)$ such that $\alpha^\top\chi(O) = \Tr\left\{\Theta^\top\mathsf{X}(O) \right\}$. Thus, the projection problem in Theorem~\ref{theorem:locally_eff} may equivalently be written as
\begin{align}
\Theta^{\mathrm{opt}}
\in
\arg\min_{\Theta\in\bbR^{J\times K}}
\bbE
\left[
\left\{
\varphi^{\mathrm{np}}(O)
-
\Tr\left(
\Theta^\top\mathsf{X}(O)
\right)
\right\}^{2}
\right].
\label{eq:matrix_projection}
\end{align}
The matrix formulation is mathematically equivalent to the vector formulation, but preserves the natural $J\times K$ indexing induced by the basis-function pairs $(f_j,g_k)$.
\end{remark}

\subsection{Indexed invariance under a mean-scale Verma restriction} \label{subsec:mean_scale_verma}

The preceding results characterize the tangent-space geometry induced by a full distributional Verma constraint. We now consider a weaker class of restrictions that constrain only a prespecified feature of the post-fixing conditional distribution. Our primary interest is a useful consequence that arises in several canonical examples: a target parameter may admit multiple identification functionals indexed by the value of an observed variable, all of which identify the same parameter.

This indexed invariance was exploited by \citet{guo2023flexible} and \citet{guo2025causal} by forming weighted combinations of the corresponding nonparametric influence functions and selecting the optimal weights to minimize asymptotic variance. The geometric framework developed above clarifies when such combinations attain the semiparametric efficiency bound and when they exploit only part of the available model restriction. 

Consider the nested conditional independence in \eqref{eq:general_verma_constraint}. For a prespecified square-integrable function $h(Y)$, define a mean-scale Verma restriction as
\begin{align}
\bbE_{\phi_R}
\left[
h(Y)\mid X,Z
\right]
=
\bbE_{\phi_R}
\left[
h(Y)\mid Z
\right], 
\label{eq:mean_scale_verma}
\end{align}
and denote the resulting model by $\calP_{\mathrm{mean},h}$. 

The mean-scale restriction is obtained from the general moment representation of the distributional Verma constraint by fixing the outcome-side test function to be $g(Y,Z)=h(Y)$. More generally, one could use $g(Y,Z)=h(Y)q(Z)$ for an arbitrary square-integrable function $q$. This does not enlarge the resulting orthocomplement space, however, because $q(Z)$ is measurable with respect to the conditioning variable and can be absorbed into the unrestricted function $f$. In particular, $\chi_{f,hq} = \chi_{fq,h}$. Thus, taking $g(Y,Z)=h(Y)$ is a convenient normalization that leaves all
weighting by functions of the conditioning variable represented through $f$. Accordingly, Proposition~\ref{prop:weighted_moment_representation}
and Theorem~\ref{thm:orthocomplement_tangent_space} yield the resulting orthocomplement space
\begin{align}
\Lambda_{\bbP,\mathrm{mean},h}^{\perp}
=
\overline{\sp}
\left\{
\chi_{f,h}(O):
f\in L^2(\bbP_{X,Z})
\right\},
\label{eq:mean_scale_orthocomplement}
\end{align}
where $\chi_{f,h}$ is defined by \eqref{eq:chi_residualized}--\eqref{eq:chi_raw} with $g(Y,Z)=h(Y)$. See Appendix~\ref{app:eq:mean_scale_verma} for details. 

Since $\calP_{\calG}\subseteq\calP_{\mathrm{mean},h}$, the mean-scale model has a larger tangent space and therefore $\Lambda_{\bbP,\mathrm{mean},h}^{\perp} \subseteq \Lambda_{\bbP}^{\perp}$. Thus, the efficiency gain available under the mean-scale restriction cannot exceed that available under the full distributional Verma constraint. Furthermore, for a pathwise differentiable parameter $\psi:\calP_{\mathrm{mean},h}\rightarrow\bbR$ with nonparametric influence function $\varphi^{\mathrm{np}}$, its efficient influence function is $\varphi_{\mathrm{mean},h}^{\mathrm{eff}} = \varphi^{\mathrm{np}} - \Pi\big(\varphi^{\mathrm{np}}\mid \Lambda_{\bbP,\mathrm{mean},h}^{\perp}\big)$. 

We now identify a broad class of target parameters for which the mean-scale Verma restriction induces an indexed family of equivalent identification functionals. Let
\begin{align*}
\eta_h(\bbP;x,z)
=
\bbE_{\phi_R}
\!\left[
h(Y)\mid X=x,Z=z
\right],
\qquad
\overline\eta_h(\bbP;z)
=
\bbE_{\phi_R}
\!\left[
h(Y)\mid Z=z
\right].
\end{align*}
Suppose the target parameter is obtained by applying a common $\bbP$-dependent identification operator $\mathcal L_{\bbP}$ to the indexed post-fixing conditional mean,
\begin{align}
\psi(\bbP;x)
=
\mathcal L_{\bbP}
\!\left\{
\eta_h(\bbP;x,\cdot)
\right\},
\qquad
x\in\cal X.
\label{eq:general_identification_operator}
\end{align}

Under the mean-scale Verma restriction \eqref{eq:mean_scale_verma}, $\eta_h(\bbP;x,z) = \overline\eta_h(\bbP;z)$, so $\psi(\bbP;x) = \mathcal L_{\bbP} \!\left\{\overline\eta_h(\bbP;\cdot)\right\}$,  for all $x\in\cal X$. Hence, the parameter admits the indexed invariance representation
\begin{align}
\psi(\bbP;x)
=
\psi(\bbP),
\qquad
x\in\cal X.
\label{eq:indexed_functional_invariance}
\end{align}
Thus, indexed functional invariance is not an independent assumption but an integrated consequence of the underlying conditional mean-scale Verma restriction. 

The indexed invariance in \eqref{eq:indexed_functional_invariance} may itself be viewed as a restriction defining a target-specific statistical model. Let
\begin{align}
\calP_{\psi}
=
\left\{
\bbP:
\psi(\bbP;x)
=
\psi(\bbP;x'),
\quad
\text{for all }
x,x'\in\cal X
\right\},
\label{eq:target_specific_invariance_model}
\end{align}
and let $\Lambda_{\bbP,\psi}$ denote its tangent space at an interior distribution $\bbP$. Let $\varphi^{\mathrm{np}}(\bbP;x)$ denote the nonparametric influence function of $\psi(\bbP;x)$ for fixed $x$. Intuitively, because differences $\psi(\bbP;x)-\psi(\bbP;x')$ vanish throughout $\calP_{\psi}$, their pathwise derivatives must also vanish. The corresponding differences of nonparametric IFs therefore generate directions orthogonal to the tangent space.

\begin{theorem}[Orthocomplement under indexed functional invariance]\label{thm:indexed_invariance_orthocomplement}
Fix $x_0\in\cal X$, and suppose that, for each $x\in\cal X$, $\psi(\bbP;x)$ is generated according to \eqref{eq:general_identification_operator} and is pathwise differentiable in the nonparametric model with influence function $\varphi^{\mathrm{np}}(\bbP;x)$. The orthocomplement to the tangent space of $\calP_{\psi}$ in \eqref{eq:target_specific_invariance_model} at $\bbP$ is $\Lambda_{\bbP,\psi}^{\perp} = \overline{\sp}\left\{\varphi^{\mathrm{np}}(\bbP;x) - \varphi^{\mathrm{np}}(\bbP;x_0): x\in\cal X \right\}$. Equivalently,
\begin{align}
\Lambda_{\bbP,\psi}^{\perp}
=
\overline{\sp}
\left\{
\int
\varphi^{\mathrm{np}}(\bbP;x) \, c(x) \, \diff x
:
\int c(x) \, \diff x=0
\right\}.
\label{eq:indexed_invariance_integral_orthocomplement}
\end{align}
Consequently, the collection of influence functions for the common value $\psi(\bbP)$ under $\calP_{\psi}$ is
\begin{align}
\calI 
\{\psi(\bbP)\}
=
\left\{
\int
\varphi^{\mathrm{np}}(\bbP;x) \, 
\widetilde p(x) \, 
\diff x
:
\int
\widetilde p(x) \, 
\diff x
=
1
\right\}.
\label{eq:indexed_invariance_IF_class}
\end{align}
The efficient influence function is therefore the minimum-variance member of this class:
\begin{align}
\varphi^{\mathrm{eff}}
=
\argmin_{
\varphi\in\calI \{\psi(\bbP)\}
}
\bbE[\varphi(O)^2].
\label{eq:indexed_invariance_variational}
\end{align}
\end{theorem}
See a proof in Appendix~\ref{app:thm:indexed_invariance_orthocomplement}.

By construction, the Verma constraint in \eqref{eq:general_verma_constraint} implies the mean-scale restriction in \eqref{eq:mean_scale_verma}, which in turn implies \eqref{eq:indexed_functional_invariance}. Hence $\calP_{\calG} \subseteq \calP_{\mathrm{mean},h} \subseteq \calP_{\psi}$, and therefore $\Lambda_{\bbP,\psi}^{\perp} \subseteq \Lambda_{\bbP,\mathrm{mean},h}^{\perp} \subseteq \Lambda_{\bbP}^{\perp}$. Thus, weighted combinations of indexed nonparametric influence functions always exploit valid orthocomplement directions under the mean-scale Verma model. If the indexed invariance restrictions exhaust the first-order implications of the mean-scale restriction, then $\Lambda_{\bbP,\psi}^{\perp} = \Lambda_{\bbP,\mathrm{mean},h}^{\perp}$, and the efficient influence function in \eqref{eq:indexed_invariance_variational} is also the semiparametric efficient influence function under $\calP_{\mathrm{mean},h}$. Otherwise, the weighting construction exploits only the restrictions encoded by the target-specific indexed invariance and need not attain the efficiency bound under the full mean-scale Verma model.

The inclusion above also clarifies the relationship between the two orthocomplement characterizations in Theorem~\ref{thm:orthocomplement_tangent_space} and Theorem~\ref{thm:indexed_invariance_orthocomplement}. Since $\Lambda_{\bbP,\psi}^{\perp} \subseteq \Lambda_{\bbP}^{\perp}$, every target-specific direction of the form $\int \varphi^{\mathrm{np}}(\bbP;x) c(x) \diff x, \, \int c(x)\diff x=0$, necessarily belongs to the closed linear span of the weighted residual directions $\chi_{f,g}$ characterized in Theorem~\ref{thm:orthocomplement_tangent_space}. The precise correspondence between these two representations, however, depends on the identification structure of the target parameter. In particular, determining whether an indexed influence-function direction can be represented by a single weighted residual direction, $\int \varphi^{\mathrm{np}}(\bbP;x)c(x)\diff x = \chi_{f_c,g_c}$, and identifying the corresponding functions $f_c$ and $g_c$, generally requires analysis of the specific identification functional. Section~\ref{sec:applications} derives this correspondence explicitly for the extended front-door and Napkin graphs.

\begin{lemma}[Optimal affine combination of influence functions]
\label{lem:optimal_affine_IF_combination}
Let $\boldsymbol G(O) = (G_1(O),\ldots,G_K(O))^{\top}$ be a vector of $K$ square-integrable influence functions for $\psi(\bbP)$ under $\calP_{\mathrm{mean},h}$. The minimum-variance influence function in the affine class $\left\{\alpha^{\top}\boldsymbol G : \alpha\in\bbR^K,\  \boldsymbol 1^{\top}\alpha=1\right\}$ is $G_{\alpha^{\mathrm{opt}}} = \left(\alpha^{\mathrm{opt}} \right)^{\top}\boldsymbol G(O)$, where
\begin{align}
\alpha^{\mathrm{opt}}
=
\argmin_{
\alpha\in\bbR^K:
\boldsymbol 1^{\top}\alpha=1
}
\bbE
\left[
\{
\alpha^{\top}\boldsymbol G(O)
\}^{2}
\right].
\label{eq:optimal_affine_IF_combination}
\end{align}
Writing $\Sigma = \bbE [\boldsymbol G(O)\boldsymbol G(O)^{\top}]$, if $\Sigma$ is nonsingular, then $\alpha^{\mathrm{opt}} = \displaystyle \frac{\Sigma^{-1}\boldsymbol 1}{\boldsymbol 1^{\top}\Sigma^{-1}\boldsymbol 1}$. 
\end{lemma}
See a proof in Appendix~\ref{app:lem:optimal_affine_IF_combination}.

If the affine class contains $\varphi_{\mathrm{mean},h}^{\mathrm{eff}}$, then $G_{\alpha^{\mathrm{opt}}} = \varphi_{\mathrm{mean},h}^{\mathrm{eff}}$. Otherwise, $G_{\alpha^{\mathrm{opt}}}$ is the minimum-variance influence function within the selected affine class.

When $\calX=\{x_1,\ldots,x_K\}$ is finite, set $G_k(O) = \varphi^{\mathrm{np}}(\bbP;x_k)$, for $k=1,\ldots,K$. Under the condition of Theorem~\ref{thm:indexed_invariance_orthocomplement}, their affine hull is the full influence-function class in \eqref{eq:indexed_invariance_IF_class}. Therefore, Lemma~\ref{lem:optimal_affine_IF_combination} yields the exact efficient influence function.

When $\calX$ is continuous, choose candidate weighting functions $\widetilde p_1,\ldots,\widetilde p_K$ satisfying $\int \widetilde p_k(x) \diff x = 1$, for $k=1,\ldots,K$, and define $G_k(O) = \int \varphi^{\mathrm{np}}(\bbP;x) \widetilde p_k(x) \diff x$. For any $\alpha$ satisfying $\boldsymbol 1^\top\alpha=1$, $\alpha^\top\boldsymbol G(O) = \int \varphi^{\mathrm{np}}(\bbP;x) \widetilde p_\alpha(x) \diff x$,  where $\widetilde p_\alpha(x) = \sum_{k=1}^K \alpha_k\widetilde p_k(x)$. Lemma~\ref{lem:optimal_affine_IF_combination} therefore gives the minimum-variance IF within the selected finite-dimensional class. Increasingly rich collections of weighting functions may be used to approximate the efficient influence function in \eqref{eq:indexed_invariance_variational}.

The preceding results establish a semiparametric framework for statistical models defined by a single nested Markov constraint. We characterized the corresponding tangent-space geometry, derived projection and minimum-variance characterizations of the semiparametric efficient influence function, and developed finite-dimensional basis approximations when exact projection is analytically intractable. We further showed how indexed invariance under mean-scale Verma restrictions recovers existing efficiency results as special cases. These results serve as the basic building blocks for the more general nested Markov models considered next.

We next consider statistical models in which an ADMG encodes multiple ordinary and nested Markov constraints simultaneously. As we will see, the single-constraint characterization extends naturally by combining the tangent-space restrictions induced by each individual constraint.

\section{Extensions to Multiple Constraints}
\label{sec:beyond_single_constraint}

An ADMG may encode multiple ordinary and nested conditional independence constraints arising from distinct fixing operations. Each such constraint induces its own collection of weighted conditional moment restrictions and therefore its own restriction on the local tangent space. In particular, both Tian's constraint finding algorithm \citet{tian2002testable} (also see Algorithm~1 of \citet{richardson2023nested}) and the  \emph{ordered local nested Markov property} (Definition~37 of \citet{richardson2023nested}) enumerate the nested Markov constraints implied by an ADMG. 

Suppose the ADMG implies $J$ generalized conditional independence constraints. For $j \in [J]$, let $R_{j} \subseteq V$ denote the subset of vertices being fixed, and $X_{j}, Y_{j}, Z_{j} \subseteq V \setminus R_{j}$ be three disjoint subsets of the remaining vertices. Then the $j$-th generalized conditional independence constraint takes the form
\begin{equation}
\label{multiple Verma}
X_{j} \indep Y_{j} \mid Z_{j} \qquad [\phi_{R_{j}} (p)].
\end{equation}

Suppose that the nested Markov model $\calP_{\calG}$ is characterized through a collection of weighted conditional moment constraints of the form \eqref{eq:general_equality_constraint} or \eqref{eq:general_equality_constraint2}, indexed by $j = 1, \dots, J$: 
\begin{equation}
\label{eq:general_equality_constraint_1}
\begin{split}
\calP_{\calG} = \left\{ \bbP \in \calP \left\vert \begin{array}{c} \bbE_{\phi_{R_{j}}} [\{f_{j} (X_{j}, Z_{j}) - \bbE_{\phi_{R_{j}}} [f_{j} (X_{j}, Z_{j}) \mid Z_{j}]\} g_{j} (Y_{j}, Z_{j})] \\
= \bbE_{\phi_{R_{j}}} [f_{j} (X_{j}, Z_{j}) \{g_{j} (Y_{j}, Z_{j}) - \bbE_{\phi_{R_{j}}} [g_{j} (Y_{j}, Z_{j}) \mid Z_{j}]\}] = 0, \\
\forall \, f_{j} \in L^{2} (\bbP_{\calX \times \calZ}), g_{j} \in L^{2} (\bbP_{\calY \times \calZ}), \text{for } j \in [J]
\end{array} \right. \right\}.
\end{split}
\end{equation}

Let $\calP_j$ denote the statistical model associated with the $j$-th constraint, and then $\calP_{\calG} = \bigcap_{j = 1}^J \calP_{j}$. The following result in turn characterizes a nontrivial superset of the tangent space at some $\bbP$ in the interior of $\calP_{\calG}$. A proof is deferred to Appendix~\ref{app:thm:multiple_constraints}. 

\begin{proposition}[Tangent space under multiple nested Markov constraints]
\label{thm:multiple_constraints}
Let $\chi^{(j)}_{f_{j}, g_{j}}$ take the following form:
\begin{align*}
& \chi^{(j)}_{f_{j}, g_{j}} (O) := \tilde{\chi}^{(j)}_{f_{j}, g_{j}} (O) - \bbE [\tilde{\chi}^{(j)}_{f_{j}, g_{j}} (O) \mid \mb_{\calG} (R_{j}), R_{j}] + \bbE [\tilde{\chi}^{(j)}_{f_{j}, g_{j}} (O) \mid \mb_{\calG} (R_{j})], \text{ where} \\
& \tilde{\chi}^{(j)}_{f_{j}, g_{j}} (O) := \omega_{R_{j}} (\M_{j}) \{f_{j} (X_{j}, Z_{j}) - \bbE_{\phi_{R_{j}}} [f_{j} (X_{j}, Z_{j}) \mid Z_{j}]\} \{g_{j} (Y_{j}, Z_{j}) - \bbE_{\phi_{R_{j}}} [g_{j} (Y_{j}, Z_{j}) \mid Z_{j}]\},
\end{align*}
with $\omega_{R_{j}} (\M_{j}) = \dfrac{\tilde{p}_{j} (R_{j})}{p (R_{j} \mid \mb_{\calG} (R_{j}))}$ and $\M_{j} := \{\mb_{\calG} (R_{j}), R_{j}\}$. Then
\begin{equation}
\Lambda_{\bbP}^{\perp} \supseteq \tilde{\Lambda}_{\bbP}^{\perp} := \overline{\sp} \Big\{ \sum_{j = 1}^{J} \chi^{(j)}_{f_{j}, g_{j}} (O): f_{j} \in L^{2} (\bbP_{X_{j}, Z_{j}}), g_{j} \in L^{2} (\bbP_{Y_{j}, Z_{j}}) \Big\}.
\end{equation}
\end{proposition}

\begin{remark}[\textit{Conjectured completeness under multiple constraints}]
\label{rem:conjecture}
In Proposition~\ref{thm:multiple_constraints}, we only obtain a partial result that characterizes a nontrivial superset $\tilde{\Lambda}_{\bbP}$ of $\Lambda_{\bbP}$, or equivalently, a nontrivial subset $\tilde{\Lambda}_{\bbP}^{\perp}$ of $\Lambda_{\bbP}^{\perp}$. We conjecture that if \eqref{multiple Verma} constitutes the complete list of equality constraints imposed by an ADMG, we should have $\Lambda_{\bbP}^{\perp} = \tilde{\Lambda}_{\bbP}^{\perp}$. Results in a similar spirit can be found in \citet{evans2016graphs}, but they only consider vertices with finite state spaces. In the concrete examples with multiple nested Markov constraints in Section~\ref{sec:ex:muiltiple_constraints}, we construct submodels so as to verify our conjecture in a case-by-case fashion, but a general theorem remains open.
\end{remark}

\begin{remark}[\textit{Recovery of the DAG tangent space}]
\label{rem:DAG}
In the special case where $\calG$ is itself a DAG or $\calP_{\calG} = \calP_{\calG^{\ast}}$ where $\calG^{\ast}$ is a DAG, the constraints are simply
\begin{align*}
X_{v} \indep X_{\pre (v) \setminus \pa (v)} \mid X_{\pa (v)}, v \in V.
\end{align*}
It is not difficult to see that $\tilde{\Lambda}_{\bbP}^{\perp}$ characterized in Proposition~\ref{thm:multiple_constraints} equals $\Lambda_{\bbP}^{\perp}$ and further implies that
\begin{equation}
\label{DAG tangent spaces}
\begin{split}
\Lambda_{\bbP} & = \Big\{ \sum_{v \in V} s_{v} (X_{v} \mid X_{\pa (v)}): \bbE \{s_{v} (X_{v} \mid X_{\pa (v)})^{2}\} < \infty, \\
& \qquad \bbE \{s_{v} (X_{v} \mid X_{\pa (v)}) \mid X_{\pa (v)}\} = 0 \text{ almost surely} \Big\},
\end{split}
\end{equation}
which is a known result \citep{guo2023variable}. See a proof in Appendix~\ref{app:DAG}. 
\end{remark}

Exact projection onto $\tilde{\Lambda}_{\bbP}^{\perp}$ is generally more challenging in the multiple-constraint setting because $\tilde{\Lambda}_{\bbP}^{\perp}$ is generated by the sum of several infinite-dimensional orthocomplement spaces. Nevertheless, the finite-dimensional basis approximation developed in Section~\ref{sec:geometry} extends directly by combining basis functions from the individual constraints. More generally, efficient influence functions may be viewed as projections onto intersections of closed subspaces of $L_0^2 (\bbP)$, suggesting the use of sequential or alternating projection algorithms. A systematic study of such computational approaches is left for future work.

In the next section, we illustrate the preceding theory through several canonical examples involving both single and multiple nested Markov constraints. These examples demonstrate how the general tangent-space and orthocomplement characterizations specialize to concrete graphical models and how the resulting semiparametric efficient influence functions can be constructed.


\section{Illustrative Examples}
\label{sec:applications}

\subsection{The extended front-door model}
\label{sec:front-door}

We first illustrate the general theory using the extended Pearl's front-door ADMG shown in Figure~\ref{fig:front-door}(a) \citep{pearl1995causal, bhattacharya2022testability, guo2025rank, dhawan2026debiased}. This graph admits a longitudinal interpretation studied by \citet{robins1999testing} in which $Z$ and $M$ represent first- and second-stage treatments. Variants of this model have also been studied in
\citet{liu2021efficient,von2025evaluation}. 

The extended front-door graph encodes a \emph{no-direct-effect} assumption $Y(z,m)=Y(m), \, \forall z,m$, which states that $Z$ affects $Y$ only through $M$. This assumption implies the nested conditional independence
\begin{align}
Z \indep Y \mid X,M
\qquad
[\phi_M(p)],
\label{eq:frontdoor_nested_constraint}
\end{align}
as illustrated via the CADMG in Figure~\ref{fig:front-door}(b). This is a direct specialization of the general framework developed in Section~\ref{sec:geometry}, obtained by identifying $X_{\mathrm{gen}}=Z,\,Y_{\mathrm{gen}}=Y,\,Z_{\mathrm{gen}}=(X,M),\,R=M$. 

\begin{figure}[t]
\centering
\begin{subfigure}[t]{0.45\textwidth}
\centering
\begin{tikzpicture}[>=stealth, node distance = 1.5cm]
    \tikzstyle{format} = [thick, circle, minimum size = 1.0mm, inner sep = 2pt]
    \tikzstyle{square} = [draw, thick, minimum size = 4.5mm, inner sep = 2pt]

    \begin{scope}[xshift = 0cm, yshift = 0cm]
	\path[->, very thick]
		
	node[shape = ellipse, draw = black] (z) {$Z$}
        node[right of = z, shape = ellipse, draw = black] (a) {$A$}
        node[above right of = a, shape = ellipse, draw = black] (x) {$X$}
        node[below right of = x, shape = ellipse, draw = black] (m) {$M$}
	node[right of = m, shape = ellipse, draw = black] (y) {$Y$}
		 
        (x) edge[blue] (z)
	   (x) edge[blue] (a) 
	   (x) edge[blue] (m) 
        (x) edge[blue] (y)
        (z) edge[blue] (a)
        (a) edge[blue] (m)
        (m) edge[blue] (y)
	   (z) edge[blue, bend right] (m)
        (a) edge[red, <->, bend right] (y)
        ;
    \end{scope}
\end{tikzpicture}
\caption{}
\label{fig:front-door graph original}
\end{subfigure}
\begin{subfigure}[t]{0.45\textwidth}
\centering
\begin{tikzpicture}[>=stealth, node distance = 1.5cm]
    \tikzstyle{format} = [thick, circle, minimum size = 1.0mm, inner sep = 2pt]
    \tikzstyle{square} = [draw, thick, minimum size = 4.5mm, inner sep = 2pt]
    
    \begin{scope}[xshift = 0cm, yshift = 0cm]
		\path[->, very thick]
        
	node[shape = ellipse, draw = black] (z) {$Z$}
    node[right of = z, shape = ellipse, draw = black] (a) {$A$}
    node[above right of = a, shape = ellipse, draw = black] (x) {$X$}
    node[below right of = x, shape = rectangle, draw = black] (m) {$M$}
	node[right of = m, shape = ellipse, draw = black] (y) {$Y$}
		 
        (x) edge[blue] (z)
	    (x) edge[blue] (a)
        (x) edge[blue] (y)
        (z) edge[blue] (a)
        (m) edge[blue] (y)
        (a) edge[red, <->, bend right] (y)
        ;
    \end{scope}
\end{tikzpicture}
\caption{}
\label{fig:front-door graph fixed}
\end{subfigure}
\caption{The extended front-door model. (a) Original ADMG. (b) CADMG obtained after fixing $M$, under which the nested conditional independence $Z \indep Y \mid X,M \ [\phi_M(p)]$ becomes an ordinary conditional independence.} 
\label{fig:front-door}
\end{figure}
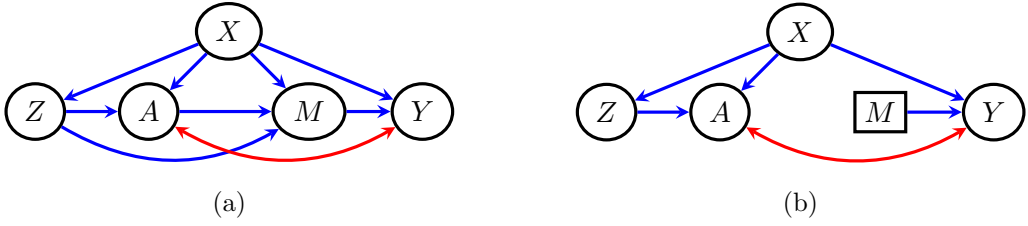

\begin{corollary}\label{cor:tangent_space_frontdoor}
The nested Markov model of the ADMG in Figure~\ref{fig:front-door}(a) is characterized as 
\begin{align}
\calP(\calG)
=
\Big\{
\bbP:
\bbE
\left[
\omega_M(\M)
\left\{
f(X,Z)
-
\bbE_{\phi_M}[f(X,Z)\mid X]
\right\}
g(X,M,Y)
\right]
=
0
\Big\},
\label{eq:frontdoor_model}
\end{align}
for every pair of $f\in L^2(\bbP_{X,Z})$ and $g\in L^2(\bbP_{X,M,Y})$, where $\M = (M, X, Z, A)$ and
\begin{align*}
\omega_M(\M)
=
\frac{1}{p(M\mid X,Z,A)},
\qquad
\bbE_{\phi_M}[h(O)\mid X]
=
\frac{
\bbE\!\left[\omega_M(\M) \, h(O)\mid X\right]
}{
\bbE\!\left[\omega_M(\M)\mid X\right]
}.
\end{align*}
Furthermore, the orthocomplement to the tangent space of $\calP(\calG)$ in \eqref{eq:frontdoor_model} at $\bbP$ is
\begin{align*}
\Lambda_{\bbP}^{\perp}
=
\overline{\sp}
\left\{
\chi_{f,g}(O):
f\in L^2(\bbP_{X,Z}),
\;
g\in L^2(\bbP_{X,M,Y})
\right\},
\end{align*}
where
\begin{align*}
\chi_{f,g}(O)
&=
\widetilde\chi_{f,g}(O)
-
\bbE
\!\left[
\widetilde\chi_{f,g}(O)
\mid
X,Z,A,M
\right]
+
\bbE
\!\left[
\widetilde\chi_{f,g}(O)
\mid
X,Z,A
\right],
\\
\widetilde\chi_{f,g}(O)
&=
\omega_M(\M)
\Big\{
f(X,Z)
-
\bbE_{\phi_M}[f(X,Z)\mid X]
\Big\}
\Big\{
g(X,M,Y)
-
\bbE_{\phi_M}[g(X,M,Y)\mid X, M]
\Big\}.
\end{align*}
\end{corollary}
See a proof in Appendix~\ref{app:cor:tangent_space_frontdoor}.

We now consider estimation of the marginal causal mean $\psi_{a_0}(\bbP) = \bbE[Y(a_0)]$. Under the extended front-door graph, this parameter is identified by a family of equivalent identification functionals indexed by a fixed value $z^*\in\cal Z$, denoted by $\psi_{a_0}(\bbP;z^*)$  \citep{tian2002general, guo2023flexible}. The explicit form is given in Appendix~\ref{app:id_npEIF_frontdoor}. The indexed identification functionals arise by applying the identification operator introduced in Section~\ref{subsec:mean_scale_verma} to the post-fixing conditional mean associated with the nested Markov constraint. Consequently,
\begin{align}
\psi_{a_0}(\bbP;z^*)
=
\bbE[Y(a_0)],
\qquad
\forall z^*\in\cal Z,
\label{eq:invariance_frontdoor}
\end{align}
which is precisely the indexed invariance condition \eqref{eq:indexed_functional_invariance}. Let $\varphi_{a_0}^{\mathrm{np}}(\bbP; z^*)$ denote the nonparametric influence function of $\psi_{a_0}(\bbP;z^*)$, derived by \citet{guo2023flexible} and provided in Appendix~\ref{app:id_npEIF_frontdoor}. 

The semiparametric efficient influence function for $\psi_{a_0}(\bbP)$ under the nested Markov model in \eqref{eq:frontdoor_model} is obtained by projecting any indexed nonparametric influence function onto the tangent space, or equivalently, $\varphi_{a_0}^{\mathrm{eff}}(O) = \varphi_{a_0}^{\mathrm{np}}(\bbP; z^*) (O) - \Pi \big(\varphi_{a_0}^{\mathrm{np}}(\bbP; z^*) \mid \Lambda_{\bbP}^{\perp} \big)(O)$, where $\Lambda_{\bbP}^{\perp}$ is characterized by Corollary~\ref{cor:tangent_space_frontdoor}. Since this projection is not available in closed form, Theorem~\ref{theorem:locally_eff} provides a finite-dimensional approximation. Choosing basis functions $f_1,\ldots,f_J\in L^2(\bbP_{X,Z})$ and $g_1,\ldots,g_K\in L^2(\bbP_{X,M,Y})$, let $\chi_{jk}(O) = \chi_{f_j,g_k}(O)$ for $j=1,\ldots,J$ and $k=1,\ldots,K$, where $\chi_{f_j,g_k}(O)$ is given explicitly in Corollary~\ref{cor:tangent_space_frontdoor}, and collect them into the vector $\chi(O)$. The resulting approximation is $\varphi_{a_0}^{\mathrm{eff, basis}}(O) = \varphi_{a_0}^{\mathrm{np}}(\bbP; z^*) (O) - \left(\alpha^{\mathrm{opt}}\right)^\top \chi(O)$, where $\alpha^{\mathrm{opt}}$ is given by Theorem~\ref{theorem:locally_eff}. Thus, the only model-specific ingredient is the collection of orthocomplement directions supplied by Corollary~\ref{cor:tangent_space_frontdoor}; the optimal projection follows immediately from the general projection theory. 

By Theorem~\ref{thm:indexed_invariance_orthocomplement}, the target-specific orthocomplement is
\begin{align}
\Lambda_{\bbP,\psi}^{\perp}
=
\overline{\sp}
\left\{
\int
\varphi_{a_0}^{\mathrm{np}}(\bbP; z^*)(O)
c(z^*)
\,\diff z^*
:
\int
c(z^*)
\,\diff z^*
=
0
\right\},
\label{eq:frontdoor_indexed_orthocomplement}
\end{align}
and the corresponding class of influence functions is
\begin{align}
\calI \{\psi_{a_0}(\bbP)\}
=
\left\{
\int
\varphi_{a_0}^{\mathrm{np}}(\bbP; z^*)(O)
\widetilde p(z^*)
\,\diff z^*
:
\int
\widetilde p(z^*)
\,\diff z^*
=
1
\right\}.
\label{eq:frontdoor_mean_IF_class}
\end{align}
The influence-function class in \eqref{eq:frontdoor_mean_IF_class} is precisely the class considered in Lemma~\ref{lem:optimal_affine_IF_combination}. Consequently, the lemma applies directly, yielding the minimum-variance influence function within the chosen affine class. This recovers the optimal weighting strategy of \citet{guo2023flexible} as a direct consequence of the indexed invariance characterization. 

The orthocomplement directions in \eqref{eq:frontdoor_indexed_orthocomplement}, derived from the indexed invariance restriction in \eqref{eq:invariance_frontdoor}, may also be recovered directly from the general orthocomplement characterization. A small distinction between the canonical mean-scale representation and the reduced representation in Corollary~\ref{cor:tangent_space_frontdoor} is important here. Under the generic specialization $X_{\mathrm{gen}}=Z, Y_{\mathrm{gen}}=Y, Z_{\mathrm{gen}}=(X,M)$, the canonical mean-scale representation fixes the outcome-side feature to be $h(Y)=Y$ while leaving the first indexing function unrestricted over the variable of interest and the full conditioning set. Thus, it would be written using $f=f(X,Z,M)$ and $g(Y)=Y$. In contrast, Corollary~\ref{cor:tangent_space_frontdoor} uses the reduced representation in which the $M$-dependence of $f(X,Z,M)$ has been transferred to the unrestricted outcome-side indexing function, yielding $f=f(X,Z)$ and $g=g(X,M,Y)$. After imposing the mean-scale restriction, this transferred $M$-dependence must be retained as a multiplicative localization of the feature $Y$. Accordingly, the reduced representation of a mean-scale direction takes the form $f=f(X,Z)$ and $g(X,M,Y)=q(X,M)Y$, for an appropriate function $q$. Equivalently, the same direction could be represented canonically using an unrestricted $f(X,Z,M)$ and $g(Y)=Y$. The following lemma uses the reduced representation because it aligns directly with Corollary~\ref{cor:tangent_space_frontdoor}.

\begin{lemma}[Embedding of the indexed influence-function class in the extended front-door model] \label{lem:frontdoor_representation}
Let $\rho(m,x) = \int p(m\mid A=a_0,z,x) p(z\mid x)\diff z$. For every square-integrable function $\alpha$ satisfying $\int\alpha(z)\diff z=0$, choose
\begin{align*}
f_\alpha(X,Z)
=
\frac{\alpha(Z)}{p(Z\mid X)},
\qquad
g_{\rho}(X,M,Y)
=
\rho(M,X)Y.
\end{align*}
Then $\chi_{f_\alpha,g_\rho}(O) \in \Lambda_{\bbP}^{\perp}$, where $\Lambda_{\bbP}^{\perp}$ is given by Corollary~\ref{cor:tangent_space_frontdoor}, and $\chi_{f_\alpha,g_\rho}(O) = \int \varphi_{a_0}^{\mathrm{np}}(\bbP; z^*)(O) \, \alpha(z^*) \,\diff z^*$.  
\end{lemma}
See Appendix~\ref{app:lem:frontdoor_representation} for a proof.

Lemma~\ref{lem:frontdoor_representation} gives an explicit realization of the abstract inclusion $\Lambda_{\bbP,\psi}^{\perp}\subseteq\Lambda_{\bbP}^{\perp}$ established in Section~\ref{subsec:mean_scale_verma}. The factor $\rho(M,X)$ in $g_{\rho}(X,M,Y)=\rho(M,X)Y$ is not part of the outcome feature defining the mean-scale restriction, which remains $h(Y)=Y$. Rather, it is the multiplicative localization required by the reduced representation in Corollary~\ref{cor:tangent_space_frontdoor}, where the first indexing function has been restricted from $f(X,Z,M)$ to $f(X,Z)$. Thus, the target-specific directions generated by the indexed identification functionals are recovered as particular weighted residual directions in the full nested Markov orthocomplement.

\subsection{The Napkin graph model}
\label{sec:napkin_graph}

We next consider the Napkin graph, shown in Figure~\ref{fig:napkin graph}(a), another canonical example of a causal model with a Verma constraint \citep{pearl2018book, guo2025causal}. The Napkin graph admits a \emph{no-direct-effect} assumption $Y(z,a) = Y(a),\, \forall z,a$, which states that $Z$ affects $Y$ only through $A$.
This assumption induces the nested conditional independence
\begin{align}
Y \indep Z \mid A
\qquad
[\phi_Z(p)],
\label{eq:napkin_nested_constraint}
\end{align}
as illustrated via the CADMG in Figure~\ref{fig:napkin graph}(b). This is a direct specialization of the general framework in Section~\ref{sec:geometry}, obtained by identifying
$X_{\mathrm{gen}}=Z$,
$Y_{\mathrm{gen}}=Y$,
$Z_{\mathrm{gen}}=A$, and
$R=Z$.

\begin{figure}[t]
\centering
\begin{subfigure}[t]{0.45\textwidth}
\centering
\begin{tikzpicture}[>=stealth, node distance = 1.5cm]
    \tikzstyle{format} = [thick, circle, minimum size = 1.0mm, inner sep = 2pt]
    \tikzstyle{square} = [draw, thick, minimum size = 4.5mm, inner sep = 2pt]

    \begin{scope}[xshift = 0cm, yshift = 0cm]
    \path[->, very thick]
    node[shape = ellipse, draw = black] (z) {$Z$}
    node[right of = z, shape = ellipse, draw = black, xshift=0.4cm] (a) {$A$}
    node[above right of = a, shape = ellipse, draw = black] (w) {$W$}
    node[below right of = w, shape = ellipse, draw = black] (y) {$Y$}

    (z) edge[blue] (a)
    (a) edge[blue] (y)
    (w) edge[blue] (z)
    (w) edge[red, <->] (a)
    (w) edge[red, <->] (y)
    ;
    \end{scope}
\end{tikzpicture}
\caption{}
\label{fig:napkin graph original}
\end{subfigure}
\begin{subfigure}[t]{0.45\textwidth}
\centering
\begin{tikzpicture}[>=stealth, node distance = 1.5cm]
    \tikzstyle{format} = [thick, circle, minimum size = 1.0mm, inner sep = 2pt]
    \tikzstyle{square} = [draw, thick, minimum size = 4.5mm, inner sep = 2pt]

    \begin{scope}[xshift = 0cm, yshift = 0cm]
    \path[->, very thick]
    node[shape = rectangle, draw = black] (z) {$Z$}
    node[right of = z, shape = ellipse, draw = black, xshift=0.4cm] (a) {$A$}
    node[above right of = a, shape = ellipse, draw = black] (w) {$W$}
    node[below right of = w, shape = ellipse, draw = black] (y) {$Y$}

    (z) edge[blue] (a)
    (a) edge[blue] (y)
    (w) edge[red, <->] (a)
    (w) edge[red, <->] (y)
    ;
    \end{scope}
\end{tikzpicture}
\caption{}
\label{fig:napkin graph fixed}
\end{subfigure}
\caption{The Napkin graph. (a) Original ADMG $\calG$. (b) CADMG obtained after fixing $Z$, under which the nested conditional independence
$Y \indep Z \mid A \ [\phi_Z(p)]$ becomes an ordinary conditional independence.}
\label{fig:napkin graph}
\end{figure}
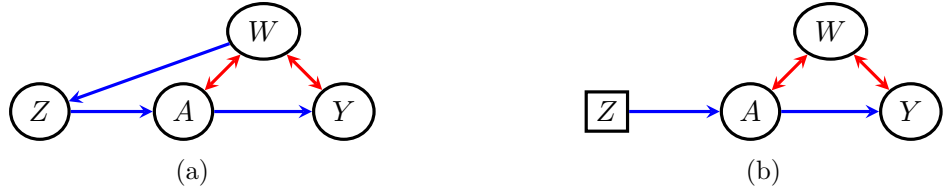

\begin{corollary}\label{cor:napkin_orthocomplement}
The nested Markov model associated with the Napkin graph in Figure~\ref{fig:napkin graph}(a) is
\begin{align}
\calP(\calG)
=
\Big\{
\bbP:
\bbE
\left[
\omega_Z(\M)
\left\{
f(Z,A)
-
\bbE_{\phi_Z}[f(Z,A)\mid A]
\right\}
g(Y, A)
\right]
=
0
\Big\},
\label{eq:napkin_model}
\end{align}
for every
$f\in L^2(\bbP_{Z,A})$
and
$g\in L^2(\bbP_{Y, A})$, where $\M = (Z, W)$ and 
\begin{align*}
\omega_Z(\M)
=
\frac{\widetilde p(Z)}{p(Z\mid W)},
\qquad
\bbE_{\phi_Z}[h(O)\mid A]
=
\frac{
\bbE
\left[
\omega_Z(\M) \, h(O)
\mid A
\right]
}{
\bbE
\left[
\omega_Z(\M)
\mid A
\right]
}.
\end{align*}
Furthermore, the orthocomplement to the tangent space of $\calP(\calG)$ in \eqref{eq:napkin_model} at $\bbP$ is
\begin{align}
\Lambda_{\bbP}^{\perp}
=
\overline{\sp}
\left\{
\chi_{f,g}(O):
f\in L^2(\bbP_{Z,A}),
\;
g\in L^2(\bbP_{Y,A})
\right\},
\label{eq:napkin_orthocomplement}
\end{align}
where
\begin{align*}
\chi_{f,g}(O)
&=
\widetilde\chi_{f,g}(O)
-
\bbE
\left[
\widetilde\chi_{f,g}(O)
\mid W,Z
\right]
+
\bbE
\left[
\widetilde\chi_{f,g}(O)
\mid W
\right],
\\
\widetilde\chi_{f,g}(O)
&=
\omega_Z(\M)
\Big\{
f(Z,A)-\bbE_{\phi_Z}[f(Z,A)\mid A]
\Big\}
\Big\{
g(Y,A)-\bbE_{\phi_Z}[g(Y,A)\mid A]
\Big\}.
\end{align*}
\end{corollary}
See a proof in Appendix~\ref{app:cor:napkin_orthocomplement}.

The efficiency analysis of the causal mean $\bbE[Y(a_0)]$ proceeds as in the previous example. Let $\varphi_{a_0}^{\mathrm{np}}(\bbP;z^*)$ denote the indexed nonparametric influence function corresponding to the identification functional $\psi_{a_0}(\bbP;z^*)$, given in Appendix~\ref{app:id_npEIF_napkin}. The semiparametric efficient influence function is obtained by projecting any indexed nonparametric influence function onto the tangent space, $\varphi_{a_0}^{\mathrm{eff}} = \varphi_{a_0}^{\mathrm{np}}(\bbP;z^*) - \Pi\big( \varphi_{a_0}^{\mathrm{np}}(\bbP;z^*) \mid \Lambda_{\bbP}^{\perp} \big)$, where $\Lambda_{\bbP}^{\perp}$ is characterized in Corollary~\ref{cor:napkin_orthocomplement}. A finite-dimensional approximation is then obtained directly from Theorem~\ref{theorem:locally_eff} using the orthocomplement basis in Corollary~\ref{cor:napkin_orthocomplement}. Thus, as in the previous example, the only model-specific ingredient is the collection of orthocomplement directions; the projection argument itself is entirely generic. 

Likewise, Theorem~\ref{thm:indexed_invariance_orthocomplement} implies that the target-specific orthocomplement and corresponding affine class of influence functions are generated by affine combinations of the indexed nonparametric influence functions. Consequently, Lemma~\ref{lem:optimal_affine_IF_combination} recovers the optimal weighting strategy of \citet{guo2025causal} directly from the indexed invariance characterization. 

\begin{lemma}[Embedding of the indexed influence-function class in the Napkin graph] \label{lem:napkin_indexed}

For every square-integrable function $\alpha$ satisfying $\int \alpha(z) \,\diff z =0$, choose 
\begin{align*}
f_\alpha(Z,A) 
= \frac{\I(A=a_0)}{\int p(a_0\mid Z,w) \, p(w) \, \diff w} \frac{\alpha(Z)}{\tilde p(Z)}, 
\qquad
g(Y,A) = Y. 
\end{align*}
Then $\chi_{f_\alpha,g}(O) \in \Lambda_{\bbP}^{\perp}$, where $\Lambda_{\bbP}^{\perp}$ is given by Corollary~\ref{cor:napkin_orthocomplement} and $\chi_{f_\alpha,g}(O) = \int \varphi_{a_0}^{\mathrm{np}}(\bbP; z^*)(O) \, \alpha(z^*) \,\diff z^*$.  
\end{lemma}
See a proof in Appendix~\ref{app:lem:napkin_indexed}. 

Lemma~\ref{lem:napkin_indexed} likewise provides an explicit realization of the abstract inclusion $\Lambda_{\bbP,\psi}^{\perp}\subseteq\Lambda_{\bbP}^{\perp}$. In particular, the indexed influence-function directions considered by \citet{guo2025causal} are recovered as a specialization of the general weighted residual directions through an explicit choice of the functions $f$ and $g$. These directions span only a proper subspace of the full orthocomplement characterized in Corollary~\ref{cor:napkin_orthocomplement}. Consequently, the optimal weighted estimator of \citet{guo2025causal} exploits only part of the efficiency gain available under the full distributional Verma constraint.

\subsection{Examples with multiple nested Markov constraints} 
\label{sec:ex:muiltiple_constraints}

Two additional examples illustrating the multiple-constraint construction of Proposition~\ref{thm:multiple_constraints} are provided in the supplementary materials. Appendix~\ref{app:ex:verma_ordinary} considers an extended front-door graph that simultaneously imposes the ordinary conditional independence $Z\indep M\mid X,A$ and the Verma constraint $Z\indep Y\mid X,M\ [\phi_M(p)]$, whereas Appendix~\ref{app:ex:sequential} considers a sequential graph with two Verma constraints. These examples demonstrate how constraint-specific orthocomplement directions are combined and why efficiency calculations generally require a joint projection that retains cross-covariances between the direction families.

\section{Simulation Studies}
\label{sec:sims}

The preceding sections characterize the semiparametric geometry induced by the indexed identifying restrictions and the Verma constraints, derive the corresponding efficient influence functions, and motivate the proposed one-step estimators. The simulation studies examine the practical implications of these theoretical results from three complementary perspectives. Experiment~1 evaluates the finite-sample performance of the proposed estimators under a fixed data-generating mechanism (DGP). Experiment~2 investigates how the relative asymptotic efficiencies of the competing estimators vary across a broad collection of DGPs. Experiment~3 examines the finite-dimensional approximation to the Verma orthocomplement by progressively enriching the projection basis.

The main text reports results for estimation of the treatment-specific mean $\psi_0(\bbP)=\bbE[Y(0)]$ under the extended front-door model of Section~\ref{sec:front-door}. Parallel experiments were conducted for $\psi_1(\bbP) = \bbE[Y(1)]$, the average treatment effect (ATE) $\tau(\bbP)=\psi_1(\bbP)-\psi_0(\bbP)$, and under the Napkin model of Section~\ref{sec:napkin_graph}; complete results are provided in Appendix~\ref{app:sims}. Code for reproducing all simulation experiments
is available on GitHub at \href{https://github.com/raziehna/Verma-efficiency}{raziehna/Verma-efficiency}.

For the extended front-door simulations, variables are generated from a latent-variable model compatible with Figure~\ref{fig:front-door}(a), with latent variable $U\sim\operatorname{Uniform}(0,1)$. The observed data are $O=(X,Z,A,M,Y)$, where we assume $X,Z,A,$ and $M$ are binary and $Y$ is continuous. The treatment mechanism includes a tunable $Z$-by-$X$ interaction $\eta_{\mathrm{int}}$, which changes the relative variances and covariance of the two fixed-$Z$ indexed influence functions and hence the potential benefit of optimal indexed weighting. The remaining treatment, mediator, and covariate models and all fixed coefficients are given in Appendix~\ref{app:sims:frontdoor}. To vary interpretable features of the outcome distribution, we define the treatment-specific mediator profile $q_a(m,x) = \sum_{z=0}^1 p(m\mid A=a,Z=z,X=x)p(z\mid x)$, for $a\in\{0,1\}$. For target treatment level $a_T$, let $a_C=1-a_T$ and write $q_T=q_{a_T}$ and $q_C=q_{a_C}$. Conditional on $(U,M,X)$, the outcome is generated as
\begin{align*}
Y
=
\mu(U,M,X)
+
\sigma(U,M,X) \, \varepsilon(U,M,X),
\end{align*}
where $\bbE[Y\mid U,M,X] = \mu(U,M,X)$, $\Var[Y\mid U,M,X] = \sigma^2(U,M,X)$, and the error $\varepsilon(U,M,X)$ is standardized to have conditional mean zero and conditional variance one, while its conditional third moment is allowed to vary through the loading $\lambda(U,M,X)$. We specify
\begin{align*}
\mu(U,M,X)
&=
\mu_0(U,M,X)
+
\left\{
\eta_{\mu,T}q_T(M,X)
+
\eta_{\mu,C}q_C(M,X)
\right\}g_\mu(U),
\\
\log \sigma(U,M,X)
&=
s_0(U)
+
\left\{
\eta_{\sigma,T}q_T(M,X)
+
\eta_{\sigma,C}q_C(M,X)
\right\}g_\sigma(U),
\\
\lambda(U,M,X)
&=
\left\{
\eta_{\kappa,T}q_T(M,X)
+
\eta_{\kappa,C}q_C(M,X)
\right\}g_\lambda(U).
\end{align*}
Thus, $(\eta_{\mu,T},\eta_{\mu,C})$ controls perturbations of the conditional mean, $(\eta_{\sigma,T},\eta_{\sigma,C})$ controls perturbations of the conditional scale, and $(\eta_{\kappa,T},\eta_{\kappa,C})$ controls the conditional third moment. The perturbations are weighted by the target and contrast mediator profiles, allowing these features of the outcome distribution to differ according to the representation of each $(M,X)$ stratum under the two treatment interventions. The functions appearing above, the standardized error distribution, and all parameter values used in the three experiments are specified in Appendix~\ref{app:sims:frontdoor}.

\subsection{Experiment 1: Finite-sample performance}
\label{sec:sim-exp1}

We compare five one-step estimators corresponding to progressively richer uses of the identifying restrictions: the two fixed-$Z$ indexed estimators, their equally weighted and optimally weighted affine combinations, and the Verma projection estimator. The first four estimators exploit only the indexed invariance of the identifying functional, whereas the Verma projection estimator additionally incorporates directions from a finite-dimensional approximation to the Verma orthocomplement developed in Section~\ref{sec:geometry}. To assess their finite-sample performance, we fix the treatment and outcome tuning parameters at a representative DGP and compare the empirical $n$-scaled variances of the resulting estimators with their corresponding theoretical asymptotic variances over an increasing sequence of sample sizes. The theoretical variances are computed from the corresponding population influence functions under the same DGP. Complete simulation details are provided in Appendix~\ref{app:sims:frontdoor}.

Figure~\ref{fig:frontdoor_var_mean_0_n_scaled_variance} compares the empirical $n$-scaled variances with their theoretical counterparts (dashed horizontal lines) for estimation of $\psi_0(\bbP)$. Across the five estimators, the empirical variances approach the corresponding theoretical asymptotic variances as the sample size increases, and the ordering predicted by the population calculations is reflected in finite samples. Table~\ref{tab:frontdoor_var_mean_0_efficiency} reports the corresponding theoretical asymptotic variances and pairwise relative efficiency comparisons. We define the relative efficiency gain of method $2$ over method $1$ as $G_{\mathrm{RE}}(1,2) = 100 (V_1/V_2 - 1)$, where $V_j$ denotes the theoretical asymptotic variance of method $j$. Thus, a positive value indicates greater efficiency of method $2$. We use this definition consistently throughout the simulation study. 

\begin{figure}[t]
\centering
\includegraphics[scale=0.7]{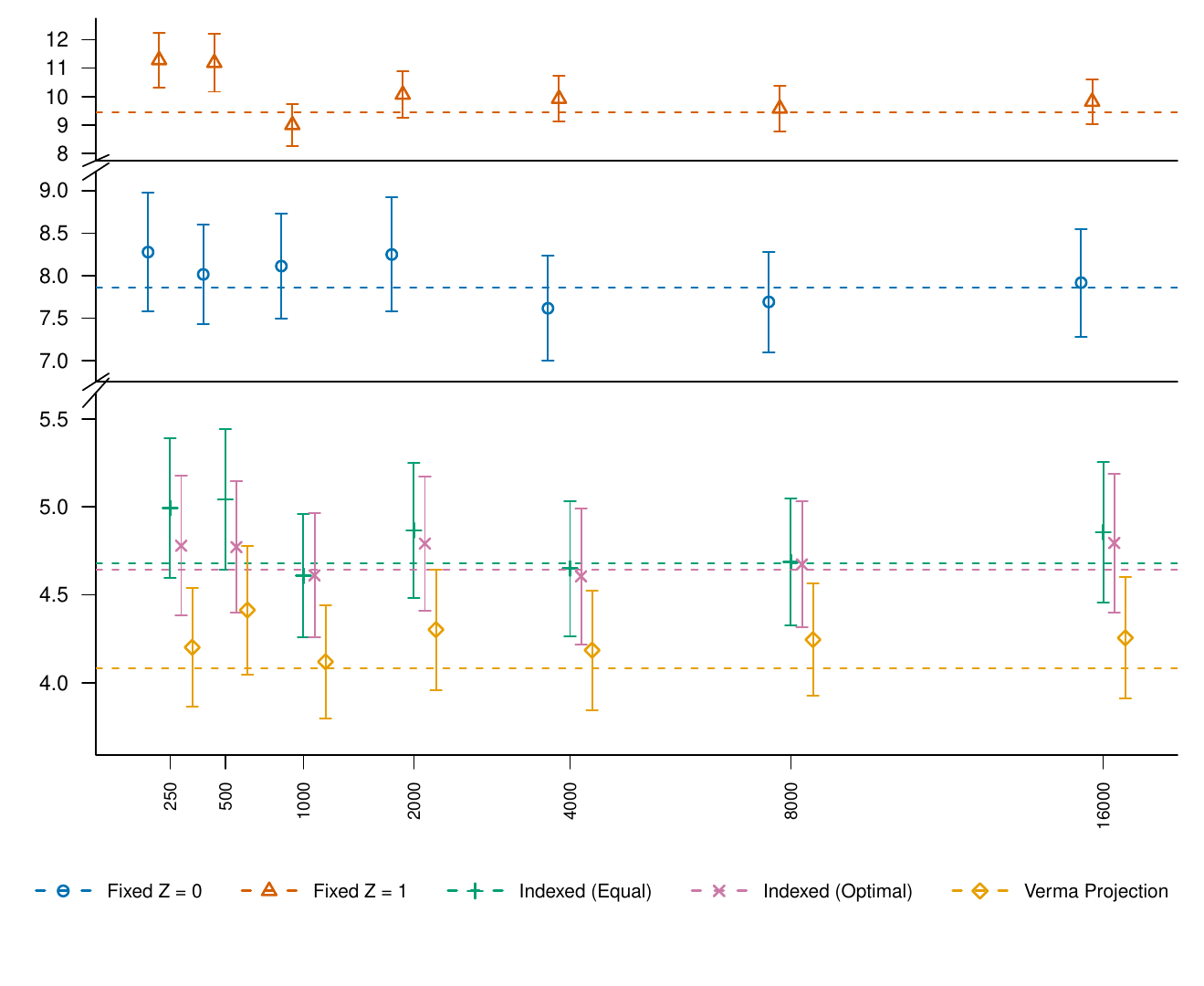}
\vspace{-1cm}
\caption{Finite-sample performance for $\psi_0(\bbP)=\bbE[Y(0)]$ under the extended front-door model. Points and vertical intervals show the empirical $n$-scaled variance and its Monte Carlo confidence interval, respectively; horizontal dashed lines show the corresponding theoretical asymptotic variances. Results are shown for the two fixed-$Z$ indexed estimators, the equally weighted indexed estimator, the optimally weighted indexed estimator, and the Verma projection estimator.}
\label{fig:frontdoor_var_mean_0_n_scaled_variance}
\end{figure}

\begin{table}[htbp]
\centering
\small
\caption{Theoretical variances and relative efficiency gains for $\psi_0(\bbP) = \mathbb{E}[Y(0)]$ under the extended front-door model. Relative efficiency gains are computed with respect to the reference method indicated by each column. }
\label{tab:frontdoor_var_mean_0_efficiency}
\begin{tabular}{lccccc}
\toprule
Method & \shortstack{Theoretical\\variance} & \shortstack{Gain vs.\\fixed $Z=0$} & \shortstack{Gain vs.\\fixed $Z=1$} & \shortstack{Gain vs.\\equal indexed} & \shortstack{Gain vs.\\optimal indexed} \\
\midrule
Fixed $Z=0$ & 7.861 & 0.000\% & 20.177\% & -40.453\% & -40.957\% \\
Fixed $Z=1$ & 9.447 & -16.789\% & 0.000\% & -50.451\% & -50.870\% \\
Indexed (equal) & 4.681 & 67.936\% & 101.820\% & 0.000\% & -0.845\% \\
Indexed (optimal) & 4.641 & 69.368\% & 103.541\% & 0.853\% & 0.000\% \\
Verma projection & 4.082 & 92.566\% & 131.420\% & 14.666\% & 13.697\% \\
\bottomrule
\end{tabular}
\end{table}

The results illustrate that combining the two fixed-$Z$ indexed estimators substantially reduces variance relative to either indexed estimator alone, and optimizing their affine combination yields a further modest improvement. The Verma projection estimator has the smallest theoretical asymptotic variance. In particular, its relative efficiency gain is roughly $14.7\%$ over the equally weighted indexed estimator and $13.7\%$ over the optimally weighted indexed estimator. The latter comparison is especially informative: even after extracting the maximal efficiency gain available within the affine class of indexed estimators, incorporating additional directions from the Verma orthocomplement yields a further substantial improvement in efficiency.

The finite-sample results for $\psi_1(\bbP)$ and the ATE under the extended front-door model are reported in Appendix~\ref{app:sims:frontdoor} (Figures~\ref{fig:frontdoor_var_ate_n_scaled_variance} and \ref{fig:frontdoor_var_mean_1_n_scaled_variance}; Table~\ref{tab:frontdoor_var_all_estimands}). The corresponding results under the Napkin model are reported in Appendix~\ref{app:sims:napkin} (Figures~\ref{fig:napkin_var_ate_n_scaled_variance}--\ref{fig:napkin_var_mean_0_n_scaled_variance}; Table~\ref{tab:napkin_var_all_estimands})).

\subsection{Experiment 2: Efficiency across data-generating mechanisms}
\label{sec:sim-exp2}

We study how the relative asymptotic variances of the competing estimators change across the DGP family described above. We vary $\eta_{\mathrm{int}}$ together with the six outcome tuning parameters $(\eta_{\mu,T},\eta_{\mu,C})$, $(\eta_{\sigma,T},\eta_{\sigma,C})$, and $(\eta_{\kappa,T},\eta_{\kappa,C})$, while holding the remaining coefficients fixed. The resulting grid therefore varies the treatment mechanism together with treatment-specific features of the conditional mean, variance, and skewness of the outcome distribution. For each of the resulting $2,187$ candidate DGPs, we compute the theoretical asymptotic variance of all five estimators and rank DGPs by the relative efficiency gain of the Verma projection estimator over the optimally weighted indexed estimator, $G_{\mathrm{RE}}(\mathrm{indexed,opt},\mathrm{Verma})$, using the gain over the equally weighted indexed estimator to break near ties. Details are provided in Appendix~\ref{app:sims:frontdoor}.

Figure~\ref{fig:frontdoor_var_grid_mean_0_top_20_method_variances_split_x} reports the theoretical asymptotic variances for the $20$ highest-ranked DGPs for estimation of $\psi_0(\bbP)$. Each row represents one DGP, and the horizontal axis is partitioned into three noncontiguous ranges to facilitate visual comparison across methods. Across all selected configurations, combining the two indexed estimators produces a substantial reduction in variance relative to either fixed-$Z$ estimator. The optimally weighted indexed estimator further improves on the equally weighted estimator, although the two variances remain close throughout the selected DGPs. In contrast, the Verma projection estimator consistently achieves the smallest theoretical asymptotic variance, yielding relative efficiency gains of roughly $14\%$ over the optimally weighted indexed estimator. The corresponding parameter configurations are reported in Appendix Table~\ref{tab:frontdoor_var_grid_all_estimands}.

\begin{figure}[t]
\centering
\includegraphics[scale=0.54]{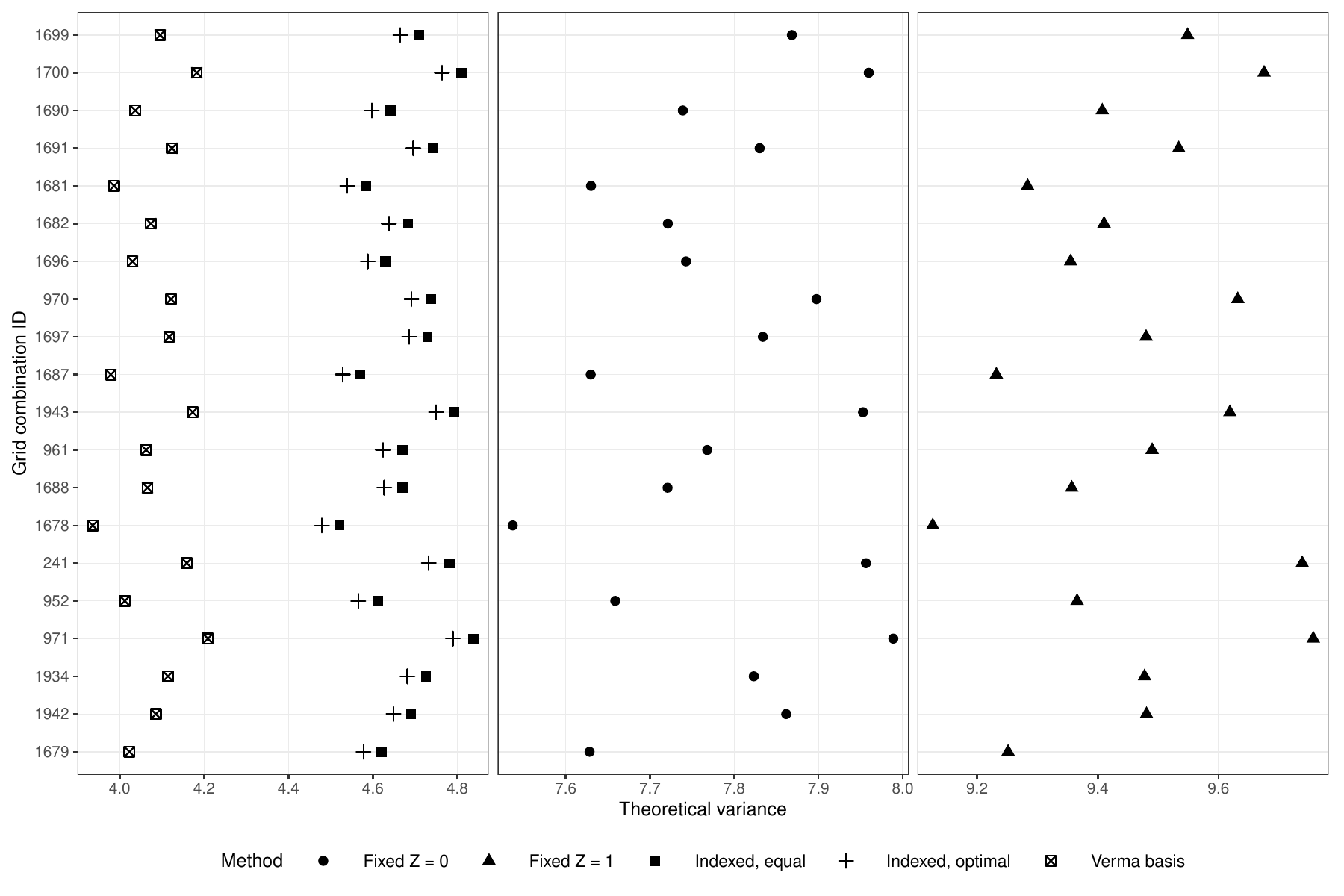}
\caption{
Theoretical asymptotic variances for  $\psi_0(\bbP)=\bbE[Y(0)]$ under the extended front-door model across the 20 highest-ranked DGPs in Experiment~2. Each row corresponds to one DGP, and symbols denote the five competing estimators. The horizontal axis is partitioned into three noncontiguous ranges to facilitate visual comparison across methods. DGPs are ranked by the relative efficiency gain of the Verma projection estimator over the optimally weighted indexed estimator. 
}
\label{fig:frontdoor_var_grid_mean_0_top_20_method_variances_split_x}
\end{figure}

The grid-search results for $\psi_1(\bbP)$ and the ATE under the extended front-door model are reported in Appendix~\ref{app:sims:frontdoor} (Figures~\ref{fig:frontdoor_var_grid_ate_top_20_method_variances_split_x} and \ref{fig:frontdoor_var_grid_mean_1_top_20_method_variances_split_x}; Table~\ref{tab:frontdoor_var_grid_all_estimands}). The corresponding results under the Napkin model are reported in Appendix~\ref{app:sims:napkin} (Figures~\ref{fig:napkin_var_grid_ate_top_20_method_variances_split_x}--\ref{fig:napkin_var_grid_mean_0_top_20_method_variances_split_x}; Table~\ref{tab:napkin_var_grid_all_estimands}).

\subsection{Experiment 3: Effect of basis richness}
\label{sec:sim-exp3}

The proposed Verma projection estimator approximates projection onto the infinite-dimensional Verma orthocomplement by projecting onto a finite-dimensional subspace generated by prespecified basis functions. Here, we investigate how the theoretical asymptotic variance of the resulting estimator changes as this approximation space is enlarged. We consider the nested outcome bases $\mathcal B_1 = \{Y\}$, $\mathcal B_2 = \{Y,Y^2\}$, $\mathcal B_3 = \{Y,Y^2,Y^3\}$. The remaining components of the orthocomplement basis are held fixed, so the corresponding projection spaces are nested. Consequently, enlarging the basis cannot increase the theoretical asymptotic variance of the resulting projection estimator. The substantive question is therefore how much additional variance reduction is obtained by enriching the basis with quadratic and cubic outcome functions, and how this incremental improvement depends on the underlying DGP. We consider a DGP grid that holds the treatment mechanism fixed while varying the six outcome tuning parameters $(\eta_{\mu,T},\eta_{\mu,C})$, $(\eta_{\sigma,T},\eta_{\sigma,C})$, $(\eta_{\kappa,T},\eta_{\kappa,C})$, over ranges allowing stronger conditional variance and skewness perturbations than in Experiment~2. For each of the resulting $1,296$ candidate DGPs, we compute the theoretical asymptotic variance of the optimally weighted indexed estimator and of the three Verma projection estimators corresponding to $\mathcal B_1$, $\mathcal B_2$, and $\mathcal B_3$. The exact grid and ranking criterion are given in Appendix~\ref{app:sims:frontdoor}. 

Figure~\ref{fig:frontdoor_basis_grid_mean_0_top_20_method_variances_split_x} displays the theoretical asymptotic variances for the $20$ highest-ranked DGPs. As implied by the nested projection spaces, the theoretical variance is nonincreasing as the basis is enriched. The largest gains occur for DGPs exhibiting strong treatment-specific perturbations of higher-order outcome features. For $\psi_0(\bbP)$, Appendix Table~\ref{tab:frontdoor_basis_grid_all_estimands} shows that the highest-ranked configurations all correspond to the largest value of the target-arm variance perturbation $\eta_{\sigma,T}$, with only negligible improvement from adding $Y^2$ alone and most of the additional efficiency arising after including the cubic basis function $Y^3$. 

\begin{figure}[!t]
\centering
\includegraphics[scale=0.55]
{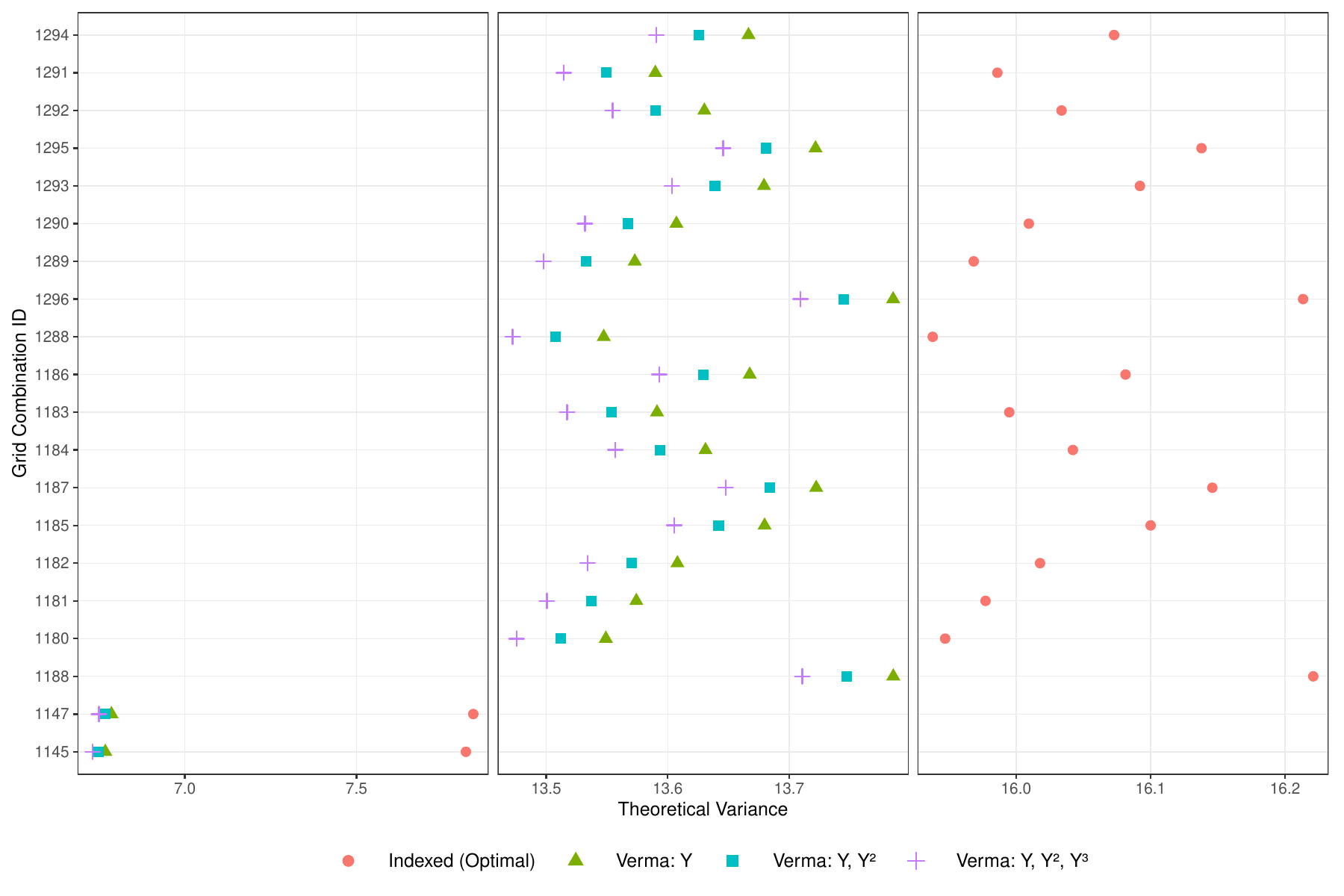}
\caption{
Effect of Verma-basis richness on the theoretical asymptotic variance for $\psi_0(\bbP)=\bbE[Y(0)]$ under the extended front-door model. Results are shown for the 20 highest-ranked DGPs in Experiment~3. The Verma projection estimators use the nested outcome bases $\{Y\}$, $\{Y,Y^2\}$, and $\{Y,Y^2,Y^3\}$; the optimally weighted indexed estimator is included as a reference. DGPs are ranked by the relative efficiency gain of the full Verma basis $\{Y,Y^2,Y^3\}$ over the $Y$-only basis. 
}
\label{fig:frontdoor_basis_grid_mean_0_top_20_method_variances_split_x}
\end{figure}

The basis-richness results for $\psi_1(\bbP)$ and the ATE under the extended front-door model are reported in Appendix~\ref{app:sims:frontdoor} (Figures~\ref{fig:frontdoor_basis_grid_ate_top_20_method_variances_split_x} and~\ref{fig:frontdoor_basis_grid_mean_1_top_20_method_variances_split_x}; Table~\ref{tab:frontdoor_basis_grid_all_estimands}). The corresponding results under the Napkin model are reported in Appendix~\ref{app:sims:napkin} (Figures~\ref{fig:napkin_basis_grid_ate_top_20_method_variances_split_x}--\ref{fig:napkin_basis_grid_mean_0_top_20_method_variances_split_x}; Table~\ref{tab:napkin_basis_grid_all_estimands}).

To complement the simulation studies, we also evaluated the proposed estimators in \textbf{two real-data applications}, with full details reported in Appendix~\ref{app:realdata} of the supplementary materials. We revisit the Framingham Heart Study application considered by \citet{bhattacharya2022semiparametric} and the Finnish Life Course application considered by \citet{guo2025causal}, examining the additional precision that can be obtained by exploiting the corresponding Verma constraints. We apply the extended front-door estimators to the Framingham data and the Napkin-model estimators to the Finnish Life Course data. Across these applications, the Verma projection estimators generally yield shorter confidence intervals than the equally weighted indexed estimator, illustrating the potential precision gains from exploiting Verma constraints in empirical settings.

\section{Concluding Remarks}
\label{sec:conclusion}


This paper progresses toward a general semiparametric efficiency theory for statistical models defined by Verma (nested Markov) constraints. These constraints admit an equivalent representation through weighted conditional moment restrictions \citep{ai2003efficient, ai2012semiparametric, chen2026local} under the post-fixing law, which can be leveraged to characterize the tangent space and its orthocomplement. With and only with a single Verma constraint, this characterization yields a unified geometric description of efficient influence functions: it is obtained by projecting the nonparametric influence function onto the tangent space, or equivalently by subtracting its projection onto the Verma-induced orthocomplement. The orthocomplement is precisely characterized by the random variables orthogonal to the scores, when differentiating the Verma constraints just as the calculus of influence functions \citep{hines2022demystifying, kennedy2024semiparametric}.

Beyond its theoretical implications, the proposed framework provides a practical route to construct more efficient estimators. By approximating the infinite-dimensional orthocomplement with finite-dimensional basis expansions, we obtain projection estimators that are straightforward to compute, preserve local efficiency within the chosen approximation space, and converge monotonically toward more improved efficiency as the basis is enriched. Simulation studies demonstrate that these estimators can achieve substantial efficiency gains over estimators that ignore the additional information encoded by Verma constraints. An important conjecture discussed in Remark~\ref{rem:conjecture} remains to be solved, before closing the broad question of characterizing the actual tangent space of nested Markov models with multiple equality constraints. In specific examples, our current strategy is to exhibit concrete parametric submodels. A possible direction is to adopt the strategy of \citet{guo2025causal}, in which they employed the implicit function theorem to prove the existence of parametric submodels, rather than explicitly constructing one. Results in a similar spirit can be found in \citet{evans2016graphs}, but they only consider vertices with finite state spaces. This more abstract approach can be more useful when multiple constraints exist. We leave this problem to future work.



\section*{Acknowledgments}

This work was supported in part by the National Science Foundation under Grant No.~2502299. R.N. and L.L. also thank the Isaac Newton Institute for Mathematical Sciences, Cambridge, for support and hospitality during the programme \emph{Causal inference: From theory to practice and back again}, which facilitated closer collaboration and substantial progress on this work. 
The authors used large language models solely to assist with computational verification, code review, and language editing. All conceptual and theoretical contributions are the authors' own.


\nocite{kannel1968framingham}

\bibliographystyle{apalike}   

\bibliography{Master}


\pagebreak
\appendix
\phantomsection
\setcounter{page}{1}

\setcounter{section}{0}
\renewcommand{\thesection}{S\arabic{section}}
\renewcommand{\thesubsection}{S\arabic{section}.\arabic{subsection}}
\renewcommand{\thesubsubsection}{S\arabic{section}.\arabic{subsection}.\arabic{subsubsection}}

\setcounter{figure}{0}
\renewcommand{\thefigure}{S\arabic{figure}}

\setcounter{table}{0}
\renewcommand{\thetable}{S\arabic{table}}

\setcounter{equation}{0}
\renewcommand{\theequation}{S\arabic{equation}}
\renewcommand{\theHequation}{S\arabic{equation}}

\begin{center}
\noindent {\bf \LARGE Supplementary Materials for ``\mytitle{}''}
\end{center}

\vspace{0.25cm}
\startcontents[supp]
{\small 
\renewcommand{\baselinestretch}{0.8}\selectfont
\printcontents[supp]{}{1}{}
}

\onehalfspacing

\clearpage
\section{Overview of Nested Markov Models} 
\label{app:NMM}

\subsection*{CADMGs, kernels, and fixing operations}

The nested Markov factorization of $\p(V)$ associated with an ADMG $\calG(V)$ is defined using conditional ADMGs (CADMGs) derived from $\calG(V)$ and kernel objects derived from $p (V)$ via the \textit{fixing} operation \citep{richardson2023nested}. We define the relevant concepts below.

\emph{Conditional ADMG (CADMG).} A CADMG $\calG(V;W)$ is an ADMG whose vertices can be partitioned into random variables $V$ and fixed/intervened variables $W.$ Only outgoing directed edges may be adjacent to variables in $W$. For any random variable $V_i \in V$ in a CADMG $\calG(V;W)$, the usual definitions of genealogical relations, such as parents and descendants, extend naturally by allowing for the inclusion of fixed variables into these sets. However, bidirected connected components a.k.a districts of a CADMG $\calG(V;W)$ are only defined for elements of $V$.

\emph{Kernel.} A kernel  $q_{V \mid W} (V \mid W)$ is a mapping from values in $W$ to normalized densities over $V$. Given a kernel $q_{V \mid W} (V \mid W)$ and a subset $X \subseteq V$,  conditioning and marginalization are defined in the usual way as $q_{V \mid W} (X \mid W) \equiv \sum_{V \setminus X} q_{V \mid W} (V \mid W)$ and $q_{V \mid W}(V \setminus X \mid X, W) \equiv q_{V \mid W} (V \mid W) / q_{V \mid W} (X \mid W)$. 

\emph{Fixing operation for a single variable.} Fixability of a single variable and the graphical and probabilistic operations of fixing it are defined as follows: 
\begin{itemize}
	\setlength{\itemsep}{0.25cm} 
	
	\item \emph{Fixability.}  A vertex $V_i \in V$ is said to be \emph{fixable} in $\calG(V;W)$ if  $V_i \rightarrow \ldots \rightarrow V_j$ and $V_i \leftrightarrow \ldots \leftrightarrow V_j$ do not both exist in $\calG$, for any $V_j \in V \setminus V_i$, or $\dis_{\calG} (V_i) \cap \de_{\calG} (V_i) = \emptyset$, where $\dis_{\calG} (\cdot)$ and $\de_{\calG} (\cdot)$ denote the district and descendant with respect to $\calG$. 
	
	\item \emph{Graphical operation of fixing.} 
	The graphical operation of fixing $V_i,$ denoted by $\phi_{V_i}(\calG),$ yields a new CADMG $\calG(V \setminus V_i; W \cup V_i)$ where all edges with arrowheads into $V_i$ are removed, and $V_i$ is fixed to a particular value $v_i.$  All other edges in $\calG$ are kept the same. 
	
	\item \emph{Probabilistic operation of fixing.} 
	Given a kernel $q_{V \mid W}(V\mid W)$, the associated CADMG $\calG(V;W),$ and $V_i \in V$, the corresponding probabilistic operation of fixing $V_i$, denoted by $\phi_{V_i}(q_{V \mid W};\calG),$ yields a new kernel defined as:
	\begin{align}
		\phi_{V_i}(q_{V \mid W}; \calG) \equiv q_{V\setminus V_i}(V \setminus V_i \mid W \cup V_i) \equiv \frac{q_{V \mid W}(V \mid W)}{q_{V \mid W}(V_i \mid \mb_\calG(V_i), W)},
		\label{eq:ordinary_fixing} 
	\end{align}
	where $\mb_\calG(V_i)$ denotes the Markov blanket of $V_i$, which consists of all vertices in the district of $V_i$ and the parents of the district of $V_i$ in $\calG(V, W)$, excluding $V_i$ itself.
\end{itemize}  

\emph{Fixing operation for a set of variables.} The above definitions can be extended to a set of vertices. 
\begin{itemize}
	\setlength{\itemsep}{0.25cm} 
	
	\item \emph{Group fixability.} A set $S \subseteq V$ is said to be fixable if there exists an ordering $(S_1, \dots, S_p)$ such that $S_1$ is fixable in $\calG,$ $S_2$ is fixable in $\phi_{S_1}(\calG),$ and so on. Such an ordering is said to form a valid fixing sequence for $S$ and we denote it by $\sigma_S$.  
	\item \emph{Graphical operation of group fixing.} It is known that applying any two valid fixing sequences on $S$ yield the same CADMG, so we denote this by $\phi_S(\calG(V;W)).$  
	\item \emph{Probabilistic operation of group fixing.} $\phi_{\sigma_{ S}}(q_{V \mid W}; \calG)$ is defined via the usual function composition to yield operators that fix all elements in ${S}$ in the order given by $\sigma_{S}$.
\end{itemize}

\emph{Intrinsic set.} A set $D$ is said to be \emph{intrinsic} in $\calG(V)$ if $V\setminus D$ is fixable in $\calG$ and $\phi_{V\setminus D}(\calG)$ contains a single district.

Finally, a distribution $p (V)$ is said to satisfy the nested Markov factorization wrt an ADMG $\calG(V)$ if there exists a set of kernels $q_{D|V\setminus D}(D \mid \pa_\calG(D))$, one for every $D$ intrinsic in $\calG(V)$, such that for every fixable set $S$ and every valid fixing sequence $\sigma_S,$
\begin{align}
	\phi_{\sigma_S}(p (V);\calG) = \prod_{D \in {\cal D}(\phi_S(\calG))} q_{D|V\setminus D}(D \mid \pa_{\calG}(D)), 
	\label{eq:nested_Markov_model}
\end{align}%
where ${\cal D}(\phi_S(\calG))$ denotes the set of all districts in the CADMG $\phi_S(\calG).$ From a causal perspective, $q_{D \mid V \setminus D} \{D \mid \pa_{\calG} (D)\}$ is identified from $p (V)$.

\subsection*{Nested Markov factorization}

The nested Markov factorization asserts that every kernel that can be derived via a valid sequence of fixing satisfies the district factorization with respect to the CADMG obtained by this sequence, and each of the kernels appearing in the factorization corresponds to intrinsic sets. 

Let ${\cal D}(\calG(V; W))$ denote the set of all bidirected connected components of random variables, commonly referred to as \emph{districts}, in the CADMG $\calG(V; W).$ The nested Markov factorization states that the observed distribution $p (V)$ satisfies the following \emph{district factorization} w.r.t. to the ADMG $\calG(V)$: 
\begin{align}
	p (V) = \prod_{D \in {\cal D(\calG)}} q_{D|V\setminus D}(D\mid \pa_\calG(D)), 
	\label{eq:fact_dist}
\end{align}%
where each kernel appearing in this factorization corresponds to intrinsic sets in $\calG(V).$  

The nested factorization further asserts that any kernel $q_{V\setminus S \mid S}(V \setminus S \mid S)$ identified from  $p (V)$ satisfies the district factorization w.r.t. to the corresponding CADMG $\calG(V\setminus S; S),$ where again each kernel in the factorization corresponds to intrinsic sets \citep{richardson2023nested}.  

\begin{example}
As an example, consider the ADMG in Fig.~\ref{fig:ex_nested_markov}(a). The kernel $q_{A, Y| Z, M}(A, Y \mid Z, M)$ is identified as $p (Z, A, M, Y) /\{p (Z)\times p(M|A, Z)\} \equiv p (A|Z) \times p (Y|Z, A, M)$. The set $\{A, Y\}$ also forms a bidirected connected component in the corresponding CADMG shown in Fig.~\ref{fig:ex_nested_markov}(c); thus, it is intrinsic. The list of all Markov kernels corresponding to intrinsic sets in Fig.~\ref{fig:ex_nested_markov}(a) is: 
\begin{equation}\label{eq:intrinsic_kernels} 
\begin{aligned}
	&q_{Z|AMY}(Z) \equiv p (Z),  \\
	&q_{A|ZMY}(A \mid Z) \equiv p (A \mid Z), \\
	&q_{M|ZAY}(M \mid A, Z) \equiv p (M \mid Z, A),  \\
	&q_{AY|ZM}(A,Y \mid Z, M) \equiv p (A \mid Z) \times p (Y \mid Z, A, M),  \\
	&q_{Y|ZAM}(Y \mid M)  \equiv \sum_A p (A \mid Z)\times p (Y \mid Z, A, M). 
\end{aligned}%
\end{equation}
\eqref{eq:fact_dist} implies: $p (V) = q_{Z|AMY}(Z) \times q_{M|ZAY}(M|A,Z) \times q_{AY|ZM}(A, Y | Z, M).$ 
\end{example}

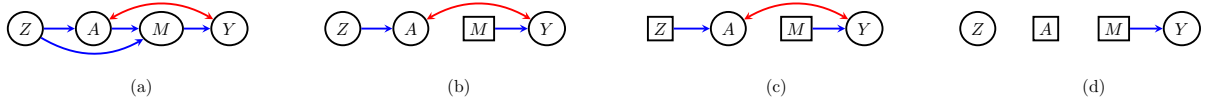
\begin{figure}[t] 
\begin{center}
\scalebox{0.6}{
\begin{tikzpicture}[>=stealth, node distance = 1.5cm]
    \tikzstyle{format} = [thick, circle, minimum size = 1.0mm, inner sep = 2pt]
    \tikzstyle{square} = [draw, thick, minimum size = 4.5mm, inner sep = 2pt]

    \begin{scope}[xshift = 0cm, yshift = 0cm]
	\path[->, very thick]
	     
        node[shape = ellipse, draw = black] (z) {$Z$}
        node[right of = z, shape = ellipse, draw = black] (a) {$A$}
        node[right of = a, shape = ellipse, draw = black] (m) {$M$}
        node[right of = m, shape = ellipse, draw = black] (y) {$Y$}
		 
        (z) edge[blue] (a)
        (a) edge[blue] (m)
        (m) edge[blue] (y)
        (z) edge[blue, bend right] (m)
        (a) edge[red, <->, bend left] (y)

        node[below right of = a, yshift = -0.25cm] () {(a)}
        ;
    \end{scope}
    \begin{scope}[xshift = 7.cm, yshift = 0cm]
	\path[->, very thick]
		
	    node[shape = ellipse, draw = black] (z) {$Z$}
        node[right of = z, shape = ellipse, draw = black] (a) {$A$}
        node[right of = a, shape = rectangle, draw = black] (m) {$M$}
        node[right of = m, shape = ellipse, draw = black] (y) {$Y$}
		 
        (z) edge[blue] (a)
        (m) edge[blue] (y)
        (a) edge[red, <->, bend left] (y)

        node[below right of = a, yshift = -0.25cm] () {(b)}
        ;
    \end{scope}

    \begin{scope}[xshift = 14cm, yshift = 0cm]
	\path[->, very thick]
		
	    node[shape = rectangle, draw = black] (z) {$Z$}
        node[right of = z, shape = ellipse, draw = black] (a) {$A$}
        node[right of = a, shape = rectangle, draw = black] (m) {$M$}
        node[right of = m, shape = ellipse, draw = black] (y) {$Y$}
		 
        (z) edge[blue] (a)
        (m) edge[blue] (y)
        (a) edge[red, <->, bend left] (y)

        node[below right of = a, yshift = -0.25cm] () {(c)}
        ;
    \end{scope}

    \begin{scope}[xshift = 21.cm, yshift = 0cm]
	\path[->, very thick]
		
	    node[shape = ellipse, draw = black] (z) {$Z$}
        node[right of = z, shape = rectangle, draw = black] (a) {$A$}
        node[right of = a, shape = rectangle, draw = black] (m) {$M$}
        node[right of = m, shape = ellipse, draw = black] (y) {$Y$}
		 
        (m) edge[blue] (y)

        node[below right of = a, yshift = -0.25cm] () {(d)}
        ;
    \end{scope}
\end{tikzpicture}}
\caption{(a) An ADMG with a Verma constraint; (b) Fixing $M$; (c) Fixing $M, Z$; (d) Fixing $M, A$.}
\label{fig:ex_nested_markov}
\end{center}
\end{figure}

\clearpage
\section{Proofs for the Main Results}
\label{app:proofs_main}

\subsection{Proof of Proposition~\ref{prop:weighted_moment_representation}}
\label{app:prop:weighted_moment_representation}

\begin{proof}
Under the post-fixing law $\bbP_{\phi_R}$, the constraint $X \indep Y \mid Z$ implies
\begin{align*}
\bbE_{\phi_R}
\!\left[
f(X,Z)\mid Y,Z
\right]
=
\bbE_{\phi_R}
\!\left[
f(X,Z)\mid Z
\right].
\end{align*}
Hence,
\begin{align*}
&
\bbE_{\phi_R}
\!\left[
\left\{
f(X,Z)-\bbE_{\phi_R}[f(X,Z)\mid Z]
\right\}
g(Y,Z)
\right]
\\
&\hspace{1.5cm}=
\bbE_{\phi_R}
\!\left[
\bbE_{\phi_R}
\!\left[
f(X,Z)-\bbE_{\phi_R}[f(X,Z)\mid Z]
\,\middle|\,
Y,Z
\right]
g(Y,Z)
\right]
\\
&\hspace{1.5cm}=0.
\end{align*}
This proves \eqref{eq:general_equality_constraint}. The equivalent representation in \eqref{eq:general_equality_constraint2} follows by the same argument with $f(X,Z)$ and $g(Y,Z)$ interchanged, or by an application of iterated expectation.

It remains to verify the weighted representation. By the probabilistic fixing operation, the post-fixing kernel is proportional to
\begin{align*}
\frac{p(O)}{p(R\mid \mb_{\calG}(R))}.
\end{align*}
After adjoining the known reference density or mass function $\tilde p(R)$ for the fixed
variable, the induced post-fixing probability law is absolutely continuous with respect
to $\bbP$, with density proportional to
\begin{align*}
\omega_R(\M) = \frac{\tilde p(R)}{p(R\mid\mb_{\calG}(R))}, \quad \text{with} \quad \M = (R, \mb_{\calG}(R)).
\end{align*}
Normalizing yields: 
\begin{align*}
\frac{\diff \bbP_{\phi_R}}{\diff \bbP} (O) = \frac{\omega_R(\M)}{\bbE[\omega_R(\M)]}.
\end{align*}
The corresponding conditional expectation satisfies 
\begin{align*}
\bbE_{\phi_R}[A]
= 
\frac{
\bbE[\omega_R(\M) \, A]
}{
\bbE[\omega_R(\M)]
}
\quad 
\text{and} 
\quad 
\bbE_{\phi_R}[A\mid B]
=
\frac{
\bbE[\omega_R(\M) \, A\mid B]
}{
\bbE[\omega_R(\M) \mid B]
}.
\end{align*}
Substituting these identities into \eqref{eq:general_equality_constraint} or \eqref{eq:general_equality_constraint2}  yields the corresponding weighted conditional moment restriction under the observed law $\bbP$.
\end{proof}

\subsection{Proof of Proposition~\ref{prop:pathwise_score_restrictions}}
\label{app:prop:pathwise_score_restrictions}

\begin{proof}
Since $\{\bbP_t: t \in (-\eps, \eps)\} \subseteq \calP_{\calG}$, Proposition~\ref{prop:weighted_moment_representation} implies that, for every $(f, g)$,
\begin{align*}
\bbE_{\phi_R(p_t)}
\left[
\left\{
f(X,Z)
-
\bbE_{\phi_R(p_t)}[f(X,Z)\mid Z]
\right\}
g(Y,Z)
\right]
=0,
\end{align*}
for all $t$ in a neighborhood of zero. Differentiating with respect to $t$ at $t=0$ gives \eqref{eq:pathwise_restriction_postfixing}.

Using the weighted representation in Proposition~\ref{prop:weighted_moment_representation},
\begin{align*}
&\bbE_{\phi_R(p_t)}
\left[
\left\{
f(X,Z)
-
\bbE_{\phi_R(p_t)}[f(X,Z)\mid Z]
\right\}
g(Y,Z)
\right] \\
&\qquad = \frac{1}{\bbE_{t}[\omega_{R,t}(\M)]} \ 
\bbE_{t}
\left[
\omega_{R,t}(\M)
\left\{
f(X,Z)
-
\frac{
\bbE_{t}[\omega_{R,t}(\M)f(X,Z)\mid Z]
}{
\bbE_{t}[\omega_{R,t}(\M)\mid Z]
}
\right\}
g(Y,Z)
\right] \\
&\qquad = \frac{1}{\bbE_{t}[\omega_{R,t}(\M)]} \ 
\bbE_{t}[\psi_{f,g,t}(O)].
\end{align*}
Since $\bbE_t[\omega_{R,t}] =1$, the normalizing constant is identically one, and therefore
\begin{align*}
\left.
\frac{\diff}{\diff t}
\bbE_t[\psi_{f,g,t}(O)]
\right|_{t=0}
=0,
\end{align*}
which proves \eqref{eq:pathwise_restriction_observed}. 

Finally, rewrite the pathwise derivative using the score. Since the submodel is regular, differentiation under the integral gives
\begin{align*}
\left.
\frac{\diff}{\diff t}
\bbE_{t}[\psi_{f,g,t}(O)]
\right|_{t=0}
&=
\bbE[\psi_{f,g,0}(O)s(O)]
+
\bbE
\left[
\left.
\frac{\diff}{\diff t}\psi_{f,g,t}(O)
\right|_{t=0}
\right].
\end{align*}
The second term captures the pathwise derivatives of the weight $\omega_{R,t}$ and of the conditional centering operator. Collecting these terms yields a mean-zero function $\chi_{f,g}\in L_0^2(\bbP)$ satisfying
\begin{align*}
\left.
\frac{\diff}{\diff t}
\bbE_{t}[\psi_{f,g,t}(O)]
\right|_{t=0}
=
\bbE[\chi_{f,g}(O)s(O)].
\end{align*}
Since the left-hand side vanishes for every regular parametric submodel in $\calP_{\calG}$,
\begin{align*}
\bbE[\chi_{f,g}(O)s(O)] = 0
\end{align*}
for every $s\in\Lambda_{\bbP}$, implying $\chi_{f,g}\in\Lambda_{\bbP}^{\perp}$.
\end{proof}

\subsection{Proof of Theorem~\ref{thm:orthocomplement_tangent_space}}
\label{app:thm:orthocomplement_tangent_space}

\begin{proof}
By Proposition~\ref{prop:weighted_moment_representation}, the constraint $X\indep Y\mid Z\ [\phi_R(p)]$ is equivalent to the collection of moment restrictions, with $\M$ denoting $(R, \mb_\calG(R))$, 
\begin{align}
\label{eq:app_weighted_constraint}
\bbE [\omega_R (\M) f (X, Z) \{g (Y, Z) - \bbE_{\phi_R} [g(Y, Z) \mid Z]\}] = 0,
\end{align}
for all square-integrable $f$ and $g$, where 
\begin{align}
\omega_R(\M)
=
\frac{\tilde p(R)}{p(R\mid \mb_{\calG}(R))}, 
\quad \text{and} \quad 
\bbE_{\phi_R}[g(Y,Z)\mid Z]
=
\frac{\bbE[\omega_R(\M)g(Y,Z)\mid Z]}
{\bbE[\omega_R(\M)\mid Z]}.
\label{eq:app_innerE_g}
\end{align}

Let $\{\bbP_t: t \in (-\eps, \eps)\} \subseteq \calP_{\calG}$ be a regular parametric submodel through $\bbP$ with score $s (O)$. In the proof, write $\diff / \diff t$ for $\left. \diff / \diff t \right|_{t = 0}$, and define
\begin{align*}
\omega_{R,t}(\M)
=
\frac{\tilde p(R)}{p_t(R\mid \mb_{\calG}(R))}.
\end{align*}
We first differentiate the conditional centering operator in
\eqref{eq:app_innerE_g}. By the quotient rule,
\begin{align*}
\frac{\diff}{\diff t}
\frac{\bbE_{t}[\omega_{R,t}(\M)g(Y,Z)\mid Z]}
{\bbE_{t}[\omega_{R,t}(\M)\mid Z]}
&=
\frac{
\bbE'_{\bbP_t}[\omega_{R,t}(\M)g(Y,Z)\mid Z]
}{
\bbE[\omega_R(\M)\mid Z]
} 
-
\bbE_{\phi_R}[g(Y,Z)\mid Z]
\frac{
\bbE'_{\bbP_t}[\omega_{R,t}(\M)\mid Z]
}{
\bbE[\omega_R(\M)\mid Z]
},
\end{align*}
where
\begin{align*}
\bbE'_{\bbP_t}[\omega_{R,t}(\M)g(Y,Z)\mid Z]
&=
\bbE
\left[
\omega_R(\M)g(Y,Z)
\{s(O)-s(Z)-s(R\mid \mb_{\calG}(R))\}
\mid Z
\right],
\\
\bbE'_{\bbP_t}[\omega_{R,t}(\M)\mid Z]
&=
\bbE
\left[
\omega_R(\M)
\{s(O)-s(Z)-s(R\mid \mb_{\calG}(R))\}
\mid Z
\right].
\end{align*}

Therefore,
\begin{align}
&\frac{\diff}{\diff t}
\frac{\bbE_{t}[\omega_{R,t}(\M)g(Y,Z)\mid Z]}
{\bbE_{t}[\omega_{R,t}(\M)\mid Z]}
\notag \\
&\quad =
\frac{
\bbE
\left[
\omega_R(\M)
\{g(Y,Z)-\bbE_{\phi_R}[g(Y,Z)\mid Z]\}
\{s(O)-s(Z)-s(R\mid \mb_{\calG}(R))\}
\mid Z
\right]
}{
\bbE[\omega_R(\M)\mid Z]
}
\notag \\
&\quad =
\frac{
\bbE
\left[
\omega_R(\M)
\{g(Y,Z)-\bbE_{\phi_R}[g(Y,Z)\mid Z]\}
\{s(O)-s(R\mid \mb_{\calG}(R))\}
\mid Z
\right]
}{
\bbE[\omega_R(\M)\mid Z]
}.
\label{eq:app_innerE_diff}
\end{align}
The final equality follows because
\begin{align*}
&\bbE
\left[
\omega_R(\M)
\{g(Y,Z)-\bbE_{\phi_R}[g(Y,Z)\mid Z]\}
s(Z)
\mid Z
\right]
\\
&\quad =
s(Z)
\left\{
\bbE[\omega_R(\M)g(Y,Z)\mid Z]
-
\bbE_{\phi_R}[g(Y,Z)\mid Z]
\bbE[\omega_R(\M)\mid Z]
\right\}
\\
&\quad =
0.
\end{align*}

We now differentiate the weighted moment restriction
\eqref{eq:app_weighted_constraint}. Since the restriction holds along the
submodel, its derivative is zero:
\begin{align}
0
&=
\frac{\diff}{\diff t}
\bbE_{t}
\left[
\omega_{R,t}(\M)
\left\{
g(Y,Z)
-
\bbE_{\phi_R(p_t)}[g(Y,Z)\mid Z]
\right\}
f(X,Z)
\right]
\notag \\
&=
\frac{\diff}{\diff t}
\bbE_{t}
\left[
\omega_{R,t}(\M)g(Y,Z)f(X,Z)
\right]
-
\frac{\diff}{\diff t}
\bbE_{t}
\left[
\omega_{R,t}(\M)f(X,Z)
\bbE_{\phi_R(p_t)}[g(Y,Z)\mid Z]
\right].
\label{eq:app_moment_diff}
\end{align}

For the first term in \eqref{eq:app_moment_diff},
\begin{align}
&\frac{\diff}{\diff t}
\bbE_{t}
\left[
\omega_{R,t}(\M)g(Y,Z)f(X,Z)
\right]
\notag \\
&\quad =
\bbE
\left[
\omega_R(\M)g(Y,Z)f(X,Z)s(O)
\right]
-
\bbE
\left[
\omega_R(\M)g(Y,Z)f(X,Z)s(R\mid \mb_{\calG}(R))
\right]
\notag \\
&\quad =
\bbE
\left[
\omega_R(\M)g(Y,Z)f(X,Z)
\{s(O)-s(R\mid \mb_{\calG}(R))\}
\right].
\label{eq:app_first_term}
\end{align}

For the second term in \eqref{eq:app_moment_diff}, using
\eqref{eq:app_innerE_diff},
\begin{align}
&\frac{\diff}{\diff t}
\bbE_{t}
\left[
\omega_{R,t}(\M)f(X,Z)
\bbE_{\phi_R(p_t)}[g(Y,Z)\mid Z]
\right]
\notag \\
&\quad =
\bbE
\left[
\omega_R(\M)f(X,Z)
\bbE_{\phi_R}[g(Y,Z)\mid Z]
\{s(O)-s(R\mid \mb_{\calG}(R))\}
\right]
\notag \\
&\qquad +
\bbE
\left[
\omega_R(\M)f(X,Z)
\frac{
\bbE
\left[
\omega_R(\M)
\{g(Y,Z)-\bbE_{\phi_R}[g(Y,Z)\mid Z]\}
\{s(O)-s(R\mid \mb_{\calG}(R))\}
\mid Z
\right]
}{
\bbE[\omega_R(\M)\mid Z]
}
\right]
\notag \\
&\quad =
\bbE
\left[
\omega_R(\M)f(X,Z)
\bbE_{\phi_R}[g(Y,Z)\mid Z]
\{s(O)-s(R\mid \mb_{\calG}(R))\}
\right]
\notag \\
&\qquad +
\bbE
\left[
\bbE_{\phi_R}[f(X,Z)\mid Z]
\omega_R(\M)
\{g(Y,Z)-\bbE_{\phi_R}[g(Y,Z)\mid Z]\}
\{s(O)-s(R\mid \mb_{\calG}(R))\}
\right]
\notag \\
&\quad =
\bbE
\Big[
\omega_R(\M)
\big\{
g(Y,Z)\bbE_{\phi_R}[f(X,Z)\mid Z]
+
f(X,Z)\bbE_{\phi_R}[g(Y,Z)\mid Z]
\notag \\
&\hspace{4.7cm}
-
\bbE_{\phi_R}[f(X,Z)\mid Z]
\bbE_{\phi_R}[g(Y,Z)\mid Z]
\big\}
\{s(O)-s(R\mid \mb_{\calG}(R))\}
\Big].
\label{eq:app_second_term}
\end{align}

Combining \eqref{eq:app_first_term} and \eqref{eq:app_second_term} gives
\begin{align*}
0
&=
\bbE
\Big[
\omega_R(\M)
\big\{
f(X,Z)-\bbE_{\phi_R}[f(X,Z)\mid Z]
\big\}
\big\{
g(Y,Z)-\bbE_{\phi_R}[g(Y,Z)\mid Z]
\big\}
\,
\big\{s(O)-s(R\mid \mb_{\calG}(R))\big\}
\Big].
\end{align*}
Define
\begin{align*}
\widetilde\chi_{f,g}(O)
&=
\omega_R(\M)
\big\{
f(X,Z)-\bbE_{\phi_R}[f(X,Z)\mid Z]
\big\}
\big\{
g(Y,Z)-\bbE_{\phi_R}[g(Y,Z)\mid Z]
\big\}.
\end{align*}
Equivalently, by Proposition~\ref{prop:weighted_moment_representation},
\begin{align*}
\widetilde\chi_{f,g}(O)
&=
\omega_R(\M)
\bigg\{
f(X,Z)
-
\frac{
\bbE[\omega_R(\M)f(X,Z)\mid Z]
}{
\bbE[\omega_R(\M)\mid Z]
}
\bigg\}
\
\bigg\{
g(Y,Z)
-
\frac{
\bbE[\omega_R(\M)g(Y,Z)\mid Z]
}{
\bbE[\omega_R(\M)\mid Z]
}
\bigg\}.
\end{align*}
The preceding calculation shows that, for every score $s\in\Lambda_{\bbP}$,
\begin{align}
\bbE
\left[
\widetilde\chi_{f,g}(O)
\{s(O)-s(R\mid \mb_{\calG}(R))\}
\right]
=
0.
\label{eq:app_raw_orthogonality}
\end{align}

It remains to express $s(R\mid \mb_{\calG}(R))$ as a conditional expectation of the full score. Since
\begin{align*}
s(R\mid \mb_{\calG}(R))
=
s(R,\mb_{\calG}(R))-s(\mb_{\calG}(R)),
\end{align*}
and since
\begin{align*}
s(R,\mb_{\calG}(R))
=
\bbE[s(O)\mid R,\mb_{\calG}(R)],
\qquad
s(\mb_{\calG}(R))
=
\bbE[s(O)\mid \mb_{\calG}(R)],
\end{align*}
we have
\begin{align*}
\bbE
\left[
\widetilde\chi_{f,g}(O)s(R\mid \mb_{\calG}(R))
\right]
&=
\bbE
\left[
\widetilde\chi_{f,g}(O)
\{
\bbE[s(O)\mid R,\mb_{\calG}(R)]
-
\bbE[s(O)\mid \mb_{\calG}(R)]
\}
\right]
\\
&=
\bbE
\left[
s(O)
\{
\bbE[\widetilde\chi_{f,g}(O)\mid R,\mb_{\calG}(R)]
-
\bbE[\widetilde\chi_{f,g}(O)\mid \mb_{\calG}(R)]
\}
\right].
\end{align*}
Substituting this identity into \eqref{eq:app_raw_orthogonality} yields
\begin{align*}
0
&=
\bbE
\left[
s(O)
\left\{
\widetilde\chi_{f,g}(O)
-
\bbE[\widetilde\chi_{f,g}(O)\mid R,\mb_{\calG}(R)]
+
\bbE[\widetilde\chi_{f,g}(O)\mid \mb_{\calG}(R)]
\right\}
\right].
\end{align*}
Thus, 
\begin{align*}
\chi_{f,g}(O)
&\coloneqq
\widetilde\chi_{f,g}(O)
-
\bbE[\widetilde\chi_{f,g}(O)\mid R,\mb_{\calG}(R)]
+
\bbE[\widetilde\chi_{f,g}(O)\mid \mb_{\calG}(R)]
\end{align*}
satisfies
\begin{align*}
\bbE[\chi_{f,g}(O)s(O)]
=
0
\end{align*}
for every $s \in \Lambda_{\bbP}$. Hence, each $\chi_{f, g}$ belongs to $\Lambda_{\bbP}^{\perp}$, and therefore
\begin{align*}
\overline{\sp} \left\{ \chi_{f, g} (O): f \in L^2 (\bbP_{X, Z}), \; g \in L^2 (\bbP_{Y, Z}) \right\} \subseteq \Lambda_{\bbP}^{\perp}.
\end{align*}

To complete the proof, we are left to show the reverse direction:
\begin{align*}
\Lambda_{\bbP}^{\perp} \subseteq \overline{\sp} \left\{ \chi_{f, g} (O): f \in L^2 (\bbP_{X, Z}), \; g \in L^2 (\bbP_{Y,Z}) \right\}.
\end{align*}
To this end, we need to explicitly exhibit a one-dimensional path $\{\bbP_{t}: t \in (-\eps, \eps), \bbP_{t = 0} = \bbP\} \subseteq \calP_{\calG}$ for some $\eps > 0$. To simplify notation, we define the following operator:
\begin{align*}
\Psi_{t} (f, g) (\cdot) := 1 + t \cdot f (\cdot) + t^{2} \cdot g (\cdot).
\end{align*}
We then consider the following path
\begin{align}
\label{single Verma path}
p_{t} (\sfM, X, Y, Z) = p (\sfM, X, Y, Z) \Psi_{t} (s, u_{t}) (\sfM, X, Y, Z),
\end{align}
where $s \in \Lambda_{\bbP}$ and recall that the conjectured orthocomplement is
\begin{align*}
\overline{\sp} \left\{ \chi_{f, g} (O): f \in L^{2} (\bbP_{X, Z}), g \in L^{2} (\bbP_{Y, Z}) \right\},
\end{align*}
and
\begin{align*}
& \chi_{f, g} (O) := \tilde{\chi}_{f, g} (O) - \bbE [\tilde{\chi}_{f, g} (O) \mid \sfM] + \bbE [\tilde{\chi}_{f, g} (O) \mid \mb (R)], \\
& \tilde{\chi}_{f, g} (O) := \frac{1}{p (R \mid \mb (R))} \left\{ f (X, Z) - \bbE_{\phi_{R}, t} [f (X, Z) \mid Z] \right\} \left\{ g (Y, Z) - \bbE_{\phi_{R}, t} [g (Y, Z) \mid Z] \right\}.
\end{align*}
Recall that
\begin{align*}
\bbE_{\phi_{R}} [f (X, Z) \mid Z] = \frac{\bbE \left[ \frac{f (X, Z)}{p (R \mid \mb (R))} \mid Z \right]}{\bbE \left[ \frac{1}{p (R \mid \mb (R))} \mid Z \right]}, \bbE_{\phi_{R}} [g (Y, Z) \mid Z] = \frac{\bbE \left[ \frac{g (Y, Z)}{p (R \mid \mb (R))} \mid Z \right]}{\bbE \left[ \frac{1}{p (R \mid \mb (R))} \mid Z \right]}.
\end{align*}

Then the path \eqref{single Verma path} is a parametric submodel if
\begin{align*}
0
&= \bbE_{\phi_{R}, t} \left[ \left\{ f (X, Z) - \bbE_{\phi_{R}, t} [f (X, Z) \mid Z] \right\}  g (Y, Z) \right]
\\
&= \bbE_{t} \left[ \frac{f (X, Z) g (Y, Z)}{p_{t} (R \mid \mb (R))} \right] - \bbE_{t} \left[ \frac{\bbE_{\phi_{R}, t} [f (X, Z) \mid Z] g (Y, Z)}{p_{t} (R \mid \mb (R))} \right].
\end{align*}
for any $f, g$.

Furthermore, we have
\begin{equation}
\label{single Verma path 1}
\begin{split}
\bbE_{t} \left[ \frac{f (X, Z) g (Y, Z)}{p_{t} (R \mid \mb (R))} \right] = \bbE \left[ \frac{f (X, Z) g (Y, Z)}{p (R \mid \mb (R))} \frac{\Psi_{t} (s, u_{t}) (\sfM, X, Y, Z) \Psi_{t} (s, u_{t}) (\mb (R))}{\Psi_{t} (s, u_{t}) (\sfM)} \right],
\end{split}
\end{equation}
and
\begin{align}
& \ \bbE_{t} \left[ \frac{\bbE_{\phi_{R}, t} [f (X, Z) \mid Z] g (Y, Z)}{p_{t} (R \mid \mb (R))} \right] \label{single Verma path 2} \\
= & \ \bbE \Big[ \frac{g (Y, Z)}{p (R \mid \mb (R))} \frac{\Psi_{t} (s, u_{t}) (\sfM, Y, Z) \Psi_{t} (s, u_{t}) (\mb (R))}{\Psi_{t} (s, u_{t}) (\sfM)} \frac{\bbE \left[ \frac{f (X, Z)}{p (R \mid \mb (R))} \frac{\Psi_{t} (s, u_{t}) (\sfM, X, Z) \Psi_{t} (s, u_{t}) (\mb (R))}{\Psi_{t} (s, u_{t}) (\sfM)} \mid Z \right]}{\bbE \left[ \frac{1}{p (R \mid \mb (R))} \frac{\Psi_{t} (s, u_{t}) (\sfM, Z) \Psi_{t} (s, u_{t}) (\mb (R))}{\Psi_{t} (s, u_{t}) (\sfM)} \mid Z \right]} \Big]. \nonumber
\end{align}
We then need to show the existence of $u_{t}$ such that $\eqref{single Verma path 1} - \eqref{single Verma path 2} = 0$ for any $f, g$. 

To simplify the above expressions, we first take $\mb (R)$ to not intersect with $\{X, Y\}$ and consider the following path
\begin{align*}
p_{t} (\sfM, X, Y, Z) = p_{t} (\sfM) p_{t} (X, Y, Z \mid \sfM),
\end{align*}
which can be achieved by specifying parametric submodels separately for $p_{t} (\sfM)$ and $p_{t} (X, Y, Z \mid \sfM)$:
\begin{align}
& \ p_{t} (\sfM) = p (\sfM) \{1 + t \cdot s (\sfM) + t^{2} \cdot h_{t} (\sfM)\}, \label{factor submodels} \\
& \ p_{t} (X, Y, Z \mid \sfM) = p (X, Y, Z \mid \sfM) \{1 + t \cdot s (X, Y, Z \mid \sfM) + t^{2} \cdot r_{t} (X, Y, Z \mid \sfM)\}, \nonumber
\end{align}
where $\bbE [h_{t} (\sfM)] = 0$ and $\bbE [r_{t} (X, Y, Z \mid \sfM) \mid \sfM] = 0$ almost surely.

Specifically, under \eqref{factor submodels}, \eqref{single Verma path 1} and \eqref{single Verma path 2}, respectively, reduce to
\begin{align*}
\eqref{single Verma path 1} & = \bbE \left[ \frac{f (X, Z) g (Y, Z)}{p (R \mid \mb (R))} \Psi_{t} (s, r_{t}) (X, Y, Z \mid \sfM) \Psi_{t} (s, h_{t}) (\mb (R)) \right], \\
& = \bbE \Big[ \frac{f (X, Z) g (Y, Z)}{p (R \mid \mb (R))} \{1 + t \cdot s (X, Y, Z \mid \sfM) + t^{2} \cdot r_{t} (X, Y, Z \mid \sfM)\} \{1 + t \cdot s (\mb (R)) + t^{2} \cdot h_{t} (\mb (R))\} \Big], \\
\eqref{single Verma path 2} & = \bbE \Big[ \frac{g (Y, Z)}{p (R \mid \mb (R))} \Psi_{t} (s, r_{t}) (Y, Z \mid \sfM) \Psi_{t} (s, h_{t}) (\mb (R)) \\
& \qquad \times \frac{\bbE \left[ \frac{f (X, Z)}{p (R \mid \mb (R)} \Psi_{t} (s, r_{t}) (X, Z \mid \sfM) \Psi_{t} (s, h_{t}) (\mb (R)) \mid Z \right]}{\bbE \left[ \frac{1}{p (R \mid \mb (R)} \Psi_{t} (s, r_{t}) (Z \mid \sfM) \Psi_{t} (s, h_{t}) (\mb (R)) \mid Z \right]} \Big] \\
& = \bbE \Big[ \frac{g (Y, Z)}{p (R \mid \mb (R))} \{1 + t \cdot s (Y, Z \mid \sfM) + t^{2} \cdot r_{t} (Y, Z \mid \sfM)\} \{1 + t \cdot s (\mb (R)) + t^{2} \cdot h_{t} (\mb (R))\} \\
& \qquad \times \frac{\bbE \left[ \frac{f (X, Z)}{p (R \mid \mb (R)} \Psi_{t} (s, r_{t}) (X, Z \mid \sfM) \Psi_{t} (s, h_{t}) (\mb (R)) \mid Z \right]}{\bbE \left[ \frac{1}{p (R \mid \mb (R)} \Psi_{t} (s, r_{t}) (Z \mid \sfM) \Psi_{t} (s, h_{t}) (\mb (R)) \mid Z \right]} \Big]
\end{align*}

We then simplify $\eqref{single Verma path 1} - \eqref{single Verma path 2} = 0$ as follows.
\begin{align*}
0 & = \bbE \Big[ \frac{g (Y, Z)}{p (R \mid \mb (R))} \{1 + t \cdot s (Y, Z \mid \sfM) + t^{2} \cdot r_{t} (Y, Z \mid \sfM)\} \{1 + t \cdot s (\mb (R)) + t^{2} \cdot h_{t} (\mb (R))\} \\
& \quad \quad \times \Big\{ f (X, Z) - \frac{\bbE \left[ \frac{f (X, Z)}{p (R \mid \mb (R)} \Psi_{t} (s, r_{t}) (X, Z \mid \sfM) \Psi_{t} (s, h_{t}) (\mb (R)) \mid Z \right]}{\bbE \left[ \frac{1}{p (R \mid \mb (R)} \Psi_{t} (s, r_{t}) (Z \mid \sfM) \Psi_{t} (s, h_{t}) (\mb (R)) \mid Z \right]} \Big\} \Big] \\
& \quad \quad + \bbE \Big[ \frac{f (X, Z) g (Y, Z)}{p (R \mid \mb (R))} \{t \cdot s (X \mid \sfM, Y, Z) + t^{2} \cdot r_{t} (X \mid \sfM, Y, Z)\} \{1 + t \cdot s (\mb (R)) + t^{2} \cdot h_{t} (\mb (R))\} \\
& \quad \quad - \frac{g (Y, Z)}{p (R \mid \mb (R))} \{1 + t \cdot s (Y, Z \mid \sfM) + t^{2} \cdot r_{t} (Y, Z \mid \sfM)\} \{1 + t \cdot s (\mb (R)) + t^{2} \cdot h_{t} (\mb (R))\} \\
& \quad \quad \times \frac{\bbE \left[ \frac{f (X, Z)}{p (R \mid \mb (R)} \Psi_{t} (s, r_{t}) (X, Z \mid \sfM) \Psi_{t} (s, h_{t}) (\mb (R)) \mid Z \right]}{\bbE \left[ \frac{1}{p (R \mid \mb (R)} \Psi_{t} (s, r_{t}) (Z \mid \sfM) \Psi_{t} (s, h_{t}) (\mb (R)) \mid Z \right]} \Big].
\end{align*}

But $\eqref{single Verma path 1} - \eqref{single Verma path 2} = 0$ is still a rational equation. Nonetheless, for $\eqref{single Verma path 1} - \eqref{single Verma path 2} = 0$ to hold, the equality needs to hold even after replacing the marginal expectations by the conditional expectations given $Z$, which yields:
\begin{align}
0 & = \bbE \Bigg[ \frac{g (Y, Z)}{p (R \mid \mb (R))} \{1 + t \cdot s (Y, Z \mid \sfM) + t^{2} \cdot r_{t} (Y, Z \mid \sfM)\} \{1 + t \cdot s (\mb (R)) + t^{2} \cdot h_{t} (\mb (R))\} \nonumber \\
& \quad \quad \times \Big\{ f (X, Z) \bbE \Big[ \frac{1}{p (R \mid \mb (R)} \{1 + t \cdot s (Z \mid \sfM) + t^{2} \cdot r_{t} (Z \mid \sfM)\} \{1 + t \cdot s (\mb (R)) + t^{2} \cdot h_{t} (\mb (R))\} \mid Z \Big] \nonumber \\
& \quad \quad \quad - \bbE \Big[ \frac{f (X, Z)}{p (R \mid \mb (R)} \{1 + t \cdot s (X, Z \mid \sfM) + t^{2} \cdot r_{t} (X, Z \mid \sfM)\} \{1 + t \cdot s (\mb (R)) + t^{2} \cdot h_{t} (\mb (R))\} \mid Z \Big] \Big\} \mid Z \Bigg] \nonumber \\
& + \bbE \Bigg[ \frac{g (Y, Z)}{p (R \mid \mb (R))} \{t \cdot s (X \mid \sfM, Y, Z) + t^{2} \cdot r_{t} (X \mid \sfM, Y, Z)\} \{1 + t \cdot s (\mb (R)) + t^{2} \cdot h_{t} (\mb (R))\} \nonumber \\
& \quad \quad \times f (X, Z) \bbE \Big[ \frac{1}{p (R \mid \mb (R)} \{1 + t \cdot s (Z \mid \sfM) + t^{2} \cdot r_{t} (Z \mid \sfM)\} \{1 + t \cdot s (\mb (R)) + t^{2} \cdot h_{t} (\mb (R))\} \mid Z \Big] \nonumber \\
& \quad \quad - \frac{g (Y, Z)}{p (R \mid \mb (R))} \{1 + t \cdot s (Y, Z \mid \sfM) + t^{2} \cdot r_{t} (Y, Z \mid \sfM)\} \{1 + t \cdot s (\mb (R)) + t^{2} \cdot h_{t} (\mb (R))\} \nonumber \\
& \quad \quad \times \bbE \Big[ \frac{f (X, Z)}{p (R \mid \mb (R)} \{1 + t \cdot s (X, Z \mid \sfM) + t^{2} \cdot r_{t} (X, Z \mid \sfM)\} \{1 + t \cdot s (\mb (R)) + t^{2} \cdot h_{t} (\mb (R))\} \mid Z \Big] \mid Z \Bigg] \nonumber \\
& =: \sum_{k = 0}^{8} \alpha_{k} t^{k}. \label{polynomial}
\end{align}
We then collect all the coefficients $\alpha_{k}$ in front of $t^{k}$ for $k = 0, 1, \cdots, 8$, considering an ansatz $h_{t} (\mb (R)) = 0$:
\begin{itemize}
\item $\alpha_{0}$:
\begin{align*}
\alpha_{0} = \bbE \Big[ \frac{f (X, Z) g (Y, Z)}{p (R \mid \mb (R))} \mid Z \Big] \bbE \Big[ \frac{1}{p (R \mid \mb (R))} \mid Z \Big] - \bbE \Big[ \frac{f (X, Z)}{p (R \mid \mb (R))} \mid Z \Big] \bbE \Big[ \frac{g (Y, Z)}{p (R \mid \mb (R))} \mid Z \Big] = 0,
\end{align*}
which is a direct consequence of the Verma constraint $X \indep Y \mid Z \quad \phi_{R} (\calG)$.

\item $\alpha_{1}$:
\begin{align*}
\alpha_{1} & = \bbE \Big[ \frac{f (X, Z) g (Y, Z)}{p (R \mid \mb (R))} \mid Z \Big] \bbE \Big[ \frac{1}{p (R \mid \mb (R))} \{ s (Z \mid \sfM) + s (\mb (R)) \} \mid Z \Big] \\ 
& \quad - \bbE \Big[ \frac{f (X, Z)}{p (R \mid \mb (R))} \mid Z \Big] \bbE \Big[ \frac{g (Y, Z)}{p (R \mid \mb (R))} \{ s (Y, Z \mid \sfM) + s (\mb (R)) \} \mid Z \Big] \\ 
& \quad +  \bbE \Big[ \frac{f (X, Z) g (Y, Z)}{p (R \mid \mb (R))} \{ s (X, Y, Z \mid \sfM) + s (\mb (R)) \} \mid Z \Big] \bbE \Big[ \frac{1}{p (R \mid \mb (R))} \mid Z \Big] \nonumber \\ & \quad - \bbE \Big[ \frac{f (X, Z)}{p (R \mid \mb (R))} \{ s (X, Z \mid \sfM) + s (\mb (R)) \} \mid Z \Big] \bbE \Big[ \frac{g (Y, Z)}{p (R \mid \mb (R))} \mid Z \Big] = 0,
\end{align*}
which is a consequence of $s (\sfM, X, Y, Z) \in \Lambda_{\bbP}$; see Lemma~\ref{lem:score} for a precise statement.

\item $\alpha_{2}$: We have applied the ansatz $h_{t} (\mb (R)) = 0$ to obtain the second equality in the formula below.
\begin{align*}
& \ \alpha_{2} \\
= & \ \bbE \Big[ \frac{f (X, Z) g (Y, Z)}{p (R \mid \mb (R))} \mid Z \Big] \bbE \Big[ \frac{1}{p (R \mid \mb (R))} \{ r_{t} (Z \mid \sfM) + h_{t} (\mb (R)) + s (Z \mid \sfM) s (\mb (R)) \} \mid Z \Big] \\ 
& - \bbE \Big[ \frac{f (X, Z)}{p (R \mid \mb (R))} \mid Z \Big] \bbE \Big[ \frac{g (Y, Z)}{p (R \mid \mb (R))} \{ r_{t} (Y, Z \mid \sfM) + h_{t} (\mb (R)) + s (Y, Z \mid \sfM) s (\mb (R)) \} \mid Z \Big] \\ 
& + \bbE \Big[ \frac{f (X, Z) g (Y, Z)}{p (R \mid \mb (R))} \{ s (X, Y, Z \mid \sfM) + s (\mb (R)) \} \mid Z \Big] \bbE \Big[ \frac{1}{p (R \mid \mb (R))} \{ s (Z \mid \sfM) + s (\mb (R)) \} \mid Z \Big] \\ 
& - \bbE \Big[ \frac{f (X, Z)}{p (R \mid \mb (R))} \{ s (X, Z \mid \sfM) + s (\mb (R)) \} \mid Z \Big] \bbE \Big[ \frac{g (Y, Z)}{p (R \mid \mb (R))} \{ s (Y, Z \mid \sfM) + s (\mb (R)) \} \mid Z \Big] \\ 
& + \bbE \Big[ \frac{f (X, Z) g (Y, Z)}{p (R \mid \mb (R))} \{ r_{t} (X, Y, Z \mid \sfM) + h_{t} (\mb (R)) + s (X, Y, Z \mid \sfM) s (\mb (R)) \} \mid Z \Big] \bbE \Big[ \frac{1}{p (R \mid \mb (R))} \mid Z \Big] \\ 
& - \bbE \Big[ \frac{f (X, Z)}{p (R \mid \mb (R))} \{ r_{t} (X, Z \mid \sfM) + h_{t} (\mb (R)) + s (X, Z \mid \sfM) s (\mb (R)) \} \mid Z \Big] \bbE \Big[ \frac{g (Y, Z)}{p (R \mid \mb (R))} \mid Z \Big] \\
= & \ \bbE \Big[ \frac{f (X, Z) g (Y, Z)}{p (R \mid \mb (R))} \mid Z \Big] \bbE \Big[ \frac{1}{p (R \mid \mb (R))} \{ r_{t} (Z \mid \sfM) + s (Z \mid \sfM) s (\mb (R)) \} \mid Z \Big] \\ 
& - \bbE \Big[ \frac{f (X, Z)}{p (R \mid \mb (R))} \mid Z \Big] \bbE \Big[ \frac{g (Y, Z)}{p (R \mid \mb (R))} \{ r_{t} (Y, Z \mid \sfM) + s (Y, Z \mid \sfM) s (\mb (R)) \} \mid Z \Big] \\
& + \bbE \Big[ \frac{f (X, Z) g (Y, Z)}{p (R \mid \mb (R))} \{ s (X, Y, Z \mid \sfM) + s (\mb (R)) \} \mid Z \Big] \bbE \Big[ \frac{1}{p (R \mid \mb (R))} \{ s (Z \mid \sfM) + s (\mb (R)) \} \mid Z \Big] \\ 
& - \bbE \Big[ \frac{f (X, Z)}{p (R \mid \mb (R))} \{ s (X, Z \mid \sfM) + s (\mb (R)) \} \mid Z \Big] \bbE \Big[ \frac{g (Y, Z)}{p (R \mid \mb (R))} \{ s (Y, Z \mid \sfM) + s (\mb (R)) \} \mid Z \Big] \\ 
& + \bbE \Big[ \frac{f (X, Z) g (Y, Z)}{p (R \mid \mb (R))} \{ r_{t} (X, Y, Z \mid \sfM) + s (X, Y, Z \mid \sfM) s (\mb (R)) \} \mid Z \Big] \bbE \Big[ \frac{1}{p (R \mid \mb (R))} \mid Z \Big] \\ 
& - \bbE \Big[ \frac{f (X, Z)}{p (R \mid \mb (R))} \{ r_{t} (X, Z \mid \sfM) + s (X, Z \mid \sfM) s (\mb (R)) \} \mid Z \Big] \bbE \Big[ \frac{g (Y, Z)}{p (R \mid \mb (R))} \mid Z \Big].
\end{align*}

\item $\alpha_{3}$: We have applied the ansatz $h_{t} (\mb (R)) = 0$ to obtain the first equality in the formula below.
\begin{align*}
\alpha_{3} & = \bbE \Big[ \frac{f (X, Z) g (Y, Z)}{p (R \mid \mb (R))} \mid Z \Big] \bbE \Big[ \frac{1}{p (R \mid \mb (R))} r_{t} (Z \mid \sfM) s (\mb (R)) \mid Z \Big] \\ 
& \quad - \bbE \Big[ \frac{f (X, Z)}{p (R \mid \mb (R))} \mid Z \Big] \bbE \Big[ \frac{g (Y, Z)}{p (R \mid \mb (R))} r_{t} (Y, Z \mid \sfM) s (\mb (R)) \mid Z \Big] \\ 
& \quad + \bbE \Big[ \frac{f (X, Z) g (Y, Z)}{p (R \mid \mb (R))} \{ s (X, Y, Z \mid \sfM) + s (\mb (R)) \} \mid Z \Big] \\
& \quad \times \bbE \Big[ \frac{1}{p (R \mid \mb (R))} \{ r_{t} (Z \mid \sfM) + s (Z \mid \sfM) s (\mb (R)) \} \mid Z \Big] \\ 
& \quad - \bbE \Big[ \frac{f (X, Z)}{p (R \mid \mb (R))} \{ s (X, Z \mid \sfM) + s (\mb (R)) \} \mid Z \Big] \\
& \quad \times \bbE \Big[ \frac{g (Y, Z)}{p (R \mid \mb (R))} \{ r_{t} (Y, Z \mid \sfM) + s (Y, Z \mid \sfM) s (\mb (R)) \} \mid Z \Big] \\ 
& \quad + \bbE \Big[ \frac{f (X, Z) g (Y, Z)}{p (R \mid \mb (R))} \{ r_{t} (X, Y, Z \mid \sfM) + s (X, Y, Z \mid \sfM) s (\mb (R)) \} \mid Z \Big] \\
& \quad \times \bbE \Big[ \frac{1}{p (R \mid \mb (R))} \{ s (Z \mid \sfM) + s (\mb (R)) \} \mid Z \Big] \\ 
& \quad - \bbE \Big[ \frac{f (X, Z)}{p (R \mid \mb (R))} \{ r_{t} (X, Z \mid \sfM) + s (X, Z \mid \sfM) s (\mb (R)) \} \mid Z \Big] \\
& \quad \times \bbE \Big[ \frac{g (Y, Z)}{p (R \mid \mb (R))} \{ s (Y, Z \mid \sfM) + s (\mb (R)) \} \mid Z \Big] \\ 
& \quad + \bbE \Big[ \frac{f (X, Z) g (Y, Z)}{p (R \mid \mb (R))} r_{t} (X, Y, Z \mid \sfM) s (\mb (R)) \mid Z \Big] \bbE \Big[ \frac{1}{p (R \mid \mb (R))} \mid Z \Big] \\ 
& \quad - \bbE \Big[ \frac{f (X, Z)}{p (R \mid \mb (R))} r_{t} (X, Z \mid \sfM) s (\mb (R)) \mid Z \Big] \bbE \Big[ \frac{g (Y, Z)}{p (R \mid \mb (R))} \mid Z \Big].
\end{align*}

\item $\alpha_{4}$: We have applied the ansatz $h_{t} (\mb (R)) = 0$ to obtain the first equality in the formula below.
\begin{align*}
\alpha_{4} & = \bbE \Big[ \frac{f (X, Z) g (Y, Z)}{p (R \mid \mb (R))} \{ s (X, Y, Z \mid \sfM) + s (\mb (R)) \} \mid Z \Big] \bbE \Big[ \frac{1}{p (R \mid \mb (R))} r_{t} (Z \mid \sfM) s (\mb (R)) \mid Z \Big] \\ 
& \quad - \bbE \Big[ \frac{f (X, Z)}{p (R \mid \mb (R))} \{ s (X, Z \mid \sfM) + s (\mb (R)) \} \mid Z \Big] \bbE \Big[ \frac{g (Y, Z)}{p (R \mid \mb (R))} r_{t} (Y, Z \mid \sfM) s (\mb (R)) \mid Z \Big] \\ 
& \quad + \bbE \Big[ \frac{f (X, Z) g (Y, Z)}{p (R \mid \mb (R))} \{ r_{t} (X, Y, Z \mid \sfM) + s (X, Y, Z \mid \sfM) s (\mb (R)) \} \mid Z \Big] \\
& \quad \times \bbE \Big[ \frac{1}{p (R \mid \mb (R))} \{ r_{t} (Z \mid \sfM) + s (Z \mid \sfM) s (\mb (R)) \} \mid Z \Big] \\ 
& \quad - \bbE \Big[ \frac{f (X, Z)}{p (R \mid \mb (R))} \{ r_{t} (X, Z \mid \sfM) + s (X, Z \mid \sfM) s (\mb (R)) \} \mid Z \Big] \\ 
& \quad \times \bbE \Big[ \frac{g (Y, Z)}{p (R \mid \mb (R))} \{ r_{t} (Y, Z \mid \sfM) + s (Y, Z \mid \sfM) s (\mb (R)) \} \mid Z \Big] \\ 
& \quad + \bbE \Big[ \frac{f (X, Z) g (Y, Z)}{p (R \mid \mb (R))} r_{t} (X, Y, Z \mid \sfM) s (\mb (R)) \mid Z \Big] \bbE \Big[ \frac{1}{p (R \mid \mb (R))} \{ s (Z \mid \sfM) + s (\mb (R)) \} \mid Z \Big] \\ 
& \quad - \bbE \Big[ \frac{f (X, Z)}{p (R \mid \mb (R))} r_{t} (X, Z \mid \sfM) s (\mb (R)) \mid Z \Big] \bbE \Big[ \frac{g (Y, Z)}{p (R \mid \mb (R))} \{ s (Y, Z \mid \sfM) + s (\mb (R)) \} \mid Z \Big].
\end{align*}

\item $\alpha_{5}$: We have applied the ansatz $h_{t} (\mb (R)) = 0$ to obtain the first equality in the formula below.
\begin{align*}
& \ \alpha_{5} \\ 
= & \ \bbE \Big[ \frac{f (X, Z) g (Y, Z)}{p (R \mid \mb (R))} \{ r_{t} (X, Y, Z \mid \sfM) + s (X, Y, Z \mid \sfM) s (\mb (R)) \} \mid Z \Big] \\
& \quad \times \bbE \Big[ \frac{1}{p (R \mid \mb (R))} r_{t} (Z \mid \sfM) s (\mb (R)) \mid Z \Big] \\ 
& - \bbE \Big[ \frac{f (X, Z)}{p (R \mid \mb (R))} \{ r_{t} (X, Z \mid \sfM) + s (X, Z \mid \sfM) s (\mb (R)) \} \mid Z \Big] \\
& \times \bbE \Big[ \frac{g (Y, Z)}{p (R \mid \mb (R))} r_{t} (Y, Z \mid \sfM) s (\mb (R)) \mid Z \Big] \\ 
& + \bbE \Big[ \frac{f (X, Z) g (Y, Z)}{p (R \mid \mb (R))} r_{t} (X, Y, Z \mid \sfM) s (\mb (R)) \mid Z \Big] \\
& \times \bbE \Big[ \frac{1}{p (R \mid \mb (R))} \{ r_{t} (Z \mid \sfM) + s (Z \mid \sfM) s (\mb (R)) \} \mid Z \Big] \\ 
& - \bbE \Big[ \frac{f (X, Z)}{p (R \mid \mb (R))} r_{t} (X, Z \mid \sfM) s (\mb (R)) \mid Z \Big] \\
& \times \bbE \Big[ \frac{g (Y, Z)}{p (R \mid \mb (R))} \{ r_{t} (Y, Z \mid \sfM) + s (Y, Z \mid \sfM) s (\mb (R)) \} \mid Z \Big].
\end{align*}

\item $\alpha_{6}$: We have applied the ansatz $h_{t} (\mb (R)) = 0$ to obtain the first equality in the formula below.
\begin{align*}
\alpha_{6} & = \bbE \Big[ \frac{f (X, Z) g (Y, Z)}{p (R \mid \mb (R))} r_{t} (X, Y, Z \mid \sfM) s (\mb (R)) \mid Z \Big] \bbE \Big[ \frac{1}{p (R \mid \mb (R))} r_{t} (Z \mid \sfM) s (\mb (R)) \mid Z \Big] 
\\ & \quad - \bbE \Big[ \frac{f (X, Z)}{p (R \mid \mb (R))} r_{t} (X, Z \mid \sfM) s (\mb (R)) \mid Z \Big] \bbE \Big[ \frac{g (Y, Z)}{p (R \mid \mb (R))} r_{t} (Y, Z \mid \sfM) s (\mb (R)) \mid Z \Big].
\end{align*}

\item $\alpha_{7}$:
\begin{align*}
\alpha_{7} & = \bbE \Big[ \frac{f (X, Z) g (Y, Z)}{p (R \mid \mb (R))} \{ s (X, Y, Z \mid \sfM) h_{t} (\mb (R)) + r_{t} (X, Y, Z \mid \sfM) s (\mb (R)) \} \mid Z \Big] \\
& \quad \times \bbE \Big[ \frac{1}{p (R \mid \mb (R))} \{ h_{t} (\mb (R)) r_{t} (Z \mid \sfM) \} \mid Z \Big] \\ 
& \quad - \bbE \Big[ \frac{f (X, Z)}{p (R \mid \mb (R))} \{ s (X, Z \mid \sfM) h_{t} (\mb (R)) + r_{t} (X, Z \mid \sfM) s (\mb (R)) \} \mid Z \Big] \\
& \quad \times \bbE \Big[ \frac{g (Y, Z)}{p (R \mid \mb (R))} \{ h_{t} (\mb (R)) r_{t} (Y, Z \mid \sfM) \} \mid Z \Big] \\ 
& \quad + \bbE \Big[ \frac{f (X, Z) g (Y, Z)}{p (R \mid \mb (R))} \{ h_{t} (\mb (R)) r_{t} (X, Y, Z \mid \sfM) \} \mid Z \Big] \\
& \quad \times \bbE \Big[ \frac{1}{p (R \mid \mb (R))} \{ s (Z \mid \sfM) h_{t} (\mb (R)) + r_{t} (Z \mid \sfM) s (\mb (R)) \} \mid Z \Big] \\ 
& \quad - \bbE \Big[ \frac{f (X, Z)}{p (R \mid \mb (R))} \{ h_{t} (\mb (R)) r_{t} (X, Z \mid \sfM) \} \mid Z \Big] \\
& \quad \times \bbE \Big[ \frac{g (Y, Z)}{p (R \mid \mb (R))} \{ s (Y, Z \mid \sfM) h_{t} (\mb (R)) + r_{t} (Y, Z \mid \sfM) s (\mb (R)) \} \mid Z \Big] = 0,
\end{align*}
by the ansatz $h_{t} (\mb (R)) = 0$.

\item $\alpha_{8}$:
\begin{align*}
\alpha_{8} & = \bbE \Big[ \frac{f (X, Z) g (Y, Z)}{p (R \mid \mb (R))} h_{t} (\mb (R)) r_{t} (X, Y, Z \mid \sfM) \mid Z \Big] \bbE \Big[ \frac{1}{p (R \mid \mb (R))} h_{t} (\mb (R)) r_{t} (Z \mid \sfM) \mid Z \Big] \\ 
& \quad - \bbE \Big[ \frac{f (X, Z)}{p (R \mid \mb (R))} h_{t} (\mb (R)) r_{t} (X, Z \mid \sfM) \mid Z \Big] \bbE \Big[ \frac{g (Y, Z)}{p (R \mid \mb (R))} h_{t} (\mb (R)) r_{t} (Y, Z \mid \sfM) \mid Z \Big] \\
& = 0,
\end{align*}
again by the ansatz $h_{t} (\mb (R)) = 0$.
\end{itemize}
Thus, we need to restrict $r_{t} (X, Y, Z \mid \sfM)$ so that
\begin{align*}
\alpha_{2} = \alpha_{3} = \alpha_{4} = \alpha_{5} = \alpha_{6} = 0.
\end{align*}
Let $\iota (\sfM, X, Y, Z) := s (X, Y, Z \mid \sfM) + s (\mb (R))$. Specifically, we consider the following ansatz:
\begin{align}
\label{ansatz 1}
r_{t} (X, Y, Z \mid \sfM) = \kappa (\mb (R), X, Y, Z) - \bbE [\kappa (\mb (R), X, Y, Z) \mid \sfM],
\end{align}
where $\kappa (\mb (R), X, Y, Z) = \beta_{0} (X, Y, Z) + \beta_{1} (X, Y, Z) s (\mb (R))$ for two functions $\beta_{0}$ and $\beta_{1}$ of $(X, Y, Z)$, satisfying:
\begin{equation}
\begin{split}
\label{extra constraint 1}
\frac{\bbE [\frac{r_{t} (X, Y, Z \mid \sfM) s (\mb (R))}{p (R \mid \mb (R))} \mid X, Y, Z]}{\bbE [\frac{1}{p (R \mid \mb (R))} \mid X, Y, Z]} = 0
\end{split}
\end{equation}
and
\begin{equation}
\label{extra constraint 2}
\begin{split}
\frac{\bbE [\frac{r_{t} (X, Y, Z \mid \sfM) + s (X, Y, Z \mid \sfM) s (\mb (R))}{p (R \mid \mb (R))} \mid X, Y, Z]}{\bbE [\frac{1}{p (R \mid \mb (R))} \mid X, Y, Z]} = \iota (X, Z) \{\iota (Y, Z) - \iota (Z)\},
\end{split}
\end{equation}
where $\iota (X, Z) := \bbE_{\phi_{R}} [\iota (\sfM, X, Y, Z) \mid X, Z]$, $\iota (Y, Z) := \bbE_{\phi_{R}} [\iota (\sfM, X, Y, Z) \mid Y, Z]$, and $\iota (Z) := \bbE_{\phi_{R}} [\iota (\sfM, X, Y, Z) \mid Z]$.

Under this particular choice of $r_{t}$, we can show the desired result. To see this, first, by statements (2) and (3) of Lemma~\ref{lem:alpha}, we have $\alpha_{6} = 0$. Next, by statements (1)--(3) of Lemma~\ref{lem:alpha}, it is easy to also see that $\alpha_{5} = 0$. By the same argument, $\alpha_{4}$ first reduces to
\begin{align*}
& \alpha_{4} = \bbE \Big[ \frac{f (X, Z) g (Y, Z)}{p (R \mid \mb (R))} r_{t} (X, Y, Z \mid \sfM) s (\mb (R)) \mid Z \Big] \bbE \Big[ \frac{1}{p (R \mid \mb (R))} \{ s (Z \mid \sfM) + s (\mb (R)) \} \mid Z \Big] \\
& \qquad - \bbE \Big[ \frac{f (X, Z)}{p (R \mid \mb (R))} \{ r_{t} (X, Z \mid \sfM) + s (X, Z \mid \sfM) s (\mb (R)) \} \mid Z \Big] \\ 
& \qquad \times \bbE \Big[ \frac{g (Y, Z)}{p (R \mid \mb (R))} \{ r_{t} (Y, Z \mid \sfM) + s (Y, Z \mid \sfM) s (\mb (R)) \} \mid Z \Big].
\end{align*}
Now, invoking \eqref{extra constraint 1} and \eqref{extra constraint 2}, we can conclude that $\alpha_{4} = 0$.

$\alpha_{3}$ can be simplified using a similar argument:
\begin{align*}
& \alpha_{3} = - \, \bbE \Big[ \frac{f (X, Z)}{p (R \mid \mb (R))} \{ s (X, Z \mid \sfM) + s (\mb (R)) \} \mid Z \Big] \\
& \qquad \times \bbE \Big[ \frac{g (Y, Z)}{p (R \mid \mb (R))} \{ r_{t} (Y, Z \mid \sfM) + s (Y, Z \mid \sfM) s (\mb (R)) \} \mid Z \Big] \\ 
& \qquad + \bbE \Big[ \frac{f (X, Z) g (Y, Z)}{p (R \mid \mb (R))} \{ r_{t} (X, Y, Z \mid \sfM) + s (X, Y, Z \mid \sfM) s (\mb (R)) \} \mid Z \Big] \\
& \qquad \times \bbE \Big[ \frac{1}{p (R \mid \mb (R))} \{ s (Z \mid \sfM) + s (\mb (R)) \} \mid Z \Big] \\ 
& \qquad - \bbE \Big[ \frac{f (X, Z)}{p (R \mid \mb (R))} \{ r_{t} (X, Z \mid \sfM) + s (X, Z \mid \sfM) s (\mb (R)) \} \mid Z \Big] \\
& \qquad \times \bbE \Big[ \frac{g (Y, Z)}{p (R \mid \mb (R))} \{ s (Y, Z \mid \sfM) + s (\mb (R)) \} \mid Z \Big] \\ 
& \qquad + \bbE \Big[ \frac{f (X, Z) g (Y, Z)}{p (R \mid \mb (R))} r_{t} (X, Y, Z \mid \sfM) s (\mb (R)) \mid Z \Big] \bbE \Big[ \frac{1}{p (R \mid \mb (R))} \mid Z \Big],
\end{align*}
which again equals $0$ by invoking \eqref{extra constraint 1} and \eqref{extra constraint 2}. Finally, as for $\alpha_{2}$, we again first use Lemma~\ref{lem:alpha} to simplify it as
\begin{align*}
& \alpha_{2} = - \bbE \Big[ \frac{f (X, Z)}{p (R \mid \mb (R))} \mid Z \Big] \bbE \Big[ \frac{g (Y, Z)}{p (R \mid \mb (R))} \{ r_{t} (Y, Z \mid \sfM) + s (Y, Z \mid \sfM) s (\mb (R)) \} \mid Z \Big] \\
& + \bbE \Big[ \frac{f (X, Z) g (Y, Z)}{p (R \mid \mb (R))} \{ s (X, Y, Z \mid \sfM) + s (\mb (R)) \} \mid Z \Big] \bbE \Big[ \frac{1}{p (R \mid \mb (R))} \{ s (Z \mid \sfM) + s (\mb (R)) \} \mid Z \Big] \\ 
& - \bbE \Big[ \frac{f (X, Z)}{p (R \mid \mb (R))} \{ s (X, Z \mid \sfM) + s (\mb (R)) \} \mid Z \Big] \bbE \Big[ \frac{g (Y, Z)}{p (R \mid \mb (R))} \{ s (Y, Z \mid \sfM) + s (\mb (R)) \} \mid Z \Big] \\ 
& + \bbE \Big[ \frac{f (X, Z) g (Y, Z)}{p (R \mid \mb (R))} \{ r_{t} (X, Y, Z \mid \sfM) + s (X, Y, Z \mid \sfM) s (\mb (R)) \} \mid Z \Big] \bbE \Big[ \frac{1}{p (R \mid \mb (R))} \mid Z \Big] \\ 
& - \bbE \Big[ \frac{f (X, Z)}{p (R \mid \mb (R))} \{ r_{t} (X, Z \mid \sfM) + s (X, Z \mid \sfM) s (\mb (R)) \} \mid Z \Big] \bbE \Big[ \frac{g (Y, Z)}{p (R \mid \mb (R))} \mid Z \Big],
\end{align*}
which can be shown to be $0$ by using \eqref{extra constraint 1} and \eqref{extra constraint 2}.

Lastly, we consider the scenario where $\mb (R)$ intersects with $\{X, Y\}$, which occurs in the example shown in Section~\ref{sec:napkin_graph}. Here, instead of decomposing the density into $p (W \mid \sfM) p (\sfM)$ with $W = \{X, Y, Z\} \setminus \sfM$ and a similar construction follows.
\end{proof}

We next state the few lemmas used in the above proof of Theorem~\ref{thm:orthocomplement_tangent_space}. The first lemma is a direct consequence of $s \in \Lambda_{\bbP}$.

\begin{lemma}
\label{lem:score}
Any valid score functions $s (\sfM, X, Y, Z)$ in $\Lambda_{\bbP}$ as defined in Theorem~\ref{thm:orthocomplement_tangent_space} must satisfy the following:
\begin{align*}
\alpha_{1} & = \bbE \Big[ \frac{f (X, Z) g (Y, Z)}{p (R \mid \mb (R))} \mid Z \Big] \bbE \Big[ \frac{1}{p (R \mid \mb (R))} \{ s (Z \mid \sfM) + s (\mb (R)) \} \mid Z \Big] \\ 
& \quad - \bbE \Big[ \frac{f (X, Z)}{p (R \mid \mb (R))} \mid Z \Big] \bbE \Big[ \frac{g (Y, Z)}{p (R \mid \mb (R))} \{ s (Y, Z \mid \sfM) + s (\mb (R)) \} \mid Z \Big] \\ 
& \quad +  \bbE \Big[ \frac{f (X, Z) g (Y, Z)}{p (R \mid \mb (R))} \{ s (X, Y, Z \mid \sfM) + s (\mb (R)) \} \mid Z \Big] \bbE \Big[ \frac{1}{p (R \mid \mb (R))} \mid Z \Big] \nonumber \\ & \quad - \bbE \Big[ \frac{f (X, Z)}{p (R \mid \mb (R))} \{ s (X, Z \mid \sfM) + s (\mb (R)) \} \mid Z \Big] \bbE \Big[ \frac{g (Y, Z)}{p (R \mid \mb (R))} \mid Z \Big] = 0,
\end{align*}
\end{lemma}

\begin{proof}
By the definition of $\Lambda_{\bbP}$, the following must hold after scaling by $\bbE [\frac{1}{p (R \mid \mb (R))} \mid Z]$:
\begin{align*}
0 & = \bbE \Big[ \iota (\sfM, X, Y, Z) \Big\{ \frac{f (X, Z) g (Y, Z)}{p (R \mid \mb (R))} \bbE [\frac{1}{p (R \mid \mb (R))} \mid Z] - \bbE [\frac{f (X, Z)}{p (R \mid \mb (R))} \mid Z] \frac{g (Y, Z)}{p (R \mid \mb (R))} \\
& \qquad - \frac{f (X, Z)}{p (R \mid \mb (R))} \bbE [\frac{g (Y, Z)}{p (R \mid \mb (R))} \mid Z] + \frac{\bbE [\frac{f (X, Z)}{p (R \mid \mb (R))} \mid Z]}{p (R \mid \mb (R))} \frac{\bbE [\frac{g (Y, Z)}{p (R \mid \mb (R))} \mid Z]}{\bbE [\frac{1}{p (R \mid \mb (R))} \mid Z]} \Big\} \mid Z \Big].
\end{align*}
Together with the constraint $\alpha_{0} = 0$, we have
\begin{align*}
& \ \bbE \Big[ \frac{f (X, Z) g (Y, Z)}{p (R \mid \mb (R))} \iota (\sfM, X, Y, Z) \mid Z \Big] \bbE \Big[ \frac{1}{p (R \mid \mb (R))} \mid Z \Big] + \bbE \Big[ \frac{f (X, Z) g (Y, Z)}{p (R \mid \mb (R))} \mid Z \Big] \bbE \Big[ \frac{\iota (\sfM, Z)}{p (R \mid \mb (R))} \mid Z \Big] \\
& - \bbE \Big[ \frac{f (X, Z)}{p (R \mid \mb (R))} \mid Z \Big] \bbE \Big[ \frac{g (Y, Z) \iota (\sfM, Y, Z)}{p (R \mid \mb (R))} \mid Z \Big] - \bbE \Big[ \frac{f (X, Z) \iota (\sfM, X, Z)}{p (R \mid \mb (R))} \mid Z \Big] \bbE \Big[ \frac{g (Y, Z)}{p (R \mid \mb (R))} \mid Z \Big] = 0.
\end{align*}
The proof is complete.
\end{proof}

The next result records a few immediate consequences of the particular construction of $r_{t} (X, Y, Z \mid \sfM)$ given in \eqref{ansatz 1}--\eqref{extra constraint 2}.

\begin{lemma}
\label{lem:alpha}
The following statements hold with the specific choice of $r_{t} (X, Y, Z \mid \sfM)$ that satisfies \eqref{ansatz 1}--\eqref{extra constraint 2}.
\begin{enumerate}[label = (\arabic*)]
\item $\bbE \Big[ \dfrac{r_{t} (Z \mid \sfM) + s (Z \mid \sfM) s (\mb (R))}{p (R \mid \mb (R))} \mid Z \Big] = 0$

\item $\bbE \Big[ \dfrac{f (X, Z) r_{t} (X, Z \mid \sfM) s (\mb (R))}{p (R \mid \mb (R))} \mid Z \Big] = \bbE \Big[ \dfrac{g (Y, Z) r_{t} (Y, Z \mid \sfM) s (\mb (R))}{p (R \mid \mb (R))} \mid Z \Big] = 0$.

\item $\bbE \Big[ \dfrac{r_{t} (Z \mid \sfM) s (\mb (R))}{p (R \mid \mb (R))} \mid Z \Big] = 0$.
\end{enumerate}
\end{lemma}

\begin{proof}
All three statements can be directly derived from \eqref{ansatz 1}--\eqref{extra constraint 2} by realizing that the left hand sides of (1)--(3) can be equivalently represented as the conditional means given $Z$ in the post-fixing distribution $\phi_{R} (p)$, once being renormalized by $\bbE [\frac{1}{p (R \mid \mb (R))} \mid Z]$.
\end{proof}

\subsection{On Basis Approximations}
\label{app:basis_approximation}

The orthocomplement characterization in
Theorem~\ref{thm:orthocomplement_tangent_space} is infinite-dimensional, since it is indexed by arbitrary square-integrable functions $f$ and $g$. In practice, one may approximate the corresponding projection by restricting $f$ and $g$ to finite-dimensional basis expansions. For example, let $\{b_j(X)\}_{j=1}^J$, $\{c_k(Z)\}_{k=1}^K$, and $\{d_\ell(Y)\}_{\ell=1}^L$ be collections of basis functions, and consider
\begin{align*}
f(X,Z)
&=
\sum_{j=1}^J
\sum_{k=1}^K
\beta_{jk}
\,b_j(X)c_k(Z),
\qquad 
g(Y,Z)
=
\sum_{\ell=1}^L
\sum_{k=1}^K
\gamma_{\ell k}
\,d_\ell(Y)c_k(Z).
\end{align*}
The resulting collection of weighted residual directions $\chi_{f,g}$ spans a finite-dimensional subspace of $\Lambda_{\bbP}^{\perp}$. Projecting the nonparametric influence function onto this subspace yields a computable approximation to the efficient influence function. Increasingly rich basis collections recover the full projection in the limit under standard sieve approximation conditions. 

\subsection{Proof of Theorem~\ref{theorem:locally_eff}}
\label{app:theorem:locally_eff}

\begin{proof}
By Theorem~\ref{thm:orthocomplement_tangent_space}, each $\chi_{jk}(O)\equiv \chi_{f_j,g_k}(O)$ belongs to $\Lambda_{\bbP}^{\perp}$. Hence the finite-dimensional subspace
\begin{align*}
\Lambda_{\bbP}^{\perp,\sub}
=
\,\{\chi_{jk}(O):j=1,\ldots,J,\;k=1,\ldots,K\}
\end{align*}
is contained in $\Lambda_{\bbP}^{\perp}$. Index the $JK$ functions $\chi_{jk}$ into the vector
\begin{align*}
\chi(O)
=
\big(\chi_{11}(O),\ldots,\chi_{JK}(O)\big)^\top .
\end{align*}
Then every element of $\Lambda_{\bbP}^{\perp,\sub}$ can be written as $\alpha^\top\chi(O)$ for some $\alpha\in\bbR^{JK}$, so
\begin{align*}
\Lambda_{\bbP}^{\perp,\sub}
=
\left\{
\alpha^\top\chi(O):\alpha\in\bbR^{JK}
\right\}.
\end{align*}

The projection of $\varphi^{\mathrm{np}}(O)$ onto $\Lambda_{\bbP}^{\perp,\sub}$ is therefore the element of this subspace that minimizes squared distance:
\begin{align*}
\Pi\big(
\varphi^{\mathrm{np}}(O)
\mid
\Lambda_{\bbP}^{\perp,\sub}
\big)
=
{\left(\alpha^{\mathrm{opt}}\right)}^{\top}\chi(O),
\end{align*}
where
\begin{align*}
\alpha^{\mathrm{opt}}
=
\arg\min_{\alpha\in\bbR^{JK}}
\bbE
\left[
\left\{
\varphi^{\mathrm{np}}(O)-\alpha^\top\chi(O)
\right\}^2
\right].
\end{align*}
Define
\begin{align*}
L(\alpha)
=
\bbE
\left[
\left\{
\varphi^{\mathrm{np}}(O)-\alpha^\top\chi(O)
\right\}^2
\right].
\end{align*}
Expanding the quadratic form gives
\begin{align*}
L(\alpha)
&=
\bbE\left[\{\varphi^{\mathrm{np}}(O)\}^2\right]
-
2\alpha^\top
\bbE\left[\chi(O)\varphi^{\mathrm{np}}(O)\right]
+
\alpha^\top
\bbE\left[\chi(O)\chi(O)^\top\right]
\alpha.
\end{align*}
Let
\begin{align*}
\Sigma
=
\bbE\left[\chi(O)\chi(O)^\top\right],
\qquad
b
=
\bbE\left[\chi(O)\varphi^{\mathrm{np}}(O)\right].
\end{align*}
Then
\begin{align*}
L(\alpha)
=
\bbE\left[\{\varphi^{\mathrm{np}}(O)\}^2\right]
-
2\alpha^\top b
+
\alpha^\top\Sigma\alpha.
\end{align*}
Differentiating with respect to $\alpha$ yields
\begin{align*}
\frac{\diff}{\diff\alpha}L(\alpha)
=
-2b+2\Sigma\alpha.
\end{align*}
Setting this derivative equal to zero gives the normal equation
\begin{align*}
\Sigma\alpha^{\mathrm{opt}}
=
b.
\end{align*}
Thus, when $\Sigma$ is nonsingular,
\begin{align*}
\alpha^{\mathrm{opt}}
=
\Sigma^{-1}b.
\end{align*}

The approximated efficient influence function obtained by projecting onto the complement of the basis approximation is
\begin{align*}
\varphi^{\mathrm{eff,basis}}(O)
=
\varphi^{\mathrm{np}}(O)
-
{\left(\alpha^{\mathrm{opt}}\right)}^{\top}\chi(O).
\end{align*}
Since both $\varphi^{\mathrm{np}}(O)$ and $\chi(O)$ are mean zero, the variance is
\begin{align*}
\Var[\varphi^{\mathrm{eff,basis}}(O)]
&=
\Var[\varphi^{\mathrm{np}}(O)]
+
{\left(\alpha^{\mathrm{opt}}\right)}^{\top}
\bbE[\chi(O)\chi(O)^\top]
\alpha^{\mathrm{opt}}
-
2{\left(\alpha^{\mathrm{opt}}\right)}^{\top}
\bbE[\chi(O)\varphi^{\mathrm{np}}(O)]
\\
&=
\Var[\varphi^{\mathrm{np}}(O)]
+
{\left(\alpha^{\mathrm{opt}}\right)}^{\top}\Sigma\alpha^{\mathrm{opt}}
-
2{\left(\alpha^{\mathrm{opt}}\right)}^{\top}b
\\
&=
\Var[\varphi^{\mathrm{np}}(O)]
-
b^\top\Sigma^{-1}b,
\end{align*}
where the last equality uses $\alpha^{\mathrm{opt}}=\Sigma^{-1}b$. Therefore,
\begin{align*}
\Var[\varphi^{\mathrm{np}}(O)]
-
\Var[\varphi^{\mathrm{eff,basis}}(O)]
=
b^\top\Sigma^{-1}b.
\end{align*}
\end{proof}

\subsection{On the Mean-Scale Model \texorpdfstring{$\calP_{\mathrm{mean},h}$}{}}
\label{app:eq:mean_scale_verma}

Equation~\eqref{eq:mean_scale_verma} holds if and only if 
\begin{align} 
\bbE_{\phi_R} \left[ \left\{ f(X,Z) - \bbE_{\phi_R}[f(X,Z)\mid Z] \right\} h(Y) \right] = 0, \qquad \text{for every } f\in L^2(\bbP_{X,Z}). 
\label{eq:mean_scale_moment} 
\end{align} 
Using Proposition~\ref{prop:weighted_moment_representation}, these restrictions may be written under the observed law as 
\begin{align} 
\bbE \left[ \omega_R(\M) \left\{ f(X,Z) - \bbE_{\phi_R}[f(X,Z)\mid Z] \right\} h(Y) \right] = 0, \qquad f\in L^2(\bbP_{X,Z}). 
\label{eq:mean_scale_weighted_moment} 
\end{align}

Thus, the mean-scale model $\calP_{\mathrm{mean},h}$ is obtained by restricting the $Y$-side function in the general moment representation to $h$. Under the same local regularity condition used in Theorem~\ref{thm:orthocomplement_tangent_space}, its tangent-space orthocomplement is given by \eqref{eq:mean_scale_orthocomplement}. 

\begin{proof}
Define
\begin{align*}
m_h(X,Z)
\coloneqq
\bbE_{\phi_R}[h(Y)\mid X,Z].
\end{align*}
For every $f\in L^2(\bbP_{X,Z})$, iterated expectation gives
\begin{align}
&\bbE_{\phi_R}
\left[
\left\{
f(X,Z)-\bbE_{\phi_R}[f(X,Z)\mid Z]
\right\}
h(Y)
\right]
\nonumber\\
&\qquad =
\bbE_{\phi_R}
\left[
\left\{
f(X,Z)-\bbE_{\phi_R}[f(X,Z)\mid Z]
\right\}
m_h(X,Z)
\right].
\label{eq:mean_scale_moment_under_fixing}
\end{align}

Suppose first that \eqref{eq:mean_scale_verma} holds. Then
\begin{align*}
m_h(X,Z)=\bbE_{\phi_R}[h(Y)\mid Z]
\end{align*}
almost surely. Substituting this equality into \eqref{eq:mean_scale_moment_under_fixing} and conditioning on $Z$ gives
\begin{align*}
\bbE_{\phi_R}
\left[
\left\{
f(X,Z)-\bbE_{\phi_R}[f(X,Z)\mid Z]
\right\}
h(Y)
\right]
=0.
\end{align*}
Hence \eqref{eq:mean_scale_moment} holds for every $f\in L^2(\bbP_{X,Z})$.

Conversely, suppose that \eqref{eq:mean_scale_moment} holds for every $f \in L^2 (\bbP_{X,Z})$, and define
\begin{align*}
\Delta_h(X,Z)
\coloneqq
m_h(X,Z)-\bbE_{\phi_R}[m_h(X,Z)\mid Z].
\end{align*}
By iterated expectation,
\begin{align*}
\Delta_h(X,Z)
=
\bbE_{\phi_R}[h(Y)\mid X,Z]
-
\bbE_{\phi_R}[h(Y)\mid Z],
\qquad
\bbE_{\phi_R}[\Delta_h(X,Z)\mid Z]=0.
\end{align*}
Moreover, $\Delta_h\in L^2(\bbP_{X,Z})$ under the stated integrability conditions. Taking $f=\Delta_h$ in
\eqref{eq:mean_scale_moment_under_fixing} yields
\begin{align*}
0
&=
\bbE_{\phi_R}[\Delta_h(X,Z)m_h(X,Z)]
\\
&=
\bbE_{\phi_R}
\left[
\Delta_h(X,Z)
\left\{
\Delta_h(X,Z)
+
\bbE_{\phi_R}[m_h(X,Z)\mid Z]
\right\}
\right]
\\
&=
\bbE_{\phi_R}[\Delta_h(X,Z)^2].
\end{align*}
Thus $\Delta_h(X,Z)=0$ almost surely, proving \eqref{eq:mean_scale_verma}. This establishes the equivalence between \eqref{eq:mean_scale_verma} and \eqref{eq:mean_scale_moment}. Proposition~\ref{prop:weighted_moment_representation} then gives the observed-law representation \eqref{eq:mean_scale_weighted_moment}.

It remains to characterize the orthocomplement. Let
$\{\bbP_t:t\in(-\epsilon,\epsilon)\}\subseteq
\calP_{\mathrm{mean},h}$ be a regular parametric submodel through
$\bbP$ with score $s$. Differentiating
\eqref{eq:mean_scale_weighted_moment} at $t=0$ and applying the
calculation in the proof of
Theorem~\ref{thm:orthocomplement_tangent_space} with
$g(Y,Z)=h(Y)$ gives
\begin{align*}
\bbE[\chi_{f,h}(O)s(O)]=0,
\qquad
f\in L^2(\bbP_{X,Z}).
\end{align*}
Therefore,
\begin{align*}
\overline{\sp}
\left\{
\chi_{f,h}:
f\in L^2(\bbP_{X,Z})
\right\}
\subseteq
\Lambda_{\bbP,\mathrm{mean},h}^{\perp}. 
\end{align*}

Conversely, the same completeness argument used in the proof of Theorem~\ref{thm:orthocomplement_tangent_space} shows that these differentiated moment restrictions exhaust the first-order restrictions defining $\calP_{\mathrm{mean},h}$. Hence a mean-zero square-integrable function is orthogonal to every score in $\Lambda_{\bbP,\mathrm{mean},h}$ if and only if it belongs to the closed linear span of the residualized moment gradients $\chi_{f,h}$. Therefore,
\begin{align*}
\Lambda_{\bbP,\mathrm{mean},h}^{\perp}
=
\overline{\sp}
\left\{
\chi_{f,h}
:
f\in L^2(\bbP_{X,Z})
\right\},
\end{align*}
which proves \eqref{eq:mean_scale_orthocomplement}. Note that, since each $\chi_{f,h}$ is obtained from the full distributional family by taking $g(Y,Z)=h(Y)$, $\Lambda_{\bbP,\mathrm{mean},h}^{\perp} \subseteq \Lambda_{\bbP}^{\perp}$. 
\end{proof}

\subsection{Proof of Theorem~\ref{thm:indexed_invariance_orthocomplement}}
\label{app:thm:indexed_invariance_orthocomplement}

\begin{proof}
For each $x\in\cal X$, define
\begin{align*}
\Gamma_x(\bbP)
\coloneqq
\psi(\bbP;x)-\psi(\bbP;x_0).
\end{align*}
By definition of $\calP_\psi$,
\begin{align*}
\Gamma_x(\bbP)=0,
\qquad
x\in\cal X,
\end{align*}
for every $\bbP\in\calP_\psi$.

Let $\{\bbP_t:t\in(-\epsilon,\epsilon)\}\subseteq\calP_\psi$ be a regular parametric submodel through $\bbP$ with score $s\in\Lambda_{\bbP,\psi}$. Since each indexed functional is pathwise differentiable in the nonparametric model,
\begin{align*}
0 = \left.
\frac{\diff}{\diff t}
\Gamma_x (\bbP_t) \right|_{t=0} = \bbE \left[
\left\{
\varphi^{\mathrm{np}}(\bbP;x)
-
\varphi^{\mathrm{np}}(\bbP;x_0)
\right\}
s(O) \right].
\end{align*}
Thus,
\begin{align*}
\varphi^{\mathrm{np}}(\bbP;x)
-
\varphi^{\mathrm{np}}(\bbP;x_0)
\in
\Lambda_{\bbP,\psi}^{\perp}
\end{align*}
for every $x\in\cal X$, and therefore
\begin{align}
\overline{\sp}
\left\{
\varphi^{\mathrm{np}}(\bbP;x)
-
\varphi^{\mathrm{np}}(\bbP;x_0)
:
x\in\cal X
\right\}
\subseteq
\Lambda_{\bbP,\psi}^{\perp}.
\label{eq:indexed_invariance_forward_inclusion}
\end{align}

Conversely, by assumption, the indexed invariance restrictions exhaust the first-order restrictions defining $\calP_\psi$ at $\bbP$. Their gradients are precisely
\begin{align*}
\varphi^{\mathrm{np}}(\bbP;x)
-
\varphi^{\mathrm{np}}(\bbP;x_0),
\qquad
x\in\cal X.
\end{align*}
Hence every element of $\Lambda_{\bbP,\psi}^{\perp}$ belongs to the closed linear span of these gradients, proving
\begin{align*}
\Lambda_{\bbP,\psi}^{\perp}
=
\overline{\sp}
\left\{
\varphi^{\mathrm{np}}(\bbP;x)
-
\varphi^{\mathrm{np}}(\bbP;x_0)
:
x\in\cal X
\right\}.
\end{align*}

To obtain the equivalent integral representation, observe that every finite linear combination of the indexed differences can be written as
\begin{align*}
\sum_{k=1}^K
a_k
\left\{
\varphi^{\mathrm{np}}(\bbP;x_k)
-
\varphi^{\mathrm{np}}(\bbP;x_0)
\right\}
=
\int
\varphi^{\mathrm{np}}(\bbP;x)c(x)\diff x,
\end{align*}
for a signed weighting function $c$ satisfying
\begin{align*}
\int c(x)\diff x=0.
\end{align*}
Conversely, every such integral belongs to the closed linear span of
the indexed differences. Taking the $L^2(\bbP)$ closure proves
\eqref{eq:indexed_invariance_integral_orthocomplement}.

It remains to characterize the influence-function class. Since
\begin{align*}
\psi(\bbP;x_0)=\psi(\bbP)
\end{align*}
throughout $\calP_\psi$,
$\varphi^{\mathrm{np}}(\bbP;x_0)$ is an influence function for
$\psi(\bbP)$ under $\calP_\psi$. Hence,
\begin{align*}
\calI \{\psi(\bbP)\}
&=
\varphi^{\mathrm{np}}(\bbP;x_0)
+
\Lambda_{\bbP,\psi}^{\perp}.
\end{align*}
Using \eqref{eq:indexed_invariance_integral_orthocomplement}, an arbitrary
element of this affine space can be written as
\begin{align*}
\varphi^{\mathrm{np}}(\bbP;x_0)
+
\int
\varphi^{\mathrm{np}}(\bbP;x)c(x)\diff x
=
\int
\varphi^{\mathrm{np}}(\bbP;x)
\widetilde p(x)\diff x,
\end{align*}
where $\int c(x)\diff x=0$ and
$\widetilde p$ satisfies
\begin{align*}
\int\widetilde p(x)\diff x=1.
\end{align*}
Conversely, for any $\widetilde p$ integrating to one,
\begin{align*}
\int
\varphi^{\mathrm{np}}(\bbP;x)
\widetilde p(x)\diff x
-
\varphi^{\mathrm{np}}(\bbP;x_0)
\end{align*}
is a zero-total-weight combination of the indexed influence functions
and therefore belongs to $\Lambda_{\bbP,\psi}^{\perp}$. This proves
\eqref{eq:indexed_invariance_IF_class}.

Finally, the efficient influence function is the minimum-variance
element of the influence-function class, so
\begin{align*}
\varphi^{\mathrm{eff}}
=
\argmin_{
\varphi\in\calI \{\psi(\bbP)\}
}
\bbE[\varphi(O)^2],
\end{align*}
which proves \eqref{eq:indexed_invariance_variational}.
\end{proof}

\subsection{Proof of Lemma~\ref{lem:optimal_affine_IF_combination}}
\label{app:lem:optimal_affine_IF_combination}

\begin{proof}
Because each $G_k$ is an influence function for $\psi(\bbP)$ under $\calP_{\mathrm{mean},h}$, every affine combination
\begin{align*}
G_\alpha(O)
=
\alpha^\top \boldsymbol G(O),
\qquad
\boldsymbol 1^\top \alpha=1,
\end{align*}
is also an influence function for $\psi(\bbP)$. Indeed, for every score $s\in\Lambda_{\bbP,\mathrm{mean},h}$,
\begin{align*}
\bbE[G_\alpha(O)s(O)] =
\sum_{k=1}^K
\alpha_k
\bbE[G_k(O)s(O)] =
\left(
\sum_{k=1}^K\alpha_k
\right)
\dot\psi_{\bbP}(s) =
\dot\psi_{\bbP}(s),
\end{align*}
where the final equality follows from $\boldsymbol 1^\top\alpha=1$.

Since influence functions are mean zero, the variance of $G_\alpha$ is
\begin{align*}
\Var_{\bbP}\{G_\alpha(O)\} = \bbE
\left[
\left\{
\alpha^\top\boldsymbol G(O)
\right\}^2
\right] =
\alpha^\top
\Sigma
\alpha.
\end{align*}
Thus, the minimum-variance element of the affine class solves
\begin{align*}
\min_{\alpha\in\bbR^K}
\alpha^\top\Sigma\alpha
\qquad
\text{subject to}
\qquad
\boldsymbol 1^\top\alpha=1.
\end{align*}

Suppose $\Sigma$ is nonsingular. Since $\Sigma$ is a covariance matrix, nonsingularity implies that it is positive definite, so the constrained optimization problem has a unique solution. Its Lagrangian is
\begin{align*}
\mathcal L(\alpha,\lambda)
=
\alpha^\top\Sigma\alpha
-
2\lambda
\left(
\boldsymbol 1^\top\alpha-1
\right).
\end{align*}
The first-order condition with respect to $\alpha$ is
\begin{align*}
2\Sigma\alpha
-
2\lambda\boldsymbol 1
=
0,
\end{align*}
and hence
\begin{align*}
\alpha
=
\lambda\Sigma^{-1}\boldsymbol 1.
\end{align*}
Imposing the constraint $\boldsymbol 1^\top\alpha=1$ gives
$1 = \lambda \boldsymbol 1^\top \Sigma^{-1} \boldsymbol 1$, 
so that
\begin{align*}
\lambda
=
\frac{1}{
\boldsymbol 1^\top
\Sigma^{-1}
\boldsymbol 1
}.
\end{align*}
Therefore,
\begin{align*}
\alpha^{\mathrm{opt}}
=
\frac{
\Sigma^{-1}\boldsymbol 1
}{
\boldsymbol 1^\top
\Sigma^{-1}\boldsymbol 1
}.
\end{align*}
This yields the unique minimum-variance influence function in the specified affine class: $G_{\alpha^{\mathrm{opt}}} = \left(\alpha^{\mathrm{opt}} \right)^{\top}\boldsymbol G(O)$.
\end{proof}

\subsection{Proof of Proposition~\ref{thm:multiple_constraints}}
\label{app:thm:multiple_constraints}

\begin{proof}
Proposition~\ref{thm:multiple_constraints} follows immediately from Theorem~\ref{thm:orthocomplement_tangent_space} and Lemma~1.7 of \citet{van2003unified}.
\end{proof}

As mentioned in Remark~\ref{rem:conjecture}, the reverse direction of Proposition~\ref{thm:multiple_constraints} remains an open problem. To close the gap, one at least needs to show that no further \emph{hidden constraints} are implied by those encoded in $\calP_{j}$ for $j = 1, 2, \cdots, J$. We expect that if $\calP$ is exactly described by some ADMG, then the reverse direction should hold given the strong evidence in \citet{evans2018margins} for vertices with finite state spaces. But to our knowledge, extending such results to infinite and uncountable state spaces remains open.

\subsection{Proof of Remark~\ref{rem:DAG}}
\label{app:DAG}

\begin{proof}
The proof is straightforward once we realize that, for a DAG model, there are only ordinary conditional independence constraints, and these constraints must follow the same ordering as in the Bayesian network decomposition. $\Lambda_{\bbP}^{\perp}$ is in fact an orthogonal summation, and each summand must be orthogonal to the corresponding conditional score function given in \eqref{DAG tangent spaces}.
\end{proof}

\clearpage

\section{Proofs for the Illustrative Examples} \label{app:proofs_applications}

\subsection{The Extended Front-Door Model}

\subsubsection{Proof of Corollary~\ref{cor:tangent_space_frontdoor}}
\label{app:cor:tangent_space_frontdoor}

\begin{proof}
The nested conditional independence induced by fixing $M$ is
\begin{align*}
Z \indep Y \mid X,M \qquad [\phi_M(p;\calG)].
\end{align*}
Thus, the result follows by specializing Proposition~\ref{prop:weighted_moment_representation} and Theorem~\ref{thm:orthocomplement_tangent_space} with
\begin{align*}
X_{\mathrm{gen}}=Z,
\qquad
Y_{\mathrm{gen}}=Y,
\qquad
Z_{\mathrm{gen}}=(X,M),
\qquad
R=M.
\end{align*}

By Proposition~\ref{prop:weighted_moment_representation},  
\begin{equation}\label{eq:app_frontdoor_full_moment}
\begin{split}
\calP(\calG)
& =
\Bigg\{
\bbP:
\bbE_{\phi_{M} (p; \calG)}
\left[ \left\{
f(X,Z, M)
-
\bbE_{\phi_{M} (p; \calG)} [f(X, Z, M) \mid X, M]
\right\} g (X,M,Y)
\right]
=
0
\Bigg\}\\
&= 
\Bigg\{
\bbP:
\bbE
\left[ \omega_M(\M)\, \left\{
f(X,Z, M)
-
\bbE_{\phi_{M} (p; \calG)} [f(X, Z, M) \mid X, M]
\right\}  \, g(X,M,Y)
\right]
=
0
\Bigg\}, 
\end{split}
\end{equation}
for every pair of square-integrable functions $f \in L^2 (\bbP_{X, Z, M})$ and $g \in L^2 (\bbP_{X, M, Y})$, where $\M = (M, X, Z, A)$ and 
\begin{equation*}
\bbE_{\phi_M} [f(X, Z, M) \mid X, M]
=
\frac{
\bbE [\, \omega_M(\M) \, f(X, Z, M) \mid X, M]
}{
\bbE [\, \omega_M(\M) \mid X, M]
},
\qquad
\omega_M(\M)
\coloneqq 
\frac{\tilde p(M)}{
p (M \mid X, Z, A)
}. 
\end{equation*}

Since $\tilde p(M)$ is known and strictly positive on the relevant support, multiplication by $\tilde p(M)$ may be absorbed into the arbitrary function $g(X,M,Y)$. Furthermore, by definition, $\tilde p(M)$ cancels our from the numerator and denominator definition of $\bbE_{\phi_M} [ f(.) \mid X, M ]$. Thus, we can ignore $\tilde p(M)$ and write: 
\begin{align*}
    \omega_M(\M) \coloneqq  \frac{1}{p (M \mid X, Z, A)}. 
\end{align*}

We next show that the indexing function $f$ in \eqref{eq:app_frontdoor_full_moment} can be reduced from $f(X,Z,M)$ to $f(X,Z)$. First consider a simple tensor-product function of the form
\begin{align*}
f(X,Z,M)
=
a(M)h(X,Z).
\end{align*}
It follows that
\begin{align*}
&\left\{
a(M)h(X,Z)
-
\bbE_{\phi_M}[a(M)h(X,Z)\mid X,M]
\right\}
g(X,M,Y)
\\
&\qquad =
\left\{
h(X,Z)
-
\bbE_{\phi_M}[h(X,Z)\mid X]
\right\}
\{a(M)g(X,M,Y)\}.
\end{align*}
where the equality holds since $\bbE_{\phi_M}[h(X,Z)\mid X,M] = \bbE_{\phi_M}[h(X,Z)\mid X]$. Since $g(X,M,Y)$ is arbitrary, the factor $a(M)$ can be absorbed into $g$.

By linearity, the same conclusion holds for finite sums
\begin{align*}
f(X,Z,M)
=
\sum_{\ell=1}^L a_\ell(M)h_\ell(X,Z).
\end{align*}
Such finite sums are dense in $L^2(\bbP_{X,Z,M})$ under the usual product-measurability and integrability conditions. Hence the closed linear span of moment functions generated by arbitrary $f(X,Z,M)$ is the same as the closed linear span generated by functions of $(X,Z)$ alone, after allowing arbitrary $g(X,M,Y)$. Therefore the model may be written as
\begin{align*}
\calP(\calG)
=
\Bigg\{
\bbP:
\bbE
\left[
\omega_M(\M)
\left\{
f(X,Z)
-
\bbE_{\phi_M}[f(X,Z)\mid X]
\right\}
g(X,M,Y)
\right]
=
0
\Bigg\},
\end{align*}
for every $f\in L^2(\bbP_{X,Z})$ and $g\in L^2(\bbP_{X,M,Y})$, where
\begin{align*}
\bbE_{\phi_M}[f(X,Z)\mid X]
=
\frac{
\bbE[\omega_M(\M)f(X,Z)\mid X]
}{
\bbE[\omega_M(\M)\mid X]
}, \qquad \omega_M(\M) \coloneqq  \frac{1}{p (M \mid X, Z, A)}.
\end{align*}

Likewise, Theorem~\ref{thm:orthocomplement_tangent_space} immediately gives
\begin{align*}
\Lambda_{\bbP}^{\perp}
=
\overline{\sp}
\left\{
\chi_{f,g}(O):
f\in L^2(\bbP_{X,Z}),
\;
g\in L^2(\bbP_{X,M,Y})
\right\},
\end{align*}
where
\begin{align*}
\widetilde\chi_{f,g}(O)
&=
\omega_M(\M)
\Big\{
f(X,Z)
-
\bbE_{\phi_M}[f(X,Z)\mid X]
\Big\}
\Big\{
g(X,M,Y)
-
\bbE_{\phi_M}[g(X,M,Y)\mid X, M]
\Big\},
\\
\chi_{f,g}(O)
&=
\widetilde\chi_{f,g}(O)
-
\bbE
[
\widetilde\chi_{f,g}(O)
\mid
X,Z,A,M
]
+
\bbE
[
\widetilde\chi_{f,g}(O)
\mid
X,Z,A
].
\end{align*}
The final expression follows because $\mb_{\calG}(M)=(X,Z,A)$.
\end{proof}

\subsubsection{Identification and NP-EIF} 
\label{app:id_npEIF_frontdoor} 

Let $V=(X,Z,A,M,Y)$, and assume treatment $A$ is binary. Fixing $M$ produces the CADMG $\calG(V\setminus M; M)$ and the corresponding post-fixing kernel
\begin{align*}
q_{V\setminus M\mid M}(X,Z,A,Y\mid M)
&= \frac{p(V)}{p(M\mid X,Z,A)} \\
&= p(Y\mid X,Z,A,M)\, p(A\mid X,Z)\, p(Z,X).
\end{align*}

The extended front-door graph implies the nested conditional independence
\begin{align*}
   Z \indep Y \mid X,M \qquad [\phi_M(p;\calG)]. 
\end{align*}
Equivalently, under the post-fixing kernel,
\begin{align*}
q_{V\setminus M\mid M}(Y\mid X,Z,M)
=
q_{V\setminus M\mid M}(Y\mid X,M),
\end{align*}
or, written in terms of the observed-data distribution,
\begin{align*}
q_{V\setminus M\mid M}(y\mid x,z,m) = \sum_{a=0}^1 p(y\mid x,z,A=a,m) \, p(A=a\mid x,z)
\end{align*}
is invariant to $z$ for every $(x,m,y)$. 

A mean-scale Verma restriction is obtained by choosing $h(Y)=Y$ in Section~\ref{subsec:mean_scale_verma}, yielding 
\begin{align*}
\bbE_{\phi_M}
\left[
Y
\mid
X,Z,M
\right]
=
\bbE_{\phi_M}
\left[
Y
\mid
X,M
\right].
\end{align*}
Equivalently,
\begin{align*}
\eta(\bbP;x,m,z)
\coloneqq
\sum_{a=0}^1
\bbE[Y\mid X=x,Z=z,A=a,M=m]\,
p(A=a\mid X=x,Z=z)
\end{align*}
does not depend on $z$.

Following the notation of Section~\ref{subsec:mean_scale_verma}, define the $\bbP$-dependent operator
\begin{align*}
\mathcal L_{\bbP}\{\eta\}
=
\iiint
\eta(\bbP; x,m,z^*)
\,f_M(m\mid A=a_0,Z=z,X=x)
\,p(z,x)
\,\diff m\,\diff z\,\diff x.
\end{align*}
Applying $\mathcal L_{\bbP}$ to the post-fixing conditional mean produces the identification functional
\begin{align*}
\psi_{a_0}(\bbP;z^*)
=
\mathcal L_{\bbP}
\left\{
\eta(\bbP;\cdot,\cdot,z^*)
\right\},
\end{align*}
which can be written explicitly as
\begin{align}
\psi_{a_0}(\bbP;z^*)
=
\iiint
\sum_{a=0}^1
\mu(m,a,z^*,x)
\,
\pi(a\mid z^*,x)
\,
p(m\mid A=a_0,z,x)
\,
p(z,x)
\,
\diff m\,\diff z\,\diff x,
\label{eq:id_ATE_zstar}
\end{align}
where $\mu(m,a,z,x) = \bbE[Y\mid M=m,A=a,Z=z,X=x]$, $\pi(a\mid z,x) = p(A=a\mid Z=z,X=x)$, and $f_M(m \mid a, z, x) = p(M=m \mid A=a, Z=z, X=x)$. 

The nonparametric influence function of \eqref{eq:id_ATE_zstar}, denoted by $\varphi^\mathrm{np}_{a_0}(\bbP; z^*)$, is 
\begin{align}
\varphi^\mathrm{np}_{a_0}(\bbP; z^*)(O_i) 
   &=  \frac{\I(Z_i=z^*)}{f_Z(z^*\mid X_i) }\sum_z f^r_{M, z^*}(M_i, A_i, z, X_i) \ f_Z(z\mid X_i) \ \big(Y_i-\mu(M_i,A_i,z^*,X_i) \big) \label{eq:verma_if}  \\
   &\hspace{0.5cm}+\frac{\I(Z=z^*)}{f_Z(z^*\mid X_i)}(A_i-\pi(1\mid z^*,X_i))\ \sum_z\big(\kappa_{1, z^*}(z,X_i)-\kappa_{0, z^*}(z,X_i)\big)f_Z(z\mid X_i) 
   \notag \\ 
    &\hspace{0.5cm}+ \frac{\I(A_i=a_0)}{\pi(a_0\mid Z_i,X_i)}\big(\xi_{z^*}(M_i,X_i) - \gamma_{z^*}(Z_i,X_i)\big) \\
   &\hspace{0.5cm}+ \gamma_{z^*}(Z_i,X_i) - \psi_{a_0}(\bbP; z^*),  \notag 
\end{align}
where 
$f_Z(z\mid x) = p(Z=z \mid X=x)$, 
$f_{M, z^*}^r(m, a, z, x) = {f_M(m \mid a_0, z, x)}/{f_M(m \mid a, z^*, x)}$, 
$\xi_{z^*}(m,x) = \sum_a \mu(m, a, z^*, x) \, \pi(a \mid z^*, x)$,
$\gamma_{z^*}(z,x) = \bbE[\xi_{z^*}(M,X) \mid  a_0,z,x]$, and 
$\kappa_{a, z^*}(z, x) = \bbE[\mu(M, a, z^*, X] \mid a_0, z, x)$.

\subsubsection{Proof of Lemma~\ref{lem:frontdoor_representation}}
\label{app:lem:frontdoor_representation} 

We show that the target-specific orthocomplement
$\Lambda^\perp_{\bbP,\psi}$, generated by the indexed invariance condition, is obtained as a specialization of the full orthocomplement to the tangent space of the extended front-door model characterized in Corollary~\ref{cor:tangent_space_frontdoor}. As discussed in the main text, the semiparametric efficient influence function within the corresponding affine class associated with $\Lambda^\perp_{\bbP,\psi}$ coincides with the optimal weighted influence function of \citet{guo2023flexible}.

\begin{proof}
Write
\begin{align*}
\mu(m,a,z,x)
&=
\bbE[Y\mid M=m,A=a,Z=z,X=x],
\\
\pi(a\mid z,x)
&=
p(A=a\mid Z=z,X=x),
\\
f_M(m\mid a,z,x)
&=
p(M=m\mid A=a,Z=z,X=x),
\\
f_Z(z\mid x)
&=
p(Z=z\mid X=x),
\end{align*}
and define
\begin{align*}
\xi_{z^*}(m,x)
&=
\sum_a
\mu(m,a,z^*,x)\pi(a\mid z^*,x),
\\
\gamma_{z^*}(z,x)
&=
\int
\xi_{z^*}(m,x)
f_M(m\mid a_0,z,x)
\diff m,
\\
\kappa_{a,z^*}(z,x)
&=
\int
\mu(m,a,z^*,x)
f_M(m\mid a_0,z,x)
\diff m.
\end{align*}

Under the post-fixing kernel, the conditional distribution of $Z$
given $X$ is $f_Z(\cdot\mid X)$. Therefore,
\begin{align*}
\bbE_{\phi_M}
\left[
f_\alpha(X,Z)\mid X
\right]
&=
\int
\frac{\alpha(z)}{f_Z(z\mid X)}
f_Z(z\mid X)
\diff z
\\
&=
\int \alpha(z)\diff z
\\
&=
0.
\end{align*}

The nested conditional independence
\begin{align*}
Z\indep Y\mid X,M
\qquad
[\phi_M(\bbP;\calG)]
\end{align*}
implies that
\begin{align*}
\xi_{z^*}(m,x)
=
\xi(m,x)
\end{align*}
does not depend on $z^*$. Hence
\begin{align*}
\bbE_{\phi_M}
\left[
g_\rho(X,M,Y)\mid X,M
\right]
&=
\rho(M,X)
\bbE_{\phi_M}[Y\mid X,M]
\\
&=
\rho(M,X)\xi(M,X).
\end{align*}

It follows that the raw residual product from
Corollary~\ref{cor:tangent_space_frontdoor} is
\begin{align}
\widetilde\chi_{f_\alpha,g_\rho}(O)
&=
\frac{\alpha(Z)\rho(M,X)}
{f_Z(Z\mid X)f_M(M\mid A,Z,X)}
\left\{
Y-\xi(M,X)
\right\}.
\label{eq:frontdoor_tilde_chi}
\end{align}

By Corollary~\ref{cor:tangent_space_frontdoor},
\begin{align}
\chi_{f_\alpha,g_\rho}(O)
&=
\widetilde\chi_{f_\alpha,g_\rho}(O)
-
\bbE[
\widetilde\chi_{f_\alpha,g_\rho}(O)
\mid X,Z,A,M
]
\nonumber\\
&\qquad+
\bbE[
\widetilde\chi_{f_\alpha,g_\rho}(O)
\mid X,Z,A
].
\label{eq:frontdoor_projection}
\end{align}

For the first two terms in \eqref{eq:frontdoor_projection},
\begin{align*}
&\widetilde\chi_{f_\alpha,g_\rho}(O)
-
\bbE[
\widetilde\chi_{f_\alpha,g_\rho}(O)
\mid X,Z,A,M
]
\\
&=
\frac{\alpha(Z)\rho(M,X)}
{f_Z(Z\mid X)f_M(M\mid A,Z,X)}
\left\{
Y-\mu(M,A,Z,X)
\right\}
\\
&=
\frac{\alpha(Z)}{f_Z(Z\mid X)}
\left\{
\int
\frac{
f_M(M\mid a_0,u,X)
}{
f_M(M\mid A,Z,X)
}
f_Z(u\mid X)\diff u
\right\}
\left\{
Y-\mu(M,A,Z,X)
\right\}.
\end{align*}
On the other hand, integrating the outcome-regression component of
$\varphi_{a_0}^{\mathrm{np}}(\bbP;z^*)$ against $\alpha(z^*)$ gives
\begin{align*}
&\int
\frac{\I(Z=z^*)}{f_Z(z^*\mid X)}
\left\{
\int
\frac{
f_M(M\mid a_0,u,X)
}{
f_M(M\mid A,z^*,X)
}
f_Z(u\mid X)\diff u
\right\}
\\
&\hspace{4cm}\times
\left\{
Y-\mu(M,A,z^*,X)
\right\}
\alpha(z^*)\diff z^*
\\
&=
\frac{\alpha(Z)}{f_Z(Z\mid X)}
\left\{
\int
\frac{
f_M(M\mid a_0,u,X)
}{
f_M(M\mid A,Z,X)
}
f_Z(u\mid X)\diff u
\right\}
\\
&\hspace{4cm}\times
\left\{
Y-\mu(M,A,Z,X)
\right\}.
\end{align*}
Thus, the first two terms in \eqref{eq:frontdoor_projection} equal the
integrated outcome-regression component.

Next, define
\begin{align*}
T_{a,z^*}(x)
=
\int
\rho(m,x)\mu(m,a,z^*,x)\diff m.
\end{align*}
Conditioning \eqref{eq:frontdoor_tilde_chi} on $(X,Z,A)$ gives
\begin{align}
\bbE[
\widetilde\chi_{f_\alpha,g_\rho}(O)
\mid X,Z,A
]
&=
\frac{\alpha(Z)}{f_Z(Z\mid X)}
\int
\rho(m,X)
\left\{
\mu(m,A,Z,X)-\xi(m,X)
\right\}
\diff m.
\notag \\
&=
\frac{\alpha(Z)}{f_Z(Z\mid X)}
\left\{
T_{A,Z}(X)
-
\sum_a
\pi(a\mid Z,X)T_{a,Z}(X)
\right\}.
\notag \\
&=
\frac{\alpha(Z)}{f_Z(Z\mid X)}
\left\{
A-\pi(1\mid Z,X)
\right\}
\left\{
T_{1,Z}(X)-T_{0,Z}(X)
\right\}
\notag \\
&=
\frac{\alpha(Z)}{f_Z(Z\mid X)}
\left\{
A-\pi(1\mid Z,X)
\right\}
\left\{
\int
\left\{
\kappa_{1,Z}(u,X)
-
\kappa_{0,Z}(u,X)
\right\}
f_Z(u\mid X)\diff u.
\right\}
\label{eq:frontdoor_treatment_projection}
\end{align} 
The second equality holds because  $\xi(m,X) = \sum_a \mu(m,a,Z,X)\pi(a\mid Z,X)$, we obtain
\begin{align*}
\int
\rho(m,X)\xi(m,X)\diff m
=
\sum_a
\pi(a\mid Z,X)T_{a,Z}(X)
\end{align*}
The third equality holds since $A$ is binary. 
The last equality holds because, by the definition of $\rho$ and
Fubini's theorem,
\begin{align*}
T_{a,z^*}(x)
&=
\int
\left\{
\int
f_M(m\mid a_0,u,x)f_Z(u\mid x)\diff u
\right\}
\mu(m,a,z^*,x)
\diff m
\\
&=
\int
\left\{
\int
\mu(m,a,z^*,x)
f_M(m\mid a_0,u,x)
\diff m
\right\}
f_Z(u\mid x)\diff u
\\
&=
\int
\kappa_{a,z^*}(u,x)
f_Z(u\mid x)\diff u.
\end{align*}
Consequently,
\begin{align*}
T_{1,z^*}(x)-T_{0,z^*}(x)
=
\int
\left\{
\kappa_{1,z^*}(u,x)
-
\kappa_{0,z^*}(u,x)
\right\}
f_Z(u\mid x)\diff u.
\end{align*}

Thus, the term \eqref{eq:frontdoor_treatment_projection} is exactly the integrated treatment-mechanism component of $\varphi_{a_0}^\mathrm{np}(\bbP; z^*)$ against $\alpha(z^*)$ 
\begin{align*}
&\int
\frac{\I(Z=z^*)}{f_Z(z^*\mid X)}
\left\{
A-\pi(1\mid z^*,X)
\right\}
 \times
\left\{
\int
\left\{
\kappa_{1,z^*}(u,X)
-
\kappa_{0,z^*}(u,X)
\right\}
f_Z(u\mid X)\diff u
\right\}
\alpha(z^*)\diff z^*
\\
&= 
\frac{\alpha(Z)}{f_Z(Z\mid X)}
\left\{
A-\pi(1\mid Z,X)
\right\}
 \times
\left\{
\int
\left\{
\kappa_{1,z^*}(u,X)
-
\kappa_{0,z^*}(u,X)
\right\}
f_Z(u\mid X)\diff u
\right\}
\end{align*}

Finally, consider the mediator and baseline components. Because the
full nested Markov restriction implies
\begin{align*}
\xi_{z^*}(m,x)=\xi(m,x)
\end{align*}
for every $z^*$, we also have
\begin{align*}
\gamma_{z^*}(z,x)
=
\gamma(z,x)
\end{align*}
for every $z^*$. Moreover,
\begin{align*}
\psi_{a_0}(\bbP;z^*)
=
\psi_{a_0}(\bbP)
\end{align*}
for every $z^*$. Hence,
\begin{align*}
&\int
\frac{\I(A=a_0)}{\pi(a_0\mid Z,X)}
\left\{
\xi_{z^*}(M,X)-\gamma_{z^*}(Z,X)
\right\}
\alpha(z^*)\diff z^*
\\
&\qquad+
\int
\left\{
\gamma_{z^*}(Z,X)
-
\psi_{a_0}(\bbP;z^*)
\right\}
\alpha(z^*)\diff z^*
\\
&=
\left[
\frac{\I(A=a_0)}{\pi(a_0\mid Z,X)}
\left\{
\xi(M,X)-\gamma(Z,X)
\right\}
+
\gamma(Z,X)
-
\psi_{a_0}(\bbP)
\right]
\int\alpha(z^*)\diff z^*
\\
&=
0.
\end{align*}

Combining the outcome-regression and treatment-mechanism components
therefore yields
\begin{align*}
\chi_{f_\alpha,g_\rho}(O)
=
\int
\varphi_{a_0}^{\mathrm{np}}(\bbP;z^*)(O)
\alpha(z^*)\diff z^*.
\end{align*}
The inclusion
\begin{align*}
\chi_{f_\alpha,g_\rho}
\in
\Lambda_{\bbP}^{\perp}
\end{align*}
follows directly from
Corollary~\ref{cor:tangent_space_frontdoor}.

For completeness, we note that for a square-integrable function $\alpha$ satisfying $\int \alpha(z) \,\diff z =0$: 
\begin{align*}
&\hspace{-1cm} \int
\varphi_{a_0}^{\mathrm{np}}(\bbP;z^*)(O_i)
\alpha(z^*)
\,\diff z^*
\\
&=
\frac{\alpha(Z_i)}{f_Z(Z_i\mid X_i)}
\sum_z
f^r_{M,Z_i}(M_i,A_i,z,X_i)\,
f_Z(z\mid X_i)\,
\Big\{
Y_i-\mu(M_i,A_i,Z_i,X_i)
\Big\}
\\
&\quad+
\frac{\alpha(Z_i)}{f_Z(Z_i\mid X_i)}
\Big\{
A_i-\pi(1\mid Z_i,X_i)
\Big\}
\sum_z
\Big\{
\kappa_{1,Z_i}(z,X_i)
-
\kappa_{0,Z_i}(z,X_i)
\Big\}
f_Z(z\mid X_i)
\\
&\quad+
\frac{\I(A_i=a_0)}{\pi(a_0\mid Z_i,X_i)}
\int
\Big\{
\xi_{z^*}(M_i,X_i)
-
\gamma_{z^*}(Z_i,X_i)
\Big\}
\alpha(z^*)
\,\diff z^*
\\
&\quad+
\int
\Big\{
\gamma_{z^*}(Z_i,X_i)
-
\psi_{a_0}(\bbP;z^*)
\Big\}
\alpha(z^*)
\,\diff z^*.
\end{align*}
Under the Verma restriction, $\xi_{z^*}(m,x)=\xi(m,x)$, $\gamma_{z^*}(z,x) = \gamma(z,x)$, and $\psi_{a_0}(\bbP;z^*) = \psi_{a_0}(\bbP)$, for every $z^*$. Together with $\int \alpha(z) \diff z$, the above simplifies to 
\begin{align*}
&\hspace{-1cm}\int
\varphi_{a_0}^{\mathrm{np}}(\bbP;z^*)(O_i)
\alpha(z^*)
\,\diff z^*
\\
&=
\frac{\alpha(Z_i)}{f_Z(Z_i\mid X_i)}
\sum_z
f^r_{M,Z_i}(M_i,A_i,z,X_i)\,
f_Z(z\mid X_i)\,
\Big\{
Y_i-\mu(M_i,A_i,Z_i,X_i)
\Big\}
\\
&\quad+
\frac{\alpha(Z_i)}{f_Z(Z_i\mid X_i)}
\Big\{
A_i-\pi(1\mid Z_i,X_i)
\Big\}
\sum_z
\Big\{
\kappa_{1,Z_i}(z,X_i)
-
\kappa_{0,Z_i}(z,X_i)
\Big\}
f_Z(z\mid X_i). 
\end{align*}

\end{proof}

\clearpage
\subsection{The Napkin Graph Model}
\label{app:examples_napkin}

\subsubsection{Proof of Corollary~\ref{cor:napkin_orthocomplement}}
\label{app:cor:napkin_orthocomplement}

\begin{proof}
The result follows by specializing Proposition~\ref{prop:weighted_moment_representation}
and Theorem~\ref{thm:orthocomplement_tangent_space} to the Napkin graph.

After fixing $Z$, the graph encodes the nested conditional independence
\begin{align*}
Y \indep Z \mid A
\qquad
[\phi_Z(p;\calG)].
\end{align*}
Thus, in the notation of Theorem~\ref{thm:orthocomplement_tangent_space}, we take
$
X_{\mathrm{gen}}=Z,
\,
Y_{\mathrm{gen}}=Y,
\,
Z_{\mathrm{gen}}=A,
\,
R=Z.
$
For this graph, the Markov blanket of $Z$ is $W$, and hence
\begin{align*}
\omega_Z(\M)
=
\frac{\tilde p(Z)}{p(Z\mid W)}, \quad \text{with} \quad \M = (Z, W). 
\end{align*}

By Proposition~\ref{prop:weighted_moment_representation}, the post-fixing conditional
expectation satisfies
\begin{align*}
\bbE_{\phi_Z}[f(Z,A)\mid A]
&=
\frac{
\bbE
\left[
\omega_Z(\M)f(Z,A)\mid A
\right]
}{
\bbE
\left[
\omega_Z(\M)\mid A
\right]
}.
\end{align*}
Therefore, the nested conditional independence
$Y\indep Z\mid A \ [\phi_Z(p;\calG)]$ is equivalent to the collection of weighted
moment restrictions
\begin{align*}
\bbE
\left[
\omega_Z(\M)
\left\{
f(Z,A)
-
\bbE_{\phi_Z}[f(Z,A)\mid A]
\right\}
g(Y,A)
\right]
=
0,
\end{align*}
for every
$f\in L^2(\bbP_{Z,A})$
and
$g\in L^2(\bbP_{Y,A})$.
This proves the model representation in \eqref{eq:napkin_model}.

It remains to specialize the orthocomplement formula. By
Theorem~\ref{thm:orthocomplement_tangent_space}, the raw weighted residual product is
\begin{align*}
\widetilde\chi_{f,g}(O)
&=
\omega_Z(\M)
\Big\{
f(Z,A)-\bbE_{\phi_Z}[f(Z,A)\mid A]
\Big\}
\Big\{
g(Y,A)-\bbE_{\phi_Z}[g(Y,A)\mid A]
\Big\}.
\end{align*}
Since $R=Z$ and $\mb_{\calG}(Z)=W$, the residualization in
Theorem~\ref{thm:orthocomplement_tangent_space} becomes
\begin{align*}
\chi_{f,g}(O)
&=
\widetilde\chi_{f,g}(O)
-
\bbE
\left[
\widetilde\chi_{f,g}(O)\mid W,Z
\right]
+
\bbE
\left[
\widetilde\chi_{f,g}(O)\mid W
\right].
\end{align*}
Thus
\begin{align*}
\Lambda_{\bbP}^{\perp}
=
\overline{\sp}
\left\{
\chi_{f,g}(O):
f\in L^2(\bbP_{Z,A}),
\;
g\in L^2(\bbP_{Y,A})
\right\},
\end{align*}
which is precisely \eqref{eq:napkin_orthocomplement}.
\end{proof}

\subsubsection{Identification and NP-EIF} 
\label{app:id_npEIF_napkin}


The ``Napkin'' graph is named after Pearl's reported napkin sketch \citep{pearl2018book}. The structural no-direct-effect statement encoded by the Napkin graph in \eqref{eq:napkin_nested_constraint} implies the absence of a direct effect of $Z$ on the mean of $Y$ after fixing $A$, that is 
\begin{align}
\bbE[Y(z,a_0)]
=
\bbE[Y(a_0)],
\qquad
\forall z \in \cal Z,
\label{eq:napkin_mean_scale}
\end{align}
which was studied by \citet{guo2025causal} to improve efficiency for the marginal causal mean $\bbE[Y(a_0)]$. 

Let $O=(W,Z,A,Y)$, and assume treatment $A$ is binary. Write
\begin{align*}
\pi(a\mid z,w)
&\coloneqq
\bbP(A=a\mid Z=z,W=w),
\\
\mu(a,z,w)
&\coloneqq
\bbE[Y\mid A=a,Z=z,W=w],
\end{align*}

Fixing $Z$ produces the CADMG $\calG(O\setminus Z; Z)$ and the corresponding post-fixing kernel
\begin{align*}
q_{O\setminus Z\mid Z}(W,A,Y\mid Z)
&= \frac{p(O)}{p(Z\mid W)} \\
&= p(Y\mid W,Z,A)\, p(A\mid W,Z)\, p(W).
\end{align*}

The Napkin graph implies the nested conditional independence $Z \indep Y \mid A \ [\phi_Z(p)]$. Equivalently, the post-fixing kernel $q_{O\setminus Z\mid Z}(Y\mid A,Z)$, written in terms of the observed-data distribution via
\begin{align*}
q_{O\setminus Z\mid Z}(y\mid a,z) = \frac{ \int p(y\mid x,z,a,m) \, p(a\mid x,z) \, p(w) \, \diff w}{\int p(a\mid x,z) \, p(w) \, \diff w},
\end{align*}
is invariant to $z$ for every $(a, y)$. 

A mean-scale Verma restriction is obtained by choosing $h(Y)=Y$ in Section~\ref{subsec:mean_scale_verma}, yielding 
\begin{align*}
\bbE_{\phi_Z}
\left[
Y
\mid
Z, A
\right]
=
\bbE_{\phi_Z}
\left[
Y
\mid
A
\right].
\end{align*}
Equivalently,
\begin{align*}
\eta(\bbP;z,a)
\coloneqq
\frac{ \int \mu(a,z,w) \, \pi(a\mid z,w) \, p(w) \, \diff w}{\int \pi(a\mid z, w) \, p(w) \, \diff w}
\end{align*}
does not depend on $z$.

Following the notation of Section~\ref{subsec:mean_scale_verma}, define the $\bbP$-dependent operator
\begin{align*}
\mathcal L_{\bbP}\{\eta(\bbP; z, a)\}
=
\eta(\bbP; z, a_0). 
\end{align*}
Applying $\mathcal L_{\bbP}$ to the post-fixing conditional mean produces the following identification functional for $\bbE[Y(a_0)]$
\begin{align*}
\psi_{a_0}(\bbP;z)
=
\mathcal L_{\bbP}
\left\{
\eta(\bbP;z,a)
\right\},
\end{align*}
which can be written explicitly as
\begin{align}
\psi_{a_0}(\bbP;z)
=
\frac{ \int \mu(a_0,z,w] \, \pi(a_0\mid z,w) \, p(w) \, \diff w}{\int \pi(a_0\mid z,w) \, p(w) \, \diff w}. 
\label{eq:id_ATE_zstar_napkin}
\end{align}

More generally, for any weighting function
$\widetilde p$ on $\cal Z$ satisfying $\int\widetilde p(z)\diff z=1$, define
\begin{align}
\psi_{a_0}(\bbP;\widetilde p)
&=
\int
\psi_{a_0}(\bbP;z)
\,\widetilde p(z)
\,\diff z
=
\mathcal L_{\bbP}
\left\{
\int
\eta(\bbP;z,a)
\,
\widetilde p(z)
\,\diff z
\right\},  
\label{eq:ID_napkin}
\end{align}
where the second equality follows from the linearity of $\mathcal L_{\bbP}$. 

The nonparametric IF for the functional in \eqref{eq:ID_napkin}, denoted as $\varphi_{a_0}^\mathrm{np}(\bbP;\tilde{p}_z)$, is given by 

{\small 
\begin{equation}\label{eq:IF_tilde_pz}
    \begin{aligned}
    \varphi_{a_0}^\mathrm{np}(\bbP;\tilde{p}_z)(O_i)
    &=\underbrace{\frac{\I(A_i=a_0)}{\kappa_{a_0}(\bbP;\tilde{p}_z)}  \frac{\tilde{p}(Z_i)}{f_Z(Z_i \mid W_i)}  \Big\{  Y_i -  \mu(a_0, Z_i, W_i) \Big\}}_{\Phi_{Y,a_0}(\bbP; \, \tilde{p}_z)(O_i)}
    \\[0.35cm] 
    &\hspace{-2.2cm} + \underbrace{\frac{1}{\kappa_{a_0}(\bbP;\tilde{p}_z)} \frac{\tilde{p}(Z_i)}{f_Z(Z_i \mid W_i)} \Big\{ \mu(a_0, Z_i, W_i) - \psi_{a_0}(\bbP;\tilde p_z)   \Big\} \Big\{ \I(A_i=a_0) -  \pi(a_0 \mid Z_i, W_i) \Big\}}_{\Phi_{A,a_0}(\bbP; \, \tilde{p}_z)(O_i)}
     \\[0.35cm]
    &\hspace{-2.2cm} + \frac{1}{\kappa_{a_0}(\bbP;\tilde{p}_z)} \underbrace{\int  \pi(a_0 \mid z, W_i) \Big\{  \mu(a_0, z, W_i) -  \psi_{a_0}(\bbP; \tilde p_z) \Big\} \, \tilde{p}(z) dz}_{\Phi_{W,a_0}(\bbP; \, \tilde{p}_z)(O_i)} \ ,
\end{aligned}
\end{equation}
}
where $\kappa_{a_0}(\bbP;z) = \int \pi(a_0\, | \, z, w) \, p(w) \, dw$. 

Under discrete $Z$, $\tilde{p}(Z)$ can be replaced by $\I(Z = z^*)$ in $\Phi_{Y,a_0}$ and $\Phi_{A,a_0}$, while $\Phi_{W,a_0}$ simplifies to $\Phi_{W,a_0}(\bbP; z^*)(O_i) = \displaystyle \frac{\pi(a_0 \mid z^*, W_i)}{\kappa_{a_0}(\bbP; z^*)} \{ \mu(a_0, z^*, W_i) - \psi_{a_0}(\bbP; z^*) \}$. The resulting np-EIF is denoted by $\varphi_{a_0}^\mathrm{np}(\bbP;z^*)$: 

{\small 
\begin{equation}\label{eq:IF_tilde_z}
    \begin{aligned}
    \varphi_{a_0}^\mathrm{np}(\bbP;z^*)(O_i)
    &=\frac{\I(A_i=a_0)}{\kappa_{a_0}(\bbP;z^*)}  \frac{\I(Z_i = z^*)}{f_Z(z^* \mid W_i)}  \Big\{  Y_i -  \mu(a_0, z^*, W_i) \Big\}
    \\[0.35cm] 
    &\hspace{-2.2cm} + \frac{1}{\kappa_{a_0}(\bbP;z^*)} \frac{\I(Z_i = z^*)}{f_Z(z^* \mid W_i)} \Big\{ \mu(a_0, z^*, W_i) - \psi_{a_0}(\bbP;z^*)   \Big\} \Big\{ \I(A_i=a_0) -  \pi(a_0 \mid z^*, W_i) \Big\}
     \\[0.35cm]
    &\hspace{-2.2cm} +   \frac{\pi(a_0 \mid z^*, W_i)}{\kappa_{a_0}(\bbP;z^*)}  \Big\{  \mu(a_0, z^*, W_i) -  \psi_{a_0}(\bbP; z^*) \Big\}  \ . 
\end{aligned}
\end{equation}
}

We also note that for a square-integrable function $\alpha$ satisfying $\int \alpha(z) \,\diff z =0$: 
\begin{align*}
&\hspace{-1cm}\int
\varphi_{a_0}^{\mathrm{np}}(\bbP;z^*)(O_i)
\alpha(z^*)
\,\diff z^*
\\
&=
\frac{\I(A_i=a_0)}
{\kappa_{a_0}(\bbP;Z_i)}
\frac{\alpha(Z_i)}
{f_Z(Z_i\mid W_i)}
\Big\{
Y_i-\mu(a_0,Z_i,W_i)
\Big\}
\\
&\quad
+
\frac{\alpha(Z_i)}
{\kappa_{a_0}(\bbP;Z_i)}
\frac{1}
{f_Z(Z_i\mid W_i)}
\Big\{
\mu(a_0,Z_i,W_i)
-
\psi_{a_0}(\bbP;Z_i)
\Big\}
\times
\Big\{
\I(A_i=a_0)
-
\pi(a_0\mid Z_i,W_i)
\Big\}
\\
&\quad
+
\int
\frac{
\pi(a_0\mid z^*,W_i)
}{
\kappa_{a_0}(\bbP;z^*)
}
\Big\{
\mu(a_0,z^*,W_i)
-
\psi_{a_0}(\bbP;z^*)
\Big\}
\alpha(z^*)
\,\diff z^*.
\end{align*}

\subsubsection{Proof of Lemma~\ref{lem:napkin_indexed}}
\label{app:lem:napkin_indexed}

We show that the target-specific orthocomplement
$\Lambda^\perp_{\bbP,\psi}$, generated by the indexed invariance condition, is obtained as a specialization of the full orthocomplement to the tangent space of the Napkin model characterized in Corollary~\ref{cor:napkin_orthocomplement}. As discussed in the main text, the semiparametric efficient influence function within the corresponding affine class associated with $\Lambda^\perp_{\bbP,\psi}$ coincides with the optimal weighted influence function of \citet{guo2025causal}.

\begin{proof}

Fix a treatment level $a_0$. Write
\begin{align*}
\pi(a\mid z,w)
&\coloneqq
\bbP(A=a\mid Z=z,W=w),
\\
\mu(a,z,w)
&\coloneqq
\bbE[Y\mid A=a,Z=z,W=w],
\end{align*}
and define
\begin{align*}
\kappa_{a_0}(\bbP;z)
&\coloneqq
\int
\pi(a_0\mid z,w)
p(w)
\diff w,
\\
\psi_{a_0}(\bbP;z)
&\coloneqq
\frac{
\int
\mu(a_0,z,w)
\pi(a_0\mid z,w)
p(w)
\diff w
}{
\kappa_{a_0}(\bbP;z)
}.
\end{align*}

Let
\begin{align*}
f_\alpha(Z,A)
&=
\frac{
\I(A=a_0)
}{
\kappa_{a_0}(\bbP;Z)
}\frac{\alpha(Z)}{\tilde p(Z)},
\qquad 
g(Y,A)
=
Y,
\end{align*}
where $\int\alpha(z)\diff z=0$. 

Under the post-fixing law,
\begin{align*}
\bbE_{\phi_Z}
\left[
f_\alpha(Z,A)
\mid A
\right]
&=
\frac{ \bbE[\omega_Z(Z, W) f_\alpha(Z,A) \mid A] }{\bbE[\omega_Z(Z, W) \mid A]} \\
&= \frac{1}{\bbE[\omega_Z(Z, W) \mid A]} \int \omega_Z(z, w) f_\alpha(z,a) p(w, z \mid A) dw dz \\
&= \frac{1}{\bbE[\omega_Z(Z, W) \mid A]} \int \frac{\tilde{p}(z)}{p(z \mid w)} \frac{\I(A=a_0)}{\kappa_{a_0}(\bbP; z)} \frac{\alpha(z)}{\tilde p(z)} p(w, z \mid A) dw dz \\
&= \frac{\I(A=a_0)}{\bbE[\omega_Z(Z, W) \mid A]} \int \frac{1}{p(z \mid w)} \frac{\alpha(z)}{\kappa_{a_0}(\bbP; z)}  p(w, z \mid A=a_0) dw dz \\
&= 
\frac{\I(A=a_0)}{
\bbE[\omega_Z(Z,W)\mid A]
}
\int
\frac{1}{p(z\mid w)}
\frac{\alpha(z)}{\kappa_{a_0}(\bbP;z)}
\frac{
\pi(a_0\mid z,w)p(z\mid w)p(w)
}{
p(A=a_0)
}
\,\diff w\,\diff z
\\
&=
\frac{\I(A=a_0)}{
p(A=a_0)\,
\bbE[\omega_Z(Z,W)\mid A]
}
\int
\frac{\alpha(z)}
{\kappa_{a_0}(\bbP;z)}
\pi(a_0\mid z,w)p(w)
\,\diff w\,\diff z
\\
&=
\frac{\I(A=a_0)}{
p(A=a_0)\,
\bbE[\omega_Z(Z,W)\mid A]
}
\int
\frac{\alpha(z)}
{\kappa_{a_0}(\bbP;z)}
\left\{
\int
\pi(a_0\mid z,w)p(w)
\,\diff w
\right\}
\,\diff z
\\
&=
\frac{\I(A=a_0)}{
p(A=a_0)\,
\bbE[\omega_Z(Z,W)\mid A]
}
\int
\alpha(z)
\,\diff z\\
&= 0, 
\end{align*}
and 
\begin{align*}
\bbE_{\phi_Z}
\left[
g(Y,A)
\mid A
\right]
&=
\frac{
\bbE\!\left[
\omega_Z(Z,W)Y
\mid A
\right]}
{
\bbE\!\left[
\omega_Z(Z,W)
\mid A
\right]
}
\\
&= 
\frac{
\int
\widetilde p(z)
\left\{
\int
\mu(A,z,w)
\pi(A\mid z,w)
p(w)
\,\diff w
\right\}
\,\diff z
}{\int
\widetilde p(z)
\left\{
\int
\pi(A\mid z,w)
p(w)
\,\diff w
\right\}
\,\diff z}
\\
&:=
\psi_{A}(\bbP)
\end{align*}

The raw weighted residual product in
Corollary~\ref{cor:napkin_orthocomplement} therefore reduces to
\begin{align}
\widetilde\chi_{\alpha}(O) 
\coloneqq
\widetilde\chi_{f_\alpha,g}(O)
&= 
\frac{1}{p(Z\mid W)}
\frac{\I(A=a_0)
}{
\int
\pi(a_0\mid Z,w)
p(w)
\,\diff w
}
\alpha(Z)
\left\{
Y-\psi_{A}(\bbP)
\right\}.
\notag \\
&=
\frac{1}{p(Z\mid W)}
\frac{\I(A=a_0)
}{
\kappa_{a_0}(\bbP;Z)
}
\alpha(Z)
\left\{
Y-\psi_{a_0}(\bbP)
\right\}.
\label{eq:napkin_raw_indexed_direction}
\end{align}
The corresponding orthocomplement element is
\begin{align}
\chi_{\alpha}(O) 
\coloneqq
\chi_{f,g}(O)
&=
\widetilde\chi_{\alpha}(O)
-
\bbE[
\widetilde\chi_{\alpha}(O)
\mid W,Z
]
+
\bbE[
\widetilde\chi_{\alpha}(O)
\mid W
].
\label{eq:napkin_residualized_indexed_direction}
\end{align}

The first conditional expectation is
\begin{align*}
\bbE[
\widetilde\chi_{\alpha}(O)
\mid W,Z
]
&=
\frac{
\alpha(Z)
}{
\kappa_{a_0}(\bbP;Z)
}
\frac{
\pi(a_0\mid Z,W)
}{
p(Z\mid W)
}
\left\{
\mu(a_0,Z,W)
-
\psi_{a_0}(\bbP)
\right\},
\end{align*}
and the second conditional expectation is
\begin{align*}
\bbE[
\widetilde\chi_{\alpha}(O)
\mid W
]
&=
\int
\frac{
\alpha(z)
}{
\kappa_{a_0}(\bbP;z)
}
\pi(a_0\mid z,W)
\left\{
\mu(a_0,z,W)
-
\psi_{a_0}(\bbP)
\right\}
\diff z.
\end{align*}
Substituting these expressions into
\eqref{eq:napkin_residualized_indexed_direction} and using
\begin{align*}
Y-\psi_{a_0}(\bbP) 
= Y-\psi_{a_0}(\bbP;Z)
=
\left\{
Y-\mu(a_0,Z,W)
\right\}
+
\left\{
\mu(a_0,Z,W)
-
\psi_{a_0}(\bbP;Z)
\right\},
\end{align*}
gives
\begin{align}
\chi_{\alpha
}(O)
&=
\frac{
\I(A=a_0)
}{
\kappa_{a_0}(\bbP;Z)
}
\frac{
\alpha(Z)
}{
p(Z\mid W)
}
\left\{
Y-\mu(a_0,Z,W)
\right\}
\notag\\
&\quad
+
\frac{
\alpha(Z)
}{
\kappa_{a_0}(\bbP;Z)
}
\frac{
\mu(a_0,Z,W)-\psi_{a_0}(\bbP;Z)
}{
p(Z\mid W)
}
\left\{
\I(A=a_0)-\pi(a_0\mid Z,W)
\right\}
\notag\\
&\quad
+
\int
\frac{
\pi(a_0\mid z,W)
}{
\kappa_{a_0}(\bbP;z)
}
\left\{
\mu(a_0,z,W)
-
\psi_{a_0}(\bbP;z)
\right\}
\alpha(z)
\diff z,
\label{eq:napkin_indexed_direction_expanded}
\end{align}
which proves:
\begin{align*}
\chi_{\alpha}(O)
\coloneqq
\chi_{f_\alpha, g}(O)
=
\int
\varphi_{a_0}^{\mathrm{np}}(\bbP; z^*)(O)
\alpha(z^*)
\diff z^*.
\end{align*}

Since every element of
$\Lambda_{\bbP,\psi}^{\perp}$ is in the closed linear span of such
zero-integral combinations,
\begin{align*}
\Lambda_{\bbP,\psi}^{\perp}
\subseteq
\Lambda_{\bbP}^{\perp}.
\end{align*}
\end{proof}

\clearpage
\section{Additional Illustrative Examples}

\subsection{Verma and Ordinary Independence Constraints}
\label{app:ex:verma_ordinary}

We return to the extended front-door model after deleting the arrow $Z\rightarrow M$, as shown in Figure~\ref{fig:one_verma_one_ordinary}(a).  The edge deletion has two distinct consequences. First, the observed-data law satisfies the ordinary conditional independence
\begin{align}
Z \indep M \mid A.
\label{eq:one_verma_ordinary_constraint}
\end{align}
Second, the no-direct-effect structure continues to imply the nested
conditional independence
\begin{align}
Z \indep Y \mid M
\qquad
[\phi_M(p)],
\label{eq:one_verma_nested_constraint}
\end{align}
which becomes an ordinary conditional independence in the post-fixing CADMG in Figure~\ref{fig:one_verma_one_ordinary}(b). Thus, this example illustrates how an ordinary and a nested restriction contribute separate, and generally nonorthogonal, directions to the tangent-space orthocomplement.

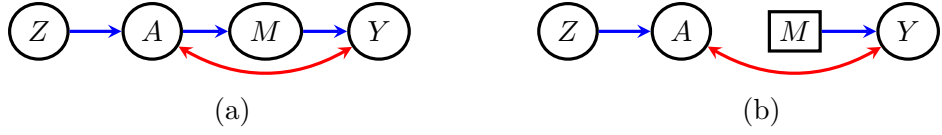
\begin{figure}[t]
\centering
\begin{tikzpicture}[>=stealth, node distance = 1.5cm]
    \tikzstyle{format} = [thick, circle, minimum size = 1.0mm, inner sep = 2pt]
    \tikzstyle{square} = [draw, thick, minimum size = 4.5mm, inner sep = 2pt]

    \begin{scope}[xshift = 0cm, yshift = 0cm]
        \path[->, very thick]
        
        node[shape = ellipse, draw = black] (z) {$Z$}
        node[right of = z, shape = ellipse, draw = black] (a) {$A$}
        node[right of = a, shape = ellipse, draw = black] (m) {$M$}
        node[right of = m, shape = ellipse, draw = black] (y) {$Y$}
        
        (z) edge[blue] (a)
        (a) edge[blue] (m)
        (m) edge[blue] (y)
        (a) edge[red, <->, bend right] (y)

        node[below right of = a] (cap) {(a)}
        ;
    \end{scope}

    \begin{scope}[xshift = 7cm, yshift = 0cm]
        \path[->, very thick]
        
        node[shape = ellipse, draw = black] (z) {$Z$}
        node[right of = z, shape = ellipse, draw = black] (a) {$A$}
        node[right of = a, shape = rectangle, draw = black] (m) {$M$}
        node[right of = m, shape = ellipse, draw = black] (y) {$Y$}
         
        (z) edge[blue] (a)
        (m) edge[blue] (y)
        (a) edge[red, <->, bend right] (y)

         node[below right of = a] (cap) {(b)}
        ;
    \end{scope}
\end{tikzpicture}
\caption{An extended front-door model with one ordinary and one nested conditional independence. (a) Original ADMG, obtained from Figure~\ref{fig:front-door}(a) by deleting $Z\rightarrow M$; the deletion implies $Z\indep M\mid A$ under the observed-data law. (b) CADMG obtained after fixing $M$, under which $Z\indep Y\mid M \ [\phi_M(p)]$ becomes an ordinary conditional independence.}
\label{fig:one_verma_one_ordinary}
\end{figure}

For the Verma restriction, $\mb_{\calG} (M)= \{A\}$ and hence
\begin{align*}
\omega_M(\M)
&=
\frac{1}{p(M\mid A)},
\qquad
\M=(M, A),
\\
\bbE_{\phi_M}[h(O)\mid C]
&=
\frac{
\bbE[\omega_M(\M)h(O)\mid C]
}{
\bbE[\omega_M(\M) \mid C]
},
\end{align*}
where $C=\emptyset$ or $C=M$, as appropriate. Define the two component models
\begin{align*}
\calP_{\mathrm V}
&=
\Bigg\{
\bbP:
\bbE\left[
\omega_M(\M)
\left\{
f(Z)-\bbE_{\phi_M}[f(Z)]
\right\}
g(M,Y)
\right]
=0,
\quad
\forall f,g
\Bigg\},
\\
\calP_{\mathrm O}
&=
\Bigg\{
\bbP:
\bbE\left[
\left\{
u(A,Z)-\bbE[u(A,Z)\mid A]
\right\}
v(A,M)
\right]
=0,
\quad
\forall u,v
\Bigg\},
\end{align*}
where $f\in L^2(\bbP_{Z})$, $g\in L^2(\bbP_{M,Y})$, $u\in L^2(\bbP_{A,Z})$, and $v\in L^2(\bbP_{A,M})$. The nested Markov model associated with Figure~\ref{fig:one_verma_one_ordinary}(a) is therefore $\calP(\calG)=\calP_{\mathrm V}\cap\calP_{\mathrm O}$. 

The corresponding orthocomplement directions take especially transparent forms. For the Verma restriction, let
\begin{align}
\widetilde\chi^{\mathrm V}_{f,g}(O)
&=
\omega_M(\M)
\Big\{
f(Z)-\bbE_{\phi_M}[f(Z)]
\Big\}
\times
\Big\{
g(M,Y)-\bbE_{\phi_M}[g(M,Y)\mid M]
\Big\},
\nonumber\\
\chi^{\mathrm V}_{f,g}(O)
&=
\widetilde\chi^{\mathrm V}_{f,g}(O)
-
\bbE[\widetilde\chi^{\mathrm V}_{f,g}(O)\mid A,M]
+
\bbE[\widetilde\chi^{\mathrm V}_{f,g}(O)\mid A].
\label{eq:one_verma_verma_direction}
\end{align}
For the ordinary conditional independence, Corollary~\ref{cor:ordinary_ci_orthocomplement} gives
\begin{align}
\chi^{\mathrm O}_{u,v}(O)
=
\Big\{
u(A,Z)-\bbE[u(A,Z)\mid A]
\Big\}
\Big\{
v(A,M)-\bbE[v(A,M)\mid A]
\Big\}.
\label{eq:one_verma_ordinary_direction}
\end{align}
Consequently, Proposition~\ref{thm:multiple_constraints} yields the combined orthocomplement subspace
\begin{align}
\mathcal H_{\mathrm{V+O}}
&\coloneqq
\overline{\sp}
\left\{
\chi^{\mathrm V}_{f,g}(O)+\chi^{\mathrm O}_{u,v}(O):
f,g,u,v\ \text{as above}
\right\}
\subseteq
\Lambda_{\bbP}^{\perp}.
\label{eq:one_verma_combined_orthocomplement}
\end{align}
If the tangent space of the intersection model equals the intersection of the two constraint-specific tangent spaces, then the inclusion in \eqref{eq:one_verma_combined_orthocomplement} is an equality. In either case, every element of $\mathcal H_{\mathrm{V+O}}$ is a valid orthocomplement direction.

Next, we prove the other direction of \eqref{eq:one_verma_combined_orthocomplement} by exhibiting an explicit one-dimensional parametric submodel $\{\bbP_{t}\} \subseteq \calP (\calG)$. We first notice the following decomposition of $p$:
\begin{equation}
\label{one_verma_combined_decomposition}
p (z, a, m, y) = p (y \mid z, a, m) p (m \mid z, a) p (z, a) = p (y \mid z, a, m) p (m \mid a) p (z, a).
\end{equation}
To this end, we consider the following form of perturbation of $p$, which enforces the constraint $Z \indep M \mid A$ under $\bbP_{t}$:
\begin{align*}
p_{t} (z, a, m, y) = p_{t} (y \mid z, a, m) p_{t} (m \mid a) p_{t} (z, a),
\end{align*}
where
\begin{align*}
& p_{t} (y \mid z, a, m) = p (y \mid z, a, m) \{1 + t \cdot s (y \mid z, a, m) + t^{2} \cdot h (y \mid z, a, m)\}, \\
& p_{t} (m \mid a) = p (m \mid a) \{1 + t \cdot s (m \mid a)\}, \\
& p_{t} (z, a) = p (z, a) \{1 + t \cdot s (z, a)\}.
\end{align*}
For simplicity, we take $p_{t}^{\ast} (z, a, m, y) = p_{t} (y \mid z, a, m) p^{\ast} (m) p_{t} (z, a)$ given any known marginal density $p^{\ast} (m)$ of $M$, as a candidate parametric submodel. We need to show that $p_{t}^{\ast} (y \mid z, m) = p_{t}^{\ast} (y \mid m)$ and we know that $p^{\ast} (y \mid z, m) = p^{\ast} (y \mid m)$. By the latter identity, we have
\begin{align*}
p^{\ast} (y \mid z, m) = \int p (y \mid z, a, m) p (a \mid z) \diff a = \int p (y \mid z, a, m) p (z, a) \diff a \diff z,
\end{align*}
so $\int p (y \mid z, a, m) p (z, a) \diff a = p (z) \int p (y \mid z, a, m) p (z, a) \diff a \diff z$. Furthermore, by $s \in \Lambda_{\bbP}$, we must conclude that
\begin{align*}
& \int p (y \mid z, a, m) p (a \mid z) \{s (y \mid z, a, m) + s (a \mid z)\} \diff a \\
& = \int p (y \mid z, a, m) p (z, a) \{s (y \mid z, a, m) + s (z, a)\} \diff a \diff z.
\end{align*}

Following elementary calculations, we have
\begin{align*}
& \ p_{t}^{\ast} (y \mid z, m) = \frac{p_{t}^{\ast} (z, m, y)}{p_{t}^{\ast} (z, m)} = \frac{\int p_{t}^{\ast} (z, a, m, y) \diff a}{\int p_{t}^{\ast} (z, a) \diff a} \\
= & \ \frac{\int p (y \mid z, a, m) p (z, a) \{1 + t \cdot s (y \mid z, a, m)\} \{1 + t \cdot s (z, a)\} \diff a}{\int p (z, a) \{1 + t \cdot s (z, a)\} \diff a} \\
& + t^{2} \cdot \frac{\int p (y \mid z, a, m) p (z, a) h (y \mid z, a, m) \{1 + t \cdot s (z, a)\} \diff a}{\int p (z, a) \{1 + t \cdot s (z, a)\} \diff a} \\
= & \ \frac{\int p (y \mid z, a, m) p (z, a) \{1 + t \cdot s (z, a)\} \diff a}{p (z) \{1 + t \cdot s (z)\}} + t \cdot \frac{\int p (y \mid z, a, m) p (z, a) s (y \mid z, a, m) \{1 + t \cdot s (z, a)\} \diff a}{p (z) \{1 + t \cdot s (z)\}} \\
& + t^{2} \cdot \frac{\int p (y \mid z, a, m) p (z, a) h (y \mid z, a, m) \{1 + t \cdot s (z, a)\} \diff a}{p (z) \{1 + t \cdot s (z)\}} \\
= & \ \int p (y \mid z, a, m) p (z, a) \diff a \diff z + t \cdot \int p (y \mid z, a, m) p (a \mid z) s (a \mid z) \diff a \\
& - t^{2} \cdot \frac{\int p (y \mid z, a, m) p (a \mid z) s (a \mid z) s (z) \diff a}{1 + t \cdot s (z)} + t^{2} \cdot \int p (y \mid z, a, m) p (a \mid z) h (y \mid z, a, m) \diff a \\
& + t^{3} \cdot \frac{\int p (y \mid z, a, m) p (a \mid z) h (y \mid z, a, m) \cdot s (a \mid z) \diff a}{1 + t \cdot s (z)} \\
& + t \cdot \frac{\int p (y \mid z, a, m) p (a \mid z) s (y \mid z, a, m) \{1 + t \cdot s (z, a)\} \diff a}{1 + t \cdot s (z)} \\
= & \ \int p (y \mid z, a, m) p (z, a) \diff a \diff z + t \cdot \int p (y \mid z, a, m) p (a \mid z) \{s (a \mid z) + s (y \mid z, a, m)\} \diff a \\
& - t^{2} \cdot \frac{\int p (y \mid z, a, m) p (a \mid z) s (a \mid z) \{s (z) - s (y \mid z, a, m)\} \diff a}{1 + t \cdot s (z)} + t^{2} \cdot \int p (y \mid z, a, m) p (a \mid z) h (y \mid z, a, m) \diff a \\
& + t^{3} \cdot \frac{\int p (y \mid z, a, m) p (a \mid z) h (y \mid z, a, m) \cdot s (a \mid z) \diff a}{1 + t \cdot s (z)},
\end{align*}
and
\begin{align*}
& \ p_{t}^{\ast} (y \mid m) = \frac{p_{t}^{\ast} (m, y)}{p_{t}^{\ast} (m)} = \frac{\int p_{t}^{\ast} (z, a, m, y) \diff a \diff z}{\int p_{t}^{\ast} (z, a) \diff a \diff z} \\
= & \ \int p (y \mid z, a, m) p (z, a) \{1 + t \cdot s (y \mid z, a, m)\} \{1 + t \cdot s (z, a)\} \diff a \diff z \\
& + t^{2} \cdot \int p (y \mid z, a, m) p (z, a) h (y \mid z, a, m) \{1 + t \cdot s (z, a)\} \diff a \diff z \\
= & \ \int p (y \mid z, a, m) p (z, a) \diff a \diff z + t \cdot \int p (y \mid z, a, m) p (z, a) \{s (z, a) + s (y \mid z, a, m)\} \diff a \diff z \\
& + t^{2} \cdot \int p (y \mid z, a, m) p (z, a) s (z, a) s (y \mid z, a, m) \diff a \diff z \\
& + t^{2} \cdot \int p (y \mid z, a, m) p (z, a) h (y \mid z, a, m) \diff a \diff z \\
& + t^{3} \cdot \int p (y \mid z, a, m) p (z, a) s (z, a) h (y \mid z, a, m) \diff a \diff z.
\end{align*}
Then $h$ can be identified as the solution to the following integral equation subject to $\bbE [h (Y \mid Z, A, M) \mid Z, A, M] = 0$ almost surely:
\begin{equation}
\label{one Verma Fredholm}
\begin{split}
& \{1 + t \cdot s (z)\} \int p (y \mid z, a, m) p (a \mid z) h (y \mid z, a, m) \diff a \\
& - \int p (y \mid z, a, m) p (a \mid z) s (a \mid z) \{s (z) - s (y \mid z, a, m)\} \diff a \\
& + t \cdot \int p (y \mid z, a, m) p (a \mid z) h (y \mid z, a, m) s (a \mid z) \diff a \\
& = \{1 + t \cdot s (z)\} \int p (y \mid z, a, m) p (z, a) s (z, a) s (y \mid z, a, m) \diff a \diff z \\
& + \{1 + t \cdot s (z)\} \int p (y \mid z, a, m) p (z, a) h (y \mid z, a, m) \diff a \diff z \\
& + t \cdot \{1 + t \cdot s (z)\} \int p (y \mid z, a, m) p (z, a) s (z, a) h (y \mid z, a, m) \diff a \diff z.
\end{split}
\end{equation}
Solving \eqref{one Verma Fredholm} yields a solution $h$ that can take the following form:
\begin{align*}
& \ p (y \mid z, a, m) h (y \mid z, a, m) \\
= & \ \frac{s (z) \int p (y \mid z, a, m) p (z, a) \{s (z, a) + s (y \mid z, a, m)\} \diff a \diff z - \int p (y \mid z, a, m) p (a \mid z) s (z, a) s (y \mid z, a, m) \diff a}{1 + t \cdot s (z, a)}.
\end{align*}
This completes the proof of the reverse direction, so $\calH_{V + O} = \Lambda_{\bbP}^{\perp}$.

For a pathwise differentiable parameter with nonparametric influence function $\varphi^{\mathrm{np}}$, projecting onto $\mathcal H_{\mathrm{V+O}}$ therefore gives the variance-improved influence function
$
\varphi^{\mathrm{np}}
-
\Pi\big(
\varphi^{\mathrm{np}}
\mid
\mathcal H_{\mathrm{V+O}}
\big).
$
A finite-dimensional implementation makes the need for a joint projection explicit. Let $\boldsymbol{\chi}_{\mathrm V}(O)$ and $\boldsymbol{\chi}_{\mathrm O}(O)$ denote finite collections of basis directions of the forms \eqref{eq:one_verma_verma_direction} and \eqref{eq:one_verma_ordinary_direction}, respectively, and stack them as
\begin{align*}
\boldsymbol{\chi}(O)
=
\begin{pmatrix}
\boldsymbol{\chi}_{\mathrm V}(O)\\
\boldsymbol{\chi}_{\mathrm O}(O)
\end{pmatrix}.
\end{align*}
The covariance matrix entering the projection in Theorem~\ref{theorem:locally_eff} then has the block form
\begin{align*}
\boldsymbol{\Sigma}
=
\bbE\left[
\boldsymbol{\chi}(O)\boldsymbol{\chi}(O)^\top
\right]
=
\begin{pmatrix}
\boldsymbol{\Sigma}_{\mathrm{VV}}
&
\boldsymbol{\Sigma}_{\mathrm{VO}}
\\
\boldsymbol{\Sigma}_{\mathrm{OV}}
&
\boldsymbol{\Sigma}_{\mathrm{OO}}
\end{pmatrix},
\end{align*}
where, in particular,
$\boldsymbol{\Sigma}_{\mathrm{VO}}
=
\bbE[
\boldsymbol{\chi}_{\mathrm V}(O)
\boldsymbol{\chi}_{\mathrm O}(O)^\top]$.
Writing
\begin{align*}
\boldsymbol b_{\mathrm V}
=
\bbE[\boldsymbol{\chi}_{\mathrm V}(O)
\varphi^{\mathrm{np}}(O)]
\qquad 
\text{and} 
\qquad 
\boldsymbol b_{\mathrm O}
=
\bbE[\boldsymbol{\chi}_{\mathrm O}(O)
\varphi^{\mathrm{np}}(O)], 
\end{align*}
the coefficients of the joint projection solve
\begin{align*}
\begin{pmatrix}
\boldsymbol{\Sigma}_{\mathrm{VV}}
&
\boldsymbol{\Sigma}_{\mathrm{VO}}
\\
\boldsymbol{\Sigma}_{\mathrm{OV}}
&
\boldsymbol{\Sigma}_{\mathrm{OO}}
\end{pmatrix}
\begin{pmatrix}
\boldsymbol{\alpha}_{\mathrm V}\\
\boldsymbol{\alpha}_{\mathrm O}
\end{pmatrix}
=
\begin{pmatrix}
\boldsymbol b_{\mathrm V}\\
\boldsymbol b_{\mathrm O}
\end{pmatrix}.
\end{align*}

By contrast, projecting separately onto the two component spaces would use $\boldsymbol{\Sigma}_{\mathrm{VV}}^{-1}\boldsymbol b_{\mathrm V}$
and $\boldsymbol{\Sigma}_{\mathrm{OO}}^{-1}\boldsymbol b_{\mathrm O}$,
thereby ignoring the off-diagonal covariance blocks. These separate coefficients generally do not solve the joint projection equations unless $\boldsymbol{\Sigma}_{\mathrm{VO}}=\boldsymbol 0$, or unless a special cancellation occurs. The conditional-centering operations in
\eqref{eq:one_verma_verma_direction} and \eqref{eq:one_verma_ordinary_direction} ensure that the resulting functions are orthogonal to the relevant tangent spaces; they do not imply that the ordinary and Verma direction families are orthogonal to one another. Consequently, the basis directions from both restrictions should be stacked and projected jointly. 

\subsection{Sequential Verma Constraints}
\label{app:ex:sequential}

Our final example illustrates a model with two Verma constraints induced by the same fixing operation. The ADMG in Figure~\ref{fig:sequential_constraints} contains two sequential variables, $Z_1$ and $Z_2$. After fixing $M$, the model implies
\begin{align}
Z_1 \indep Y \mid M
&\qquad [\phi_M(p)],
\nonumber\\
Z_2 \indep Y \mid Z_1,L_1,M
&\qquad [\phi_M(p)].
\label{eq:sequential_verma_constraints}
\end{align}
The two constraints share the post-fixing law but induce distinct families of weighted residual directions. 

\begin{figure}[!t]
\centering
\scalebox{0.9}{
\begin{tikzpicture}[>=stealth, node distance = 1.5cm]
    \tikzstyle{format} = [thick, circle, minimum size = 1.0mm, inner sep = 2pt]
    \tikzstyle{square} = [draw, thick, minimum size = 4.5mm, inner sep = 2pt]
    
    \begin{scope}[xshift = 0cm, yshift = 0cm]
        \path[->, very thick]
        
        node[shape = ellipse, draw = black] (z1) {$Z_{1}$}
        node[right of = z1, shape = ellipse, draw = black] (a1) {$L_{1}$}
        node[right of = a1, shape = ellipse, draw = black] (z2) {$Z_{2}$}
        node[right of = z2, shape = ellipse, draw = black] (a2) {$L_{2}$}
        node[right of = a2, shape = ellipse, draw = black] (m) {$M$}
        node[right of = m, shape = ellipse, draw = black] (y) {$Y$}
         
        (z1) edge[blue] (a1)
        (z1) edge[blue, bend right] (z2)
        (z1) edge[blue, bend right] (a2)
        (z1) edge[blue, bend right] (m)
        (a1) edge[blue] (z2)
        (a1) edge[blue, bend right] (a2)
        (a1) edge[blue, bend right] (m)
        (z2) edge[blue] (a2)
        (z2) edge[blue, bend right] (m)
        (a2) edge[blue] (m)
        (m) edge[blue] (y)

        (a1) edge[red, <->, bend left] (y)
        (a2) edge[red, <->, bend left] (y)

         node[below right of = z2, yshift=-0.5cm, xshift=-0.2cm] (cap) {(a)}
        ;
    \end{scope}
    \begin{scope}[xshift = 10cm, yshift = 0cm]
        \path[->, very thick]
        
        node[shape = ellipse, draw = black] (z1) {$Z_{1}$}
        node[right of = z1, shape = ellipse, draw = black] (a1) {$L_{1}$}
        node[right of = a1, shape = ellipse, draw = black] (z2) {$Z_{2}$}
        node[right of = z2, shape = ellipse, draw = black] (a2) {$L_{2}$}
        node[right of = a2, shape = rectangle, draw = black] (m) {$M$}
        node[right of = m, shape = ellipse, draw = black] (y) {$Y$}
         
        (z1) edge[blue] (a1)
        (z1) edge[blue, bend right] (z2)
        (z1) edge[blue, bend right] (a2)
        (a1) edge[blue] (z2)
        (a1) edge[blue, bend right] (a2)
        (z2) edge[blue] (a2)
        (m) edge[blue] (y)

        (a1) edge[red, <->, bend left] (y)
        (a2) edge[red, <->, bend left] (y)

        node[below right of = z2, yshift=-0.5cm, xshift=-0.2cm] (cap) {(b)}
        ;
    \end{scope}
\end{tikzpicture}
}
\caption{An ADMG with two sequential Verma constraints. (a) Original ADMG. (b) CADMG obtained after fixing $M$, under which $Z_1\indep Y\mid X,M\ [\phi_M(p)]$ and $Z_2\indep Y\mid X,Z_1,L_1,M\ [\phi_M(p)]$ become ordinary conditional independences.
}
\label{fig:sequential_constraints}
\end{figure}
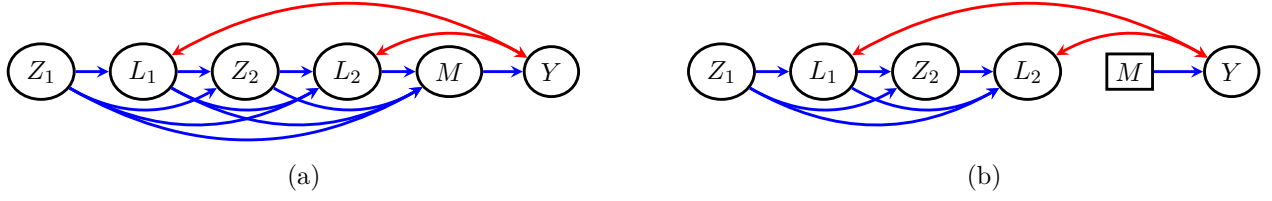

Let $H_1=\emptyset$, $H_2=(Z_1,L_1)$, $B_M=(Z_1,L_1,Z_2,L_2)$, $\M=(M,B_M),$ and $\omega_M(\M) = \frac{1}{p(M\mid B_M)}$. For $j\in\{1,2\}$, write
\begin{align*}
r^{(j)}_{f_j}(O)
=
f_j(H_j,Z_j)
-
\bbE_{\phi_M}[f_j(H_j,Z_j)\mid H_j],
\end{align*}
and define
\begin{align}
\calP_j
=
\Big\{
\bbP:
\bbE\left[
\omega_M(\M)
r^{(j)}_{f_j}(O)
g_j(H_j,M,Y)
\right]
=0,
\quad
\forall f_j,g_j
\Big\},
\label{eq:sequential_component_models}
\end{align}
where
$f_j\in L^2(\bbP_{H_j,Z_j})$ and
$g_j\in L^2(\bbP_{H_j,M,Y})$.  Here and below,
\begin{align*}
\bbE_{\phi_M}[h(O)\mid C]
=
\frac{
\bbE[\omega_M(\M)h(O)\mid C]
}{
\bbE[\omega_M(\M)\mid C]
},
\end{align*}
for the conditioning sets $C$ appearing in the relevant expression.  The
full model is the intersection
\begin{align}
\calP(\calG)=\calP_1\cap\calP_2.
\label{eq:sequential_intersection_model}
\end{align}

To obtain the constraint-specific orthocomplement directions, set
\begin{align}
\widetilde\chi^{(j)}_{f_j,g_j}(O)
&=
\omega_M(\M)
r^{(j)}_{f_j}(O)
\Big\{
g_j(H_j,M,Y)-\bbE_{\phi_M}[g_j(H_j,M,Y)\mid H_j,M]
\Big\},
\nonumber\\
\chi^{(j)}_{f_j,g_j}(O)
&=
\widetilde\chi^{(j)}_{f_j,g_j}(O)
-
\bbE[\widetilde\chi^{(j)}_{f_j,g_j}(O)\mid B_M,M]
+
\bbE[\widetilde\chi^{(j)}_{f_j,g_j}(O)\mid B_M],
\qquad j=1,2.
\label{eq:sequential_verma_directions}
\end{align}
Proposition~\ref{thm:multiple_constraints} then gives
\begin{align}
\mathcal H_{1+2}
&\coloneqq
\overline{\sp}
\left\{
\chi^{(1)}_{f_1,g_1}(O)
+
\chi^{(2)}_{f_2,g_2}(O):
f_j,g_j\ \text{as above},\ j=1,2
\right\}
\subseteq
\Lambda_{\bbP}^{\perp}.
\label{eq:sequential_combined_orthocomplement}
\end{align}
As in the preceding example, equality holds whenever the tangent space of
the intersection model equals the intersection of the two
constraint-specific tangent spaces.

Next, we prove the other direction of \eqref{eq:sequential_combined_orthocomplement} by exhibiting an explicit one-dimensional parametric submodel $\{\bbP_{t}\} \subseteq \calP (\calG)$. In fact, this example is a special case of the problem considered in \citet{liu2021efficient}. There, the authors neither explicitly constructed nor proved the existence of parametric submodels, which we fill the gap in the representative example described by Figure~\ref{fig:sequential_constraints}. We decompose $p$ as
\begin{align*}
p (H_{2}, B_{M}, M, Y) = p (Y \mid H_{2}, B_{M}, M) p (M \mid H_{2}, B_{M}) p (H_{2}, B_{M}),
\end{align*}
and thus the post-fixing law is
\begin{align*}
p^{\ast} (H_{2}, B_{M}, M, Y) = p (Y \mid H_{2}, B_{M}, M) p^{\ast} (M) p (H_{2}, B_{M}),
\end{align*}
where $p^{\ast} (M)$ is an arbitrary known probability density function of $M$. Similar to the previous example, we consider the following perturbation.
\begin{align*}
& p_{t} (y \mid z_{1}, l_{1}, z_{2}, l_{2}, m) = p (y \mid z_{1}, l_{1}, z_{2}, l_{2}, m) \{1 + t \cdot s (y \mid z_{1}, l_{1}, z_{2}, l_{2}, m) + t^{2} \cdot h (y \mid z_{1}, l_{1}, z_{2}, l_{2}, m)\}, \\
& p_{t} (z_{1}, l_{1}, z_{2}, l_{2}) = p (z_{1}, l_{1}, z_{2}, l_{2}) \{1 + t \cdot s (z_{1}, l_{1}, z_{2}, l_{2})\}.
\end{align*}
Based on the two constraints on $\bbP$, we have
\begin{align*}
& \int p (y \mid z_{1}, l_{1}, z_{2}, l_{2}, m) p (z_{1}, l_{1}, z_{2}, l_{2}) \diff z_{1} \diff l_{1} \diff z_{2} \diff l_{2} = \int p (y \mid z_{1}, l_{1}, z_{2}, l_{2}, m) p (l_{1}, z_{2}, l_{2} \mid z_{1}) \diff l_{1} \diff z_{2} \diff l_{2}, \\
& \int p (y \mid z_{1}, l_{1}, z_{2}, l_{2}, m) p (z_{2}, l_{2} \mid z_{1}, l_{1}) \diff z_{2} \diff l_{2} = \int p (y \mid z_{1}, l_{1}, z_{2}, l_{2}, m) p (l_{2} \mid z_{1}, l_{1}, z_{2}) \diff l_{2}.
\end{align*}
Since $s \in \Lambda_{\bbP}$, we must also conclude that
\begin{align*}
& \int p (y \mid z_{1}, l_{1}, z_{2}, l_{2}, m) p (z_{1}, l_{1}, z_{2}, l_{2}) \{s (y \mid z_{1}, l_{1}, z_{2}, l_{2}, m) + s (z_{1}, l_{1}, z_{2}, l_{2})\} \diff z_{1} \diff l_{1} \diff z_{2} \diff l_{2} \\
& = \int p (y \mid z_{1}, l_{1}, z_{2}, l_{2}, m) p (l_{1}, z_{2}, l_{2} \mid z_{1}) \{s (y \mid z_{1}, l_{1}, z_{2}, l_{2}, m) + s (l_{1}, z_{2}, l_{2} \mid z_{1})\} \diff l_{1} \diff z_{2} \diff l_{2}, \\
& \int p (y \mid z_{1}, l_{1}, z_{2}, l_{2}, m) p (z_{2}, l_{2} \mid z_{1}, l_{1}) \{s (y \mid z_{1}, l_{1}, z_{2}, l_{2}, m) + s (z_{2}, l_{2} \mid z_{1}, l_{1})\} \diff z_{2} \diff l_{2} \\
& = \int p (y \mid z_{1}, l_{1}, z_{2}, l_{2}, m) p (l_{2} \mid z_{1}, l_{1}, z_{2}) \{s (y \mid z_{1}, l_{1}, z_{2}, l_{2}, m) + s (l_{2} \mid z_{1}, l_{1}, z_{2})\} \diff l_{2}.
\end{align*}
Let $X := \{Z_{1}, L_{1}, Z_{2}, L_{2}\}$. Under $\bbP_{t}$, following the same calculations, the corresponding constraints on $h$ are
\begin{equation}
\label{sequential Verma Fredholm}
\begin{split}
& \int p_{y} \{ s_{x} s_{y} + (1+ t \cdot s_{x}) h\} p (l_{1}, z_{2}, l_{2} \mid z_{1}) \diff l_{1} \diff z_{2} \diff l_{2} \\
& \qquad = \{1+ t \cdot s (z_{1})\} \int p_{y} \{s_{x} s_{y} + (1+ t \cdot s_{x}) h\} p (z_{1}, l_{1}, z_{2}, l_{2}) \diff z_{1} \diff l_{1} \diff z_{2} \diff l_{2} \\
& \qquad \quad + s (z_{1}) \int p_{y} (s_{x} + s_{y}) p (z_{1}, l_{1}, z_{2}, l_{2}) \diff z_{1} \diff l_{1} \diff z_{2} \diff l_{2}, \\
& \{1+t \cdot s(z_{1}, l_{1})\} \int p_{y} \{s_{x} s_{y} + (1 + t \cdot s_{x}) h\} p (l_{2} \mid z_{1}, l_{1}, z_{2}) \diff l_{2} \\
& \qquad + s (z_{1}, l_{1}) \int p_{y} (s_{x} + s_{y}) p (l_{2} \mid z_{1}, l_{1}, z_{2}) \diff l_{2} \\
& \qquad = \{1 + t \cdot s (z_{1}, l_{1}, z_{2})\} \int p_{y} \{s_{x} s_{y} + (1+ t \cdot s_{x}) h\} p (z_{2}, l_{2} \mid z_{1}, l_{1}) \diff z_{2} \diff l_{2} \\
& \qquad \quad + s (z_{1}, l_{1}, z_{2}) \int p_{y} (s_{x} + s_{y}) p(z_{2}, l_{2} \mid z_{1}, l_{1}) \diff z_{2} \diff l_{2},
\end{split}
\end{equation}
where we introduce the following short-hand notation:
\begin{align*}
p_{y} := p (y \mid z_{1}, l_{1}, z_{2}, l_{2}, m), \quad s_{y} := s (y \mid z_{1} , l_{1}, z_{2}, l_{2}, m) \\
s_{x} := s (z_{1}, l_{1}, z_{2}, l_{2}), \quad h := h (y \mid z_{1}, l_{1}, z_{2}, l_{2}, m).
\end{align*}
We exhibit a solution $h$ to equation~\eqref{sequential Verma Fredholm} to complete the proof of the tangent space. Let
\begin{align*}
& f_{1} (m, y) := \int p_{y} (s_{y} + s_{x}) p (z_{1}, l_{1}, z_{2}, l_{2}) \diff z_{1} \diff l_{1} \diff z_{2} \diff l_{2}, \\
& f_{2} (z_{1}, l_{1}, m, y) = \int p_{y} (s_{y} + s (z_{2}, l_{2} \mid z_{1}, l_{1})) p (z_{2},l_{2} \mid z_{1},l_{1}) \diff z_{2} \diff l_{2},
\end{align*}
Define
\begin{align*}
\omega (z_{1}, l_{1}, z_{2}, m, y) := f_{2} s (z_{1}, l_{1}, z_{2}) + \{s (z_{1}) f_{1} - f_{2} s (z_{1}, l_{1})\} \frac{1+t \cdot s (z_{1}, l_{1}, z_{2})}{1+ t \cdot s (z_{1}, l_{1})}.
\end{align*}
Solving \eqref{sequential Verma Fredholm} therefore yields a solution $h$ that can take the following form:
\begin{align*}
p (y \mid x, m) h (y \mid x, m) = \frac{\omega (z_{1}, l_{1}, z_{2}, m, y) - p_{y} s_{x} s_{y}}{1 + t \cdot s (z_{1},l_{1},z_{2},l_{2})}.
\end{align*}
Because $\int f_{1} (m, y) \diff y = \int f_{2} (z_{1}, l_{1}, m, y) \diff y = 0$ and $\int p_{y} s_{y} \diff y = 0$, $h$ is a valid second-order perturbation. This completes the proof of the reverse direction.

For any pathwise differentiable parameter with nonparametric influence function $\varphi^{\mathrm{np}}$, the residual
$
\varphi^{\mathrm{np}}
-
\Pi\left(
\varphi^{\mathrm{np}}
\mid
\mathcal H_{1+2}
\right)
$
is a valid influence function under the two restrictions and has variance no greater than that of $\varphi^{\mathrm{np}}$; it is efficient when $\mathcal H_{1+2}=\Lambda_{\bbP}^{\perp}$. A finite-dimensional approximation is obtained by letting $\boldsymbol{\chi}_j(O)$ collect selected basis directions from the $j$th family in \eqref{eq:sequential_verma_directions}, stacking them as
\begin{align*}
\boldsymbol{\chi}(O)
=
\left\{
\boldsymbol{\chi}_1(O)^\top,
\boldsymbol{\chi}_2(O)^\top
\right\}^\top,
\end{align*}
and applying the joint least-squares projection in Theorem~\ref{theorem:locally_eff}. This projection uses the full block covariance matrix
\begin{align*}
\boldsymbol{\Sigma}
&=
\bbE\left[
\boldsymbol{\chi}(O)\boldsymbol{\chi}(O)^\top
\right]
=
\begin{pmatrix}
\boldsymbol{\Sigma}_{11} & \boldsymbol{\Sigma}_{12}\\
\boldsymbol{\Sigma}_{21} & \boldsymbol{\Sigma}_{22}
\end{pmatrix},
\quad \text{where } \
\boldsymbol{\Sigma}_{jk}
=
\bbE\left[
\boldsymbol{\chi}_j(O)\boldsymbol{\chi}_k(O)^\top
\right].
\end{align*}
Although the two direction families arise from the same fixing operation and use the same fixing weight, this does not imply that $\boldsymbol{\Sigma}_{12}$ vanishes. The construction ensures that each family lies in the orthocomplement associated with its corresponding restriction, but it does not imply that the two families are mutually orthogonal. Indeed, both families involve the same mediator and outcome variables and are indexed by overlapping histories, with $H_1\subset H_2$, so their cross-covariances may be nonzero. The optimal projection must therefore retain the off-diagonal blocks $\boldsymbol{\Sigma}_{12}$ and $\boldsymbol{\Sigma}_{21}$. Only when these blocks vanish does the joint projection reduce to the sum of the two separate projections.

\clearpage
\section{Additional Simulation Details and Results}
\label{app:sims}

This appendix provides the complete simulation details and additional results corresponding to Section~\ref{sec:sims}. The main text presents representative results for estimation of $\psi_0(\bbP)=\bbE[Y(0)]$ under the extended front-door model. Here we provide the complete simulation study for both the extended front-door and Napkin models, including results for the treatment-specific means $\psi_0(\bbP)$ and $\psi_1(\bbP)$ and the average treatment effect (ATE) $\tau(\bbP) = \psi_1(\bbP) - \psi_0(\bbP)$.

All simulations follow the same general implementation strategy. For each data-generating mechanism (DGP), we first generated a large synthetic population from the known DGP to approximate the population quantities appearing in the asymptotic efficiency theory. Oracle nuisance functions, projection coefficients, and theoretical asymptotic variances were then computed from this synthetic population. Finite-sample performance was evaluated using independent Monte Carlo replications, within which nuisance functions were estimated from the observed data before constructing each estimator. Throughout the simulation study, empirical performance is summarized by the empirical $n$-scaled variance, and compared with the corresponding theoretical asymptotic variance obtained from the appropriate influence function. Figures displaying finite-sample performance report Monte Carlo confidence intervals for the empirical $n$-scaled variances together with the corresponding theoretical asymptotic variances.

For both graphical models we consider three simulation experiments.
Experiment~1 evaluates finite-sample performance under a fixed DGP by comparing empirical $n$-scaled variances with their corresponding theoretical asymptotic variances over increasing sample sizes.
Experiment~2 investigates how the relative asymptotic variances of the five competing estimators change across a broad collection of DGPs generated by varying interpretable features of the treatment and outcome mechanisms.
Experiment~3 studies the finite-dimensional approximation to the Verma orthocomplement by progressively enriching the outcome basis used for the projection estimator.

The two graphical models share a common simulation philosophy but differ in the underlying structural equations. Sections~\ref{app:sims:frontdoor} and~\ref{app:sims:napkin} therefore describe only the graph-specific components of each simulation design.

\subsection{Simulation studies for the extended front-door model}
\label{app:sims:frontdoor}

\subsubsection{Data-generating mechanism}
\label{app:sim:frontdoor:dgp}

This section provides the complete specification of the extended front-door simulation design described in Section~\ref{sec:sims}. 

The observed data consist of $O=(X,Z,A,M,Y)$, where $X$, $Z$, $A$, and $M$ are binary and $Y$ is continuous, and are generated according to the latent-variable extended front-door graph shown in Figure~\ref{fig:front-door}(a), with latent variable $U$. Let
\begin{align*}
U 
&\sim \operatorname{Uniform}(0,1), 
\\  
X 
&\sim
\operatorname{Bernoulli}\!\big(
\operatorname{expit}
\left\{
\beta_{X0}
\right\}
\big),
\\
Z \mid X
&\sim
\operatorname{Bernoulli}
\left(
\operatorname{expit}
\{
\beta_{Z0}
+
\beta_{ZX}X
\}
\right),
\\
A \mid U, X, Z
&\sim
\operatorname{Bernoulli}
\left(
\operatorname{expit}
\{
\beta_{A0}
+
\beta_{AZ}Z
+
\beta_{AX}X
+
\beta_{AU}(U-\tfrac12)
+
\eta_{\mathrm{int}}ZX
\}
\right),
\\
M \mid X, Z, A
&\sim
\operatorname{Bernoulli}
\left(
\operatorname{expit}
\{
\beta_{M0}
+
\beta_{MA}A
+
\beta_{MZ}Z
+
\beta_{MX}X
+
\beta_{MAZ}AZ
\}
\right),
\end{align*}

The outcome model was constructed so that perturbations of the conditional mean, conditional variance, and conditional skewness could be varied independently through arm-specific tuning parameters. Specifically, six nonnegative tuning parameters, $\eta_{\mu,T}, \eta_{\mu,C}, \eta_{\sigma,T}, \eta_{\sigma,C}, \eta_{\kappa,T}$, and $\eta_{\kappa,C}$, govern treatment-specific perturbations of the conditional mean, conditional variance, and conditional skewness, respectively. The subscripts $T$ and $C$ denote the target treatment level $a_0$ and its contrast $1-a_0$. Setting a tuning parameter equal to zero removes the corresponding treatment-specific perturbation while leaving the remaining components of the outcome model unchanged. This construction permits a systematic investigation of how progressively richer Verma basis functions exploit increasingly rich features of the conditional outcome distribution. The continuous outcome is then generated according to
\begin{align}
Y
=
\mu(U,M,X)
+
\sigma(U,M,X) \, 
\varepsilon(U,M,X),
\label{eq:frontdoor_outcome}
\end{align}
where the three components are specified below.

\vspace{0.25cm} \noindent 
\underline{\bf Conditional mean.} For each treatment level $a\in\{0,1\}$, we define treatment-specific mediator profile 
\begin{align}
q_a(m,x)
&=
\sum_{z=0}^1
\Pr(M=m\mid A=a,Z=z,X=x)
\Pr(Z=z\mid X=x).
\label{eq:frontdoor_profile}
\end{align}
For the estimand of interest, define the target treatment level by $a_T\in\{0,1\}$,  and let $a_C=1-a_T$ denote the contrast treatment level. We write $q_T(m,x)=q_{a_T}(m,x)$ and $q_C(m,x)=q_{a_C}(m,x)$, for the corresponding mediator profiles under the target and contrast interventions. The conditional mean of the outcome is generated as
\begin{align}
\mu(U,M,X)
&=
\beta_{Y0}
+
\beta_{YM}M
+
\beta_{YX}X
+
\beta_{YU}U
+
\beta_{YU^2}
\left(
U^2-\frac13
\right)
\nonumber\\
&
\quad
+ 
\Big\{
\eta_{\mu_T}
q_T(M,X)
+
\eta_{\mu_C}
q_C(M,X) \Big\} \, 
g_\mu(U),
\label{eq:frontdoor_mean}
\end{align}
where
$g_\mu(U) = \gamma_{\mu0} + \gamma_{\mu1} \left(U-\frac12\right) + \gamma_{\mu2} \left(U^2-\frac13\right)$, and the parameters $\eta_{\mu_T}$ and $\eta_{\mu_C}$ control the magnitude of treatment-specific perturbations to the conditional mean. 

\vspace{0.25cm} \noindent 
\underline{\bf Conditional variance.} The conditional standard deviation is generated through a log-linear model,
\begin{align}
\log\sigma(U,M,X)
&=
\beta_{\sigma0}
+
\beta_{\sigma U}
\left(U-\frac12\right)
\nonumber\\
&
\quad
+
\Big\{
\eta_{\sigma,T}
q_T(M,X)
+
\eta_{\sigma,C}
q_C(M,X)
\Big\} \, 
g_\sigma(U),
\label{eq:frontdoor_sd}
\end{align}
where $g_\sigma(U) = \gamma_{\sigma0} + \gamma_{\sigma1} \left(U-\frac12\right) + \gamma_{\sigma2} \left(U^2-\frac13\right)$. To avoid numerically extreme variances, the resulting log-standard deviation is truncated to the interval $[-3, 3]$ before exponentiation.

\vspace{0.25cm} \noindent 
\underline{\bf Skewness loading.} Finally, skewness is introduced through the loading
\begin{align}
\lambda(U,M,X)
=
\Big\{ \eta_{\kappa,T}
q_T(M,X)
+
\eta_{\kappa,C}
q_C(M,X)
\Big\} \, 
g_\lambda(U),
\label{eq:frontdoor_lambda}
\end{align}
where
$g_\lambda(U) = \gamma_{\lambda0} + \gamma_{\lambda1} \left(U-\frac12\right) + \gamma_{\lambda2} \left(U^2-\frac13\right)$. The parameters $\eta_{\kappa,T}$ and $\eta_{\kappa,C}$
control the strength of treatment-specific departures from conditional Gaussianity.

Let $\varepsilon_N \sim \operatorname{Normal}(0,1)$, and define the standardized chi-square random variable $\varepsilon_\chi = \frac{\chi_3^2-3}{\sqrt6}$. The outcome error is constructed as
\begin{align}
\varepsilon(U, M, X)
=
\frac{
\varepsilon_N
+
\lambda(U,M,X)\varepsilon_\chi
}
{
\sqrt{1+\lambda(U,M,X)^2}
},
\label{eq:frontdoor_error}
\end{align}
which has conditional mean zero and conditional variance one while allowing the conditional third moment to vary with $\lambda(U,M,X)$.

The coefficients defining the baseline DGP remain fixed throughout all three experiments and are listed in Table~\ref{tab:frontdoor_baseline_parameters}. The tuning parameters used to generate the simulation grids are summarized in Table~\ref{tab:frontdoor_tuning_parameters}. Experiment~1 fixes the tuning parameters at the baseline configuration, 
\begin{align*}
\eta_{\mathrm{int}}=-0.5,\quad
\eta_{\mu_T}=\eta_{\mu_C}=1.0,\quad
\eta_{\sigma,T}=\eta_{\sigma,C}=1.5,\quad
\eta_{\kappa,T}=\eta_{\kappa,C}=2.0.
\end{align*}
whereas Experiments~2 and~3 vary them over the grids reported in the table. The Cartesian product of the tuning parameter values yields $\mathbf{2,187}$ \textbf{candidate DGPs} in Experiment~2, and  $\mathbf{1,296}$ \textbf{candidate DGPs} in Exepriment~3.

\begin{table}[!t]
\caption{Fixed coefficients defining the baseline data-generating mechanism for the extended front-door simulations. These coefficients remained unchanged throughout all simulation studies.}
\label{tab:frontdoor_baseline_parameters}
\begin{tabular}{lll}
\toprule
Component & Parameter(s) & Value \\
\midrule
$X$
&
$(\beta_{X0})$
&
$(-0.15)$
\\[0.6ex]

$Z \mid X$
&
$(\beta_{Z0},\beta_{ZX})$
&
$(-0.55, 1.00)$
\\[0.6ex]

$A \mid U, X, Z$
&
$(\beta_{A0},\beta_{AZ},\beta_{AX},\beta_{AU})$
&
$(-0.20,\;1.10,\;0.45,\;1.00)$
\\[0.6ex]

$M \mid X, Z, A$
&
$(\beta_{M0},\beta_{MA},\beta_{MZ},\beta_{MX},\beta_{MAZ})$
&
$(-0.35,\;1.20,\;0.55,\;0.50,\;0.30)$
\\[0.6ex]

$Y \mid U, X, M$: $\mu(U, M, X)$
&
$(\beta_{Y0},\beta_{YM},\beta_{YX},\beta_{YU},\beta_{YU^2})$
&
$(3.00,\;1.00,\;0.40,\;0.70,\;0.25)$
\\[0.6ex]

&
$(\gamma_{\mu0},\gamma_{\mu1},\gamma_{\mu2})$
&
$(0,\;0.40,\;0.40)$
\\[0.6ex]

$Y \mid U, X, M$: $\sigma(U, M, X)$
&
$(\beta_{\sigma0},\beta_{\sigma U})$
&
$(-0.10,\;0.15)$
\\[0.6ex]

&
$(\gamma_{\sigma0},\gamma_{\sigma1},\gamma_{\sigma2})$
&
$(0,\;0.90,\;0.70)$
\\[0.6ex] 

$Y \mid U, X, M$: $\lambda(U, M, X) $
& 
$(\gamma_{\lambda0},\gamma_{\lambda1},\gamma_{\lambda2})$
&
$(0,\;0.750,\;0.90)$
\\
\bottomrule
\end{tabular}
\end{table}

\begin{table}[!t]
\centering
\caption{Simulation tuning parameters varied across the grid-search experiments. The Cartesian product of the parameter values defines the candidate data-generating mechanisms considered in Experiments~2 and~3.}
\label{tab:frontdoor_tuning_parameters}
\begin{tabular}{llll}
\toprule
Distributional feature & Parameter & Experiment~2 & Experiment~3\\
\midrule
Treatment mechanism: interaction
&
$\eta_{\mathrm{int}}$
&
$-0.5,\;0.0,\;0.5$
&
--- \\
Outcome mechanism: target mean
&
$\eta_{\mu_T}$
&
$0.0,\;0.5,\;1.0$
&
$0.0,\;0.5,\;1.0$
\\
\textcolor{white}{Outcome mechanism:} contrast mean
&
$\eta_{\mu_C}$
&
$0.0,\;0.5,\;1.0$
&
$-1.0,\;0.0,\;1.0$
\\
\textcolor{white}{Outcome mechanism} target variance
&
$\eta_{\sigma,T}$
&
$0.0,\;1.0,\;1.5$
&
$1.0,\;2.0,\;3.0,\;4.0$
\\
\textcolor{white}{Outcome mechanism} contrast variance
&
$\eta_{\sigma,C}$
&
$0.0,\;1.0,\;1.5$
&
$-1.0,\;0.0,\;1.0$
\\
\textcolor{white}{Outcome mechanism} target skewness
&
$\eta_{\kappa,T}$
&
$0.0,\;1.0,\;2.0$
&
$1.0,\;2.0,\;3.0,\;4.0$
\\
\textcolor{white}{Outcome mechanism} contrast skewness
&
$\eta_{\kappa,C}$
&
$0.0,\;1.0,\;2.0$
&
$-1.0,\;0.0,\;1.0$
\\
\bottomrule
\end{tabular}
\end{table}

\subsubsection{Oracle calculations and simulation implementation}
\label{app:frontdoor:implementation}

Throughout the simulations we considered estimation of the treatment-specific means and the average treatment effect using the five one-step estimators described in Section~\ref{sec:front-door}: the two fixed-$Z$ indexed estimators, their equally weighted and optimally weighted affine combinations, and the Verma projection estimator.

The true treatment-specific mean was computed from the identifying front-door functional derived in Appendix~\ref{app:id_npEIF_frontdoor},
\begin{align}
\psi_a(\bbP)
=
\sum_{x,m,z}
p(x)\,
p(z\mid x)\,
p(m\mid A=a,z,x)
\left\{
\sum_{a'}
p(a'\mid x,z')
\bbE[Y\mid m,a',z',x]
\right\},
\label{eq:true_frontdoor_mean}
\end{align}
where
$
\bbE(Y\mid m,a,z,x)
=
\int
\mu(u,m,x)\,
p(u\mid a,z,x)\,du,
$
and
$
p(a\mid x,z)
=
\int
p(a\mid x,z,u)\,
p(u)\,du.
$
The posterior density $p(u\mid a,z,x)$ was obtained from Bayes' rule under the known data-generating mechanism.

Oracle quantities were computed directly from the known latent-variable DGP using deterministic Gauss--Legendre quadrature over $U\in[0,1]$. This numerical integration was used to evaluate the true causal parameters, all observed-data nuisance functions, and the higher conditional moments required for constructing the Verma basis functions. The mediator and covariate distributions were available analytically and therefore required no numerical approximation.

For each candidate DGP, oracle calculations were based on an independent synthetic population of size $250{,}000$, which was used to approximate population nuisance functions, projection coefficients, and theoretical asymptotic variances. Finite-sample performance was evaluated separately using $1{,}250$ independent Monte Carlo replications over an increasing sequence of sample sizes ranging from $250$ to $16,000$ observations. Unless otherwise stated, nuisance functions in each Monte Carlo sample were estimated using correctly specified saturated regression models.

\subsubsection{Experiment 1: Finite-sample performance}

Experiment~1 evaluates finite-sample performance under the baseline DGP. The tuning parameters were fixed at the configuration reported above, and independent Monte Carlo samples were generated over an increasing sequence of sample sizes. For each sample, nuisance functions were estimated using correctly specified saturated regression models, and the empirical $n$-scaled variance of each estimator was compared with its theoretical asymptotic variance.

Table~\ref{tab:frontdoor_var_all_estimands} reports the theoretical asymptotic variances of the five estimators for the three estimands. Across all estimands, combining the indexed estimators substantially improves efficiency relative to either fixed-$Z$ estimator, while the Verma projection estimator achieves the smallest theoretical asymptotic variance. The magnitude of the additional improvement depends on the estimand. Relative to the equally weighted indexed estimator, the Verma projection estimator yields roughly 9\% greater efficiency for the ATE and $\bbE[Y(1)]$ and roughly 15\% for $\bbE[Y(0)]$.

Figures~\ref{fig:frontdoor_var_ate_n_scaled_variance} and \ref{fig:frontdoor_var_mean_1_n_scaled_variance} present the corresponding empirical $n$-scaled variances. Across all estimands, the empirical variances converge toward the theoretical asymptotic variances while preserving the same ordering of the competing estimators.

\subsubsection{Experiment 2: Efficiency across data-generating mechanisms}

Experiment~2 examines the relative efficiency of the competing estimators across the 2,187 candidate DGPs generated from the tuning grid in Table~\ref{tab:frontdoor_tuning_parameters}. For each DGP in the tuning grid, we compute the theoretical asymptotic variances of the five estimators and the corresponding efficiency gains. Among DGPs with finite, rankable variance values, configurations are ranked primarily by the efficiency gain of the Verma projection estimator over
the optimally weighted indexed estimator, $V_{\rm indexed,opt}/V_{\rm Verma}-1$. When these gains are nearly tied, we break ties using the efficiency gain over the equally weighted indexed estimator, $V_{\rm indexed,equal}/V_{\rm Verma}-1$. 

Table~\ref{tab:frontdoor_var_grid_all_estimands} summarizes the twenty highest-ranked data-generating mechanisms for the average treatment effect, $\bbE[Y(0)]$, and $\bbE[Y(1)]$, together with the corresponding tuning parameters. Figures~\ref{fig:frontdoor_var_grid_ate_top_20_method_variances_split_x} and \ref{fig:frontdoor_var_grid_mean_1_top_20_method_variances_split_x} display the theoretical asymptotic variances of the five competing estimators for these same configurations for the ATE and $\bbE[Y(1)]$, respectively. 

Across all three estimands, the ordering of the competing estimators is remarkably stable. The Verma projection estimator consistently attains the smallest theoretical asymptotic variance, followed by the optimally weighted indexed estimator, the equally weighted indexed estimator, and the two fixed-$Z$ estimators.

The tuning parameters associated with the highest-ranked DGPs differ across estimands. For the ATE, the largest gains typically occur under strong treatment-specific perturbations of the conditional variance and skewness. For the treatment-specific means, the most influential tuning parameters depend on the target treatment arm. Thus, although the Verma projection consistently improves upon the indexed estimators, the magnitude of this improvement depends on both the estimand and the underlying outcome mechanism.
\subsubsection{Experiment 3: Effect of basis richness}

Experiment~3 investigates how enriching the finite-dimensional basis used to approximate the Verma orthocomplement affects the theoretical asymptotic variance of the resulting Verma projection estimator. The comparison is performed over the 1,296 candidate DGPs generated from the tuning grid in Table~\ref{tab:frontdoor_tuning_parameters}. For each DGP in the tuning grid, we compute the theoretical asymptotic variances for the $Y$-only basis, the $(Y,Y^2)$ basis, and the full $(Y,Y^2,Y^3)$ basis. Among DGPs with finite, rankable variance values, configurations are ranked by the efficiency gain of the full Verma basis over the $Y$-only basis, $V_Y/V_{(Y,Y^2,Y^3)}-1$. 

Table~\ref{tab:frontdoor_basis_grid_all_estimands} reports the twenty highest-ranked configurations together with the theoretical variances under each basis specification and the corresponding relative efficiency gains. Figures~\ref{fig:frontdoor_basis_grid_ate_top_20_method_variances_split_x} and \ref{fig:frontdoor_basis_grid_mean_1_top_20_method_variances_split_x} display the theoretical variances for these configurations separately for the average treatment effect and $\bbE[Y(1)]$.

Across all three estimands, enlarging the basis produces the expected monotone reduction in theoretical asymptotic variance. The improvement depends on both the estimand and the underlying DGP. The largest gains occur when the target-arm conditional variance perturbation is strongest, with nearly all of the highest-ranked configurations satisfying $\eta_{\sigma,T}=4$.

Table~\ref{tab:frontdoor_basis_grid_all_estimands} further shows that both quadratic and cubic outcome terms contribute to the additional efficiency, with the cubic enrichment providing an improvement comparable to that obtained from adding the quadratic basis. The overall gain is largest for $\bbE[Y(1)]$ (approximately $1.14\%$ relative to the $\{Y\}$ basis), followed by the ATE (approximately $1.10\%$) and $\bbE[Y(0)]$ (approximately $0.56\%$). These results indicate that richer basis expansions are most beneficial in data-generating mechanisms exhibiting pronounced treatment-specific heterogeneity in the conditional outcome distribution.

\clearpage

\begin{table}[htbp]
\centering
\scriptsize
\setlength{\tabcolsep}{2.1pt}
\renewcommand{\arraystretch}{0.92}
\caption{Theoretical variances and relative efficiency gains for the five estimators across the three parameters under the extended front-door model. Relative efficiency gains are computed with respect to the reference method indicated by each column. }
\label{tab:frontdoor_var_all_estimands}
\resizebox{\textwidth}{!}{%
\begin{tabular}{llccccc}
\toprule
 & Method & \shortstack{Theoretical\\variance} & \shortstack{Gain vs.\\fixed $Z=0$} & \shortstack{Gain vs.\\fixed $Z=1$} & \shortstack{Gain vs.\\equal indexed} & \shortstack{Gain vs.\\optimal indexed} \\
\midrule
$\mathrm{ATE}$ & Fixed $Z=0$ & 2.680 & 0.000\% & 28.812\% & -25.831\% & -27.121\% \\
 & Fixed $Z=1$ & 3.453 & -22.368\% & 0.000\% & -42.421\% & -43.422\% \\
 & Indexed (equal) & 1.988 & 34.828\% & 73.675\% & 0.000\% & -1.739\% \\
 & Indexed (optimal) & 1.953 & 37.213\% & 76.748\% & 1.769\% & 0.000\% \\
 & Verma projection & 1.832 & 46.277\% & 88.423\% & 8.492\% & 6.606\% \\
\addlinespace[6pt]
$\mathbb{E}[Y(1)]$ & Fixed $Z=0$ & 11.681 & 0.000\% & -21.185\% & -53.778\% & -54.428\% \\
 & Fixed $Z=1$ & 9.206 & 26.880\% & 0.000\% & -41.354\% & -42.178\% \\
 & Indexed (equal) & 5.399 & 116.348\% & 70.514\% & 0.000\% & -1.405\% \\
 & Indexed (optimal) & 5.323 & 119.432\% & 72.944\% & 1.425\% & 0.000\% \\
 & Verma projection & 4.992 & 133.990\% & 84.418\% & 8.154\% & 6.634\% \\
\addlinespace[6pt]
$\mathbb{E}[Y(0)]$ & Fixed $Z=0$ & 7.861 & 0.000\% & 20.177\% & -40.453\% & -40.957\% \\
 & Fixed $Z=1$ & 9.447 & -16.789\% & 0.000\% & -50.451\% & -50.870\% \\
 & Indexed (equal) & 4.681 & 67.936\% & 101.820\% & 0.000\% & -0.845\% \\
 & Indexed (optimal) & 4.641 & 69.368\% & 103.541\% & 0.853\% & 0.000\% \\
 & Verma projection & 4.082 & 92.566\% & 131.420\% & 14.666\% & 13.697\% \\
\bottomrule
\end{tabular}%
}
\begin{flushleft}
\footnotesize
\textit{Note:} Each gain column reports $100\{V_{\mathrm{reference}}/V_{\mathrm{method}}-1\}\%$, with the reference estimator named in the column header. Positive values indicate smaller theoretical variance than the reference.
\end{flushleft}
\end{table}

\begin{figure}[t]
\centering
\includegraphics[scale=0.7]{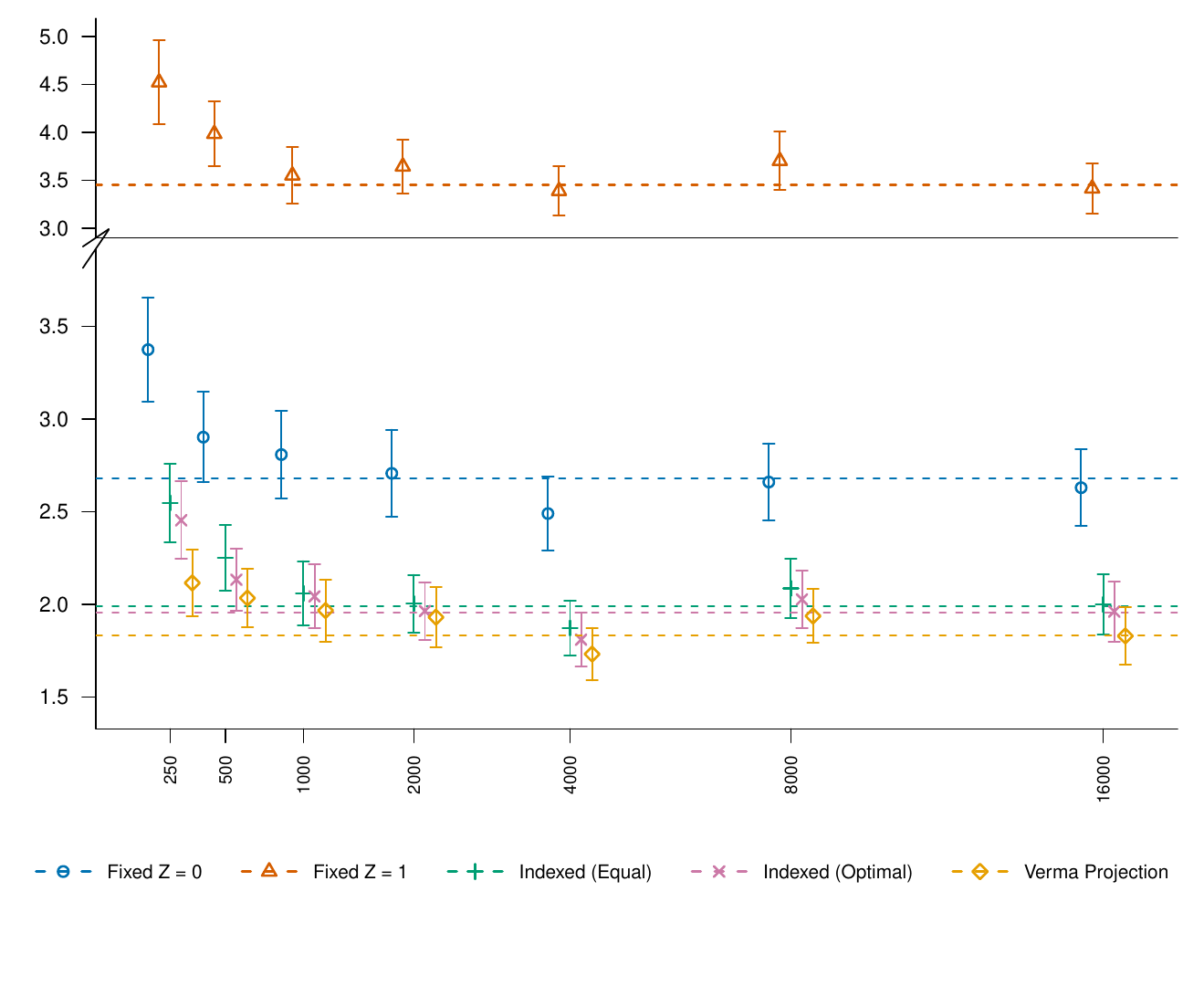}
\vspace{-1cm}
\caption{Finite-sample performance for $\text{ATE} = \psi_1(\bbP) - \psi_0(\bbP)$ under the extended front-door model. Points and vertical intervals show the empirical $n$-scaled variance and its Monte Carlo confidence interval, respectively; horizontal dashed lines show the corresponding theoretical asymptotic variances. Results are shown for the two fixed-$Z$ indexed estimators, the equally weighted indexed estimator, the optimally weighted indexed estimator, and the Verma projection estimator.}
\label{fig:frontdoor_var_ate_n_scaled_variance}
\end{figure}

\begin{figure}[t]
\centering
\includegraphics[scale=0.7]{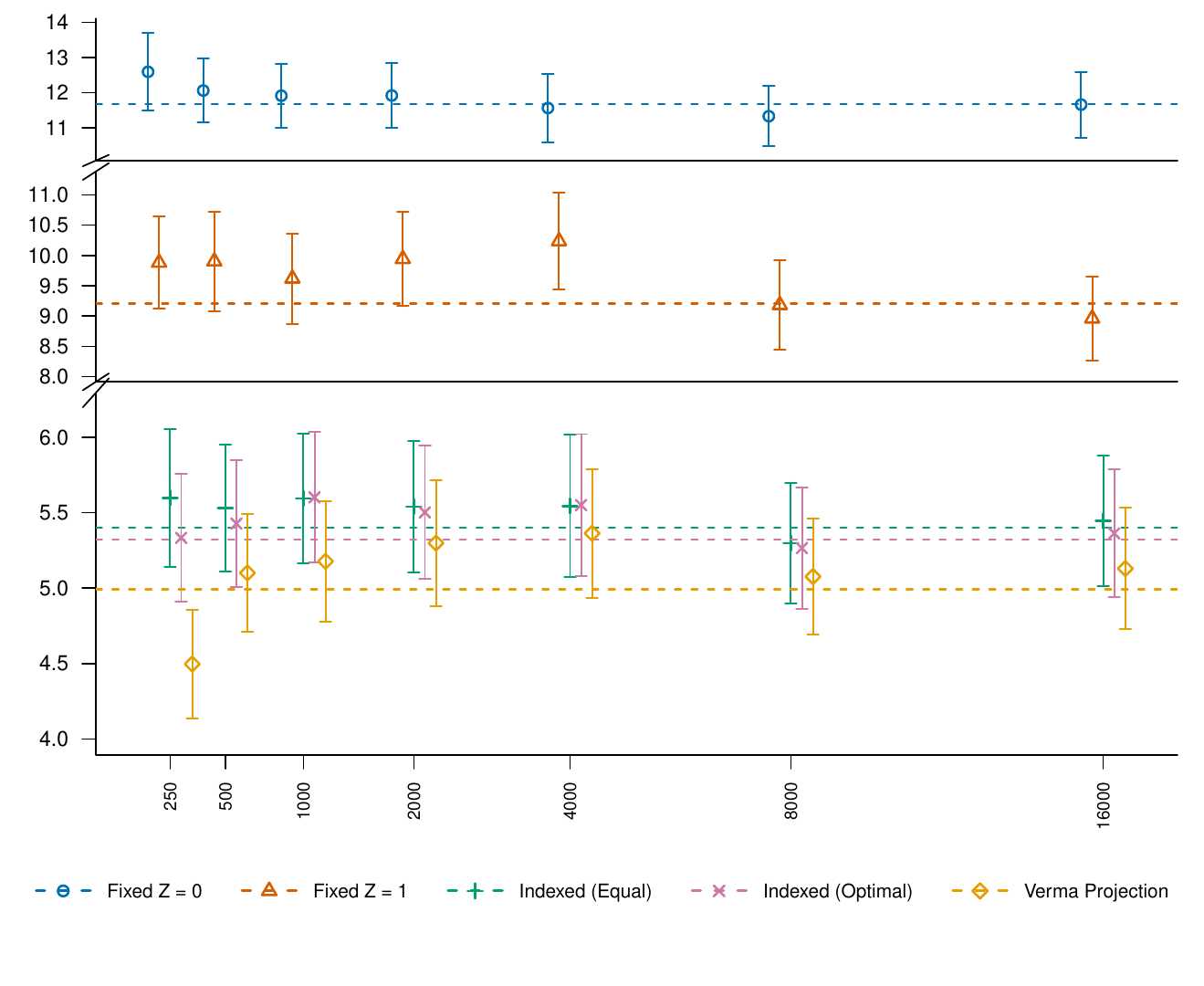}
\vspace{-1cm}
\caption{Finite-sample performance for $\psi_1(\bbP) = \bbE[Y(1)]$ under the extended front-door model. Points and vertical intervals show the empirical $n$-scaled variance and its Monte Carlo confidence interval, respectively; horizontal dashed lines show the corresponding theoretical asymptotic variances. Results are shown for the two fixed-$Z$ indexed estimators, the equally weighted indexed estimator, the optimally weighted indexed estimator, and the Verma projection estimator.}
\label{fig:frontdoor_var_mean_1_n_scaled_variance}
\end{figure}

\begin{table}[!t]
\centering
\setlength{\abovecaptionskip}{2pt}
\setlength{\belowcaptionskip}{1pt}
\scriptsize
\setlength{\tabcolsep}{2.1pt}
\renewcommand{\arraystretch}{0.80}
\caption{Top DGP grid settings under the extended front-door model: theoretical variances of the five estimators and Verma projection efficiency gains across the three parameters. Settings are ranked primarily by the Verma projection gain over the optimally weighted indexed estimator, with the gain over the equally weighted indexed estimator used to order near ties.}
\label{tab:frontdoor_var_grid_all_estimands}
\resizebox{0.9\textwidth}{!}{%
\begin{tabular}{crcccccccccccccc}
\toprule
 & ID & $\eta_{\mathrm{int}}$ & $\eta_{\mu,T}$ & $\eta_{\mu,C}$ & $\eta_{\sigma,T}$ & $\eta_{\sigma,C}$ & $\eta_{\kappa,T}$ & $\eta_{\kappa,C}$ & \shortstack{Fixed\\$Z=0$} & \shortstack{Fixed\\$Z=1$} & \shortstack{Indexed\\equal} & \shortstack{Indexed\\optimal} & \shortstack{Verma\\projection} & \shortstack{Gain\\(\%)} \\
\midrule
\multirow{20}{*}{\rotatebox[origin=c]{90}{$\mathrm{ATE}$}} & 727 & -0.500 & 1.000 & 1.000 & 1.500 & 1.500 & 2.000 & 0.000 & 2.676 & 3.483 & 1.994 & 1.957 & 1.834 & 6.733\% \\
 & 484 & -0.500 & 1.000 & 1.000 & 1.500 & 1.500 & 1.000 & 0.000 & 2.679 & 3.508 & 2.001 & 1.962 & 1.839 & 6.678\% \\
 & 724 & -0.500 & 0.500 & 1.000 & 1.500 & 1.500 & 2.000 & 0.000 & 2.646 & 3.443 & 1.977 & 1.939 & 1.818 & 6.684\% \\
 & 241 & -0.500 & 1.000 & 1.000 & 1.500 & 1.500 & 0.000 & 0.000 & 2.689 & 3.537 & 2.010 & 1.969 & 1.846 & 6.665\% \\
 & 718 & -0.500 & 1.000 & 0.500 & 1.500 & 1.500 & 2.000 & 0.000 & 2.644 & 3.426 & 1.972 & 1.936 & 1.815 & 6.666\% \\
 & 1456 & -0.500 & 1.000 & 1.000 & 1.500 & 1.500 & 2.000 & 1.000 & 2.678 & 3.462 & 1.990 & 1.954 & 1.832 & 6.665\% \\
 & 721 & -0.500 & 0.000 & 1.000 & 1.500 & 1.500 & 2.000 & 0.000 & 2.619 & 3.408 & 1.961 & 1.924 & 1.805 & 6.641\% \\
 & 481 & -0.500 & 0.500 & 1.000 & 1.500 & 1.500 & 1.000 & 0.000 & 2.648 & 3.467 & 1.983 & 1.944 & 1.823 & 6.630\% \\
 & 1213 & -0.500 & 1.000 & 1.000 & 1.500 & 1.500 & 1.000 & 1.000 & 2.678 & 3.477 & 1.993 & 1.956 & 1.835 & 6.626\% \\
 & 715 & -0.500 & 0.500 & 0.500 & 1.500 & 1.500 & 2.000 & 0.000 & 2.617 & 3.391 & 1.956 & 1.921 & 1.801 & 6.623\% \\
 & 238 & -0.500 & 0.500 & 1.000 & 1.500 & 1.500 & 0.000 & 0.000 & 2.658 & 3.496 & 1.992 & 1.952 & 1.831 & 6.617\% \\
 & 475 & -0.500 & 1.000 & 0.500 & 1.500 & 1.500 & 1.000 & 0.000 & 2.647 & 3.450 & 1.979 & 1.941 & 1.820 & 6.611\% \\
 & 1453 & -0.500 & 0.500 & 1.000 & 1.500 & 1.500 & 2.000 & 1.000 & 2.647 & 3.422 & 1.972 & 1.937 & 1.816 & 6.616\% \\
 & 232 & -0.500 & 1.000 & 0.500 & 1.500 & 1.500 & 0.000 & 0.000 & 2.656 & 3.479 & 1.988 & 1.948 & 1.828 & 6.599\% \\
 & 709 & -0.500 & 1.000 & 0.000 & 1.500 & 1.500 & 2.000 & 0.000 & 2.616 & 3.376 & 1.953 & 1.918 & 1.799 & 6.607\% \\
 & 2185 & -0.500 & 1.000 & 1.000 & 1.500 & 1.500 & 2.000 & 2.000 & 2.680 & 3.453 & 1.988 & 1.953 & 1.832 & 6.606\% \\
 & 478 & -0.500 & 0.000 & 1.000 & 1.500 & 1.500 & 1.000 & 0.000 & 2.622 & 3.432 & 1.968 & 1.929 & 1.810 & 6.587\% \\
 & 712 & -0.500 & 0.000 & 0.500 & 1.500 & 1.500 & 2.000 & 0.000 & 2.595 & 3.360 & 1.943 & 1.908 & 1.790 & 6.587\% \\
 & 1447 & -0.500 & 1.000 & 0.500 & 1.500 & 1.500 & 2.000 & 1.000 & 2.645 & 3.405 & 1.967 & 1.933 & 1.813 & 6.597\% \\
 & 235 & -0.500 & 0.000 & 1.000 & 1.500 & 1.500 & 0.000 & 0.000 & 2.631 & 3.460 & 1.977 & 1.936 & 1.817 & 6.576\% \\
\addlinespace[5pt]
\multirow{20}{*}{\rotatebox[origin=c]{90}{$\mathbb{E}[Y(1)]$}} & 2127 & 0.500 & 0.000 & 1.000 & 0.000 & 1.500 & 2.000 & 2.000 & 3.617 & 3.395 & 1.927 & 1.925 & 1.772 & 8.633\% \\
 & 1398 & 0.500 & 0.000 & 1.000 & 0.000 & 1.500 & 2.000 & 1.000 & 3.619 & 3.401 & 1.929 & 1.927 & 1.774 & 8.630\% \\
 & 669 & 0.500 & 0.000 & 1.000 & 0.000 & 1.500 & 2.000 & 0.000 & 3.622 & 3.412 & 1.933 & 1.931 & 1.778 & 8.631\% \\
 & 1884 & 0.500 & 0.000 & 1.000 & 0.000 & 1.500 & 1.000 & 2.000 & 3.619 & 3.403 & 1.930 & 1.928 & 1.775 & 8.590\% \\
 & 2118 & 0.500 & 0.000 & 0.500 & 0.000 & 1.500 & 2.000 & 2.000 & 3.502 & 3.292 & 1.872 & 1.870 & 1.723 & 8.583\% \\
 & 1389 & 0.500 & 0.000 & 0.500 & 0.000 & 1.500 & 2.000 & 1.000 & 3.504 & 3.298 & 1.874 & 1.873 & 1.725 & 8.580\% \\
 & 660 & 0.500 & 0.000 & 0.500 & 0.000 & 1.500 & 2.000 & 0.000 & 3.507 & 3.309 & 1.878 & 1.876 & 1.728 & 8.582\% \\
 & 1155 & 0.500 & 0.000 & 1.000 & 0.000 & 1.500 & 1.000 & 1.000 & 3.623 & 3.414 & 1.934 & 1.932 & 1.779 & 8.573\% \\
 & 426 & 0.500 & 0.000 & 1.000 & 0.000 & 1.500 & 1.000 & 0.000 & 3.630 & 3.434 & 1.940 & 1.939 & 1.786 & 8.575\% \\
 & 2130 & 0.500 & 0.500 & 1.000 & 0.000 & 1.500 & 2.000 & 2.000 & 3.813 & 3.551 & 2.015 & 2.013 & 1.854 & 8.565\% \\
 & 1401 & 0.500 & 0.500 & 1.000 & 0.000 & 1.500 & 2.000 & 1.000 & 3.815 & 3.557 & 2.017 & 2.015 & 1.856 & 8.563\% \\
 & 672 & 0.500 & 0.500 & 1.000 & 0.000 & 1.500 & 2.000 & 0.000 & 3.817 & 3.568 & 2.021 & 2.018 & 1.859 & 8.564\% \\
 & 2109 & 0.500 & 0.000 & 0.000 & 0.000 & 1.500 & 2.000 & 2.000 & 3.407 & 3.206 & 1.827 & 1.825 & 1.682 & 8.542\% \\
 & 651 & 0.500 & 0.000 & 0.000 & 0.000 & 1.500 & 2.000 & 0.000 & 3.412 & 3.223 & 1.833 & 1.831 & 1.687 & 8.542\% \\
 & 183 & 0.500 & 0.000 & 1.000 & 0.000 & 1.500 & 0.000 & 0.000 & 3.644 & 3.464 & 1.951 & 1.950 & 1.796 & 8.551\% \\
 & 1875 & 0.500 & 0.000 & 0.500 & 0.000 & 1.500 & 1.000 & 2.000 & 3.504 & 3.300 & 1.875 & 1.873 & 1.726 & 8.539\% \\
 & 1380 & 0.500 & 0.000 & 0.000 & 0.000 & 1.500 & 2.000 & 1.000 & 3.409 & 3.212 & 1.829 & 1.827 & 1.684 & 8.540\% \\
 & 1887 & 0.500 & 0.500 & 1.000 & 0.000 & 1.500 & 1.000 & 2.000 & 3.815 & 3.559 & 2.018 & 2.015 & 1.857 & 8.524\% \\
 & 2121 & 0.500 & 0.500 & 0.500 & 0.000 & 1.500 & 2.000 & 2.000 & 3.669 & 3.426 & 1.948 & 1.945 & 1.793 & 8.520\% \\
 & 1392 & 0.500 & 0.500 & 0.500 & 0.000 & 1.500 & 2.000 & 1.000 & 3.671 & 3.432 & 1.950 & 1.948 & 1.795 & 8.517\% \\
\addlinespace[5pt]
\multirow{20}{*}{\rotatebox[origin=c]{90}{$\mathbb{E}[Y(0)]$}} & 1699 & -0.500 & 1.000 & 1.000 & 1.500 & 1.500 & 0.000 & 2.000 & 7.868 & 9.549 & 4.709 & 4.665 & 4.096 & 13.897\% \\
 & 1700 & 0.000 & 1.000 & 1.000 & 1.500 & 1.500 & 0.000 & 2.000 & 7.960 & 9.676 & 4.809 & 4.763 & 4.183 & 13.887\% \\
 & 1690 & -0.500 & 1.000 & 0.500 & 1.500 & 1.500 & 0.000 & 2.000 & 7.739 & 9.407 & 4.642 & 4.597 & 4.037 & 13.880\% \\
 & 1691 & 0.000 & 1.000 & 0.500 & 1.500 & 1.500 & 0.000 & 2.000 & 7.830 & 9.534 & 4.742 & 4.696 & 4.124 & 13.864\% \\
 & 1681 & -0.500 & 1.000 & 0.000 & 1.500 & 1.500 & 0.000 & 2.000 & 7.630 & 9.284 & 4.583 & 4.539 & 3.987 & 13.863\% \\
 & 1682 & 0.000 & 1.000 & 0.000 & 1.500 & 1.500 & 0.000 & 2.000 & 7.721 & 9.410 & 4.684 & 4.638 & 4.074 & 13.840\% \\
 & 1696 & -0.500 & 0.500 & 1.000 & 1.500 & 1.500 & 0.000 & 2.000 & 7.743 & 9.355 & 4.629 & 4.588 & 4.031 & 13.831\% \\
 & 970 & -0.500 & 1.000 & 1.000 & 1.500 & 1.500 & 0.000 & 1.000 & 7.898 & 9.632 & 4.738 & 4.691 & 4.122 & 13.822\% \\
 & 1697 & 0.000 & 0.500 & 1.000 & 1.500 & 1.500 & 0.000 & 2.000 & 7.834 & 9.480 & 4.729 & 4.686 & 4.117 & 13.817\% \\
 & 1687 & -0.500 & 0.500 & 0.500 & 1.500 & 1.500 & 0.000 & 2.000 & 7.630 & 9.232 & 4.570 & 4.529 & 3.979 & 13.818\% \\
 & 1943 & 0.000 & 1.000 & 1.000 & 1.500 & 1.500 & 1.000 & 2.000 & 7.953 & 9.619 & 4.793 & 4.750 & 4.173 & 13.813\% \\
 & 961 & -0.500 & 1.000 & 0.500 & 1.500 & 1.500 & 0.000 & 1.000 & 7.768 & 9.490 & 4.670 & 4.623 & 4.063 & 13.805\% \\
 & 1688 & 0.000 & 0.500 & 0.500 & 1.500 & 1.500 & 0.000 & 2.000 & 7.721 & 9.357 & 4.670 & 4.627 & 4.066 & 13.798\% \\
 & 1678 & -0.500 & 0.500 & 0.000 & 1.500 & 1.500 & 0.000 & 2.000 & 7.537 & 9.126 & 4.521 & 4.479 & 3.936 & 13.804\% \\
 & 241 & -0.500 & 1.000 & 1.000 & 1.500 & 1.500 & 0.000 & 0.000 & 7.956 & 9.739 & 4.781 & 4.732 & 4.159 & 13.787\% \\
 & 952 & -0.500 & 1.000 & 0.000 & 1.500 & 1.500 & 0.000 & 1.000 & 7.659 & 9.366 & 4.612 & 4.565 & 4.012 & 13.786\% \\
 & 971 & 0.000 & 1.000 & 1.000 & 1.500 & 1.500 & 0.000 & 1.000 & 7.989 & 9.757 & 4.838 & 4.789 & 4.209 & 13.794\% \\
 & 1934 & 0.000 & 1.000 & 0.500 & 1.500 & 1.500 & 1.000 & 2.000 & 7.823 & 9.477 & 4.725 & 4.682 & 4.114 & 13.788\% \\
 & 1942 & -0.500 & 1.000 & 1.000 & 1.500 & 1.500 & 1.000 & 2.000 & 7.862 & 9.481 & 4.690 & 4.649 & 4.086 & 13.786\% \\
 & 1679 & 0.000 & 0.500 & 0.000 & 1.500 & 1.500 & 0.000 & 2.000 & 7.628 & 9.251 & 4.621 & 4.578 & 4.023 & 13.780\% \\
\bottomrule
\end{tabular}%
}
\vspace{-0.50em}
\begin{flushleft}
\tiny
\textit{Note:} Variance columns are theoretical. Gain is $100\{V_{\mathrm{opt}}/V_{\mathrm{Verma}}-1\}\%$, where $V_{\mathrm{opt}}$ is for the optimally weighted indexed estimator. Rows are ranked by this gain; near ties are ordered by $100\{V_{\mathrm{equal}}/V_{\mathrm{Verma}}-1\}\%$.
\end{flushleft}
\end{table}

\begin{figure}[t]
\centering
\includegraphics[scale=0.54]{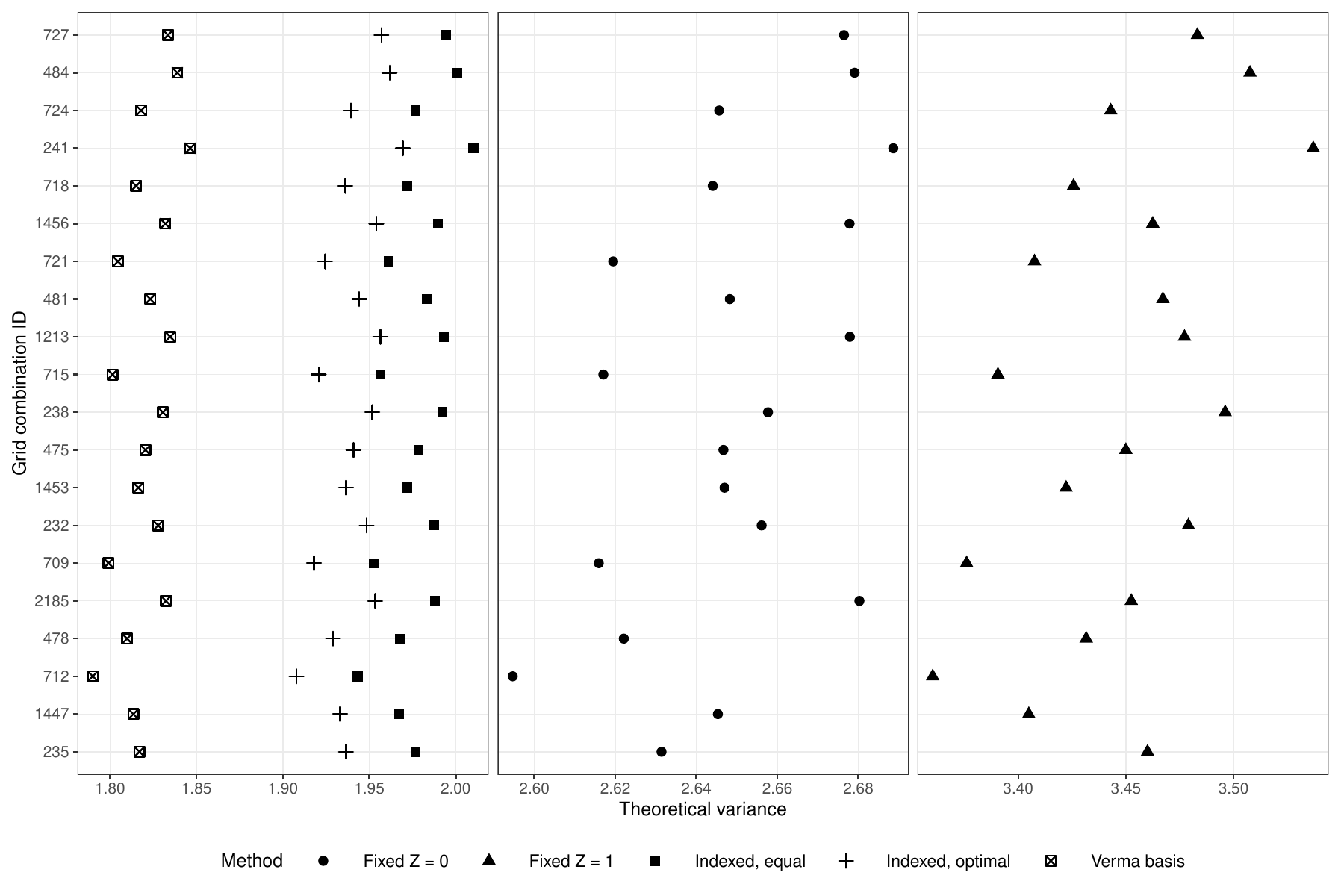}
\caption{
Theoretical asymptotic variances for  $\text{ATE} = \psi_1(\bbP) - \psi_0(\bbP)$ under the extended front-door model across the 20 highest-ranked DGPs in Experiment~2. Each row corresponds to one DGP, and symbols denote the five competing estimators. The horizontal axis is partitioned into three noncontiguous ranges to facilitate visual comparison across methods. DGPs are ranked by the relative efficiency gain of the Verma projection estimator over the optimally weighted indexed estimator. 
}
\label{fig:frontdoor_var_grid_ate_top_20_method_variances_split_x}
\end{figure}

\begin{figure}[t]
\centering
\includegraphics[scale=0.54]{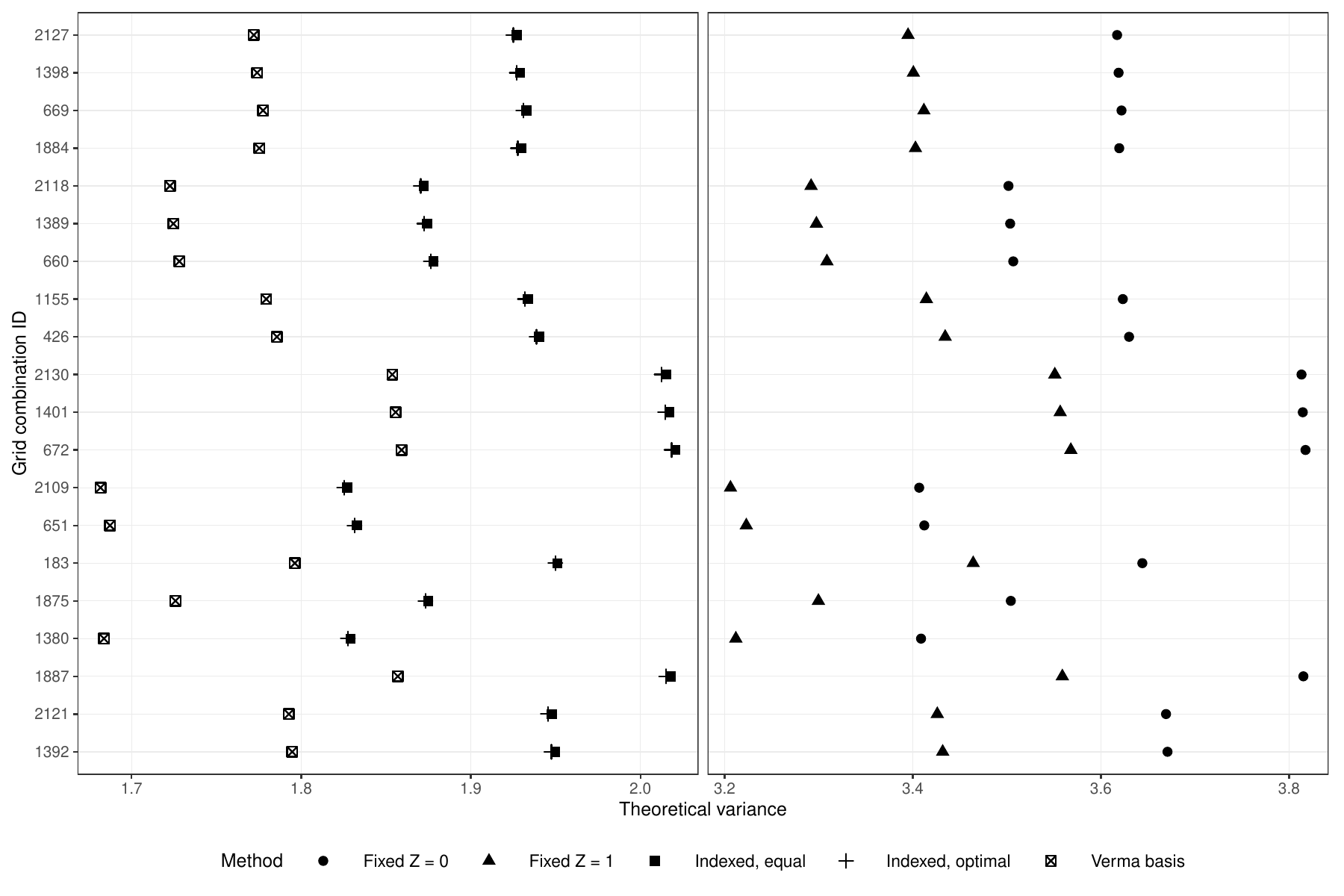}
\caption{
Theoretical asymptotic variances for  $\psi_1(\bbP) = \bbE[Y(1)]$ under the extended front-door model across the 20 highest-ranked DGPs in Experiment~2. Each row corresponds to one DGP, and symbols denote the five competing estimators. The horizontal axis is partitioned into three noncontiguous ranges to facilitate visual comparison across methods. DGPs are ranked by the relative efficiency gain of the Verma projection estimator over the optimally weighted indexed estimator. 
}
\label{fig:frontdoor_var_grid_mean_1_top_20_method_variances_split_x}
\end{figure}

\begin{table}[htbp]
\centering
\setlength{\abovecaptionskip}{2pt}
\setlength{\belowcaptionskip}{1pt}
\scriptsize
\setlength{\tabcolsep}{2.2pt}
\renewcommand{\arraystretch}{0.70}
\caption{Top DGP grid settings under the extended front-door model: theoretical variances for basis enrichments of the Verma projection across the three parameters, with efficiency gains relative to the $Y$-only basis. Settings are ranked by the efficiency gain of the full $(Y,Y^2,Y^3)$ basis over the $Y$-only basis.}
\label{tab:frontdoor_basis_grid_all_estimands}
\resizebox{0.85\textwidth}{!}{%
\begin{tabular}{crrrrrrrrrr@{\hspace{10pt}}rr}
\toprule
 & ID & $\eta_{\mu,T}$ & $\eta_{\mu,C}$ & $\eta_{\sigma,T}$ & $\eta_{\sigma,C}$ & $\eta_{\kappa,T}$ & $\eta_{\kappa,C}$ & $Y$ & $(Y,Y^2)$ & $(Y,Y^2,Y^3)$ & \shortstack{$(Y,Y^2)$\\vs. $Y$} & \shortstack{$(Y,Y^2,Y^3)$\\vs. $Y$} \\
\midrule
\multirow{20}{*}{\rotatebox[origin=c]{90}{$\mathrm{ATE}$}} & 1290 & 1.000 & -1.000 & 4.000 & 1.000 & 4.000 & 1.000 & 5.863 & 5.829 & 5.800 & 0.590\% & 1.099\% \\
 & 1293 & 1.000 & 0.000 & 4.000 & 1.000 & 4.000 & 1.000 & 5.888 & 5.853 & 5.824 & 0.591\% & 1.099\% \\
 & 1292 & 0.500 & 0.000 & 4.000 & 1.000 & 4.000 & 1.000 & 5.876 & 5.842 & 5.813 & 0.589\% & 1.098\% \\
 & 1289 & 0.500 & -1.000 & 4.000 & 1.000 & 4.000 & 1.000 & 5.856 & 5.822 & 5.792 & 0.588\% & 1.097\% \\
 & 1291 & 0.000 & 0.000 & 4.000 & 1.000 & 4.000 & 1.000 & 5.867 & 5.833 & 5.804 & 0.587\% & 1.097\% \\
 & 1296 & 1.000 & 1.000 & 4.000 & 1.000 & 4.000 & 1.000 & 5.922 & 5.887 & 5.858 & 0.591\% & 1.097\% \\
 & 1295 & 0.500 & 1.000 & 4.000 & 1.000 & 4.000 & 1.000 & 5.907 & 5.872 & 5.843 & 0.590\% & 1.097\% \\
 & 1294 & 0.000 & 1.000 & 4.000 & 1.000 & 4.000 & 1.000 & 5.894 & 5.859 & 5.830 & 0.588\% & 1.096\% \\
 & 1288 & 0.000 & -1.000 & 4.000 & 1.000 & 4.000 & 1.000 & 5.851 & 5.817 & 5.787 & 0.586\% & 1.096\% \\
 & 858 & 1.000 & -1.000 & 4.000 & 1.000 & 4.000 & 0.000 & 5.863 & 5.830 & 5.803 & 0.561\% & 1.038\% \\
 & 861 & 1.000 & 0.000 & 4.000 & 1.000 & 4.000 & 0.000 & 5.887 & 5.854 & 5.827 & 0.562\% & 1.038\% \\
 & 860 & 0.500 & 0.000 & 4.000 & 1.000 & 4.000 & 0.000 & 5.876 & 5.843 & 5.816 & 0.561\% & 1.037\% \\
 & 857 & 0.500 & -1.000 & 4.000 & 1.000 & 4.000 & 0.000 & 5.856 & 5.823 & 5.796 & 0.559\% & 1.037\% \\
 & 859 & 0.000 & 0.000 & 4.000 & 1.000 & 4.000 & 0.000 & 5.867 & 5.834 & 5.807 & 0.559\% & 1.036\% \\
 & 864 & 1.000 & 1.000 & 4.000 & 1.000 & 4.000 & 0.000 & 5.921 & 5.888 & 5.860 & 0.562\% & 1.036\% \\
 & 863 & 0.500 & 1.000 & 4.000 & 1.000 & 4.000 & 0.000 & 5.906 & 5.873 & 5.845 & 0.561\% & 1.036\% \\
 & 862 & 0.000 & 1.000 & 4.000 & 1.000 & 4.000 & 0.000 & 5.893 & 5.860 & 5.833 & 0.559\% & 1.036\% \\
 & 856 & 0.000 & -1.000 & 4.000 & 1.000 & 4.000 & 0.000 & 5.850 & 5.818 & 5.790 & 0.557\% & 1.035\% \\
 & 1182 & 1.000 & -1.000 & 4.000 & 1.000 & 3.000 & 1.000 & 5.865 & 5.834 & 5.807 & 0.541\% & 0.998\% \\
 & 1185 & 1.000 & 0.000 & 4.000 & 1.000 & 3.000 & 1.000 & 5.889 & 5.858 & 5.831 & 0.542\% & 0.997\% \\
\addlinespace[5pt]
\multirow{20}{*}{\rotatebox[origin=c]{90}{$\mathbb{E}[Y(1)]$}} & 1293 & 1.000 & 0.000 & 4.000 & 1.000 & 4.000 & 1.000 & 40.298 & 40.058 & 39.843 & 0.598\% & 1.140\% \\
 & 1296 & 1.000 & 1.000 & 4.000 & 1.000 & 4.000 & 1.000 & 40.428 & 40.187 & 39.973 & 0.599\% & 1.139\% \\
 & 1295 & 0.500 & 1.000 & 4.000 & 1.000 & 4.000 & 1.000 & 40.334 & 40.095 & 39.880 & 0.598\% & 1.139\% \\
 & 1290 & 1.000 & -1.000 & 4.000 & 1.000 & 4.000 & 1.000 & 40.201 & 39.963 & 39.748 & 0.595\% & 1.139\% \\
 & 1292 & 0.500 & 0.000 & 4.000 & 1.000 & 4.000 & 1.000 & 40.227 & 39.989 & 39.774 & 0.596\% & 1.139\% \\
 & 1294 & 0.000 & 1.000 & 4.000 & 1.000 & 4.000 & 1.000 & 40.257 & 40.019 & 39.804 & 0.596\% & 1.139\% \\
 & 1291 & 0.000 & 0.000 & 4.000 & 1.000 & 4.000 & 1.000 & 40.173 & 39.936 & 39.721 & 0.594\% & 1.138\% \\
 & 1289 & 0.500 & -1.000 & 4.000 & 1.000 & 4.000 & 1.000 & 40.154 & 39.917 & 39.702 & 0.593\% & 1.138\% \\
 & 1288 & 0.000 & -1.000 & 4.000 & 1.000 & 4.000 & 1.000 & 40.123 & 39.887 & 39.672 & 0.591\% & 1.136\% \\
 & 861 & 1.000 & 0.000 & 4.000 & 1.000 & 4.000 & 0.000 & 40.318 & 40.091 & 39.888 & 0.566\% & 1.080\% \\
 & 864 & 1.000 & 1.000 & 4.000 & 1.000 & 4.000 & 0.000 & 40.449 & 40.221 & 40.017 & 0.568\% & 1.079\% \\
 & 863 & 0.500 & 1.000 & 4.000 & 1.000 & 4.000 & 0.000 & 40.355 & 40.128 & 39.924 & 0.566\% & 1.079\% \\
 & 858 & 1.000 & -1.000 & 4.000 & 1.000 & 4.000 & 0.000 & 40.222 & 39.996 & 39.792 & 0.564\% & 1.079\% \\
 & 860 & 0.500 & 0.000 & 4.000 & 1.000 & 4.000 & 0.000 & 40.247 & 40.022 & 39.818 & 0.564\% & 1.079\% \\
 & 862 & 0.000 & 1.000 & 4.000 & 1.000 & 4.000 & 0.000 & 40.278 & 40.052 & 39.848 & 0.564\% & 1.079\% \\
 & 859 & 0.000 & 0.000 & 4.000 & 1.000 & 4.000 & 0.000 & 40.194 & 39.969 & 39.765 & 0.562\% & 1.078\% \\
 & 857 & 0.500 & -1.000 & 4.000 & 1.000 & 4.000 & 0.000 & 40.174 & 39.950 & 39.746 & 0.562\% & 1.078\% \\
 & 856 & 0.000 & -1.000 & 4.000 & 1.000 & 4.000 & 0.000 & 40.144 & 39.920 & 39.716 & 0.560\% & 1.076\% \\
 & 1185 & 1.000 & 0.000 & 4.000 & 1.000 & 3.000 & 1.000 & 40.331 & 40.114 & 39.918 & 0.542\% & 1.035\% \\
 & 1188 & 1.000 & 1.000 & 4.000 & 1.000 & 3.000 & 1.000 & 40.462 & 40.243 & 40.048 & 0.544\% & 1.035\% \\
\addlinespace[5pt]
\multirow{20}{*}{\rotatebox[origin=c]{90}{$\mathbb{E}[Y(0)]$}} & 1294 & 0.000 & 1.000 & 4.000 & 1.000 & 4.000 & 1.000 & 13.666 & 13.626 & 13.591 & 0.298\% & 0.558\% \\
 & 1291 & 0.000 & 0.000 & 4.000 & 1.000 & 4.000 & 1.000 & 13.590 & 13.550 & 13.514 & 0.296\% & 0.558\% \\
 & 1292 & 0.500 & 0.000 & 4.000 & 1.000 & 4.000 & 1.000 & 13.630 & 13.590 & 13.555 & 0.296\% & 0.557\% \\
 & 1295 & 0.500 & 1.000 & 4.000 & 1.000 & 4.000 & 1.000 & 13.722 & 13.681 & 13.646 & 0.298\% & 0.557\% \\
 & 1293 & 1.000 & 0.000 & 4.000 & 1.000 & 4.000 & 1.000 & 13.679 & 13.639 & 13.603 & 0.296\% & 0.557\% \\
 & 1290 & 1.000 & -1.000 & 4.000 & 1.000 & 4.000 & 1.000 & 13.607 & 13.567 & 13.532 & 0.294\% & 0.557\% \\
 & 1289 & 0.500 & -1.000 & 4.000 & 1.000 & 4.000 & 1.000 & 13.573 & 13.533 & 13.498 & 0.294\% & 0.557\% \\
 & 1296 & 1.000 & 1.000 & 4.000 & 1.000 & 4.000 & 1.000 & 13.786 & 13.745 & 13.709 & 0.297\% & 0.556\% \\
 & 1288 & 0.000 & -1.000 & 4.000 & 1.000 & 4.000 & 1.000 & 13.547 & 13.508 & 13.472 & 0.294\% & 0.556\% \\
 & 1186 & 0.000 & 1.000 & 4.000 & 1.000 & 3.000 & 1.000 & 13.668 & 13.630 & 13.593 & 0.279\% & 0.547\% \\
 & 1183 & 0.000 & 0.000 & 4.000 & 1.000 & 3.000 & 1.000 & 13.591 & 13.554 & 13.517 & 0.277\% & 0.547\% \\
 & 1184 & 0.500 & 0.000 & 4.000 & 1.000 & 3.000 & 1.000 & 13.631 & 13.593 & 13.557 & 0.277\% & 0.547\% \\
 & 1187 & 0.500 & 1.000 & 4.000 & 1.000 & 3.000 & 1.000 & 13.722 & 13.684 & 13.648 & 0.278\% & 0.547\% \\
 & 1185 & 1.000 & 0.000 & 4.000 & 1.000 & 3.000 & 1.000 & 13.680 & 13.642 & 13.605 & 0.277\% & 0.547\% \\
 & 1182 & 1.000 & -1.000 & 4.000 & 1.000 & 3.000 & 1.000 & 13.608 & 13.571 & 13.534 & 0.276\% & 0.546\% \\
 & 1181 & 0.500 & -1.000 & 4.000 & 1.000 & 3.000 & 1.000 & 13.574 & 13.537 & 13.500 & 0.275\% & 0.546\% \\
 & 1180 & 0.000 & -1.000 & 4.000 & 1.000 & 3.000 & 1.000 & 13.549 & 13.512 & 13.476 & 0.275\% & 0.546\% \\
 & 1188 & 1.000 & 1.000 & 4.000 & 1.000 & 3.000 & 1.000 & 13.786 & 13.747 & 13.711 & 0.278\% & 0.546\% \\
 & 1147 & 0.000 & 0.000 & 4.000 & 0.000 & 3.000 & 1.000 & 6.786 & 6.767 & 6.749 & 0.280\% & 0.545\% \\
 & 1145 & 0.500 & -1.000 & 4.000 & 0.000 & 3.000 & 1.000 & 6.768 & 6.749 & 6.731 & 0.279\% & 0.544\% \\
\bottomrule
\end{tabular}%
}
\vspace{-0.50em}
\begin{flushleft}
\tiny
\textit{Note:} Variance columns correspond to the displayed bases for the Verma projection. Gains are $100\{V_Y/V_{\mathrm{rich}}-1\}\%$, relative to the $Y$-only basis. Rows are ranked by the gain for the full $(Y,Y^2,Y^3)$ basis.
\end{flushleft}
\end{table}

\begin{figure}[!t]
\centering
\includegraphics[scale=0.55]
{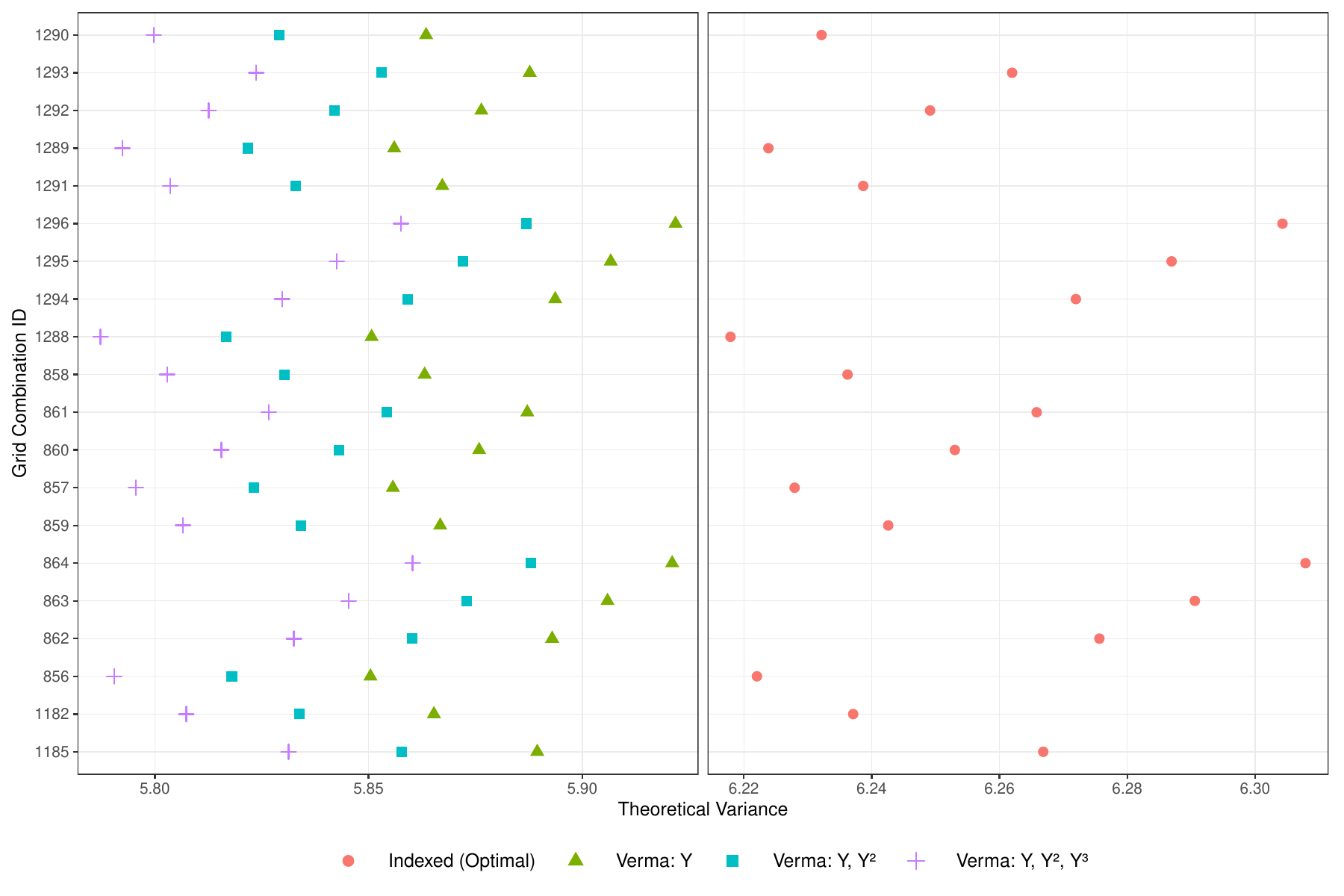}
\caption{
Effect of Verma-basis richness on the theoretical asymptotic variance for $\text{ATE} = \psi_1(\bbP) - \psi_0(\bbP)$ under the extended front-door model. Results are shown for the 20 highest-ranked DGPs in Experiment~3. The Verma projection estimators use the nested outcome bases $\{Y\}$, $\{Y,Y^2\}$, and $\{Y,Y^2,Y^3\}$; the optimally weighted indexed estimator is included as a reference. DGPs are ranked by the relative efficiency gain of the full Verma basis $\{Y,Y^2,Y^3\}$ over the $Y$-only basis. 
}
\label{fig:frontdoor_basis_grid_ate_top_20_method_variances_split_x}
\end{figure}

\begin{figure}[!t]
\centering
\includegraphics[scale=0.55]
{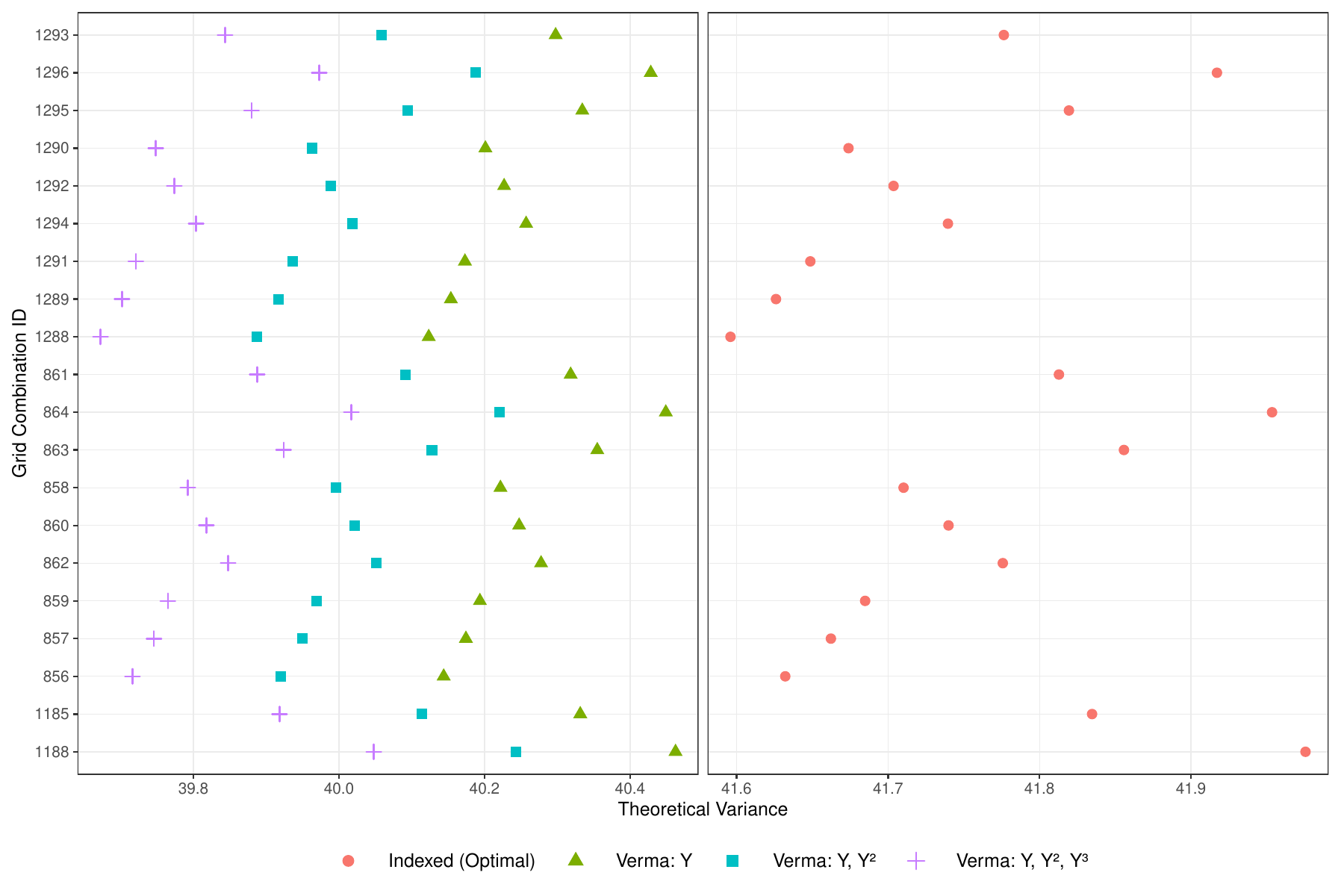}
\caption{
Effect of Verma-basis richness on the theoretical asymptotic variance for $\psi_1(\bbP)=\bbE[Y(1)]$ under the extended front-door model. Results are shown for the 20 highest-ranked DGPs in Experiment~3. The Verma projection estimators use the nested outcome bases $\{Y\}$, $\{Y,Y^2\}$, and $\{Y,Y^2,Y^3\}$; the optimally weighted indexed estimator is included as a reference. DGPs are ranked by the relative efficiency gain of the full Verma basis $\{Y,Y^2,Y^3\}$ over the $Y$-only basis. 
}
\label{fig:frontdoor_basis_grid_mean_1_top_20_method_variances_split_x}
\end{figure}


\clearpage
\subsection{Simulation studies for the Napkin model}
\label{app:sims:napkin}

This section provides the complete simulation specification for the Napkin model considered in Section~\ref{sec:napkin_graph}. The simulation design mirrors that of the extended front-door model, differing only in the underlying graphical model and the associated structural equations. Unless otherwise stated, oracle calculations, Monte Carlo implementation, sample sizes, nuisance estimation, and performance measures are identical to those described in Section~\ref{app:frontdoor:implementation}.

\subsubsection{Data-generating mechanism} 

The observed data consist of $O=(W,Z,A,Y)$, where $W$, $Z$, and $A$ are binary and $Y$ is continuous, and are generated according to the latent-variable Napkin graph shown in Figure~\ref{fig:napkin graph}(a), with latent variables $U_1$ and $U_2$. Specifically,
\begin{align}
U_1
&\sim
\operatorname{Uniform}(0,1),
\qquad
U_2
\sim
\operatorname{Normal}(0,1),
\label{eq:sim_latent}
\\
W\mid U_1,U_2
&\sim
\operatorname{Bernoulli}
\left[
\operatorname{expit}
\left\{
\beta_{W,0}
+
\beta_{W,1}U_1
+
\beta_{W,2}U_2
\right\}
\right],
\label{eq:sim_w}
\\
Z\mid W
&\sim
\operatorname{Bernoulli}
\left[
\operatorname{expit}
\left\{
\beta_{Z,0}
+
\beta_{Z,1}W
\right\}
\right],
\label{eq:sim_z}
\\
A\mid Z,U_2
&\sim
\operatorname{Bernoulli}
\left[
\operatorname{expit}
\left\{
\beta_{A,0}
+
\beta_{A,1}Z
+
\beta_{A,2}U_2
+
\eta_{\mathrm{int}}ZU_2
\right\}
\right].
\label{eq:sim_a}
\end{align}
The interaction parameter $\eta_{\mathrm{int}}$ controls the extent to which the association induced by the latent variable $U_2$ varies across levels of $Z$. Varying $\eta_{\mathrm{int}}$ changes the relative precision and covariance structure of the two fixed-$Z$ indexed influence functions. It can therefore move the variance-minimizing indexed combination away from equal weighting, without changing the equality of the indexed identifying functionals implied by the Napkin restriction.

The outcome is generated according to
\begin{align*}
Y = \mu(A,U_1) + \sigma(A,U_1) \, \varepsilon(A, U_1), 
\end{align*}
where the three components are specified below.

\vspace{0.25cm} \noindent 
\underline{\bf Conditional mean.} 
Let $\widetilde U_1 = U_1-\frac12$ and $\widetilde U_1^{\,2} = U_1^2-\frac13$, which satisfy $\bbE[\widetilde U_1] = \bbE[\widetilde U_1^{\,2}]=0$. Define $m(U_1) = \gamma_0 + \gamma_1\widetilde U_1 + \gamma_2\widetilde U_1^{\,2}$. The conditional mean is 
\begin{align}
\mu(A,U_1)
&=
\beta_0
+
\beta_1A
+
\beta_2U_1
+
\beta_3\widetilde U_1^{\,2}
\nonumber\\
&
\quad
+
\big\{
\eta_{\mu,T}
I(A=a_0)
+
\eta_{\mu,C}
I(A\neq a_0)
\big\} \, 
m(U_1). 
\label{eq:mean_model}
\end{align}

Centering $\widetilde U_1$ and $\widetilde U_1^{\,2}$ ensures that the perturbation parameters modify treatment-specific heterogeneity without inducing unintended shifts in the marginal causal means beyond those introduced through the intercept component of $m(U_1)$.

\vspace{0.25cm} \noindent 
\underline{\bf Conditional variance.} Define $v(U_1) = \rho_0 + \rho_1\widetilde U_1 + \rho_2\widetilde U_1^{\,2}.$ The conditional log-standard deviation is 
\begin{align}
\log\sigma(A,U_1)
&=
\lambda_0
+
\lambda_1\widetilde U_1
\nonumber\\
&
\quad
+
\big\{ 
\eta_{\sigma,T}
I(A=a_0)
+
\eta_{\sigma,C}
I(A\neq a_0)
\big\} \, v(U_1). 
\label{eq:variance_model}
\end{align}

\vspace{0.25cm} \noindent 
\underline{\bf Skewness loading.} Let $s(U_1) = \alpha_0 + \alpha_1 \widetilde U_1 + \alpha_2 \widetilde U_1^{\,2}$, and define 
\begin{align}
\lambda(A,U_1)
=
\big\{ 
\eta_{\kappa,T}
I(A=a_0)
+
\eta_{\kappa,C}
I(A\neq a_0)
\big\} \, 
s(U_1). 
\label{eq:shape_loading}
\end{align}

Let $\varepsilon_N\sim \operatorname{Normal}(0,1)$, and let $\varepsilon_\chi = \frac{\chi^2_3-3}{\sqrt6}$, where $\chi^2_3$ denotes a chi-square random variable with three degrees of freedom, centered and rescaled to have mean zero and variance one. The standardized error is constructed as
\begin{align}
\varepsilon(A, U_1)
=
\frac{
\varepsilon_N
+
\lambda(A,U_1)\varepsilon_\chi
}
{
\sqrt{1+\lambda(A,U_1)^2}
},
\label{eq:error}
\end{align}
which has conditional mean zero and conditional variance one while allowing the conditional third moment to vary with $\lambda(A,U_1)$. 

The tuning parameters play the same roles as in the extended front-door simulations described in Section~\ref{app:sim:frontdoor:dgp}. The baseline DGP is defined by the fixed coefficients in Table~\ref{tab:napkin_baseline_parameters}, while the tuning parameter grids used in the Experiments~2 and ~3 are summarized in
Table~\ref{tab:napkin_tuning_parameters}. Experiment~1 fixes the tuning parameters at the baseline
configuration,
\begin{align*}
\eta_{\mathrm{int}} = 0,\quad
\eta_{\mu,T} = 0.5,\quad
\eta_{\mu,C} = 0.5,\quad
\eta_{\sigma,T} = 1.5,\quad
\eta_{\sigma,C} = 0,\quad
\eta_{\kappa,T} = 2,\quad
\eta_{\kappa,C} = 2,
\end{align*}
to evaluate finite-sample convergence to the theoretical asymptotic variances. Experiment~2 varies the tuning parameters over the grid in Table~\ref{tab:napkin_tuning_parameters}, yielding $\mathbf{2,187}$ \textbf{candidate DGPs} used to compare the theoretical asymptotic variances of the five estimators. Experiment~3 fixes the treatment mechanism and varies the outcome tuning parameters over an expanded grid, yielding $\mathbf{1,296}$ \textbf{candidate DGPs}. 

\begin{table}[!t]
\centering
\caption{Fixed coefficients defining the baseline data-generating mechanism for the Napkin simulations. These coefficients remained unchanged throughout all simulation studies.}
\label{tab:napkin_baseline_parameters}
\begin{tabular}{lll}
\toprule
Component & Parameter & Value\\
\midrule
$W \mid U_1, U_2$ &
$(\beta_{W,0},\beta_{W,1},\beta_{W,2})$
&
$(0.5,1.0,0.5)$
\\
$Z \mid W$ &
$(\beta_{Z_0},\beta_{Z_1})$
&
$(-0.7,1.0)$
\\
$A \mid Z, U_2$ &
$(\beta_{A_0},\beta_{A_1},\beta_{A_2})$
&
$(0.5,1.2,0.4)$
\\
$Y \mid A, U_1$: $\mu(A, U_1)$ &
$(\beta_0,\beta_1,\beta_2,\beta_3)$
&
$(3.0,1.0,0.7,0.25)$
\\
&
$(\gamma_0,\gamma_1,\gamma_2)$
&
$(0,0.4,0.4)$
\\
$Y \mid A, U_1$: $\sigma(A, U_1)$ &
$(\lambda_0,\lambda_1)$
&
$(-0.1,0.15)$
\\
 &
$(\rho_0,\rho_1,\rho_2)$
&
$(0,0.9,0.7)$
\\
$Y \mid A, U_1$: $\lambda(A, U_1)$ & 
$(\alpha_0, \alpha_1, \alpha_2)$
& 
$(0, 0.75, 0.9)$
\\
\bottomrule
\end{tabular}
\end{table}

\begin{table}[t]
\centering
\caption{Simulation tuning parameters varied across the grid-search experiments. The Cartesian product of the parameter values defines the candidate data-generating mechanisms considered in Experiments~2 and~3.}
\label{tab:napkin_tuning_parameters}
\begin{tabular}{llll}
\toprule
Distributional feature & Parameter & Experiment~2 & Experiment~3\\
\midrule
$A \mid Z, U_2$: \ interaction
&
$\eta_{\mathrm{int}}$
&
$-0.5,\;0.0,\;0.5$
&
--- \\
$Y \mid A, U_1$: target mean
&
$\eta_{\mu, T}$
&
$0.0,\;0.5,\;1.0$
&
$0.0,\;0.5,\;1.0$
\\
\textcolor{white}{$Y \mid A, U_1$:} contrast mean
&
$\eta_{\mu, C}$
&
$0.0,\;0.5,\;1.0$
&
$-1.0,\;0.0,\;1.0$
\\
\textcolor{white}{$Y \mid A, U_1$:} target variance
&
$\eta_{\sigma, T}$
&
$0.0,\;1.0,\;1.5$
&
$1.0,\;2.0,\;3.0,\;4.0$
\\
\textcolor{white}{$Y \mid A, U_1$:} contrast variance
&
$\eta_{\sigma, C}$
&
$0.0,\;1.0,\;1.5$
&
$-1.0,\;0.0,\;1.0$
\\
\textcolor{white}{$Y \mid A, U_1$:} target skewness
&
$\eta_{\kappa, T}$
&
$0.0,\;1.0,\;2.0$
&
$1.0,\;2.0,\;3.0,\;4.0$
\\
\textcolor{white}{$Y \mid A, U_1$:} contrast skewness
&
$\eta_{\kappa, C}$
&
$0.0,\;1.0,\;2.0$
&
$-1.0,\;0.0,\;1.0$
\\
\bottomrule
\end{tabular}
\end{table}

\subsubsection{Oracle calculations and simulation implementation} 

As described above, oracle calculations, Monte Carlo implementation, sample sizes, nuisance estimation, and performance measures are identical to those described in Section~\ref{app:frontdoor:implementation}.

Under the Napkin data-generating mechanism described above, intervention on the treatment variable gives the treatment-specific counterfactual mean
\begin{align}
\psi_{a_0}(\bbP)
&=
\beta_{0}
+
\beta_{1}a_0
+
\frac{1}{2}\beta_{2}
+ 
\eta_{\mu, T} \gamma_0
\nonumber \\
\psi_{1-a_0}(\bbP)
&=
\beta_{0}
+
\beta_{1}(1-a_0)
+
\frac{1}{2}\beta_{2}
+ 
\eta_{\mu, C} \gamma_0
,
\label{eq:sim_truth}
\end{align}
and consequently the average treatment effect equals 
\begin{align}
  \tau(\bbP) = \psi_1(\bbP) - \psi_0(\bbP) = \beta_1 + (2a_0 - 1)(\eta_{\mu, T} - \eta_{\mu, C}) \gamma_0. 
  \label{eq:sim_ate_truth}
\end{align}
For the parameter values used in our simulations, $\gamma_0=0$, so the mean-perturbation parameters alter treatment-specific outcome heterogeneity without changing the marginal causal means, and $\tau(\bbP)=\beta_1$.

Both estimands admit indexed identifying representations under the assumptions implied by the Napkin graph, with corresponding indexed influence-function classes developed in Appendix~\ref{app:id_npEIF_napkin}. Moreover, Section~\ref{sec:geometry} establishes that the indexed influence-function class is contained within the affine class generated by the Verma orthocomplement. Consequently, the Verma-basis estimator cannot have larger asymptotic variance than the corresponding indexed estimator. The simulation studies below investigate the practical magnitude of this improvement for both treatment-specific counterfactual means and the average treatment effect.

For each simulated data set we computed five estimators: the two fixed indexed estimators, the equal-weight indexed estimator, the optimally weighted indexed estimator, and the Verma-basis projection estimator. We briefly summarize their implementation below for pedological purposes.

\vspace{0.25cm} \noindent 
\underline{\bf Indexed identification functionals.} 
For $z^\ast\in\{0,1\}$, let $\psi_{a_0}(\bbP;z^\ast)$ denote the indexed identifying functional. Under the target-specific invariance restriction,
\begin{align*}
\psi_{a_0}(\bbP;0)
=
\psi_{a_0}(\bbP;1)
=
\psi_{a_0}(\bbP).
\end{align*}
 
Writing
\begin{align*}
\mu(a,z,w)
&=
\bbE(Y\mid A=a,Z=z,W=w),
\\
\pi(a\mid z,w)
&=
p(A=a\mid Z=z,W=w),
\\
\lambda(z\mid w)
&=
p(Z=z\mid W=w),
\\
\kappa_{a_0}(z^\ast)
&=
\bbE[\pi(a_0\mid z^\ast,W)],
\end{align*}
the indexed functional used in the simulations was
\begin{align}
\psi_{a_0}(\bbP; z^\ast)
=
\frac{
\bbE
\left[
\mu(a_0,z^\ast,W)
\pi(a_0\mid z^\ast,W)
\right]
}{
\kappa_{a_0}(z^\ast)
}.
\label{eq:sim_indexed_functional}
\end{align}

Let $\varphi_{a_0}^{\mathrm{np}}(\bbP;z^\ast)(O)$ denote the corresponding nonparametric influence function: 
\begin{align}
\varphi_{a_0}^{\mathrm{np}}(\bbP;z^\ast)(O)
&=
\frac{
\I(A=a_0)\I(Z=z^\ast)
}{
\kappa_{a_0}(z^\ast)
\lambda(z^\ast\mid W)
}
\left\{
Y-\mu(a_0,z^\ast,W)
\right\}
\nonumber\\
&\quad+
\frac{
\I(Z=z^\ast)
}{
\kappa_{a_0}(z^\ast)
\lambda(z^\ast\mid W)
}
\left\{
\mu(a_0,z^\ast,W)-\psi_{a_0}(z^\ast)
\right\}
\left\{
\I(A=a_0)-\pi(a_0\mid z^\ast,W)
\right\}
\nonumber\\
&\quad+
\frac{
\pi(a_0\mid z^\ast,W)
}{
\kappa_{a_0}(z^\ast)
}
\left\{
\mu(a_0,z^\ast,W)-\psi_{a_0}(z^\ast)
\right\}.
\label{eq:sim_indexed_if}
\end{align}

For each value of $z^\ast$, the corresponding fixed indexed estimator was constructed as the one-step estimator 
\begin{align}
\widehat\psi_{a_0}(z^\ast)
=
\psi_{a_0}(\widehat\eta; z^\ast)
+
\bbP_n
\varphi_{a_0}^{\mathrm{np}}
(\widehat\eta;z^\ast),
\label{eq:sim_fixed_estimator}
\end{align}
where $\eta$ denotes the nuisance functions.

\vspace{0.25cm} \noindent 
\underline{\bf Combinations of indexed influence functions.} The equal-weight and optimally weighted indexed estimators are obtained by taking affine combinations of the two fixed indexed estimators. For $\alpha\in\mathbb R$, define 
\begin{align}
\widehat\psi_{a_0,\alpha}
=
(1-\alpha)\widehat\psi_{a_0}(0)
+
\alpha\widehat\psi_{a_0}(1).
\label{eq:sim_affine_estimator}
\end{align}
The equal-weight estimator sets $\alpha=1/2$. For the average treatment effect, the indexed class is constructed analogously by combining the indexed representations of the treatment effect,
\begin{align*}
\widehat\psi_{\mathrm{ATE},\alpha}
=
(1-\alpha)
\left\{
\widehat\psi_{1}(0)
-
\widehat\psi_{0}(0)
\right\}
+
\alpha
\left\{
\widehat\psi_{1}(1)
-
\widehat\psi_{0}(1)
\right\},
\end{align*}
whose corresponding influence function is
\begin{align*}
\widehat\varphi_{\mathrm{ATE},\alpha}
=
(1-\alpha)
\left\{
\widehat\varphi_{1}^{\mathrm{np}}(0)
-
\widehat\varphi_{0}^{\mathrm{np}}(0)
\right\}
+
\alpha
\left\{
\widehat\varphi_{1}^{\mathrm{np}}(1)
-
\widehat\varphi_{0}^{\mathrm{np}}(1)
\right\}.
\end{align*}

The optimal indexed estimator chooses the coefficient that minimizes the empirical variance of the corresponding estimated influence function. Since this objective is quadratic in $\alpha$, the minimizer has the closed-form expression
\begin{align}
\widehat\alpha_{\mathrm{opt}}
=
\frac{
\bbP_n
\left[
\widehat\varphi_{a_0}^{\mathrm{np}}(0)
\left\{
\widehat\varphi_{a_0}^{\mathrm{np}}(0)
-
\widehat\varphi_{a_0}^{\mathrm{np}}(1)
\right\}
\right]
}{
\bbP_n
\left[
\left\{
\widehat\varphi_{a_0}^{\mathrm{np}}(0)
-
\widehat\varphi_{a_0}^{\mathrm{np}}(1)
\right\}^2
\right]
},
\label{eq:sim_alpha_opt}
\end{align}
where $\widehat\varphi_{a_0}^{\mathrm{np}}(z) = \varphi_{a_0}^{\mathrm{np}}(\widehat\eta;z)$ denotes the estimated influence function corresponding to the indexed estimand under consideration (either a treatment-specific mean or the average treatment effect). The resulting estimator exploits the invariance of the indexed identifying functionals while remaining within the indexed affine class. It therefore represents the variance-minimizing estimator obtainable using only indexed influence functions.  

Because the influence functions have mean zero at the population level, centering does not affect the population-optimal coefficient. In finite samples, however, the estimated influence functions are empirically centered before optimization so that $\widehat\alpha_{\mathrm{opt}}$ minimizes the empirical variance rather than the empirical second moment.

\vspace{0.25cm} \noindent 
\underline{\bf Finite-dimensional approximation to the Verma orthocomplement.} 
To approximate the additional directions contained in the Verma orthocomplement, we constructed the finite-dimensional basis
\begin{align}
f_a(Z,A)
&=
\I(A=a)Z,
\qquad
a\in\{0,1\},
\label{eq:sim_f_basis}
\\
g_r(Y,A)
&=
Y^r,
\qquad
r=1,\ldots,K.
\label{eq:sim_g_basis}
\end{align}
For each pair $(a,r)$, let $\chi_{a,r}(O)$ denote the corresponding orthocomplement direction obtained from the characterization in Corollary~\ref{cor:napkin_orthocomplement}. Throughout the simulations, the fixing distribution was taken to be $\widetilde p(Z=1)=\widetilde p(Z=0)=\frac{1}{2}$. The finite-dimensional basis was therefore
\begin{align}
\boldsymbol\chi_K(O)
=
\left\{
\chi_{0,1}(O),
\chi_{1,1}(O),
\ldots,
\chi_{0,K}(O),
\chi_{1,K}(O)
\right\}^{\top}.
\label{eq:sim_chi_basis}
\end{align}

To guarantee that the projection space contains the indexed affine direction, we additionally include
\begin{align}
D(O)
=
\varphi_{a_0}^{\mathrm{np}}(\bbP;0)(O)
-
\varphi_{a_0}^{\mathrm{np}}(\bbP;1)(O). 
\label{eq:sim_indexed_difference}
\end{align}
Thus, the projection is taken onto
$\operatorname{span}
\left\{
D,
\chi_{0,1},
\chi_{1,1},
\ldots,
\chi_{0,K},
\chi_{1,K}
\right\},
$ which is a finite-dimensional subspace of the full Verma orthocomplement. Define $\boldsymbol H_K(O) = \left\{D(O) \, , \, \boldsymbol\chi_K(O)^\top \right\}^{\top}$. The empirical estimator of the projection coefficient is
\begin{align}
\widehat{\boldsymbol\gamma}_K
=
\argmin_{\boldsymbol\gamma}
\bbP_n
\left[
\left\{
\widehat\varphi_{a_0,c}^{\mathrm{np}}(\bbP; 0)
-
\boldsymbol\gamma^\top
\widehat{\boldsymbol H}_{K, c}
\right\}^2
\right],
\label{eq:sim_empirical_projection}
\end{align}
where the subscript $c$ denotes empirical centering. 

Equation~\eqref{eq:sim_empirical_projection}  is the empirical least-squares projection of the indexed influence function onto the chosen Verma basis. Empirical centering again ensures that the projection minimizes variance rather than second moment.

The resulting influence-function estimate is
\begin{align}
\widehat\varphi_{a_0,K}^{\mathrm{eff, basis}}(O)
=
\widehat\varphi_{a_0}^{\mathrm{np}}(\bbP;0)(O)
-
\widehat{\boldsymbol\gamma}_K^\top
\widehat{\boldsymbol H}_K(O).
\label{eq:sim_projected_if}
\end{align}
The associated point estimator retains the empirical means of the uncentered directions:
\begin{align}
\widehat\psi_{a_0,K}^{\mathrm{eff, basis}}
&=
\left(
1-\widehat\gamma_{K,1}
\right)
\widehat\psi_{a_0}(\bbP; 0)
+
\widehat\gamma_{K,1}
\widehat\psi_{a_0}(\bbP; 1)
-
\sum_{j=2}^{2K+1}
\widehat\gamma_{K,j}
\bbP_n
\widehat H_{K,j}.
\label{eq:sim_projected_estimator}
\end{align}
To improve numerical stability, the empirical projection was computed using a rank-revealing QR decomposition. When the projection design was rank deficient, the coefficients were estimated using a maximal linearly independent subset of the design columns, with coefficients for dropped columns set to zero.  

We considered three nuisance estimation strategies. The \emph{saturated} approach estimates each observed-data nuisance function using a saturated model for the binary covariates. We fit a logistic regression for $p(Z=1\mid W)$, a logistic regression containing the $Z$-by-$W$ interaction for $p(A=1\mid Z,W)$, and fully interacted cell-mean regressions for $\bbE[Y^r\mid A,Z,W]$. Since $A$, $Z$, and $W$ are binary, these models permit a distinct fitted conditional mean in every observed covariate cell. The \emph{working} approach instead fits simpler additive models that omit interaction terms, thereby introducing deliberate model misspecification. Finally, the \emph{oracle} approach evaluates the true observed-data nuisance functions implied by the latent-variable DGP using deterministic Gaussian quadrature. Oracle nuisance functions are used solely to compute theoretical efficiency bounds and are never supplied to the finite-sample estimators.

For each estimator, the estimated standard error was
\begin{align}
\widehat{\operatorname{se}}
\left(
\widehat\psi
\right)
=
\left[
\frac{1}{n}
\bbP_n
\left\{
\widehat\varphi(O)
-
\bbP_n\widehat\varphi(O)
\right\}^2
\right]^{1/2},
\label{eq:sim_se}
\end{align}
where $\widehat\varphi$ is the estimator-specific estimated influence function. The corresponding $95\%$ Wald confidence interval was
$
\widehat\psi
\pm
1.96\,
\widehat{\operatorname{se}}
\left(
\widehat\psi
\right).
$

\subsubsection{Experiment~1: Finite-sample performance}

Experiment~1 follows the same finite-sample design as the extended front-door simulations. Oracle quantities were computed from a synthetic population of size $250{,}000$, and finite-sample performance was evaluated using $1{,}250$ independent Monte Carlo replications over an increasing sequence of sample sizes ranging from 250 to 16,000 observations.

Table~\ref{tab:napkin_var_all_estimands} reports the theoretical asymptotic variances of the five estimators together with their relative efficiency gains, using the equally weighted indexed estimator as the reference. Figures~\ref{fig:napkin_var_ate_n_scaled_variance}--\ref{fig:napkin_var_mean_0_n_scaled_variance} compare the corresponding empirical $n$-scaled variances with their theoretical asymptotic counterparts. 

Across all three estimands, the empirical $n$-scaled variances converge toward their theoretical asymptotic limits. The agreement between empirical and theoretical variances provides finite-sample support for the projection-based variance calculations developed in Sections~\ref{sec:geometry} and~\ref{sec:napkin_graph}. Furthermore, across all three estimands, the theoretical variances exhibit the expected ordering. Averaging the two fixed-$Z$ indexed estimators substantially improves efficiency relative to either indexed estimator, and optimizing the affine combination yields a further reduction. The additional improvement obtained from the Verma projection, however, depends strongly on the estimand. For the average treatment effect, the Verma projection estimator substantially improves upon the optimally weighted indexed estimator, reducing the theoretical asymptotic variance from $14.578 $ to $13.413$, corresponding to an additional relative efficiency gain of roughly $9\%$. In contrast, for the treatment-specific means,
the optimally weighted indexed and Verma projection estimators are nearly indistinguishable, indicating that optimization within the indexed class captures nearly all of the efficiency available from the Verma restriction under this representative DGP.

\subsubsection{Experiment 2: Efficiency across data-generating mechanisms}

Experiment~2 uses the tuning grid in Table~\ref{tab:napkin_tuning_parameters} together with the ranking
criterion introduced in Section~\ref{app:sims:frontdoor}. The twenty highest-ranked DGPs are reported in Table~\ref{tab:napkin_var_grid_all_estimands}, and the corresponding theoretical asymptotic variances of the five estimators are displayed in Figures~\ref{fig:napkin_var_grid_ate_top_20_method_variances_split_x}--\ref{fig:napkin_var_grid_mean_0_top_20_method_variances_split_x}.

Across the highest-ranked configurations, the competing estimators exhibit a remarkably stable efficiency ordering. Averaging the two fixed-$Z$ indexed estimators substantially reduces the theoretical variance relative to either indexed estimator, and optimizing the affine combination provides a further improvement. The Verma projection
estimator consistently attains the smallest theoretical asymptotic variance, demonstrating that the additional directions in the Verma orthocomplement provide information beyond that available from optimization within the indexed class.

The tuning parameters associated with the highest-ranked DGPs differ across the three estimands, although configurations exhibiting stronger treatment-specific perturbations of the conditional variance and higher
moments of the outcome appear prominently among the top-ranked settings. For the average treatment effect, the largest gains typically occur when the target-arm variance perturbation is strongest ($\eta_{\sigma,T}=1.5$), whereas the mean perturbation parameters vary more substantially across the highest-ranked configurations.
Accordingly, these results should be viewed as descriptive of the simulation grid rather than as a general characterization of when the Verma projection yields the greatest efficiency gains.

For the highest-ranked DGP under the ATE, the theoretical asymptotic variance decreases from $7.528$ for the optimally weighted indexed estimator to $6.807$ for the Verma projection estimator, corresponding to an additional relative efficiency gain of approximately $11\%$. Across the remaining top-ranked configurations, similar improvements are observed, indicating that the efficiency gains seen in Experiment~1 are not isolated to a single DGP but persist across a broad collection of observational regimes. 

\subsubsection{Experiment~3: Effect of basis richness}

Experiment~3 follows the same basis-enrichment strategy described for the extended front-door model. The Verma projection estimator is constructed using the nested outcome bases
$\{Y\}$, $\{Y,Y^2\}$, and $\{Y,Y^2,Y^3\}$.

Table~\ref{tab:napkin_basis_grid_all_estimands} reports the twenty highest-ranked DGPs for each estimand, together with the theoretical asymptotic variances under the three basis specifications and the corresponding relative efficiency gains over the $\{Y\}$ basis.
Figures~\ref{fig:napkin_basis_grid_ate_top_20_method_variances_split_x}--\ref{fig:napkin_basis_grid_mean_0_top_20_method_variances_split_x} display the corresponding theoretical asymptotic variances of the optimally weighted indexed estimator and the three Verma projection estimators.

Across all three estimands, enlarging the outcome basis produces the expected nonincreasing theoretical asymptotic variance. The incremental gains from basis enrichment are modest but consistent across the highest-ranked DGPs, reflecting the increasingly rich approximation to the Verma orthocomplement. In each case, the Verma projection estimator constructed using $\{Y,Y^2,Y^3\}$ attains the smallest theoretical variance among the three basis specifications. 

The magnitude of the improvement depends on the estimand. The largest relative efficiency gains are observed for $\bbE[Y(0)]$, followed by $\bbE[Y(1)]$ and the average treatment effect. Across the top-ranked configurations, both the quadratic and cubic outcome terms contribute to the additional efficiency, although the overall gains remain below approximately $1.7\%$ relative to the $\{Y\}$ basis. Thus, for the Napkin model, the one-dimensional basis already captures most of the efficiency available from the Verma restriction, while richer basis expansions recover a modest but consistent additional improvement.

\clearpage

\begin{table}[htbp]
\centering
\scriptsize
\setlength{\tabcolsep}{2.1pt}
\renewcommand{\arraystretch}{0.92}
\caption{Theoretical variances and relative efficiency gains for the five estimators across the three parameters under the Napkin model. Relative efficiency gains are computed with respect to the reference method indicated by each column. }
\label{tab:napkin_var_all_estimands}
\resizebox{\textwidth}{!}{%
\begin{tabular}{llccccc}
\toprule
 & Method & \shortstack{Theoretical\\variance} & \shortstack{Gain vs.\\fixed $Z=0$} & \shortstack{Gain vs.\\fixed $Z=1$} & \shortstack{Gain vs.\\equal indexed} & \shortstack{Gain vs.\\optimal indexed} \\
\midrule
$\mathrm{ATE}$ & Fixed $Z=0$ & 22.810 & 0.000\% & 77.097\% & -30.726\% & -36.088\% \\
 & Fixed $Z=1$ & 40.395 & -43.534\% & 0.000\% & -60.883\% & -63.912\% \\
 & Indexed (equal) & 15.801 & 44.354\% & 155.646\% & 0.000\% & -7.741\% \\
 & Indexed (optimal) & 14.578 & 56.466\% & 177.097\% & 8.391\% & 0.000\% \\
 & Verma projection & 13.413 & 70.059\% & 201.171\% & 17.808\% & 8.688\% \\
\addlinespace[6pt]
$\mathbb{E}[Y(1)]$ & Fixed $Z=0$ & 8.827 & 0.000\% & -32.031\% & -57.986\% & -59.514\% \\
 & Fixed $Z=1$ & 6.000 & 47.126\% & 0.000\% & -38.187\% & -40.434\% \\
 & Indexed (equal) & 3.709 & 138.016\% & 61.777\% & 0.000\% & -3.636\% \\
 & Indexed (optimal) & 3.574 & 146.998\% & 67.882\% & 3.774\% & 0.000\% \\
 & Verma projection & 3.570 & 147.277\% & 68.072\% & 3.891\% & 0.113\% \\
\addlinespace[6pt]
$\mathbb{E}[Y(0)]$ & Fixed $Z=0$ & 13.978 & 0.000\% & 145.963\% & -13.494\% & -28.892\% \\
 & Fixed $Z=1$ & 34.380 & -59.343\% & 0.000\% & -64.830\% & -71.090\% \\
 & Indexed (equal) & 12.092 & 15.599\% & 184.331\% & 0.000\% & -17.800\% \\
 & Indexed (optimal) & 9.939 & 40.632\% & 245.902\% & 21.655\% & 0.000\% \\
 & Verma projection & 9.843 & 42.006\% & 249.281\% & 22.843\% & 0.977\% \\
\bottomrule
\end{tabular}%
}
\begin{flushleft}
\footnotesize
\textit{Note:} Each gain column reports $100\{V_{\mathrm{reference}}/V_{\mathrm{method}}-1\}\%$, with the reference estimator named in the column header. Positive values indicate smaller theoretical variance than the reference.
\end{flushleft}
\end{table}

\begin{figure}[t]
\centering
\includegraphics[scale=0.7]{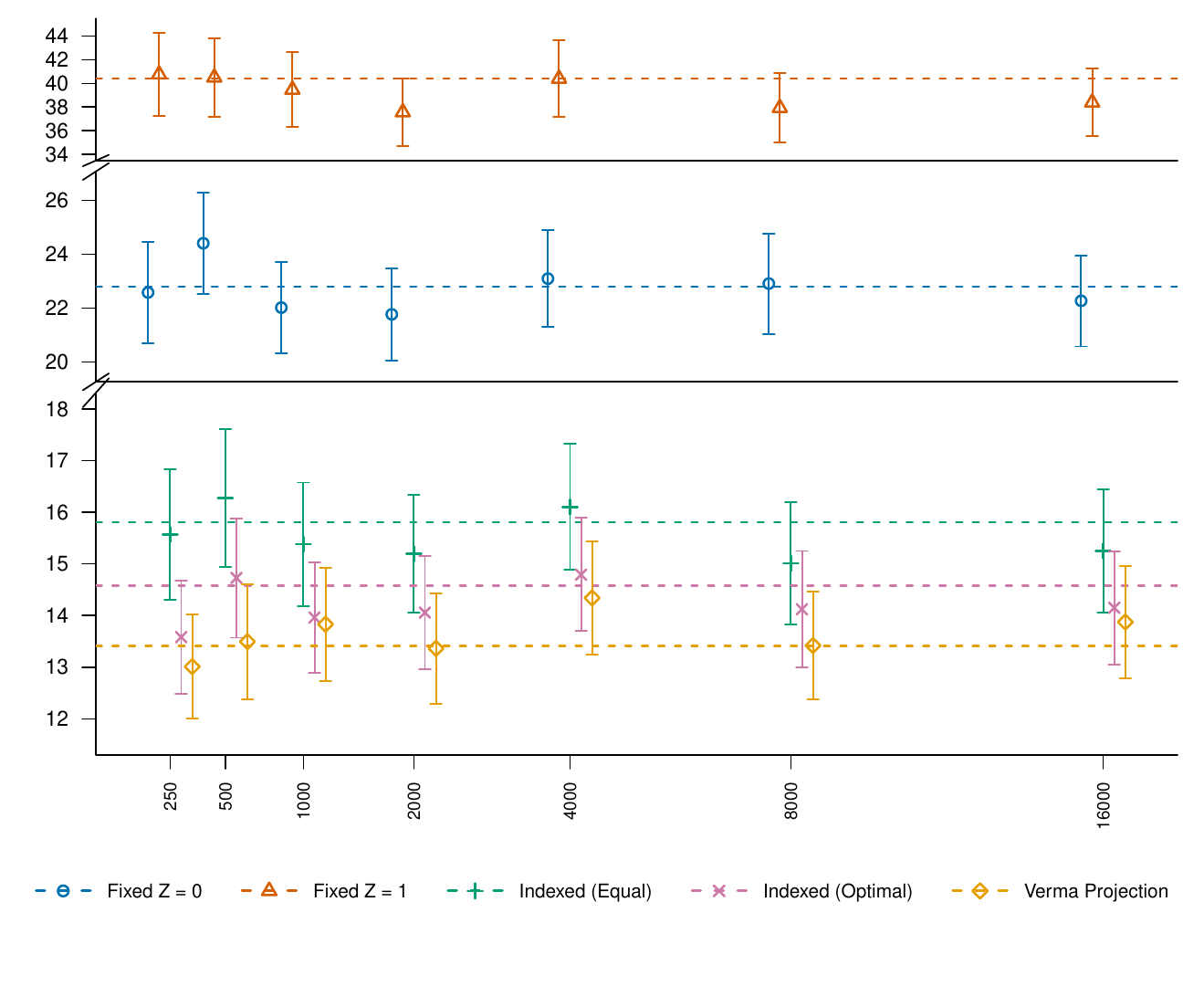}
\vspace{-1cm}
\caption{Finite-sample performance for $\text{ATE} = \psi_1(\bbP) - \psi_0(\bbP)$ under the Napkin model. Points and vertical intervals show the empirical $n$-scaled variance and its Monte Carlo confidence interval, respectively; horizontal dashed lines show the corresponding theoretical asymptotic variances. Results are shown for the two fixed-$Z$ indexed estimators, the equally weighted indexed estimator, the optimally weighted indexed estimator, and the Verma projection estimator.}
\label{fig:napkin_var_ate_n_scaled_variance}
\end{figure}

\begin{figure}[t]
\centering
\includegraphics[scale=0.7]{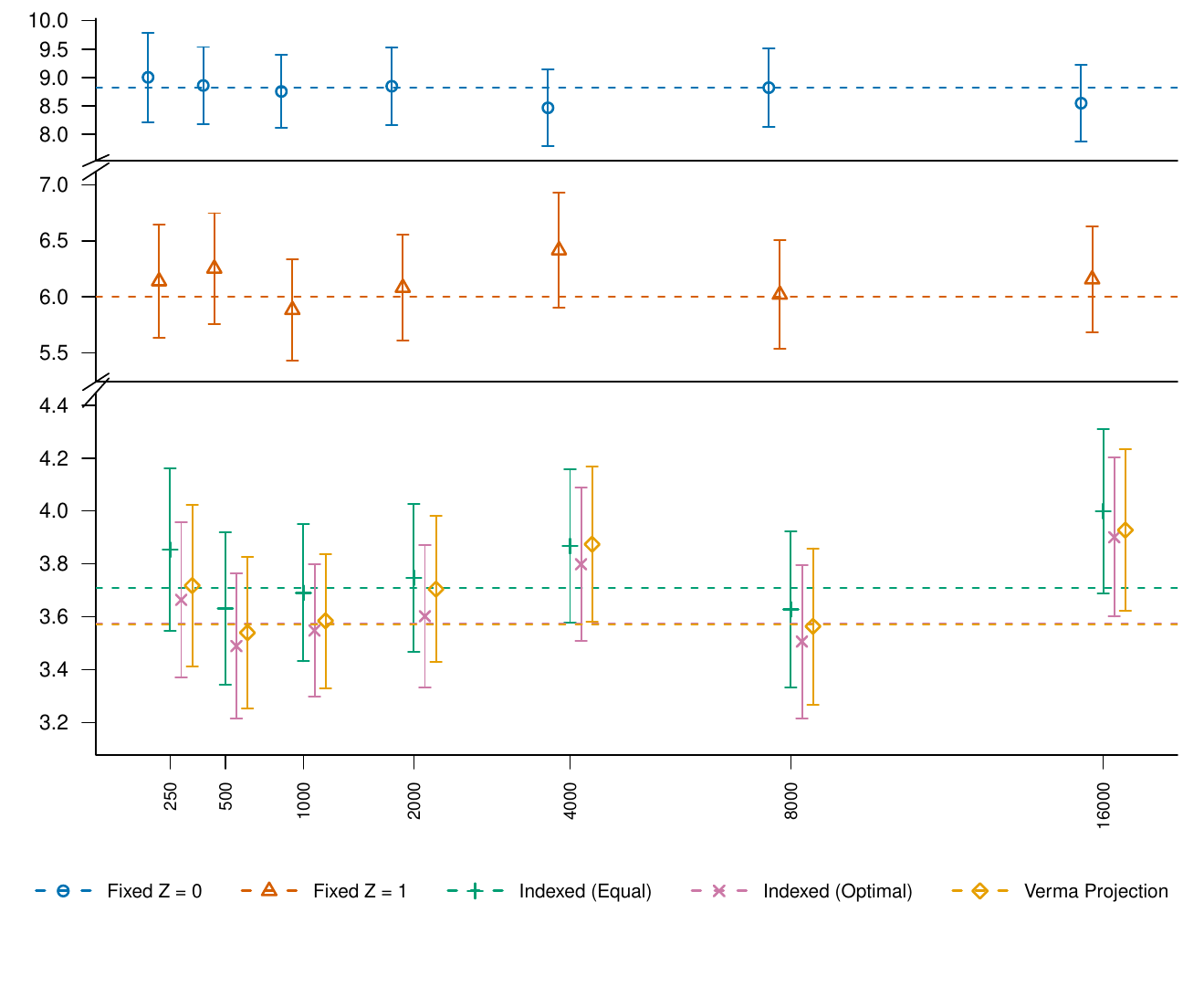}
\vspace{-1cm}
\caption{
Finite-sample performance for $\psi_1(\bbP)=\bbE[Y(1)]$ under the Napkin model. Points and vertical intervals show the empirical $n$-scaled variance and its Monte Carlo confidence interval, respectively; horizontal dashed lines show the corresponding theoretical asymptotic variances. Results are shown for the two fixed-$Z$ indexed estimators, the equally weighted indexed estimator, the optimally weighted indexed estimator, and the Verma projection estimator.  
}
\label{fig:napkin_var_mean_1_n_scaled_variance}
\end{figure}

\begin{figure}[t]
\centering
\includegraphics[scale=0.7]{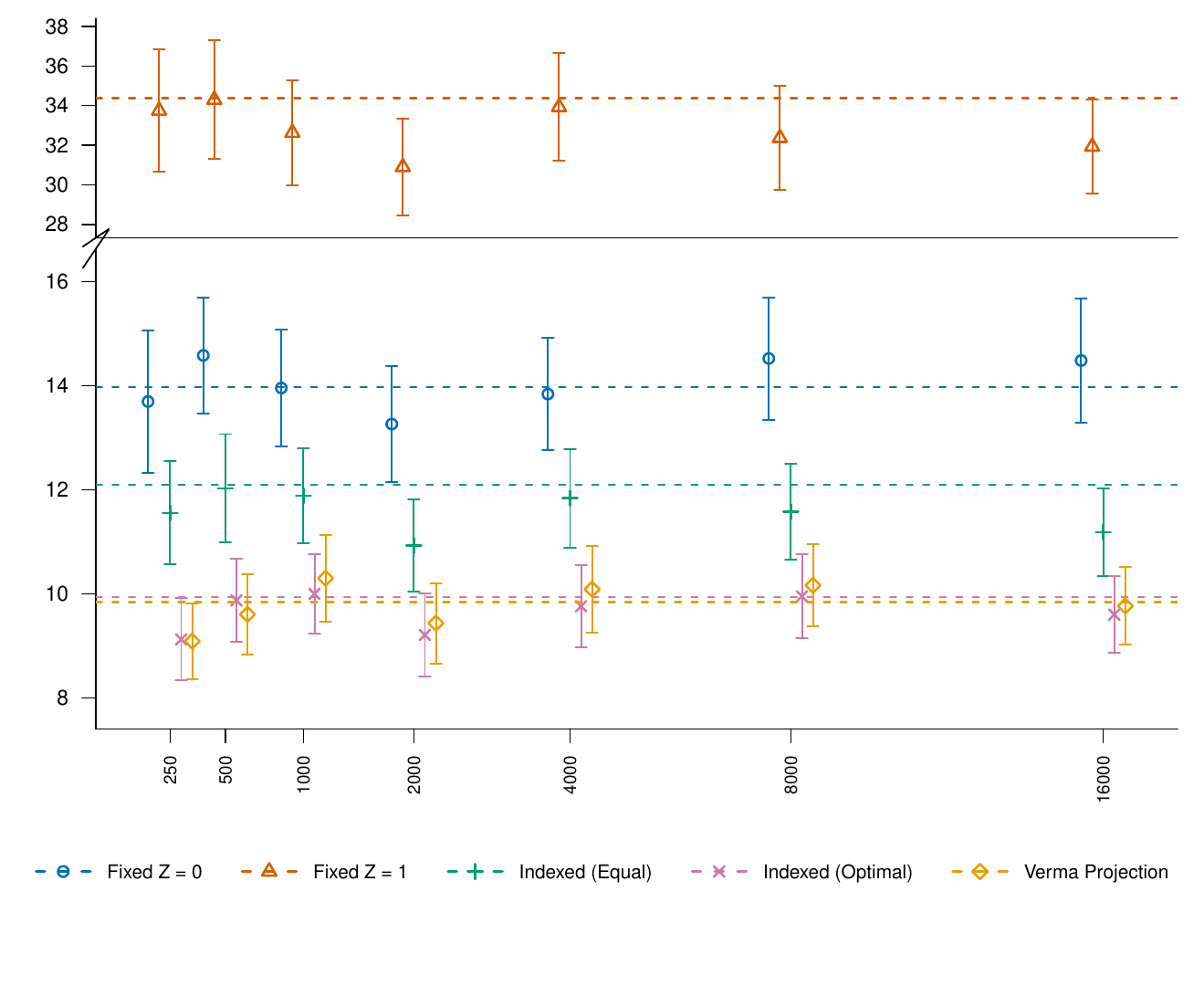}
\vspace{-1cm}
\caption{Finite-sample performance for $\psi_0(\bbP)=\bbE[Y(0)]$ under the Napkin model. Points and vertical intervals show the empirical $n$-scaled variance and its Monte Carlo confidence interval, respectively; horizontal dashed lines show the corresponding theoretical asymptotic variances. Results are shown for the two fixed-$Z$ indexed estimators, the equally weighted indexed estimator, the optimally weighted indexed estimator, and the Verma projection estimator.}
\label{fig:napkin_var_mean_0_n_scaled_variance}
\end{figure}

\begin{table}[!t]
\centering
\setlength{\abovecaptionskip}{2pt}
\setlength{\belowcaptionskip}{1pt}
\scriptsize
\setlength{\tabcolsep}{2.1pt}
\renewcommand{\arraystretch}{0.80}
\caption{Top DGP grid settings under the Napkin model: theoretical variances of the five estimators and Verma projection efficiency gains across the three parameters. Settings are ranked primarily by the Verma projection gain over the optimally weighted indexed estimator, with the gain over the equally weighted indexed estimator used to order near ties.}
\label{tab:napkin_var_grid_all_estimands}
\resizebox{0.9\textwidth}{!}{%
\begin{tabular}{crcccccccccccccc}
\toprule
 & ID & $\eta_{\mathrm{int}}$ & $\eta_{\mu,T}$ & $\eta_{\mu,C}$ & $\eta_{\sigma,T}$ & $\eta_{\sigma,C}$ & $\eta_{\kappa,T}$ & $\eta_{\kappa,C}$ & \shortstack{Fixed\\$Z=0$} & \shortstack{Fixed\\$Z=1$} & \shortstack{Indexed\\equal} & \shortstack{Indexed\\optimal} & \shortstack{Verma\\projection} & \shortstack{Gain\\(\%)} \\
\midrule
\multirow{20}{*}{\rotatebox[origin=c]{90}{$\mathrm{ATE}$}} & 298 & -0.500 & 0.000 & 0.000 & 1.500 & 0.000 & 1.000 & 0.000 & 13.092 & 17.713 & 7.701 & 7.528 & 6.807 & 10.586\% \\
 & 301 & -0.500 & 0.500 & 0.000 & 1.500 & 0.000 & 1.000 & 0.000 & 13.353 & 17.900 & 7.813 & 7.648 & 6.916 & 10.584\% \\
 & 304 & -0.500 & 1.000 & 0.000 & 1.500 & 0.000 & 1.000 & 0.000 & 13.708 & 18.152 & 7.966 & 7.811 & 7.064 & 10.575\% \\
 & 310 & -0.500 & 0.500 & 0.500 & 1.500 & 0.000 & 1.000 & 0.000 & 13.778 & 18.934 & 8.178 & 7.975 & 7.216 & 10.510\% \\
 & 313 & -0.500 & 1.000 & 0.500 & 1.500 & 0.000 & 1.000 & 0.000 & 14.133 & 19.188 & 8.330 & 8.139 & 7.364 & 10.519\% \\
 & 55 & -0.500 & 0.000 & 0.000 & 1.500 & 0.000 & 0.000 & 0.000 & 13.124 & 17.738 & 7.715 & 7.543 & 6.825 & 10.516\% \\
 & 307 & -0.500 & 0.000 & 0.500 & 1.500 & 0.000 & 1.000 & 0.000 & 13.517 & 18.747 & 8.066 & 7.854 & 7.108 & 10.497\% \\
 & 58 & -0.500 & 0.500 & 0.000 & 1.500 & 0.000 & 0.000 & 0.000 & 13.380 & 17.925 & 7.826 & 7.661 & 6.933 & 10.500\% \\
 & 61 & -0.500 & 1.000 & 0.000 & 1.500 & 0.000 & 0.000 & 0.000 & 13.729 & 18.179 & 7.977 & 7.822 & 7.080 & 10.477\% \\
 & 64 & -0.500 & 0.000 & 0.500 & 1.500 & 0.000 & 0.000 & 0.000 & 13.548 & 18.771 & 8.080 & 7.869 & 7.126 & 10.432\% \\
 & 67 & -0.500 & 0.500 & 0.500 & 1.500 & 0.000 & 0.000 & 0.000 & 13.804 & 18.959 & 8.191 & 7.988 & 7.234 & 10.430\% \\
 & 70 & -0.500 & 1.000 & 0.500 & 1.500 & 0.000 & 0.000 & 0.000 & 14.154 & 19.214 & 8.342 & 8.150 & 7.381 & 10.425\% \\
 & 322 & -0.500 & 1.000 & 1.000 & 1.500 & 0.000 & 1.000 & 0.000 & 14.706 & 20.586 & 8.823 & 8.578 & 7.770 & 10.405\% \\
 & 319 & -0.500 & 0.500 & 1.000 & 1.500 & 0.000 & 1.000 & 0.000 & 14.351 & 20.332 & 8.671 & 8.413 & 7.622 & 10.373\% \\
 & 316 & -0.500 & 0.000 & 1.000 & 1.500 & 0.000 & 1.000 & 0.000 & 14.089 & 20.144 & 8.559 & 8.291 & 7.514 & 10.344\% \\
 & 79 & -0.500 & 1.000 & 1.000 & 1.500 & 0.000 & 0.000 & 0.000 & 14.727 & 20.612 & 8.835 & 8.590 & 7.787 & 10.316\% \\
 & 76 & -0.500 & 0.500 & 1.000 & 1.500 & 0.000 & 0.000 & 0.000 & 14.377 & 20.357 & 8.684 & 8.426 & 7.640 & 10.298\% \\
 & 73 & -0.500 & 0.000 & 1.000 & 1.500 & 0.000 & 0.000 & 0.000 & 14.121 & 20.169 & 8.573 & 8.306 & 7.532 & 10.284\% \\
 & 299 & 0.000 & 0.000 & 0.000 & 1.500 & 0.000 & 1.000 & 0.000 & 13.092 & 17.417 & 7.627 & 7.474 & 6.785 & 10.151\% \\
 & 302 & 0.000 & 0.500 & 0.000 & 1.500 & 0.000 & 1.000 & 0.000 & 13.353 & 17.604 & 7.739 & 7.593 & 6.894 & 10.146\% \\
\addlinespace[5pt]
\multirow{20}{*}{\rotatebox[origin=c]{90}{$\mathbb{E}[Y(1)]$}} & 559 & -0.500 & 0.000 & 1.000 & 1.500 & 0.000 & 2.000 & 0.000 & 8.203 & 5.570 & 3.444 & 3.318 & 3.313 & 0.148\% \\
 & 640 & -0.500 & 0.000 & 1.000 & 1.500 & 1.000 & 2.000 & 0.000 & 8.203 & 5.570 & 3.444 & 3.318 & 3.313 & 0.148\% \\
 & 631 & -0.500 & 0.000 & 0.500 & 1.500 & 1.000 & 2.000 & 0.000 & 8.203 & 5.570 & 3.444 & 3.318 & 3.313 & 0.148\% \\
 & 622 & -0.500 & 0.000 & 0.000 & 1.500 & 1.000 & 2.000 & 0.000 & 8.203 & 5.570 & 3.444 & 3.318 & 3.313 & 0.148\% \\
 & 550 & -0.500 & 0.000 & 0.500 & 1.500 & 0.000 & 2.000 & 0.000 & 8.203 & 5.570 & 3.444 & 3.318 & 3.313 & 0.148\% \\
 & 541 & -0.500 & 0.000 & 0.000 & 1.500 & 0.000 & 2.000 & 0.000 & 8.203 & 5.570 & 3.444 & 3.318 & 3.313 & 0.148\% \\
 & 562 & -0.500 & 0.500 & 1.000 & 1.500 & 0.000 & 2.000 & 0.000 & 8.469 & 5.756 & 3.557 & 3.428 & 3.423 & 0.150\% \\
 & 643 & -0.500 & 0.500 & 1.000 & 1.500 & 1.000 & 2.000 & 0.000 & 8.469 & 5.756 & 3.557 & 3.428 & 3.423 & 0.150\% \\
 & 634 & -0.500 & 0.500 & 0.500 & 1.500 & 1.000 & 2.000 & 0.000 & 8.469 & 5.756 & 3.557 & 3.428 & 3.423 & 0.150\% \\
 & 553 & -0.500 & 0.500 & 0.500 & 1.500 & 0.000 & 2.000 & 0.000 & 8.469 & 5.756 & 3.557 & 3.428 & 3.423 & 0.150\% \\
 & 625 & -0.500 & 0.500 & 0.000 & 1.500 & 1.000 & 2.000 & 0.000 & 8.469 & 5.756 & 3.557 & 3.428 & 3.423 & 0.150\% \\
 & 544 & -0.500 & 0.500 & 0.000 & 1.500 & 0.000 & 2.000 & 0.000 & 8.469 & 5.756 & 3.557 & 3.428 & 3.423 & 0.149\% \\
 & 565 & -0.500 & 1.000 & 1.000 & 1.500 & 0.000 & 2.000 & 0.000 & 8.827 & 6.007 & 3.711 & 3.577 & 3.571 & 0.151\% \\
 & 646 & -0.500 & 1.000 & 1.000 & 1.500 & 1.000 & 2.000 & 0.000 & 8.827 & 6.007 & 3.711 & 3.577 & 3.571 & 0.151\% \\
 & 637 & -0.500 & 1.000 & 0.500 & 1.500 & 1.000 & 2.000 & 0.000 & 8.827 & 6.007 & 3.711 & 3.577 & 3.571 & 0.151\% \\
 & 556 & -0.500 & 1.000 & 0.500 & 1.500 & 0.000 & 2.000 & 0.000 & 8.827 & 6.007 & 3.711 & 3.577 & 3.571 & 0.151\% \\
 & 628 & -0.500 & 1.000 & 0.000 & 1.500 & 1.000 & 2.000 & 0.000 & 8.827 & 6.007 & 3.711 & 3.577 & 3.571 & 0.151\% \\
 & 547 & -0.500 & 1.000 & 0.000 & 1.500 & 0.000 & 2.000 & 0.000 & 8.827 & 6.007 & 3.711 & 3.577 & 3.571 & 0.151\% \\
 & 316 & -0.500 & 0.000 & 1.000 & 1.500 & 0.000 & 1.000 & 0.000 & 8.182 & 5.582 & 3.442 & 3.319 & 3.314 & 0.127\% \\
 & 397 & -0.500 & 0.000 & 1.000 & 1.500 & 1.000 & 1.000 & 0.000 & 8.182 & 5.582 & 3.442 & 3.319 & 3.314 & 0.127\% \\
\addlinespace[5pt]
\multirow{20}{*}{\rotatebox[origin=c]{90}{$\mathbb{E}[Y(0)]$}} & 398 & 0.000 & 0.000 & 1.000 & 1.500 & 1.000 & 1.000 & 0.000 & 13.061 & 31.169 & 11.058 & 9.205 & 9.139 & 0.721\% \\
 & 317 & 0.000 & 0.000 & 1.000 & 1.500 & 0.000 & 1.000 & 0.000 & 13.061 & 31.169 & 11.058 & 9.205 & 9.139 & 0.721\% \\
 & 389 & 0.000 & 0.000 & 0.500 & 1.500 & 1.000 & 1.000 & 0.000 & 13.061 & 31.169 & 11.058 & 9.205 & 9.139 & 0.721\% \\
 & 380 & 0.000 & 0.000 & 0.000 & 1.500 & 1.000 & 1.000 & 0.000 & 13.061 & 31.169 & 11.058 & 9.205 & 9.139 & 0.721\% \\
 & 308 & 0.000 & 0.000 & 0.500 & 1.500 & 0.000 & 1.000 & 0.000 & 13.061 & 31.169 & 11.058 & 9.205 & 9.139 & 0.721\% \\
 & 299 & 0.000 & 0.000 & 0.000 & 1.500 & 0.000 & 1.000 & 0.000 & 13.061 & 31.169 & 11.058 & 9.205 & 9.139 & 0.721\% \\
 & 383 & 0.000 & 0.500 & 0.000 & 1.500 & 1.000 & 1.000 & 0.000 & 13.478 & 32.237 & 11.430 & 9.506 & 9.439 & 0.703\% \\
 & 392 & 0.000 & 0.500 & 0.500 & 1.500 & 1.000 & 1.000 & 0.000 & 13.478 & 32.237 & 11.430 & 9.506 & 9.439 & 0.703\% \\
 & 401 & 0.000 & 0.500 & 1.000 & 1.500 & 1.000 & 1.000 & 0.000 & 13.478 & 32.237 & 11.430 & 9.506 & 9.439 & 0.703\% \\
 & 320 & 0.000 & 0.500 & 1.000 & 1.500 & 0.000 & 1.000 & 0.000 & 13.478 & 32.237 & 11.430 & 9.506 & 9.439 & 0.703\% \\
 & 311 & 0.000 & 0.500 & 0.500 & 1.500 & 0.000 & 1.000 & 0.000 & 13.478 & 32.237 & 11.430 & 9.506 & 9.439 & 0.703\% \\
 & 302 & 0.000 & 0.500 & 0.000 & 1.500 & 0.000 & 1.000 & 0.000 & 13.478 & 32.237 & 11.430 & 9.506 & 9.439 & 0.703\% \\
 & 386 & 0.000 & 1.000 & 0.000 & 1.500 & 1.000 & 1.000 & 0.000 & 14.043 & 33.667 & 11.930 & 9.911 & 9.844 & 0.680\% \\
 & 395 & 0.000 & 1.000 & 0.500 & 1.500 & 1.000 & 1.000 & 0.000 & 14.043 & 33.667 & 11.930 & 9.911 & 9.844 & 0.680\% \\
 & 404 & 0.000 & 1.000 & 1.000 & 1.500 & 1.000 & 1.000 & 0.000 & 14.043 & 33.667 & 11.930 & 9.911 & 9.844 & 0.680\% \\
 & 323 & 0.000 & 1.000 & 1.000 & 1.500 & 0.000 & 1.000 & 0.000 & 14.043 & 33.667 & 11.930 & 9.911 & 9.844 & 0.680\% \\
 & 314 & 0.000 & 1.000 & 0.500 & 1.500 & 0.000 & 1.000 & 0.000 & 14.043 & 33.667 & 11.930 & 9.911 & 9.844 & 0.680\% \\
 & 305 & 0.000 & 1.000 & 0.000 & 1.500 & 0.000 & 1.000 & 0.000 & 14.043 & 33.667 & 11.930 & 9.911 & 9.844 & 0.680\% \\
 & 399 & 0.500 & 0.000 & 1.000 & 1.500 & 1.000 & 1.000 & 0.000 & 13.061 & 26.168 & 9.808 & 8.713 & 8.660 & 0.610\% \\
 & 318 & 0.500 & 0.000 & 1.000 & 1.500 & 0.000 & 1.000 & 0.000 & 13.061 & 26.168 & 9.808 & 8.713 & 8.660 & 0.610\% \\
\bottomrule
\end{tabular}%
}
\vspace{-0.50em}
\begin{flushleft}
\tiny
\textit{Note:} Variance columns are theoretical. Gain is $100\{V_{\mathrm{opt}}/V_{\mathrm{Verma}}-1\}\%$, where $V_{\mathrm{opt}}$ is for the optimally weighted indexed estimator. Rows are ranked by this gain; near ties are ordered by $100\{V_{\mathrm{equal}}/V_{\mathrm{Verma}}-1\}\%$.
\end{flushleft}
\end{table}

\begin{figure}[t]
\centering
\includegraphics[scale=0.54]{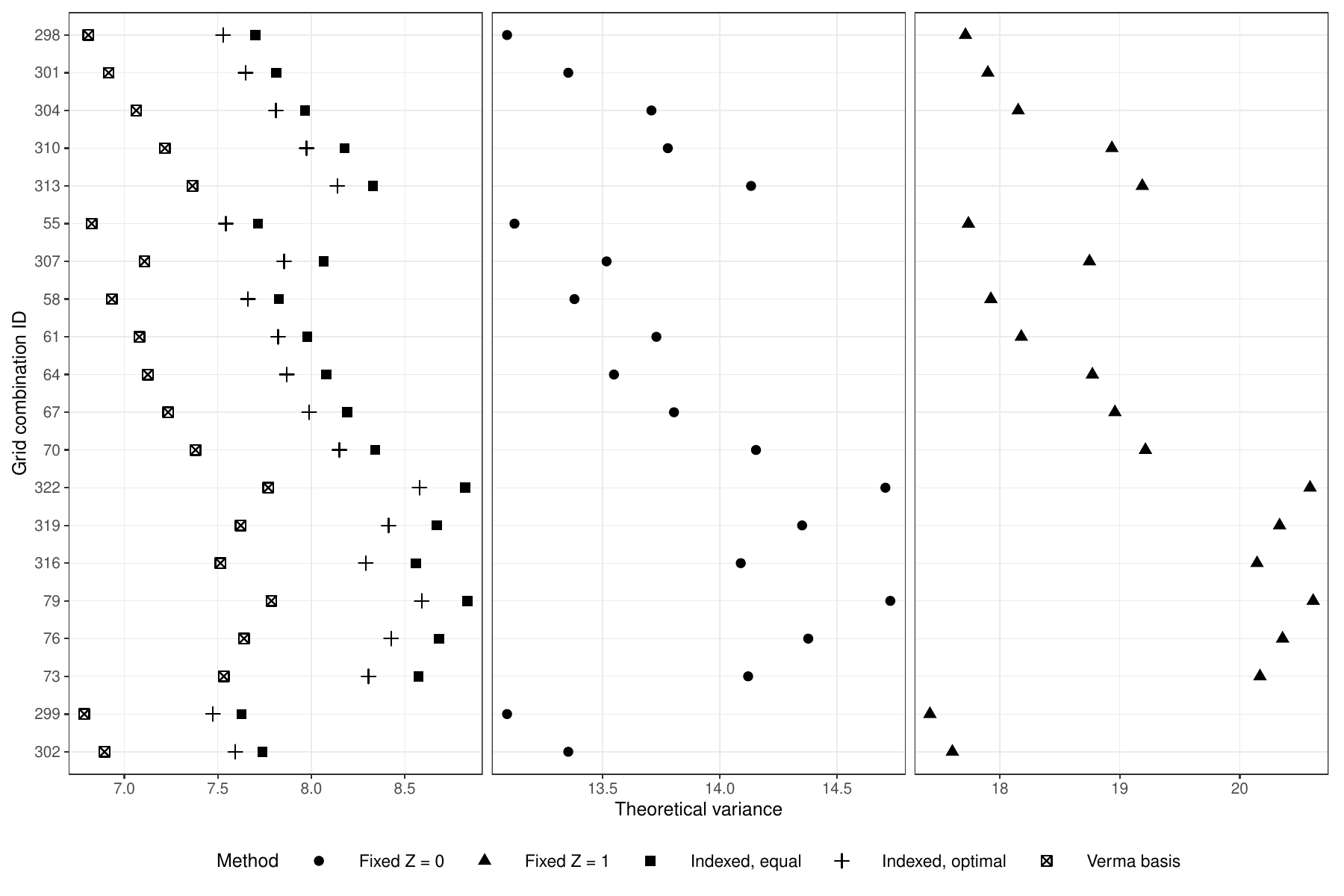}
\caption{
Theoretical asymptotic variances for  $\text{ATE} = \psi_1(\bbP) - \psi_0(\bbP)$ under the Napkin model across the 20 highest-ranked DGPs in Experiment~2. Each row corresponds to one DGP, and symbols denote the five competing estimators. The horizontal axis is partitioned into three noncontiguous ranges to facilitate visual comparison across methods. DGPs are ranked by the relative efficiency gain of the Verma projection estimator over the optimally weighted indexed estimator. 
}
\label{fig:napkin_var_grid_ate_top_20_method_variances_split_x}
\end{figure}

\begin{figure}[t]
\centering
\includegraphics[scale=0.54]{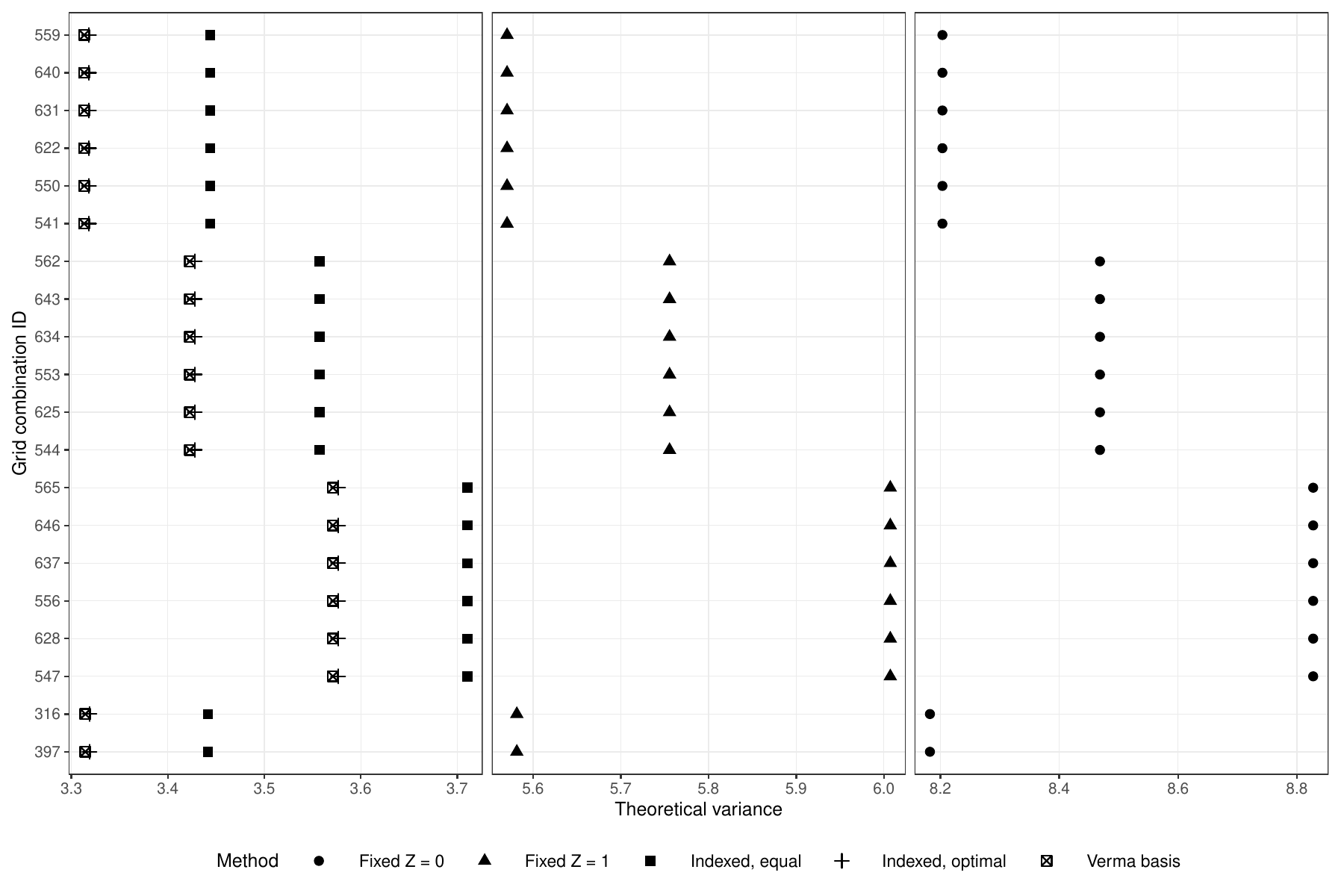}
\caption{
Theoretical asymptotic variances for  $\psi_1(\bbP) = \bbE[Y(1)]$ under the Napkin model across the 20 highest-ranked DGPs in Experiment~2. Each row corresponds to one DGP, and symbols denote the five competing estimators. The horizontal axis is partitioned into three noncontiguous ranges to facilitate visual comparison across methods. DGPs are ranked by the relative efficiency gain of the Verma projection estimator over the optimally weighted indexed estimator, using the gain over the equally weighted indexed estimator to break near ties.
}
\label{fig:napkin_var_grid_mean_1_top_20_method_variances_split_x}
\end{figure}

\begin{figure}[t]
\centering
\includegraphics[scale=0.54]{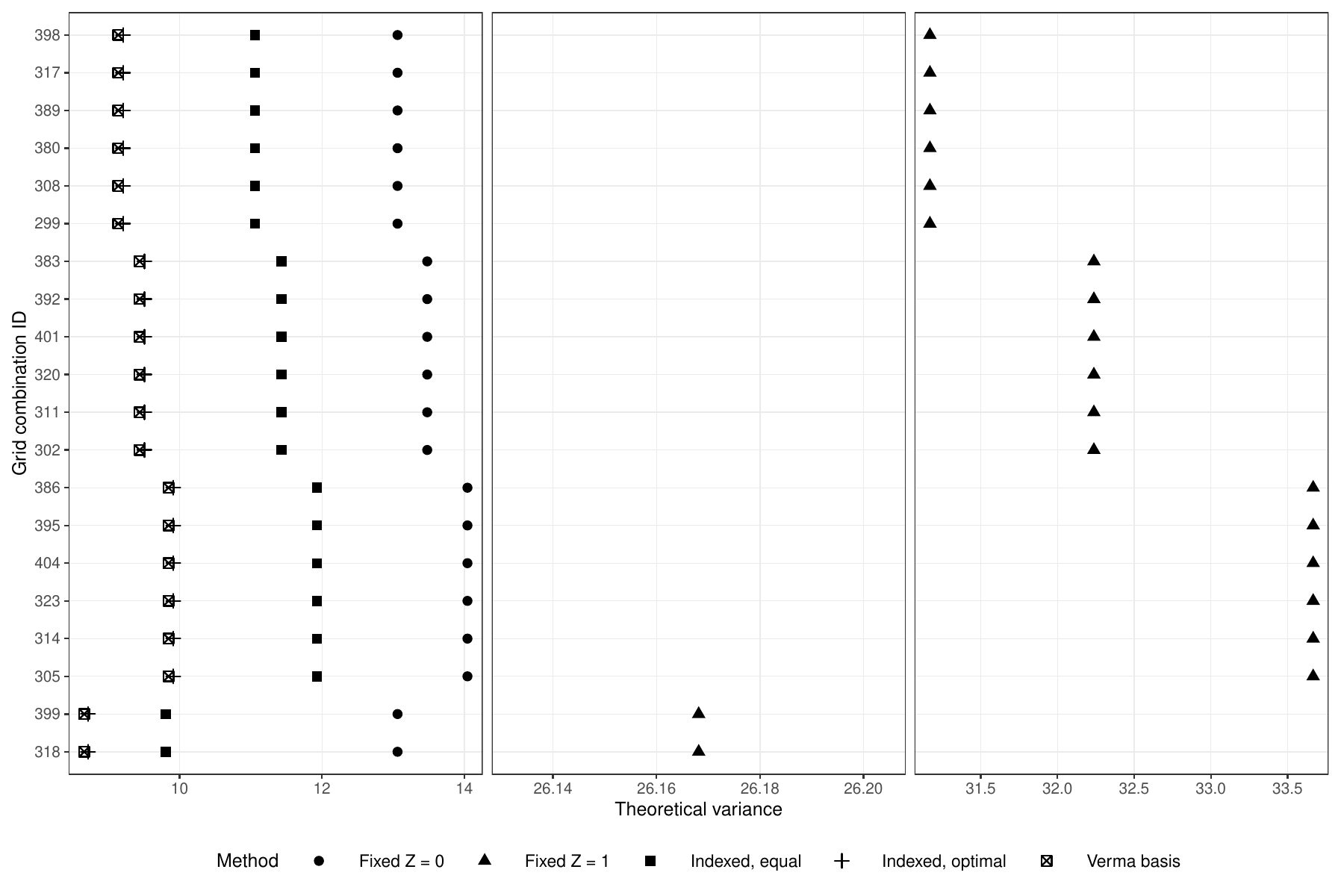}
\caption{
Theoretical asymptotic variances for  $\psi_0(\bbP) = \bbE[Y(0)]$ under the Napkin model across the 20 highest-ranked DGPs in Experiment~2. Each row corresponds to one DGP, and symbols denote the five competing estimators. The horizontal axis is partitioned into three noncontiguous ranges to facilitate visual comparison across methods. DGPs are ranked by the relative efficiency gain of the Verma projection estimator over the optimally weighted indexed estimator, using the gain over the equally weighted indexed estimator to break near ties.
}
\label{fig:napkin_var_grid_mean_0_top_20_method_variances_split_x}
\end{figure}

\begin{table}[htbp]
\centering
\setlength{\abovecaptionskip}{2pt}
\setlength{\belowcaptionskip}{1pt}
\scriptsize
\setlength{\tabcolsep}{2.2pt}
\renewcommand{\arraystretch}{0.70}
\caption{Top DGP grid settings under the Napkin model: theoretical variances for basis enrichments of the Verma projection across the three parameters, with efficiency gains relative to the $Y$-only basis. Settings are ranked by the efficiency gain of the full $(Y,Y^2,Y^3)$ basis over the $Y$-only basis.}
\label{tab:napkin_basis_grid_all_estimands}
\resizebox{0.85\textwidth}{!}{%
\begin{tabular}{crrrrrrrrrr@{\hspace{10pt}}rr}
\toprule
 & ID & $\eta_{\mu,T}$ & $\eta_{\mu,C}$ & $\eta_{\sigma,T}$ & $\eta_{\sigma,C}$ & $\eta_{\kappa,T}$ & $\eta_{\kappa,C}$ & $Y$ & $(Y,Y^2)$ & $(Y,Y^2,Y^3)$ & \shortstack{$(Y,Y^2)$\\vs. $Y$} & \shortstack{$(Y,Y^2,Y^3)$\\vs. $Y$} \\
\midrule
\multirow{20}{*}{\rotatebox[origin=c]{90}{$\mathrm{ATE}$}} & 55 & 0.000 & -1.000 & 3.000 & 0.000 & 1.000 & -1.000 & 27.543 & 27.529 & 27.430 & 0.050\% & 0.411\% \\
 & 56 & 0.500 & -1.000 & 3.000 & 0.000 & 1.000 & -1.000 & 27.644 & 27.630 & 27.531 & 0.051\% & 0.410\% \\
 & 58 & 0.000 & 0.000 & 3.000 & 0.000 & 1.000 & -1.000 & 27.824 & 27.809 & 27.710 & 0.052\% & 0.410\% \\
 & 59 & 0.500 & 0.000 & 3.000 & 0.000 & 1.000 & -1.000 & 27.926 & 27.911 & 27.812 & 0.053\% & 0.409\% \\
 & 57 & 1.000 & -1.000 & 3.000 & 0.000 & 1.000 & -1.000 & 27.785 & 27.771 & 27.672 & 0.051\% & 0.408\% \\
 & 60 & 1.000 & 0.000 & 3.000 & 0.000 & 1.000 & -1.000 & 28.066 & 28.052 & 27.953 & 0.053\% & 0.407\% \\
 & 61 & 0.000 & 1.000 & 3.000 & 0.000 & 1.000 & -1.000 & 28.525 & 28.510 & 28.411 & 0.055\% & 0.404\% \\
 & 62 & 0.500 & 1.000 & 3.000 & 0.000 & 1.000 & -1.000 & 28.627 & 28.611 & 28.512 & 0.056\% & 0.404\% \\
 & 63 & 1.000 & 1.000 & 3.000 & 0.000 & 1.000 & -1.000 & 28.768 & 28.752 & 28.653 & 0.056\% & 0.402\% \\
 & 19 & 0.000 & -1.000 & 3.000 & -1.000 & 1.000 & -1.000 & 28.602 & 28.589 & 28.489 & 0.049\% & 0.396\% \\
 & 91 & 0.000 & -1.000 & 3.000 & 1.000 & 1.000 & -1.000 & 29.728 & 29.711 & 29.611 & 0.058\% & 0.396\% \\
 & 20 & 0.500 & -1.000 & 3.000 & -1.000 & 1.000 & -1.000 & 28.704 & 28.690 & 28.591 & 0.049\% & 0.396\% \\
 & 92 & 0.500 & -1.000 & 3.000 & 1.000 & 1.000 & -1.000 & 29.830 & 29.813 & 29.713 & 0.058\% & 0.395\% \\
 & 21 & 1.000 & -1.000 & 3.000 & -1.000 & 1.000 & -1.000 & 28.845 & 28.831 & 28.732 & 0.049\% & 0.394\% \\
 & 93 & 1.000 & -1.000 & 3.000 & 1.000 & 1.000 & -1.000 & 29.971 & 29.954 & 29.854 & 0.059\% & 0.393\% \\
 & 94 & 0.000 & 0.000 & 3.000 & 1.000 & 1.000 & -1.000 & 29.999 & 29.981 & 29.881 & 0.058\% & 0.393\% \\
 & 95 & 0.500 & 0.000 & 3.000 & 1.000 & 1.000 & -1.000 & 30.100 & 30.083 & 29.983 & 0.059\% & 0.392\% \\
 & 22 & 0.000 & 0.000 & 3.000 & -1.000 & 1.000 & -1.000 & 28.890 & 28.877 & 28.778 & 0.047\% & 0.391\% \\
 & 23 & 0.500 & 0.000 & 3.000 & -1.000 & 1.000 & -1.000 & 28.992 & 28.978 & 28.879 & 0.048\% & 0.391\% \\
 & 96 & 1.000 & 0.000 & 3.000 & 1.000 & 1.000 & -1.000 & 30.241 & 30.223 & 30.124 & 0.059\% & 0.390\% \\
\addlinespace[5pt]
\multirow{20}{*}{\rotatebox[origin=c]{90}{$\mathbb{E}[Y(1)]$}} & 91 & 0.000 & -1.000 & 3.000 & 1.000 & 1.000 & -1.000 & 24.361 & 24.348 & 24.249 & 0.056\% & 0.464\% \\
 & 19 & 0.000 & -1.000 & 3.000 & -1.000 & 1.000 & -1.000 & 24.361 & 24.348 & 24.249 & 0.056\% & 0.464\% \\
 & 22 & 0.000 & 0.000 & 3.000 & -1.000 & 1.000 & -1.000 & 24.361 & 24.348 & 24.249 & 0.056\% & 0.464\% \\
 & 94 & 0.000 & 0.000 & 3.000 & 1.000 & 1.000 & -1.000 & 24.361 & 24.348 & 24.249 & 0.056\% & 0.464\% \\
 & 25 & 0.000 & 1.000 & 3.000 & -1.000 & 1.000 & -1.000 & 24.361 & 24.348 & 24.249 & 0.056\% & 0.464\% \\
 & 97 & 0.000 & 1.000 & 3.000 & 1.000 & 1.000 & -1.000 & 24.361 & 24.348 & 24.249 & 0.056\% & 0.464\% \\
 & 61 & 0.000 & 1.000 & 3.000 & 0.000 & 1.000 & -1.000 & 24.361 & 24.348 & 24.249 & 0.056\% & 0.464\% \\
 & 58 & 0.000 & 0.000 & 3.000 & 0.000 & 1.000 & -1.000 & 24.361 & 24.348 & 24.249 & 0.056\% & 0.464\% \\
 & 55 & 0.000 & -1.000 & 3.000 & 0.000 & 1.000 & -1.000 & 24.361 & 24.348 & 24.249 & 0.056\% & 0.464\% \\
 & 92 & 0.500 & -1.000 & 3.000 & 1.000 & 1.000 & -1.000 & 24.463 & 24.449 & 24.351 & 0.056\% & 0.462\% \\
 & 20 & 0.500 & -1.000 & 3.000 & -1.000 & 1.000 & -1.000 & 24.463 & 24.449 & 24.351 & 0.056\% & 0.462\% \\
 & 23 & 0.500 & 0.000 & 3.000 & -1.000 & 1.000 & -1.000 & 24.463 & 24.449 & 24.351 & 0.056\% & 0.462\% \\
 & 95 & 0.500 & 0.000 & 3.000 & 1.000 & 1.000 & -1.000 & 24.463 & 24.449 & 24.351 & 0.056\% & 0.462\% \\
 & 26 & 0.500 & 1.000 & 3.000 & -1.000 & 1.000 & -1.000 & 24.463 & 24.449 & 24.351 & 0.056\% & 0.462\% \\
 & 98 & 0.500 & 1.000 & 3.000 & 1.000 & 1.000 & -1.000 & 24.463 & 24.449 & 24.351 & 0.056\% & 0.462\% \\
 & 62 & 0.500 & 1.000 & 3.000 & 0.000 & 1.000 & -1.000 & 24.463 & 24.449 & 24.351 & 0.056\% & 0.462\% \\
 & 59 & 0.500 & 0.000 & 3.000 & 0.000 & 1.000 & -1.000 & 24.463 & 24.449 & 24.351 & 0.056\% & 0.462\% \\
 & 56 & 0.500 & -1.000 & 3.000 & 0.000 & 1.000 & -1.000 & 24.463 & 24.449 & 24.351 & 0.056\% & 0.462\% \\
 & 93 & 1.000 & -1.000 & 3.000 & 1.000 & 1.000 & -1.000 & 24.604 & 24.590 & 24.491 & 0.057\% & 0.460\% \\
 & 21 & 1.000 & -1.000 & 3.000 & -1.000 & 1.000 & -1.000 & 24.604 & 24.590 & 24.491 & 0.057\% & 0.460\% \\
\addlinespace[5pt]
\multirow{20}{*}{\rotatebox[origin=c]{90}{$\mathbb{E}[Y(0)]$}} & 127 & 0.000 & -1.000 & 3.000 & -1.000 & 2.000 & -1.000 & 68.700 & 68.143 & 67.542 & 0.817\% & 1.714\% \\
 & 199 & 0.000 & -1.000 & 3.000 & 1.000 & 2.000 & -1.000 & 68.700 & 68.143 & 67.542 & 0.817\% & 1.714\% \\
 & 130 & 0.000 & 0.000 & 3.000 & -1.000 & 2.000 & -1.000 & 68.700 & 68.143 & 67.542 & 0.817\% & 1.714\% \\
 & 133 & 0.000 & 1.000 & 3.000 & -1.000 & 2.000 & -1.000 & 68.700 & 68.143 & 67.542 & 0.817\% & 1.714\% \\
 & 163 & 0.000 & -1.000 & 3.000 & 0.000 & 2.000 & -1.000 & 68.700 & 68.143 & 67.542 & 0.817\% & 1.714\% \\
 & 166 & 0.000 & 0.000 & 3.000 & 0.000 & 2.000 & -1.000 & 68.700 & 68.143 & 67.542 & 0.817\% & 1.714\% \\
 & 169 & 0.000 & 1.000 & 3.000 & 0.000 & 2.000 & -1.000 & 68.700 & 68.143 & 67.542 & 0.817\% & 1.714\% \\
 & 128 & 0.500 & -1.000 & 3.000 & -1.000 & 2.000 & -1.000 & 69.011 & 68.451 & 67.851 & 0.818\% & 1.710\% \\
 & 131 & 0.500 & 0.000 & 3.000 & -1.000 & 2.000 & -1.000 & 69.011 & 68.451 & 67.851 & 0.818\% & 1.710\% \\
 & 200 & 0.500 & -1.000 & 3.000 & 1.000 & 2.000 & -1.000 & 69.011 & 68.451 & 67.851 & 0.818\% & 1.710\% \\
 & 134 & 0.500 & 1.000 & 3.000 & -1.000 & 2.000 & -1.000 & 69.011 & 68.451 & 67.851 & 0.818\% & 1.710\% \\
 & 170 & 0.500 & 1.000 & 3.000 & 0.000 & 2.000 & -1.000 & 69.011 & 68.451 & 67.851 & 0.818\% & 1.710\% \\
 & 167 & 0.500 & 0.000 & 3.000 & 0.000 & 2.000 & -1.000 & 69.011 & 68.451 & 67.851 & 0.818\% & 1.710\% \\
 & 164 & 0.500 & -1.000 & 3.000 & 0.000 & 2.000 & -1.000 & 69.011 & 68.451 & 67.851 & 0.818\% & 1.710\% \\
 & 129 & 1.000 & -1.000 & 3.000 & -1.000 & 2.000 & -1.000 & 69.427 & 68.864 & 68.265 & 0.818\% & 1.702\% \\
 & 135 & 1.000 & 1.000 & 3.000 & -1.000 & 2.000 & -1.000 & 69.427 & 68.864 & 68.265 & 0.818\% & 1.702\% \\
 & 132 & 1.000 & 0.000 & 3.000 & -1.000 & 2.000 & -1.000 & 69.427 & 68.864 & 68.265 & 0.818\% & 1.702\% \\
 & 171 & 1.000 & 1.000 & 3.000 & 0.000 & 2.000 & -1.000 & 69.427 & 68.864 & 68.265 & 0.818\% & 1.702\% \\
 & 168 & 1.000 & 0.000 & 3.000 & 0.000 & 2.000 & -1.000 & 69.427 & 68.864 & 68.265 & 0.818\% & 1.702\% \\
 & 165 & 1.000 & -1.000 & 3.000 & 0.000 & 2.000 & -1.000 & 69.427 & 68.864 & 68.265 & 0.818\% & 1.702\% \\
\bottomrule
\end{tabular}%
}
\vspace{-0.50em}
\begin{flushleft}
\tiny
\textit{Note:} Variance columns correspond to the displayed bases for the Verma projection. Gains are $100\{V_Y/V_{\mathrm{rich}}-1\}\%$, relative to the $Y$-only basis. Rows are ranked by the gain for the full $(Y,Y^2,Y^3)$ basis.
\end{flushleft}
\end{table}

\begin{figure}[!t]
\centering
\includegraphics[scale=0.55]
{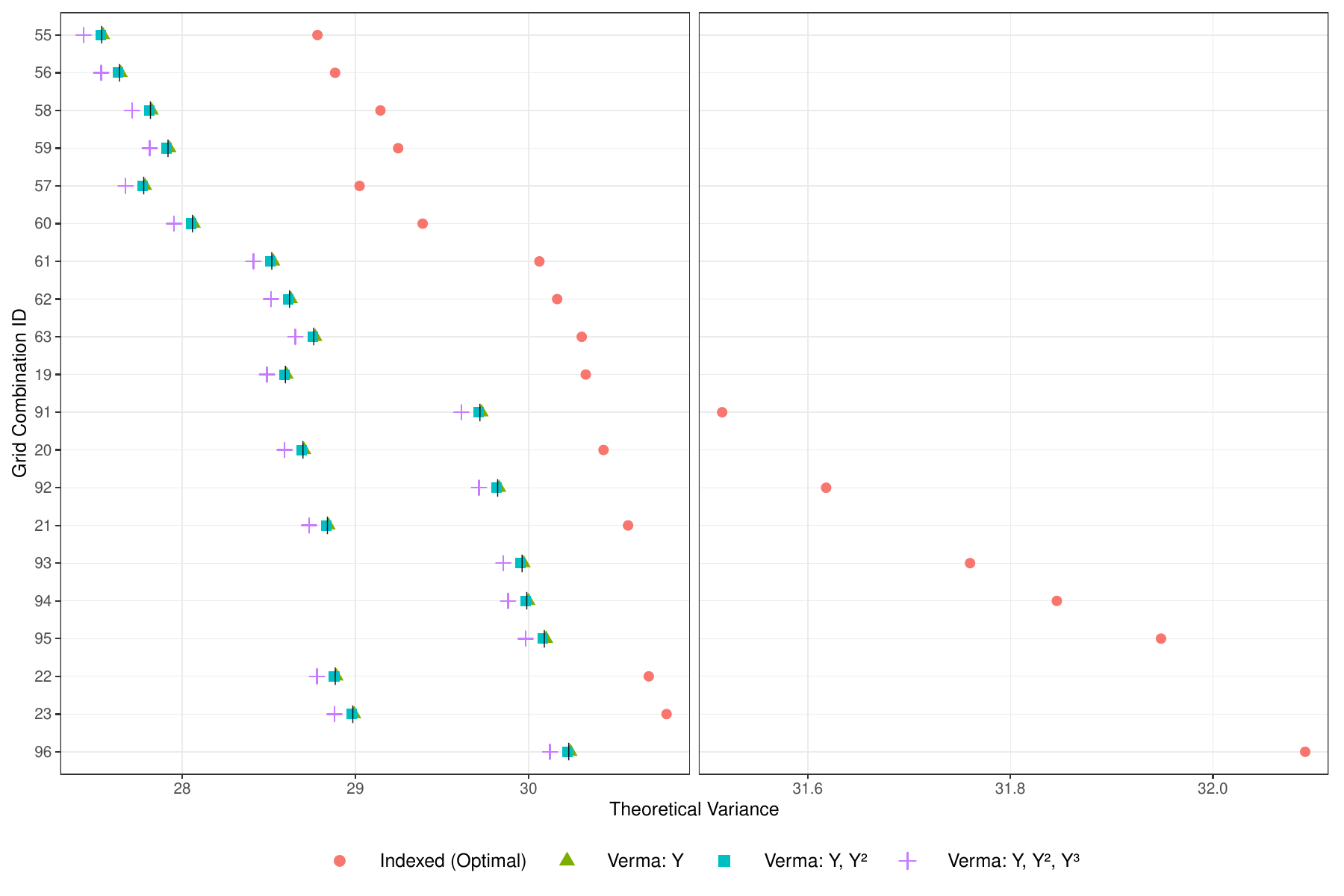}
\caption{
Effect of Verma-basis richness on the theoretical asymptotic variance for $\text{ATE} = \psi_1(\bbP) - \psi_0(\bbP)$ under the Napkin model. Results are shown for the 20 highest-ranked DGPs in Experiment~3. The Verma projection estimators use the nested outcome bases $\{Y\}$, $\{Y,Y^2\}$, and $\{Y,Y^2,Y^3\}$; the optimally weighted indexed estimator is included as a reference. DGPs are ranked by the relative efficiency gain of the full Verma basis $\{Y,Y^2,Y^3\}$ over the $Y$-only basis. Vertical ticks mark DGPs for which two or more method-specific theoretical variances are visually indistinguishable at the plotted scale. 
}
\label{fig:napkin_basis_grid_ate_top_20_method_variances_split_x}
\end{figure}

\begin{figure}[!t]
\centering
\includegraphics[scale=0.55]
{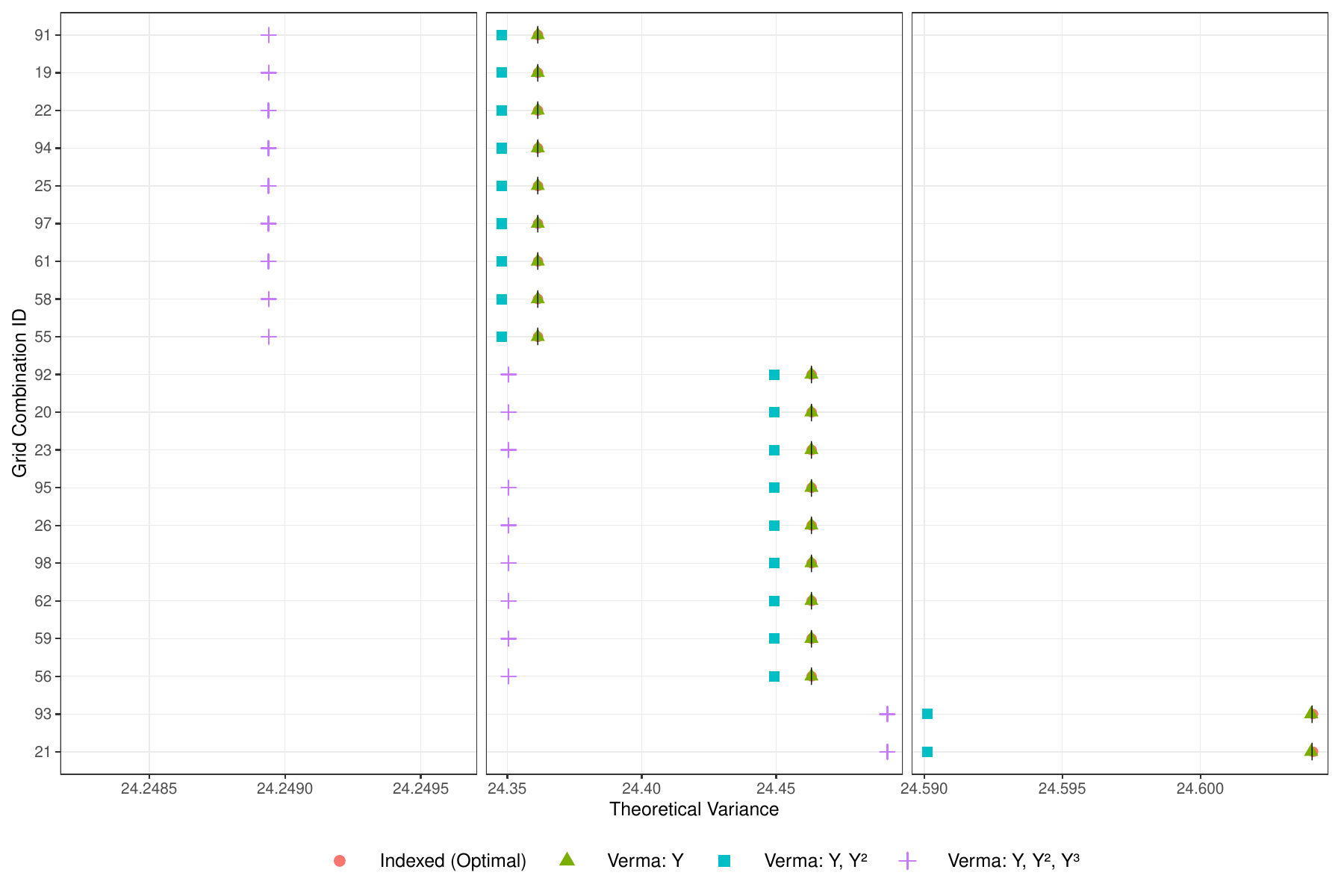}
\caption{
Effect of Verma-basis richness on the theoretical asymptotic variance for $\psi_1(\bbP) = \bbE[Y(1)]$ under the Napkin model. Results are shown for the 20 highest-ranked DGPs in Experiment~3. The Verma projection estimators use the nested outcome bases $\{Y\}$, $\{Y,Y^2\}$, and $\{Y,Y^2,Y^3\}$; the optimally weighted indexed estimator is included as a reference. DGPs are ranked by the relative efficiency gain of the full Verma basis $\{Y,Y^2,Y^3\}$ over the $Y$-only basis. Vertical ticks mark DGPs for which two or more method-specific theoretical variances are visually indistinguishable at the plotted scale. 
}
\label{fig:napkin_basis_grid_mean_1_top_20_method_variances_split_x}
\end{figure}

\begin{figure}[!t]
\centering
\includegraphics[scale=0.55]
{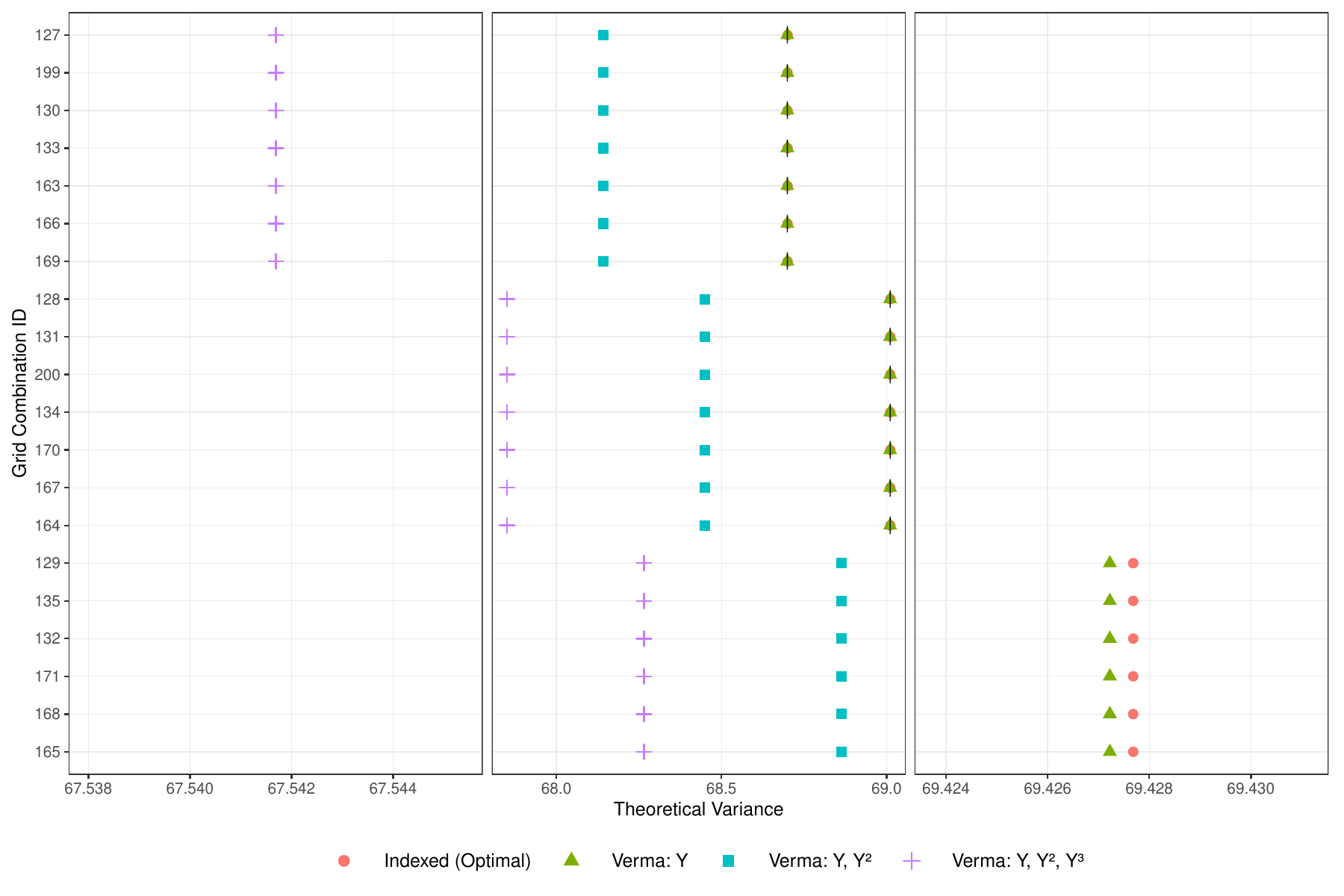}
\caption{
Effect of Verma-basis richness on the theoretical asymptotic variance for $\psi_0(\bbP) = \bbE[Y(0)]$ under the Napkin model. Results are shown for the 20 highest-ranked DGPs in Experiment~3. The Verma projection estimators use the nested outcome bases $\{Y\}$, $\{Y,Y^2\}$, and $\{Y,Y^2,Y^3\}$; the optimally weighted indexed estimator is included as a reference. DGPs are ranked by the relative efficiency gain of the full Verma basis $\{Y,Y^2,Y^3\}$ over the $Y$-only basis. Vertical ticks mark DGPs for which two or more method-specific theoretical variances are visually indistinguishable at the plotted scale. 
}
\label{fig:napkin_basis_grid_mean_0_top_20_method_variances_split_x}
\end{figure}

\clearpage
\section{Real-data Applications}
\label{app:realdata}

\subsection{Application of the Extended Front-Door Model}

We applied the front-door estimators to the Framingham data \citep{kannel1968framingham} analyzed by \citet{bhattacharya2022semiparametric}. The data were preprocessed following the structure of the original Framingham application. We retained individuals with at least one follow-up exam and constructed one row per subject. Variables measured at the first exam were used for the treatment, anchor, and baseline covariates, while variables measured at the second exam were used for the mediator and outcome. 

The treatment $A$ is current smoking at the first exam, defined by $\texttt{CURSMOKE}=1$. The mediator $M$ is hypertension at the second exam, defined by $\texttt{HYPERTEN}=1$. The anchor variable $Z$ is prior history of hypertension at the first exam, defined by $\texttt{PREVHYP}=1$. Baseline covariates $X$ include first-exam age, sex, BMI, and prior coronary heart disease, $\texttt{PREVCHD}$. To align the real-data analysis with the continuous-outcome simulation setting, we used a continuous time-to-event variable as the outcome. Specifically, $Y$ is the recorded time to coronary heart disease (CHD) or censoring, measured in years as $\texttt{TIMECHD}/365.25$. This outcome is analyzed as a continuous variable; the analysis is therefore not a censoring-aware survival analysis. The resulting complete-case analysis sample contains $3977$ subjects.

The target estimand is the pooled average causal effect of current smoking on the continuous time-to-CHD outcome,
$\psi = \bbE[Y(1)-Y(0)]$, where positive values correspond to a longer mean recorded time to CHD or censoring under current smoking, and negative values correspond to a shorter mean recorded time. For each estimator, nuisance functions were estimated using working parametric models for the anchor, treatment, mediator, and outcome regressions, with baseline covariates included in each nuisance model. Since the outcome is continuous, the Verma projection bases use raw outcome moments $Y$, $(Y,Y^2)$, and $(Y,Y^2,Y^3)$, matching the basis choices used in the simulation studies.

Table~\ref{tab:realdata_frontdoor_estimates} reports the resulting ATE estimates. Entries are point estimates with Wald-type $95\%$ confidence intervals, measured in years. The table includes the two fixed-anchor estimators, the equally weighted indexed estimator, the optimally weighted indexed estimator, and the three Verma projection estimators corresponding to the three outcome-moment bases. The estimated effects are small and negative across the estimators, suggesting a shorter mean recorded time to CHD or censoring under current smoking in this application.

Table~\ref{tab:realdata_frontdoor_CI} compares estimator precision using confidence interval lengths. For each estimator, the confidence interval length is divided by the confidence interval length of the equally weighted indexed estimator for the same pooled ACE estimand. Thus, values below one indicate shorter confidence intervals than the equally weighted indexed estimator, while values above one indicate longer intervals. In this application, the Verma projection estimators yield shorter intervals than the equally weighted indexed estimator, with relative interval lengths approximately $0.75$, $0.74$, and $0.73$ for the $Y$, $(Y,Y^2)$, and $(Y,Y^2,Y^3)$ bases, respectively.


\begin{table}[!htbp]
\centering
\caption{Front-door real-data estimates of the ACE of current smoking on time to coronary heart disease. Entries are mean differences in years with 95\% confidence intervals.}
\label{tab:realdata_frontdoor_estimates}
\begin{tabular}{lcc}
\toprule
Estimator & Estimate (95\% CI) & SE \\
\midrule
Fixed $Z=0$ & -0.002 (-0.012, 0.008) & 0.005 \\
Fixed $Z=1$ & -0.010 (-0.020, -0.000) & 0.005 \\
Equal indexed & -0.006 (-0.013, 0.001) & 0.004 \\
Optimal indexed & -0.006 (-0.013, 0.001) & 0.004 \\
Verma $Y$ & -0.017 (-0.022, -0.011) & 0.003 \\
Verma $Y,Y^2$ & -0.016 (-0.021, -0.010) & 0.003 \\
Verma $Y,Y^2,Y^3$ & -0.016 (-0.021, -0.011) & 0.003 \\
\bottomrule
\end{tabular}
\end{table}


\begin{table}[!htbp]
\centering
\caption{Relative confidence interval lengths for the front-door ATE estimates. The reference is the equally weighted indexed estimator.}
\label{tab:realdata_frontdoor_CI}
\begin{tabular}{lccc}
\toprule
Estimator & CI length & Relative length & Tighter than reference \\
\midrule
Fixed $Z=0$ & 0.019 & 1.36 & No \\
Fixed $Z=1$ & 0.019 & 1.36 & No \\
Equal indexed & 0.014 & 1.00 & -- \\
Optimal indexed & 0.014 & 1.00 & Yes \\
Verma $Y$ & 0.011 & 0.75 & Yes \\
Verma $Y,Y^2$ & 0.011 & 0.74 & Yes \\
Verma $Y,Y^2,Y^3$ & 0.010 & 0.73 & Yes \\
\bottomrule
\end{tabular}
\begin{minipage}{0.82\linewidth}
\footnotesize Values below 1 indicate shorter confidence intervals than the equally weighted indexed estimator for the same pooled ATE estimand.
\end{minipage}
\end{table}

\clearpage
\subsection{Application of the Napkin Model}

We applied the proposed estimators to the Life Course 1971--2002 study from the Finnish Social Science Data Archive\footnote{Available upon application at \href{https://services.fsd.tuni.fi/catalogue/FSD2007}{https://services.fsd.tuni.fi/catalogue/FSD2007}.}, following the educational-attainment and income application considered by \citet{guo2025causal}. The outcome $Y$ is annual income in 2000 euros. Educational attainment $X$ was coded as secondary or less, lower tertiary, and higher tertiary; for estimation, each target education level was represented by the binary indicator $I(X=x)$. We report two pairwise contrasts: lower tertiary versus secondary or less, and higher tertiary versus lower tertiary.

The analysis used $509$ complete cases out of $634$ observations. Parental socioeconomic status $W$ was binarized as low versus middle/high parental SES. Primary-school GPA $Z$ was binarized at the complete-case mean, $8.0226$. Sex and childhood verbal intelligence, measured by the ITPA score, were used only to define subgroups; ITPA was binarized at the complete-case mean, $36.1100$. Results are reported separately for four subgroups: female below-mean ITPA ($n=137$), female above/equal-mean ITPA ($n=133$), male below-mean ITPA ($n=126$), and male above/equal-mean ITPA ($n=113$).

Within each subgroup and contrast, we computed seven estimators of the ATE: the two fixed-$Z$ estimators, the equally weighted indexed estimator, the estimated optimally weighted indexed estimator, and three Verma projection estimators using basis choices $Y$, $(Y,Y^2)$, and $(Y,Y^2,Y^3)$. Nuisance functions were estimated using the working nuisance specification to improve stability in the subgroup analyses. Point estimates are reported in euros with Wald-type $95\%$ confidence intervals based on the estimated influence functions. 

Results are provided in Table~\ref{tab:realdata_napkin_estimates}. The table reports subgroup-specific ATE estimates, stratified by sex and by whether the ITPA score is below or above/equal to the complete-case mean. Each entry is a point estimate with a 95\% confidence interval, reported in euros. The first block estimates the effect of lower tertiary education relative to secondary education or less, and the second block estimates the effect of higher tertiary education relative to lower tertiary education. Positive values indicate higher expected annual income under the higher education level in the contrast, whereas negative values indicate lower expected annual income. Confidence intervals that include zero should be interpreted as providing limited evidence against no effect in that subgroup and contrast.

Table~\ref{tab:realdata_napkin_CI} summarizes the relative lengths of the confidence intervals in Table~\ref{tab:realdata_napkin_estimates}. For each estimator, the confidence interval length is divided by the corresponding length for the equally weighted indexed estimator within the same subgroup and contrast. Thus, values below one indicate shorter confidence intervals than the equally weighted indexed estimator. The Verma projection estimators yield shorter confidence intervals in all eight subgroup-by-contrast comparisons, with mean relative lengths of 0.65, 0.57, and 0.52 for the $Y$, $(Y,Y^2)$, and $(Y,Y^2,Y^3)$ bases, respectively. 

Table~\ref{tab:realdata_napkin_CI} compares estimator precision using confidence interval lengths. For each estimator, subgroup, and education contrast, we computed the ratio of its confidence interval length to the confidence interval length of the equally weighted indexed estimator for the same estimand. Specifically, if $\widehat{\mathrm{CI}}_{m,g,c}$ denotes the confidence interval for estimator $m$, subgroup $g$, and contrast $c$, the ratios are based on $|\widehat{\mathrm{CI}}_{m,g,c}| / |\widehat{\mathrm{CI}}_{\mathrm{eq},g,c}|$, where $|\widehat{\mathrm{CI}}|$ is the upper endpoint minus the lower endpoint, and $\mathrm{eq}$ denotes the equally weighted indexed estimator. These ratios are computed over the eight subgroup-by-contrast estimands, corresponding to four sex-by-ITPA subgroups and two education contrasts. The ``Mean rel. length'' column reports the average of these eight ratios, while the ``Median rel. length'' column reports their median. The ``Range'' column gives the minimum and maximum ratio observed across the eight estimands. Values below one indicate shorter confidence intervals than the equally weighted indexed estimator, and values above one indicate longer confidence intervals. The ``Tighter cells'' column gives the number of subgroup-by-contrast estimands, out of eight, for which the ratio was below one.

Overall, the real-data results illustrate the practical efficiency gains delivered by the Verma projection approach. Across all eight subgroup-by-contrast estimands, the Verma projection estimators produced shorter confidence intervals than the equally weighted indexed estimator, and the higher-dimensional projection bases yielded the shortest average intervals in this application. These findings suggest that, Verma projection estimators may be preferable in practice for improving precision while preserving the same identified causal estimand. At the same time, the choice of projection basis should be guided by finite-sample stability, since richer bases can be more sensitive to sparse cells and model extrapolation.

\clearpage

\begin{table}[!htbp]
\centering
\caption{Napkin real-data estimates of the ATE by sex and ITPA score. Entries are point estimates with 95\% confidence intervals in euros.}
\label{tab:realdata_napkin_estimates}
\resizebox{\textwidth}{!}{%
\begin{tabular}{lcccc}
\toprule
\textbf{Sex} & \multicolumn{2}{c}{\textbf{Female}} & \multicolumn{2}{c}{\textbf{Male}} \\
\textbf{ITPA score ($n$)} & Below avg (137) & Above avg (133) & Below avg (126) & Above avg (113) \\
\midrule
\multicolumn{5}{c}{\textbf{Lower tertiary vs secondary or less}} \\
Fixed $Z=0$ & 5,897 (1,639, 10,156) & -9,923 (-28,604, 8,758) & 949 (-5,701, 7,600) & -1,455 (-11,531, 8,621) \\
Fixed $Z=1$ & 258 (-3,703, 4,220) & 3,306 (-1,607, 8,219) & 3,234 (-6,478, 12,946) & 1,258 (-8,318, 10,833) \\
Equal indexed & 3,078 (145, 6,010) & -3,309 (-12,967, 6,350) & 2,092 (-3,794, 7,977) & -99 (-7,052, 6,854) \\
Optimal indexed & 3,266 (428, 6,105) & 3,474 (-1,053, 8,000) & 2,453 (-2,731, 7,637) & 1,541 (-5,167, 8,248) \\
Verma $Y$ & 4,288 (2,093, 6,483) & -5,152 (-8,099, -2,206) & 2,451 (-2,722, 7,625) & 1,562 (-5,148, 8,273) \\
Verma $Y,Y^2$ & 4,930 (2,915, 6,946) & -3,325 (-5,808, -841) & 2,469 (-2,626, 7,563) & -373 (-5,762, 5,017) \\
Verma $Y,Y^2,Y^3$ & 5,168 (3,178, 7,158) & -5,825 (-8,177, -3,474) & 4,760 (309, 9,211) & -156 (-5,527, 5,216) \\
\addlinespace
\multicolumn{5}{c}{\textbf{Higher vs lower tertiary}} \\
Fixed $Z=0$ & -2,918 (-16,046, 10,210) & 31,432 (4,044, 58,819) & 3,916 (-6,842, 14,674) & 11,129 (1,166, 21,091) \\
Fixed $Z=1$ & 9,510 (1,601, 17,418) & 7,208 (-24, 14,440) & -2,658 (-15,439, 10,123) & 7,492 (-4,770, 19,754) \\
Equal indexed & 3,296 (-4,389, 10,981) & 19,320 (5,149, 33,491) & 629 (-7,735, 8,993) & 9,311 (1,408, 17,213) \\
Optimal indexed & 5,968 (-785, 12,722) & 8,632 (1,646, 15,617) & 1,117 (-7,120, 9,353) & 7,213 (503, 13,922) \\
Verma $Y$ & -1,923 (-4,298, 452) & 21,947 (17,318, 26,576) & 1,365 (-6,223, 8,954) & 6,863 (749, 12,978) \\
Verma $Y,Y^2$ & -2,493 (-4,686, -299) & 18,193 (14,928, 21,459) & 5,193 (-2,007, 12,393) & 8,777 (3,983, 13,571) \\
Verma $Y,Y^2,Y^3$ & -2,127 (-4,222, -32) & 20,559 (17,597, 23,522) & 3,063 (-2,151, 8,276) & 8,424 (3,658, 13,189) \\
\bottomrule
\end{tabular}%
}
\end{table}

\begin{table}[!htbp]
\centering
\caption{Relative confidence interval lengths for the subgroup ATE estimates. The reference is the equally weighted indexed estimator within the same subgroup and contrast.}
\label{tab:realdata_napkin_CI}
\begin{tabular}{lcccc}
\toprule
Estimator & Mean rel. length & Median rel. length & Range & Tighter cells \\
\midrule
Fixed $Z=0$ & 1.52 & 1.45 & 1.13--1.93 & 0/8 \\
Fixed $Z=1$ & 1.19 & 1.36 & 0.51--1.65 & 2/8 \\
Equal indexed & 1.00 & 1.00 & 1.00--1.00 & -- \\
Optimal indexed & 0.81 & 0.88 & 0.47--0.98 & 8/8 \\
Verma $Y$ & 0.65 & 0.76 & 0.31--0.97 & 8/8 \\
Verma $Y,Y^2$ & 0.57 & 0.65 & 0.23--0.87 & 8/8 \\
Verma $Y,Y^2,Y^3$ & 0.52 & 0.61 & 0.21--0.77 & 8/8 \\
\bottomrule
\end{tabular}
\begin{minipage}{0.94\linewidth}
\footnotesize Values below 1 indicate shorter confidence intervals than the equally weighted indexed estimator. Ratios are computed over the eight subgroup-by-contrast estimates in Table~\ref{tab:realdata_napkin_estimates}.
\end{minipage}
\end{table}

\end{document}